\documentclass[11pt,a4paper]{article}

\usepackage[utf8]{inputenc}
\usepackage{styling}

\hypersetup{
  pdftitle={Balanced Allocations with the Choice of Noise},
  pdfauthor={Dimitrios Los, Thomas Sauerwald}
  }

\allowdisplaybreaks

\begin{document}

\title{Balanced Allocations with the Choice of Noise}

\author[]{Dimitrios Los\thanks{\texttt{dimitrios.los@cl.cam.ac.uk}} }
\author[]{Thomas Sauerwald\thanks{\texttt{thomas.sauerwald@cl.cam.ac.uk}}}
\affil[]{Department of Computer Science \& Technology, University of Cambridge}

\maketitle

\begin{abstract}
We consider the allocation of $m$ balls (jobs) into $n$ bins (servers). In the standard \TwoChoice process, at each step $t=1,2,\ldots,m$ we first sample {\em two} randomly chosen bins, compare their two loads and then place a ball in the least loaded bin. It is well-known that for any $m \geq n$, this results in a gap (difference between the maximum and average load) of $\log_2 \log n + \Theta(1)$ (with high probability).

In this work, we consider \TwoChoice in different settings with noisy load comparisons. One key setting involves an adaptive adversary whose power is limited by some threshold $g \in \mathbb{N}$. In each step, such adversary can determine the result of any load comparison between two bins whose loads differ by at most $g$, while if the load difference is greater than $g$, the comparison is correct. 

For this adversarial setting, we first prove that for any $m \geq n$ the gap is $\Oh(g+\log n)$ with high probability. Then through a refined analysis we prove that if $g \leq \log n$, then for any $m \geq n$ the gap is $\Oh\big(\frac{g}{\log g} \cdot \log \log n\big)$. For constant values of $g$, this generalizes the heavily loaded analysis of \cite{BCSV06,TW14} for the \TwoChoice process, and establishes that asymptotically the same gap bound holds even if load comparisons among ``similarly loaded'' bins are wrong. Finally, we complement these upper bounds with tight lower bounds, which establish an interesting phase transition on how the parameter $g$ impacts the gap.

The analysis also applies to settings with outdated and delayed information. For example, for the setting of \cite{BCEFN12}~where balls are allocated in consecutive batches of size $b = n$, we present an improved and tight gap bound of $\Theta\big(\frac{\log n}{\log \log n}\big)$. This bound also extends for a range of values of $b$ and applies to a relaxed setting where the reported load of a bin can be any load value from the last $b$ steps.

\end{abstract}

\clearpage

\clearpage
\tableofcontents
~
\clearpage

\section{Introduction}

\subsubsection*{Motivation} In this work we examine balls-and-bins processes where the goal is to allocate $m$ balls (jobs or tasks) sequentially into $n$ bins (processors or servers). The balls-and-bins framework a.k.a.~balanced allocations~\cite{ABKU99}~is a popular abstraction for various resource allocation and storage problems such as load balancing, scheduling or hashing (see surveys~\cite{MRS01,W17}). In order to allocate the balls in an efficient and decentralized way, randomized strategies are usually employed which are based on
sampling a number of bins for each ball, and then allocating the ball into one of those bins.

It is well-known that if each ball is placed in a bin chosen independently and uniformly at random (called \OneChoice), then the maximum load is $\Theta( \log n / \log \log n)$ \Whp\footnote{In general, with high probability refers to probability of at least $1 - n^{-c}$ for some constant $c > 0$.} for $m=n$, and $m/n + \Theta( \sqrt{ (m/n) \log n})$ \Whp~for $m \geq n \log n$. 
Azar, Broder, Karlin and Upfal~\cite{ABKU99} (and implicitly Karp, Luby and Meyer auf der Heide~\cite{KLM96}) proved the remarkable result that if each ball is placed in the lesser loaded of $d \geq 2$ randomly chosen bins, then the maximum load drops to $\log_d \log n + \Oh(1)$ \Whp, if $m=n$. This dramatic improvement from $d=1$ (\OneChoice) to $d=2$ (\TwoChoice) is known as ``power-of-two-choices'', and similar effects have been observed in other problems including routing, hashing and randomized rounding~\cite{MRS01}. 

Later, Berenbrink, Czumaj, Steger and V\"ocking~\cite{BCSV06} extended the analysis of \DChoice to the so-called ``heavily loaded case'', where $m \geq n$ can be arbitrarily large. In particular, for \TwoChoice an upper bound on the gap (the difference between the maximum and average load) of $ \log_2 \log n + \Oh(1)$ \Whp~was shown. A simpler but slightly weaker analysis was later given by Talwar and Wieder~\cite{TW14}.

A crucial ability of \TwoChoice (or \DChoice) is to quickly recover from a bad load vector, thanks to the information and flexibility provided by the additional bin sample(s).
Hence these processes typically reach an equilibrium with load vectors involving only a small number of different load values (around the mean). This can be seen as some self-stabilizing property, and was exploited in other balls-into-bins settings where balls can be removed, e.g.,~\cite{CFMMRSU98,BFKMNW18}. Also there are many distributed algorithms for voting and consensus that make use of a similar ``power-of-two-choices'' effect, see e.g.,~\cite{DGM11,CER14,BCN14,BCNPT16}. Other applications in distributed computing include population protocols \cite{AGR22} and distributed data structures \cite{N21}, whose analyses rely on a noisy version of~\TwoChoice.

It should be noted that the theoretical results for \TwoChoice assume that balls arrive sequentially one by one, so that each bin has perfect knowledge about its load at any point in time. This assumption may not be always met in practice, e.g., in a concurrent setting, bins may not be able to update their load immediately (even if this were possible, in some applications bins might still prefer not to reveal their true load).
These issues had been observed as early as 2001 in a survey on power-of-two-choices by Mitzenmacher, Richa and Sitaraman \cite{MRS01}, who remarked:
\begin{quote}\textit{``...the general problem of how to cope with incomplete or inaccurate load information and still achieve good load balancing performance appears to be an area with a great deal of research potential.''}
\end{quote}

In this work we are trying to improve our understanding of this problem and consider the following questions:
\begin{enumerate}
    \item What if the load information of a bin at time $t$ is outdated, e.g., it might be as small as the load at an earlier time $t-\tau$ for some parameter $\tau$?
    \item What if the load information of a bin at time $t$ is subject to some adversarial noise, e.g., the reported load of a bin might be an adversarial perturbation from the exact load within some range $g$?
    \item What if instead of having adversarial noise, there is a Gaussian perturbation with standard deviation $\sigma$ on the reported loads?
\end{enumerate}

A closely related setting to the first question called \emph{batching} was studied by Berenbrink, Czumaj, Englert, Friedetzky and Nagel~\cite{BCEFN12} (and later in \cite{BFKMNW18}). Here the allocation of the balls proceeds in consecutive batches of size $b$, and the load is only updated at the end of each batch. The authors proved that for $b = n$ the gap is $\Oh(\log n)$ w.h.p. Recently, the authors of this work proved some further results in the batched setting for a broader class of processes~\cite{LS22Batched}. However, applying these to \TwoChoice only gives a tight gap bound for $b \geq n \log n$, which is $\Theta(b/n$).
Mitzenmacher~\cite{M00} studied the two-choice-paradigm in queuing systems, where the batched setting is referred to as \emph{periodic update model}. His results include some fixed-point analysis and simulations, but no quantitative bounds on the gap are derived. Along similar lines, Dahlin~\cite{D00} investigated several \DChoice-based processes, and demonstrated through experiments that allocation strategies which interpret load information based on its age, can outperform simpler strategies.

In this work we will focus on noise settings related to the three questions above, but below we will also discuss other ``robustness'' aspects of \TwoChoice that were studied in previous works.

\subsubsection*{Further Related Work}
Peres, Talwar and Wieder~\cite{PTW15} introduced the $(1+\beta)$-process, in which two choices are available with probability $\beta \in (0, 1]$, and otherwise only one. This process mixes \OneChoice with \TwoChoice steps and is useful in applications where it is costly to always take two bin samples. Their upper bound of $\Oh( \log n/\beta + \log(1/\beta)/\beta)$ on the gap for any $m \geq n$ shows that we can achieve a gap that is independent of $m$ for $\beta := \beta(n)$.
A natural generalization of \TwoChoice studied in the same work \cite{PTW15} (and earlier in \cite{KP06} for $m=n$) is the so-called \emph{graphical allocation setting}. In this setting, bins correspond to vertices of a graph, and for each ball we sample an edge uniformly at random and place the ball in the lesser loaded bin of the two endpoints. The results of \cite{PTW15} show that for any connected and regular graph, the gap remains independent of $n$. In \cite{BF21}, Bansal and Feldheim analyzed a sophisticated algorithm that achieves a poly-logarithmic gap on sparse regular graphs in the graphical allocation setting. On a related problem where averaging is allowed, a tighter bound for the cycle was shown by Alistarh, Nadiradze and Sabour~\cite{ANS20}. Also recently, Greenhill, Mans and Pourmiri studied the graphical allocation setting on dynamic hypergraphs~\cite{GMP20}.

In the $(1+\beta)$-process the decision whether to take a second bin sample does not depend on the first bin sample. In contrast to that, the \TwoThinning process allocates in a two-stage procedure: Firstly, sample a random bin $i$. Secondly, based on the load of bin $i$ (and additional information based on the history of the process), we can either place the ball into $i$,  or place the ball into another randomly chosen bin $j$ (without comparing its load with $i$). This process has received a lot of attention lately, and several variations were studied in \cite{FL20} for $m=n$ and \cite{FGL21,LS21,LSS21} for $m \geq n$.  Finally, Czumaj and Stemann \cite{CS01} investigated so-called \emph{adaptive allocation schemes} for $m=n$. In contrast to \Thinning, after having taken a certain number of bin samples, the ball is allocated into the least loaded bin among all samples. In another related model recently studied by the authors of this work, the load of a sampled bin can only be approximated through binary queries of the form ``Is your load at least $g$?''~\cite{LS21}. It was shown that by using $1 \leq k = \Oh(\log \log n)$ queries for each of the two samples, for any $m \geq n$, the gap is $\Oh( k \cdot (\log n)^{1/k})$ w.h.p.

A setting relaxing the \emph{uniform sampling} assumption was investigated by Wieder~\cite{W07}, which shows that for any $d > 1$, the \DChoice gap bounds continue to hold as long as the probability by which the $d$ bins are sampled is close enough to uniform. A setting with heterogeneous bin capacities was studied in Berenbrink, Brinkmann, Friedetzky and Nagel~\cite{BBFN14}, who showed the gap bound of $\log_{d} \log n+\Oh(1)$ for \DChoice continues to hold.

There is also a rich line of work investigating randomized allocation schemes which use fewer random bits. For example, Alon, Gurel-Gurevich and Lubetzky~\cite{AGL10} established a trade-off between the number of bits used for the representation of the load and the number of $d$ bin choices. For $d=2$ choices, Benjamini and Makarychev~\cite{IM12} presented some tight results relating the gap of \TwoChoice to the available memory.

We remark that parallel versions of \TwoChoice exist that involve only a very small number of rounds of interactions between all balls and bins. Most of the studies, e.g., \cite{ACMR98,LW11} focus on $m=n$, and only recently the heavily loaded case was addressed in \cite{LPY19}. In comparison to our settings, the gap bounds are stronger, however, these algorithms require more coordination and do not handle any noisy or outdated load information.

From a higher perspective, investigating the complexity of problems in the presence of noisy data is a popular area in algorithms, machine learning and distributed computing. For example, of similar flavor as our noisy allocation setting are studies on sorting and ranking with unreliable information \cite{ACN05,BM08}. In the paradigm of algorithms with predictions, different scheduling algorithms under noisy job size estimates have been analyzed~\cite{MV20,M20,SGM22}. %

\subsubsection*{Our Contribution}
Our first contribution is to present a general framework of four different classes of noise settings (a formal description of these settings can be found in \cref{sec:settings}). This framework contains some of the previously studied processes and settings (like batching \cite{BCEFN12,BCN19,LS22Batched} and \GBounded \cite{N21}), but it also leads to new settings, which were not studied before.

Perhaps the most important setting involves \TwoChoice in the presence of an \textbf{adaptive adversary} with parameter $g \in \mathbb{N}$. In each step, if the two bin samples $i_1,i_2$ have a load difference of at most $g$, the adversary can manipulate the outcome of the load comparison arbitrarily and thereby decide whether the ball will be placed in $i_1$ or $i_2$. However, if the load difference is more than $g$, the comparison will be correct and \TwoChoice will place the ball in the less loaded bin. A slightly weaker setting, which we call \textbf{myopic setting}, works similarly, but now, in case the load difference is at most $g$, the ball is placed into a \emph{random} bin (among $\{i_1,i_2\}$).
Through a combination of different lower and upper bounds, we establish the following phase transition in how the parameter $g$ impacts the gap:
\begin{enumerate}
    \item If $\log n \leq g$, then for all $m \geq n$, $\Gap(m)=\Oh( g )$ in the adversarial setting. Further, there is a matching lower bound in the myopic setting.
    \item If $1 < g \leq \log n$, then for all $m \geq n$, $\Gap(m) = \Oh( \frac{g}{\log g} \cdot \log \log n)$ in the adversarial setting\footnote{For $g=1$, we prove a bound of $\Theta(\log \log n)$ and for $g = 0$, the process is equivalent to the \TwoChoice process without noise.}. Further, there is a matching lower bound in the myopic setting.
\end{enumerate}
Equivalently, we could say that $\Gap(m)=\Theta( \frac{g}{\log g} \cdot \log \log n + g)$ for any $g > 1$. Both upper bounds improve and generalize a bound of $\Oh(g \cdot \log(ng))$ which was shown by Nadiradze~\cite{N21} for the so-called \GBounded process, where an adversary ``greedily'' reverts all comparisons if the load difference is at most $g$. We believe that our new bounds could be helpful in obtaining a tighter analysis of the \emph{multi-counter data structure} studied in \cite{ABKLN18,N21}. %
Note that for constant $g$, our upper bound shows that \TwoChoice maintains a $\Oh(\log \log n)$ gap even if bins with constant load difference cannot be compared correctly. While it may seem intuitive that such a result should hold, we are not aware of any simple argument. Even in the most basic case $g=1$, there seems to be no ad-hoc method, e.g., using  couplings or majorization, which would extend the upper bound of $\Oh(\log \log n)$ for \TwoChoice from \cite{BCSV06,TW14} to this noise setting -- not to mention a more general albeit loose upper bound of $\Oh(g \cdot \log \log n)$ for any $g > 1$.

We then proceed to delay settings with outdated load information. In the \textbf{delay setting}, each load information of a bin can be outdated by at most $\tau \geq 1$ rounds. This is a generalization of the \textbf{batched setting}~\cite{BCEFN12}, where the load information of all bins is updated after rounds $0,b,2 b,\ldots$. For a range of values for $\tau$ and $b$ around $n$, we prove tight bounds on the gap in the delay and batched setting. This upper bound is matched by a trivial lower bound for the batched setting, which is based on the gap created by \OneChoice when allocating the first $b$ balls randomly into the $n$ bins. Our upper bound demonstrates that the special property of batching to reset all load values to their correct value at the beginning of a batch is not crucial, and it suffices if bin loads are updated asynchronously.
Finally, for the batched setting with $b=n$, our results improve the gap bound of $\Oh(\log n)$ from \cite{BCEFN12} to a tight bound of $\Theta\big(\frac{\log n}{\log \log n}\big)$. 

Finally, complementing the adversarial and myopic settings, we also investigate a \textbf{probabilistic noise} setting. In this setting, whenever two bins with load difference $\delta$ are compared, the comparison will be correct with probability $\rho(\delta) \in [0,1]$. One natural instance of this setting is when the reported bin loads are randomly perturbed by some Gaussian noise with variance $\sigma^2$, which essentially leads to $\rho(\delta) = 1 - \frac{1}{2} \exp(-(\delta/\sigma)^2 )$. We do not present a tight analysis of this setting, but our upper and lower bounds demonstrate that the gap is polynomial in $\sigma$ and poly-logarithmic in $n$.

In \cref{sec:settings}, we formally define the various processes and settings, and in \cref{tab:overview_bounds} we give a summary of upper and lower bounds proved in this paper.

\subsubsection*{Our Techniques}

Since most of our upper bounds use similar techniques as the upper bound for the adversarial setting, we only outline the proof of the latter. As with most of the previous works, e.g., \cite{PTW15,BF21,LS21}, we make extensive use of exponential potential functions in order to prove that the process stabilizes. This technique essentially suffices to prove a weaker gap bound of $\Oh(g \cdot \log (ng))$ for the adversarial setting. However, in order to prove stronger gap bounds, we require that a technical precondition on the load distribution is satisfied in most of the rounds. Roughly, this precondition says that a linear potential function is small (even though the exponential potential function may be very large). Extending~\cite{LSS21}, we establish this precondition by studying the interplay between three potential functions: the linear, quadratic and exponential potential functions.

For sub-logarithmic values of $g$, we prove sub-logarithmic gap bounds by employing a type of layered induction argument over a series of super-exponential potential functions. These potential functions are extensions of those used in \cite{LS21} to circumvent rounding issues and prove more fine-grained gap bounds. %

\subsubsection*{Organization}
In \cref{sec:settings} we present in more detail the different noise settings considered in this work. Then, in \cref{sec:notation} we introduce some basic notation and definitions used in the later analysis. In~\cref{sec:g_adv_warm_up}, we present our first bound of $\Oh(g \cdot \log (n g))$ on the gap for the adversarial setting. In \cref{sec:g_adv_g_plus_logn_bound}, this gap bound is then refined to $\Oh(g + \log n)$ (\cref{thm:g_adv_g_plus_logn_gap}). In \cref{sec:g_adv_upper_bound_for_small_g_outline}, we outline the use of a layered-induction inspired technique to prove a gap bound of $\Oh(\frac{g}{\log g} \cdot \log \log n)$ for $g = o(\log n)$ (\cref{thm:g_adv_strongest_bound}) and in the following three sections we give the proof of this bound. More specifically, in \cref{sec:g_adv_base_case}, we strengthen the results of \cref{sec:g_adv_g_plus_logn_bound} to obtain the base case, in \cref{sec:super_exponential_potentials} we provide a general analysis of super-exponential potentials and in \cref{sec:g_adv_layered_induction}, we use this to complete the layered induction step. Next, in \cref{sec:g_adv_delay_noise_settings} we derive upper bounds for both the probabilistic noise and the delay settings, essentially through reductions to the adversarial setting. \cref{sec:lower_bounds} contains our lower bounds for the different settings, including a tight lower bound of $\Omega(g + \frac{g}{\log g} \cdot \log \log n)$ for \GMyopicComp (\cref{cor:g_myopic_combined_lower_bound}). \cref{sec:experiments} presents some experimental results and in \cref{sec:conclusion} we conclude with a brief summary of the main results and some open questions.

\section{Our Settings}\label{sec:settings}

All of our settings are based on running the \TwoChoice process with $m$ balls into $n$ bins, labeled $[n]:=\{1,2,\ldots,n\}$. 

In the normal setting, referred to as ``without noise'', the load comparisons between the two bin samples are correct in all steps. Here, we investigate several noise settings where the outcome of the load comparisons may not always be correct. We distinguish between three classes of noise settings: (1) Adversarial Load and Comparison, (2) Adversarial Delay and (3) Probabilistic Load and Comparison. 
A schematic figure with the connections between the different settings is shown in \cref{fig:g_adv_settings_reductions}.

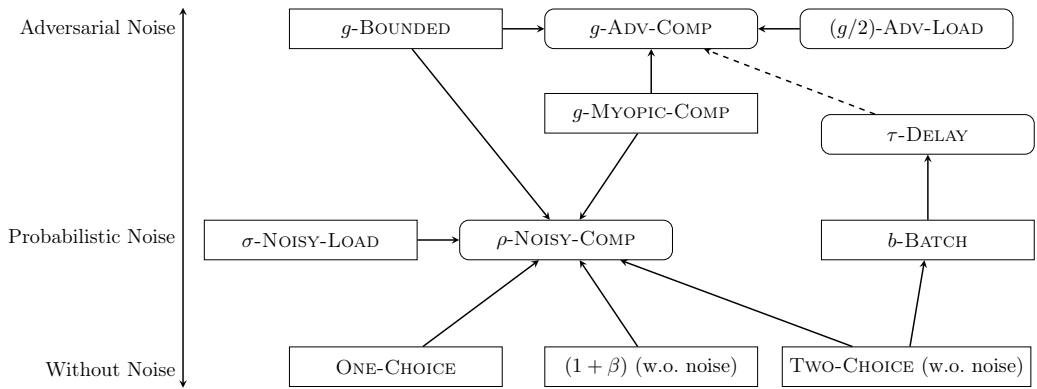
\begin{figure}[h]
    \centering
\scalebox{0.7}{
\begin{tikzpicture}[scale=0.8,process/.style={rectangle,draw=black,minimum width=4cm,minimum height=0.75cm},
setting/.style={rectangle,draw=black,minimum width=4cm,minimum height=0.75cm,rounded corners=5pt},ed/.style={black,thick,-stealth}]

\node[setting] (AC) at (6,3) {$g$-\textsc{Adv-Comp}};

\node[process] (B) at (0,3) {\GBounded};

\node[setting] (AL) at (12,3) {$(g/2)$-\textsc{Adv-Load}};

\node[process] (MC) at (6,1) {\GMyopicComp};

\node[setting] (TD) at (12.5,0.5) {\TauDelay};

\node[process] (BB) at (12.5,-2) {\BBatch};

\draw[ed] (AL) to (AC);
\draw[ed] (MC) to (AC);
\draw[ed] (B) to (AC);
\draw[ed, dashed] (TD) to (AC);
\draw[ed] (BB) to (TD);

\node[process] (OC) at (0,-5) {\OneChoice};
\node[process] (BE) at (6,-5) {$(1+\beta)$ (w.o.~noise)};
\node[process] (TC) at (12,-5) {\TwoChoice (w.o.~noise)};

\node[setting] (NC) at (4,-2) {$\rho$-\textsc{Noisy-Comp}};

\node[process] (NL) at (-2,-2) {\SigmaNoisyLoad};

\draw[ed] (OC) to (NC);
\draw[ed] (BE) to (NC);
\draw[ed] (TC) to (NC);
\draw[ed] (NL) to (NC);
\draw[ed] (B) to (NC);
\draw[ed] (MC) to (NC);
\draw[ed] (TC) to (BB);

\draw[thick,stealth-stealth] (-5,3.5) to node[pos=0.05,left] {Adversarial Noise} node[pos=0.6,left] {Probabilistic Noise} node[pos=0.95,left] {Without Noise} (-5,-5.5);

\end{tikzpicture}
}
\caption{Overview of settings (rounded rectangles) and processes (rectangles). A directed arrow from setting (process) $A$ to setting (process) $B$ means that $B$ is stronger than $A$ (that is, $B$ can simulate $A$). For \TauDelay, a dashed arrow is used for the connection to \GAdvComp, as the (strong) connection is slightly more involved.}
\label{fig:g_adv_settings_reductions}
\end{figure}

\begin{framed} \label{def:two_choice_without_noise}
	\vspace{-.45em} 
    \noindent \underline{\TwoChoice Process (without noise):}\\
    \textsf{Iteration:} For each step $t \geq 1$, 
    \begin{enumerate}[topsep=0.5pt, itemsep=0pt]
      \item Sample two bins $i_1 = i_1^t$ and $i_2 = i_2^t$  with replacement, independently and uniformly at random. 
      \item Let $i^t \in \{ i_1, i_2 \}$ be a bin such that
      $
        x_{i^t}^{t-1} = \min\big\{ x_{i_1}^{t-1}, x_{i_2}^{t-1} \big\}, 
      $ breaking ties arbitrarily.
      \item Allocate one ball to bin $i^t$.
    \end{enumerate}
\end{framed}

We now present our framework where \TwoChoice runs in a noisy setting, meaning that step~$2$ is subject to some noise. 
One possibility is to have an adversary that is able to replace the load values $x_{i_1}^{t-1}$ and $x_{i_2}^{t-1}$ by estimates $\tilde{x}_{i_1}^{t-1}$ and $\tilde{x}_{i_2}^{t-1}$, while the bin $i^t$ will always be the one with the better load estimate. A more general framework is to have a two-sample process with a constrained decision function that decides in which of the two samples $i_1$ and $i_2$ to allocate to. This function may depend on the filtration $\mathfrak{F}^{t-1}$ corresponding to the first $t-1$ allocations of the process.\footnote{In particular, filtration $\mathfrak{F}^{t}$ reveals the load vector $x^{t}$.}

\begin{framed}
	\vspace{-.45em} \noindent
	\underline{\TwoChoice Process with Noise:}\\
  \textsf{Noise:} The adversary $A^t : \mathfrak{F}^{t-1} \times [n] \times [n] \to [n]$ \\
    \textsf{Iteration:} For each step $t \geq 1$,
    \begin{enumerate}[topsep=0.5pt, itemsep=0pt]
      \item Sample two bins $i_1$ and $i_2$ with replacement, independently and uniformly at random. 
      \item Let $i^t = A^t(\mathfrak{F}^{t-1},i_1,i_2) \in \{i_1,i_2\}$ be the bin determined by the adversary.
      \item Allocate one ball to bin $i^t$.
    \end{enumerate}
\end{framed}

This is a very general setting, and in order to prove any meaningful bounds, we will have to restrict the power of the adversary $A^t$, giving rise to different settings defined below. Also note that the adversary $A^t$ is allowed to use coin-flips, meaning that the above framework encompasses settings in which, for instance, load estimates or comparisons are determined probabilistically. 

Clearly we can recover \TwoChoice without noise as a special case, if the adversary uses \[
A^t(\mathfrak{F}^{t-1}, i_1, i_2) = \mathrm{arg\,min}_{k \in \{i_1, i_2 \}} \ x_k^{t-1},
\]
breaking ties arbitrarily. On the other hand, for the \GBounded process, \[
A^t(\mathfrak{F}^{t-1}, i_1, i_2) = \begin{cases}
  \mathrm{arg\,max}_{k \in \{i_1, i_2\}} \ x_k^{t-1} & \text{if }|x_{i_1}^{t-1} - x_{i_2}^{t-1}| \leq g, \\
  \mathrm{arg\,min}_{k \in \{i_1, i_2\}} \ x_k^{t-1} & \text{otherwise},
\end{cases}
\]
\vspace{-0.2cm}breaking ties arbitrarily.

\subsection*{Adversarial Load and Comparison}

\textbf{Settings.} We first present a setting with an \emph{adaptive adversary}, who has direct control on the outcome of load comparisons (provided that the load difference is small).
Specifically, in  \GAdvLoad for $g \in \mathbb{N}$, at each step $t=1,2,\ldots$ the adversary first determines for each bin $k$ a load estimate $\tilde{x}_k^{t-1} \in [x_k^{t-1}-g,x_k^{t-1}+g]$. Then our process samples two bins $i_1$ and $i_2$, and allocates into the bin with the smaller load estimate (ties can be broken arbitrarily). 
A slightly stronger adversary appears in the setting \GAdvComp. Here, at each step $t=1,2,\ldots$ the process samples two bins $i_1$ and $i_2$, and then the adversary is able to determine the outcome of the comparison if $|x_{i_1}^{t-1} - x_{i_2}^{t-1}| \leq g$, thereby deciding where \TwoChoice places the ball. Otherwise, if $|x_{i_1}^{t-1} - x_{i_2}^{t-1}| > g$, \TwoChoice will place the ball in the less loaded bin. For $g = 0$, we recover the \TwoChoice process without noise. %

We remark that \GAdvLoad can be simulated by \textsc{$(2g)$-Adv-Comp}, which is why we will only study the slightly more powerful \textsc{Adv-Comp} setting here. 

\textbf{Processes.} A special instance of  \GAdvComp is one that forces \TwoChoice to allocate the ball to the heavier of the two bins when $|x_{i_1}^{t-1} - x_{i_2}^{t-1}| \leq g$. This process was studied in \cite{N21} under the name \GBounded and served as an analysis tool for the multi-counter distributed data structure.

Another special instance of the \emph{adversarial} setting is a \textit{myopic} process, where the outcomes of ``tight'' load comparisons is decided uniformly at random. More precisely, in \GMyopicComp, at each step $t=1,2,\ldots$, if $|x_{i_1}^{t-1}-x_{i_2}^{t-1}| \leq g$, the ball is allocated to a random bin among $\{i_1,i_2\}$, otherwise, the ball is allocated to the less loaded among $\{i_1,i_2\}$. By using random coin flips, the adversary of \GAdvComp can trivially simulate \GMyopicComp.

\subsection*{Adversarial Delay}

\textbf{Setting.} Next we turn to adversarial settings where the load information a bin reports may be \emph{outdated}. The first setting, called $\tau$-\textsc{Delay} for $\tau \in \mathbb{N}_{\geq 1}$, is similar to \GAdvComp, but here the range of load values the adversary can choose from is based on a sliding time interval. That is, after bins $i_1$ and $i_2$ are sampled in step $t$, an adaptive adversary provides us with load estimates that must satisfy $\tilde{x}_{i_1}^{t-1} \in [x_{i_1}^{t-\tau},x_{i_1}^{t-1}]$ and
$\tilde{x}_{i_2}^{t-1} \in [x_{i_2}^{t-\tau},x_{i_2}^{t-1}]$, and the ball is then allocated to the bin with a smaller load estimate, breaking ties arbitrarily. As an example, this includes an adversary that has the power of delaying load updates to each bin arbitrarily, but any load update (i.e., allocation) that is $\tau$ or more steps in the past must be processed.

\textbf{Processes.} A special case of \textsc{$\tau$-Delay} is \BBatch for $b \in \mathbb{N}_{\geq 1}$, introduced in \cite{BCEFN12}, where balls are allocated in consecutive batches of size $b$ each. When the load of a bin $i$ is queried, the bin $i$ will report the load the bin had at the beginning of the batch (i.e., all allocations within the current batch will not be considered). More formally, when two bins $i_1,i_2$ are sampled in step $t$, the decision where to allocate is based on comparing $x_{i_1}^{\lfloor (t-1)/b \rfloor \cdot b}$ and $x_{i_2}^{\lfloor (t-1)/b \rfloor \cdot b}$, breaking ties randomly. We recall that in \cite{BCEFN12}, only the special case $b=n$ was studied, and a gap bound of $\Oh(\log n)$ was shown. Clearly, for $\tau=b$, \textsc{$\tau$-Delay} can simulate \textsc{$b$-Batch} and for $\tau = b = 1$, both are equivalent to the \TwoChoice process without noise. Note that \BBatch assumes that load updates are perfectly synchronized, while the \TauDelay setting relaxes this assumption, therefore covering a wider class of asynchronous update schemes. The reduction of the \TauDelay (and \BBatch) settings to a relaxed \GAdvComp setting will be formalized in \cref{sec:upper_bounds_for_delay_settings}.

\subsection*{Probabilistic Noise}

\textbf{Setting.} We now turn to non-adversarial settings where comparisons are subject to \emph{probabilistic noise}, with a larger chance of getting the correct comparison the larger the load difference. Specifically, we consider the \textsc{Noisy-Comp} setting, where we have an arbitrary non-decreasing function $\rho: \mathbb{N} \rightarrow [0,1]$ such that for any $\delta \in \mathbb{N}$, $\rho(\delta)$ is the probability that a comparison between two bins with absolute load difference $\delta$ will be correct. 
As before we assume independence, i.e., for all $t \geq 1$, the event of a correct comparison in step $t$ only depends on the load difference of the two sampled bins. Thus, for a specific function $\rho$, $\rho$\textsc{-Noisy-Comp} defines a process.

\textbf{Processes.} This is a very expressive setting, encompassing, for example the
\GBounded and \GMyopicComp processes by choosing $\rho$ as step functions with values in $\{0,1/2,1\}$ (see \cref{fig:rho_noisy_comp_settings}). Similarly, \TwoChoice (without noise), \OneChoice and $(1+\beta)$ correspond to $\rho$ being the constant $1$, $1/2$ and $1/2+\beta/2$, respectively.

As a concrete example consider the \SigmaNoisyLoad process. Here, when a bin $i \in \{ i_1, i_2 \}$ is sampled at step $t$, it reports an unbiased load estimate $\tilde{x}_i^{t-1} = x_i^{t-1} + Z_i^{t-1}$, where $Z_i^{t-1}$ has a normal distribution $\mathcal{N}(0,\sigma^2)$ (and all $\{Z_i^{t}\}_{i \in [n], t \geq 0}$ are mutually independent). Then \TwoChoice allocates a ball to the bin that reports the smallest load estimate.  Thus if $x_{i_2}^{t-1} - x_{i_1}^{t-1} = \delta > 0$, the probability to allocate to the lighter bin can be computed as follows:
\begin{align*}
 \Pro{ \tilde{x}_{i_1}^{t-1} \leq \tilde{x}_{i_2}^{t-1} } & = \Pro{ Z_{i_1}^{t-1} - Z_{i_2}^{t-1} \leq \delta } \\
 &= 1 - \Pro{ \mathcal{N}(0,2 \sigma^2) > \delta} \\
 &= 1 - \Pro{ \mathcal{N}(0,1) > \delta/(\sqrt{2} \sigma)} \\
 &= 1 - \Phi( \delta / ( \sqrt{2} \sigma ) ).
 \end{align*}
 Note that $\Phi(z) = 1/2$ for $z=0$ and $\Phi(z)$ is increasing in $z$.
As shown in \cite[p.\,17]{IM74}, for $z \geq 0$,
\[
 \frac{1}{\sqrt{2\pi}} \cdot \frac{2}{\sqrt{z^2+4}+z} \cdot e^{-z^2/2} \leq 1 - \Phi(z) \leq \frac{1}{\sqrt{2\pi}} \cdot \frac{2}{\sqrt{z^2+2}+z} \cdot e^{-z^2/2}.
 \]
By setting $z:=\delta/(\sqrt{2} \sigma)$, ignoring the linear term in $1/z$, and re-scaling $\sigma$, we can define \SigmaNoisyLoad as the process satisfying for all steps $t$ and samples $i_1, i_2$ with $x_{i_2}^{t-1} -x_{i_1}^{t-1}=\delta>0$,
\begin{align}
\Pro{ \tilde{x}_{i_1}^{t-1} \leq \tilde{x}_{i_2}^{t-1} } = \rho(\delta) := 1 - \frac{1}{2} \cdot \exp \left( - \left( \frac{\delta}{\sigma} \right)^2 \right), \label{eq:gaussian}
\end{align}
meaning that the correct comparison probability exhibits a Gaussian tail behavior.

\begin{figure}[t]
\begin{center}
\begin{tikzpicture}[scale=1.5]
\def\normaltwo{\x,1}
\def\z{1.6}
\def\y{4.4}
\draw (0.1,1) -- (-0.1,1) node[left] {$1$};
\draw (0.1,0.5) -- (-0.1,0.5) node[left] {$0.5$};
\draw (0.1,0) -- (-0.1,0) node[left] {$0$};
\def\fy{1-exp(\y)}
\draw[color=blue,thick,domain=0:1] plot ({\x,0}) node[right] {};
\draw[color=blue,thick,domain=1:2.2] plot ({\x,1}) node[right] {};
\draw[color=blue,thick, fill=white] (1,1) circle (0.025cm);
\draw[color=blue,thick, fill=blue] (1,0) circle (0.025cm);

\draw (1,0.1) -- (1,-0.1) node[pos=2.2] {$g$};

\draw[->] (0,0) -- (2.2,0) node[right] {$\delta$};
\draw[->] (0,0) -- (0,1.2) node[above] {$\rho(\delta)$};
\node () at (1,-0.75) {$(a)$ \GBounded};
\end{tikzpicture}
\begin{tikzpicture}[scale=1.5]
\def\normaltwo{\x,1}
\def\z{1.6}
\def\y{4.4}
\draw (0.1,1) -- (-0.1,1) node[left] {$1$};
\draw (0.1,0.5) -- (-0.1,0.5) node[left] {$0.5$};
\draw (0.1,0) -- (-0.1,0) node[left] {$0$};
\def\fy{1-exp(\y)}
\draw[color=blue,thick,domain=0:1] plot ({\x,0.5}) node[right] {};
\draw[color=blue,thick,domain=1:2.2] plot ({\x,1}) node[right] {};
\draw[color=blue,thick, fill=white] (1,1) circle (0.025cm);
\draw[color=blue,thick, fill=blue] (1,0.5) circle (0.025cm);

\draw (1,0.1) -- (1,-0.1) node[pos=2.2] {$g$};

\draw[->] (0,0) -- (2.2,0) node[right] {$\delta$};
\draw[->] (0,0) -- (0,1.2) node[above] {$\rho(\delta)$};
\node () at (1,-0.75) {$(b)$ \GMyopicComp};
\end{tikzpicture}
\begin{tikzpicture}[scale=1.5]
\def\normaltwo{\x,1}
\def\z{1.6}
\def\y{4.4}
\draw (0.1,1) -- (-0.1,1) node[left] {$1$};
\draw (0.1,0.5) -- (-0.1,0.5) node[left] {$0.5$};
\draw (0.1,0) -- (-0.1,0) node[left] {$0$};
\def\fy{1-exp(\y)}
\draw[color=blue,thick,domain=0:2.2] plot ({\x},{1-0.5*exp(-\x)}) node[right] {};

\draw (1,0.1) -- (1,-0.1) node[pos=2.2] {$\sigma$};
\draw (2,0.1) -- (2,-0.1) node[pos=2.2] {$2 \sigma$};

\draw[->] (0,0) -- (2.2,0) node[right] {$\delta$};
\draw[->] (0,0) -- (0,1.2) node[above] {$\rho(\delta)$};
\node () at (1,-0.75) {$(c)$ \SigmaNoisyLoad};
\end{tikzpicture}
\end{center}
\caption{In the graphs above, $\delta=|x_{i_1}^{t-1}-x_{i_2}^{t-1}|$ is the load difference among the two sampled bins and $\rho(\delta)$ is the probability that the load comparison is correct for the \GBounded, \GMyopicComp and \SigmaNoisyLoad processes, respectively.
}
\label{fig:rho_noisy_comp_settings}
\end{figure}
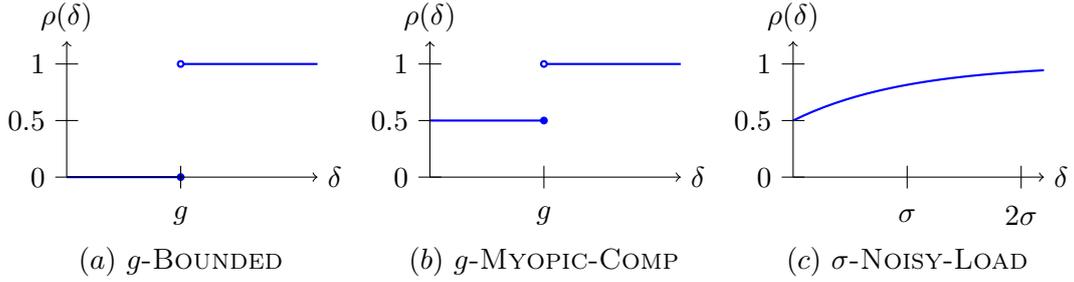

\renewcommand{\arraystretch}{1.8}
 \begin{table}[H]
\resizebox{\textwidth}{!}{
 \begin{tabular}{cccccc}
 \textbf{Setting/Process} & \textbf{Range} & \textbf{Lower Bound} & \textbf{Reference} & \textbf{Upper Bound} & \textbf{Reference} \\
 \hline 
 \rowcolor{Gray} & & &  &  & \\[-0.2cm]
 \rowcolor{Gray} \multirow{-2}{*}{\GBounded} & \multirow{-2}{*}{$1 \leq g $} & \multirow{-2}{*}{--} & \multirow{-2}{*}{--} & \multirow{-2}{*}{$\Oh(g \cdot \log (ng))$} & \multirow{-2}{*}{\makecell{\cite[Thm~2.5.12]{N21}\\ Thm~\ref{thm:g_adv_warm_up_gap}}} \\ \hline
 \rowcolor{Greenish} \GAdvComp & $1 \leq g $ & -- & -- & $\Oh(g \cdot \log (ng))$ & Thm~\ref{thm:g_adv_warm_up_gap} \\ \hline
 \rowcolor{Greenish} \GAdvComp & $1 \leq g$ & -- &-- & $\Oh(g + \log n)$ & Thm~\ref{thm:g_adv_g_plus_logn_gap} \\ \hline
 \rowcolor{Greenish} \GAdvComp & $1 < g \leq \log n$ & -- & -- & $\Oh\big(\frac{g}{\log g} \cdot \log \log n \big)$ &  Thm~\ref{thm:g_adv_strongest_bound} \\ \hline
 \rowcolor{Greenish} \GMyopicComp & $\frac{\log n}{\log \log n} \leq g$ & $\Omega(g)$ & Pro~\ref{pro:g_myopic_g_lower_bound} & -- & -- \\ \hline
 \rowcolor{Greenish} & & &  &  & \\[-0.2cm]
 \rowcolor{Greenish}\multirow{-2}{*}{\GMyopicComp} & \multirow{-2}{*}{$1 < g \leq \frac{\log n}{\log \log n}$} & \multirow{-2}{*}{$\Omega\big( \frac{g}{\log g} \cdot \log \log n \big)$} &  \multirow{-2}{*}{\makecell{Obs~\ref{obs:g_adv_loglogn_lower_bound}\\Thm~\ref{thm:g_myopic_layered_induction_lower_bound}}} & \multirow{-2}{*}{--} & \multirow{-2}{*}{--} \\[-0.2cm] \hline
 \rowcolor{Gray} \BBatch & $b = \Omega(n \log n)$ & $\Omega\big(\frac{b}{n}\big)$ & \cite[Prop 7.4]{LS22Batched} & $\Oh\big(\frac{b}{n}\big)$ & \cite[Thm 5.1]{LS22Batched} \\ \hline
 \rowcolor{Gray} \BBatch & $b=n$ & $\Omega\big( \frac{\log n}{\log \log n}\big)$ & Obs~\ref{obs:two_choice_batching_lower_bound_small_b} & $\Oh( \log n)$ & \cite[Thm 1]{BCEFN12} \\ \hline
 \rowcolor{Greenish} \TauDelay & $\tau = n$ & -- & -- & $\Oh\big(\frac{\log n}{\log \log n}\big)$ & Thm~\ref{thm:batching} \\ \hline
 \rowcolor{Greenish} \TauDelay & $\tau \in [n \cdot e^{-\log^c n}, n \log n]$ & -- & -- & $\Oh\left(\frac{\log n}{\log((4n/\tau) \log n)} \right)$ & Cor~\ref{cor:delay_polylog_n_upper_bound} \\ \hline
 \rowcolor{Greenish} \TauDelay & $\tau = n^{1-\eps}$ & -- & -- & $\Oh(\log \log n)$ & Rem~\ref{rem:delay_very_very_lightly_loaded} \\ \hline
 \rowcolor{Gray} \BBatch & $b \in [n \cdot e^{-\log^c n}, n \log n]$ & $\Omega\left(\frac{\log n}{\log((4n/b) \log n)}\right)$ & Obs~\ref{obs:two_choice_batching_lower_bound_small_b} & -- & -- \\ \hline
 \rowcolor{Gray} \BBatch & $b = n^{1-\eps}$ & $\Omega( \log \log n )$ & Obs~\ref{obs:g_adv_loglogn_lower_bound} & -- & -- \\ \hline
 \rowcolor{Greenish} \SigmaNoisyLoad & $1 \leq \sigma$ & -- & -- & $\Oh\big( \sigma \sqrt{\log n} \cdot \log (n \sigma) \big)$ & Pro~\ref{pro:prob_upper_bound} \\ \hline
 \rowcolor{Greenish} \SigmaNoisyLoad & $2 \cdot (\log n)^{-1/3} \leq \sigma$ & $\Omega( \min\{ 1, \sigma \} \cdot (\log n)^{1/3})$ & Pro~\ref{pro:sigma_lower_bound} & -- & --\\ \hline
 \rowcolor{Greenish} \SigmaNoisyLoad & $32 \leq \sigma$ & $\Omega(\min\{\sigma^{4/5}, \sigma^{2/5} \cdot \sqrt{\log n}\})$ & Pro~\ref{pro:sigma_lower_bound} & -- & -- \\ \hline
 \end{tabular}
 }
 \caption{Overview of the lower and upper bounds for different noise settings derived in previous works (rows in \hlgray{\,Gray\,}) and in this work (rows in \hlgreenish{\,Green\,}). All upper bounds hold for all values of $m \geq n$ \Whp, while lower bounds may only hold for a suitable value of $m$ \Whp. Recall that the \GBounded and \GMyopicComp processes are instances of the \GAdvComp setting, and \BBatch an instance of \TauDelay, for $\tau = b$. The parameters $c, \eps$ are any constants in $(0,1)$.
 }
 \label{tab:overview_bounds}
 \end{table}

\section{Notation and Preliminaries} \label{sec:notation}

We sequentially allocate $m$ balls (jobs) to $n$ bins (servers) with labels in $[n]:=\{1,2,\ldots,n\}$.
The \textbf{load vector} after $t \geq 0$ steps, i.e., $t$ allocations, is $x^{t}=(x_1^{t},x_2^{t},\ldots,x_n^{t})$ and at the beginning, $x_i^{0} = 0$ for any $i \in [n]$. For any $t \geq 0$, we also define the normalized load by sorting the entries of $x^t$ non-increasingly and setting $y_i^t:=x_i^t - \frac{t}{n}$; so $y_1^t \geq y_2^t \geq \cdots \geq y_n^t$.
We will analyze the performance of allocation processes via the \textbf{gap},  defined for any $t \geq 0$ as
\[
\Gap(t):= \max_{1 \leq i \leq n} x_i^{t} - \frac{t}{n} = y_1^t,
\]
i.e., the difference between maximum and average load after step $t \geq 0$. It is well-known that even for \TwoChoice without noise, the difference between the maximum and minimum load is $\Omega(\log n)$ for $m = \Omega(n \log n)$.
Further, for any step $t \geq 0$, we define $B_+^t := \{ i \in [n] : y_i^t \geq 0 \}$ as the set of \textbf{overloaded} bins and $B_-^t := \{ i \in [n]: y_i^t < 0 \}$ as the set of \textbf{underloaded} bins.

Following~\cite{PTW15}, many allocation processes can be described by a \textbf{probability allocation vector} $r^{t} =(r_1^{t},r_2^{t},\ldots,r_n^{t})$ for step $t$, where $r_i^{t}$ is the probability for incrementing the load of the $i$-th most loaded bin.

Recall the definition of the \TwoChoice process from \cref{def:two_choice_without_noise}. The \TwoChoice process without noise has a time-independent probability allocation vector $p$, where the probability to allocate to the $i$-th most loaded bin is given by $
p_i = \frac{2i - 1}{n^2},
$
for $i \in [n]$. Recall that in \OneChoice, each ball is allocated into a bin sampled independently and uniformly at random. This corresponds to $r$ being the uniform distribution. We say that a probability vector $q$ majorizes another probability vector $r$ if for all $k \in [n]$, it satisfies $\sum_{j=1}^k q_j \geq \sum_{j=1}^k r_j$. The same definition applies unchanged to sorted load vectors.

For a sequence of random variables $(X^t)_{t \geq 0}$, we define $\Delta X^{t+1} := X^{t+1} - X^t$. We also use the shorthands $u^+ := \max\{u, 0\}$ and $u_1 \wedge u_2 := \min\{ u_1, u_2 \}$.

Many statements in this work hold only for sufficiently large $n$, and several constants are chosen generously with the intention of making it easier to verify some technical inequalities.

\section{Warm-Up: Upper Bound of \texorpdfstring{$\Oh(g \log (ng))$}{O(g log(ng))} for \texorpdfstring{\GAdvComp}{g-Adv-Comp}} \label{sec:g_adv_warm_up}

In this section we will prove the $\Oh(g \log(ng))$ gap bound for the \GAdvComp setting (with $g\geq 1$ being arbitrary), also recovering the $\Oh(g \log(ng))$ gap bound for the \GBounded process proven in \cite{N21}. The proof relies on a version of the exponential potential function~\cite{PTW15}, called \textit{hyperbolic cosine potential} $\Gamma := \Gamma(\gamma)$, which is defined as\begin{align} \label{eq:gamma_def}
\Gamma^t := \Gamma^t(\gamma) := \sum_{i = 1}^n \Gamma_i^t = \sum_{i = 1}^n \left[ e^{\gamma \cdot y_i^t} + e^{-\gamma \cdot y_i^t} \right],
\end{align}
for a smoothing parameter $\gamma \in (0,1)$, to be specified in \cref{thm:g_adv_warm_up_gap}. Recall that $\Delta\Gamma^{t+1} := \Gamma^{t+1} - \Gamma^t$. We make use of the following general lemma from~\cite{PTW15} regarding the expected change of $\Gamma$.
\begin{lem}[\textbf{\cite[Lemmas 2.1 \& 2.3]{PTW15}}] \label{lem:ptw_gamma_change}
Consider any allocation process with probability allocation vector $r^t$ and the potential $\Gamma := \Gamma(\gamma)$ with any $\gamma \in (0,1)$. Then, for any step $t \geq 0$,
\begin{align*}
\lefteqn{\Ex{\left. \Delta\Gamma^{t+1} \,\right|\, y^t}} \\ &\leq \sum_{i = 1}^n \left[ \Big( r_i^t \cdot (\gamma + \gamma^2) - \Big(\frac{\gamma}{n} - \frac{\gamma^2}{n^2}\Big)\Big) \cdot e^{\gamma y_i^t} + \Big( r_i^t \cdot (-\gamma + \gamma^2) + \Big(\frac{\gamma}{n} + \frac{\gamma^2}{n^2}\Big)\Big) \cdot e^{-\gamma y_i^t} \right].
\end{align*}
\end{lem}

For convenience, we rewrite \cref{lem:ptw_gamma_change} by decomposing the upper bound into the components that are independent of the probability allocation vector $r^t$ and those that are not, i.e.,
\begin{align}
\Ex{\left. \Delta\Gamma^{t+1} \,\right|\, y^t} 
  \leq h(y^t) + \sum_{i = 1}^n r_i^t \cdot f(y_i^t), \label{eq:lemma_four_one}
\end{align}
where $h(y^t) := \sum_{i = 1}^n - \big(\frac{\gamma}{n} - \frac{\gamma^2}{n^2}\big) \cdot e^{\gamma y_i^t} + \big(\frac{\gamma}{n} + \frac{\gamma^2}{n^2}\big) \cdot e^{-\gamma y_i^t}$ and $f(y_i^t) := (\gamma + \gamma^2) \cdot e^{\gamma y_i^t} + (-\gamma + \gamma^2) \cdot e^{-\gamma y_i^t}$.

We will also make use of the following drop inequality for \TwoChoice without noise (see also \cite[Theorem 3.1]{LS22Batched}). %
\begin{lem}[\textbf{implied by \cite[Theorem 2.9]{PTW15}}] \label{lem:ptw_gamma_drop_two_choice}
Consider the \TwoChoice process without noise with probability allocation vector $p$ and the potential $\Gamma := \Gamma(\gamma)$ with any $\gamma \in (0, \frac{1}{6 \cdot 12})$. Then, there exists a constant $c > 0$, such that for any step $t \geq 0$,%
\begin{align*}
\Ex{\left. \Delta\Gamma^{t+1} \,\right|\, y^t} 
 & \leq h(y^t) + \sum_{i = 1}^n p_i \cdot f(y_i^t)
  \leq - \frac{\gamma}{48n} \cdot \Gamma^t + c.
\end{align*}
\end{lem}

We will analyze $\Delta\Gamma^{t+1}$ for the \GAdvComp setting by relating it to the change $\Delta\Gamma^{t+1}$ for the \TwoChoice process without noise. To this end, it will be helpful to define all pairs of bins (of unequal load), whose comparison is under the control of the adversary:
\begin{align}
 R^t := \left\{ (i, j) \in [n] \times [n] \colon y_j^t < y_i^t \leq y_j^t + g \right\}. \label{def:rt_definition}
\end{align}
So for each pair $(i,j) \in R^t$, the adversary determines the outcome of the load comparison assuming $\{i,j\}$ are the two bin samples in step $t+1$, which happens with probability $2 \cdot \frac{1}{n} \cdot \frac{1}{n} = \frac{2}{n^2}$. This can be seen as moving a probability of up to $\frac{2}{n^2}$ from bin $j$ to bin $i$, if we relate the probability allocation vector $p$ (of \TwoChoice without noise) to the probability allocation vector $q^t=q^t(\mathfrak{F}^t)$ of \TwoChoice with noise.

\begin{figure}

\begin{tikzpicture}[scale=.9]

\begin{scope}[scale=0.9,xshift=8.5cm,yshift=-2cm]

\node (0) at (0,3) {$x_i^t$:};
\node (1) at (1,3) {$21$};
\node (2) at (2,3) {$19$};
\node (3) at (3,3) {$13$};
\node (4) at (4,3) {$12$};
\node (5) at (5,3) {$12$};
\node (6) at (6,3) {$11$};
\node (7) at (7,3) {$8$};
\node (8) at (8,3) {$6$};

\node (0a) at (0,2.3) {\scriptsize{$i$:}};
\node (1a) at (1,2.3) {\scriptsize{$1$}};
\node (2a) at (2,2.3) {\scriptsize{$2$}};
\node (3a) at (3,2.3) {\scriptsize{$3$}};
\node (4a) at (4,2.3) {\scriptsize{$4$}};
\node (5a) at (5,2.3) {\scriptsize{$5$}};
\node (6a) at (6,2.3) {\scriptsize{$6$}};
\node (7a) at (7,2.3) {\scriptsize{$7$}};
\node (8a) at (8,2.3) {\scriptsize{$8$}};

\draw[-stealth] (2.120) to [bend right=50] (1.60);
\draw[-stealth] (4) to [bend right=50] (3.70);
\draw[-stealth] (5) to [bend right=50] (3.85);
\draw[-stealth] (6.105) to [bend right=50] (3.100);

\draw[-stealth] (6.115) to [bend right=50] (5.60);

\node() at (4.7,5) {\small{$R^t = \{ (1,2), (3,4), (3,5), (3,6), (4,6), (5,6), (6,7), (7,8) \}$}};

\end{scope}

\begin{axis}[domain=0:8.5,
  samples=40,
  grid=both,xtick distance=1,xmin=0.5,xlabel=$i$,ylabel={$\Pro{\cdot}$},xmax=8.5,ymin=0,ymax=0.3, legend pos=north west]
\addplot[domain=1:8,opacity=0.8,color=green,draw=green, mark=*,mark size=2pt] plot coordinates
    {
    (1,0.015625)
    (2,0.046875)
    (3,0.078125)
    (4,0.109375)
    (5,0.140625)
    (6,0.171875)
    (7,0.203125)
    (8,0.234375)
    };
     \addlegendentry{\TwoChoice without noise $p$};
     \addplot[domain=1:10,opacity=0.8,color=red,draw=red, mark=*,mark size=2pt] plot coordinates
    {
    (1,0.015625+2/64)
    (2,0.046875-2/64)
    (3,0.078125+6/64)
    (4,0.109375-2/64)
    (5,0.140625)
    (6,0.171875-2/64-2/64)
    (7,0.203125)
    (8,0.234375)
    };
     \addlegendentry{\GAdvComp $q^t$};

     \draw[draw=black,thick,fill=blue!50,opacity=0.4] (axis cs:0.85,0.015625) rectangle (axis cs:1.15,0.015625+2/64);
     \draw[draw=black,thick,fill=blue!50,opacity=0.4] (axis cs:1.85,0.046875) rectangle (axis cs:2.15,0.046875-2/64);
     
      \draw[draw=black,thick,fill=blue!50,opacity=0.4] (axis cs:2.85,0.078125) rectangle (axis cs:3.15,0.078125+2/64);
      \draw[draw=black,thick,fill=blue!50,opacity=0.4] (axis cs:2.85,0.078125+2/64) rectangle (axis cs:3.15,0.078125+4/64);
      \draw[draw=black,thick,fill=blue!50,opacity=0.4] (axis cs:2.85,0.078125+4/64) rectangle (axis cs:3.15,0.078125+6/64); 
      \draw[draw=black,thick,fill=blue!50,opacity=0.4] (axis cs:3.85,0.109375) rectangle (axis cs:4.15,0.109375-2/64); 
      \draw[draw=black,thick,fill=blue!50,opacity=0.4] (axis cs:5.85,0.171875) rectangle (axis cs:6.15,0.171875-2/64); 
      \draw[draw=black,thick,fill=blue!50,opacity=0.4] (axis cs:5.85,0.171875-2/64) rectangle (axis cs:6.15,0.171875-4/64); 
\end{axis}

\end{tikzpicture}

\caption{Illustration of the set $R^t$ and the change in the probability allocation vector from $p$ to $q^t$, where $n=8$ and $g=3$. In the example, each directed arrow moves a probability of $\frac{2}{n^2}$ (indicated by the blue rectangles) from a bin $j$ to a heavier bin $i < j$. Note that in this example, the adversary decides not to reverse some of the comparisons, e.g., between bins $7$ and $8$.}
\end{figure}
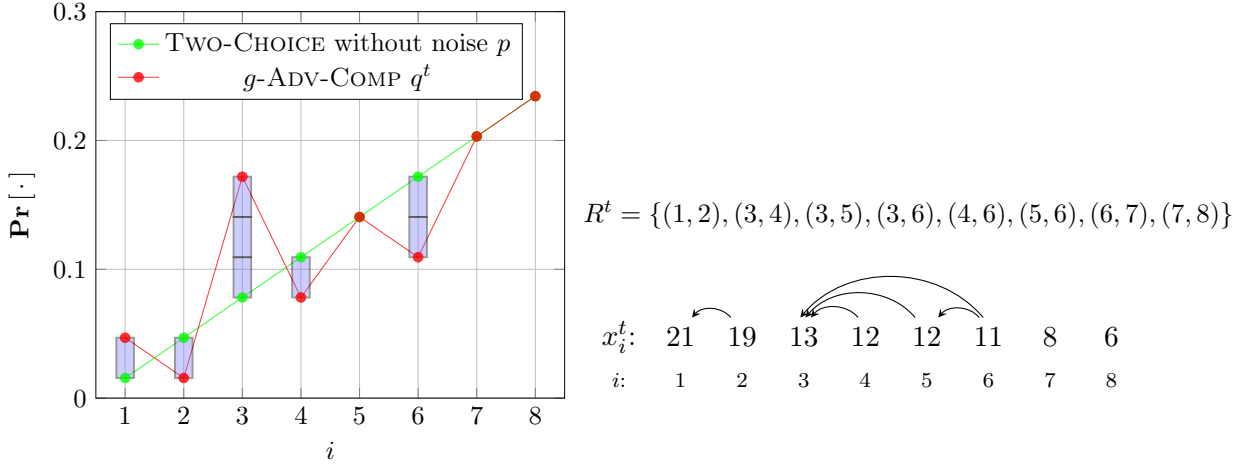

\begin{thm} \label{thm:g_adv_warm_up_gap}
Consider the \GAdvComp setting for any $g\geq 1$, and the potential $\Gamma := \Gamma(\gamma)$ with $\gamma := - \log(1 - \frac{1}{8 \cdot 48})/g < \frac{1}{6 \cdot 12}$. Then, there exist constants $c_1 \geq 1, c_2 > 0, c_3 \geq 2$, such that the following three statements hold for all steps $t \geq 0$:
\begin{align*}
(i) & \qquad \Ex{\left. \Delta\Gamma^{t+1} \,\right|\, y^t} \leq - \frac{\gamma}{96n} \cdot \Gamma^t + c_1, \\
(ii) & \qquad \Ex{\Gamma^t} \leq c_2 ng, \\
(iii) & \qquad \Pro{\max_{i \in [n]} \left|y_i^t\right| \leq c_3 g\log(ng)} \geq 1 - (ng)^{-14}.
\end{align*}
\end{thm}
\begin{proof}
\textit{First statement.} Consider the probability allocation vector $q^t$ at step $t$ in the \GAdvComp setting. By \cref{eq:lemma_four_one} we have
\begin{align*}
\Ex{\left.\Delta\Gamma^{t+1} \,\right|\, y^t} 
 & \leq h(y^t) + \sum_{i = 1}^n q_i^t \cdot f(y_i^t).
\end{align*} 
Recall that $p$ is the probability allocation vector of \TwoChoice without noise. Then, %
\[
 q^t := p + \sum_{ (i,j) \in R^t } \left( \mathbf{e}_i - \mathbf{e}_j \right) \cdot \gamma_{i,j}^t
 + \sum_{ (i,j) \in [n] \times [n] \colon y_i^t = y_j^t} \left( \mathbf{e}_i - \mathbf{e}_j \right) \cdot \gamma_{i,j}^t,
\]
where $\mathbf{e}_i \in \R^n$ is the $i$-th unit vector, and $\gamma_{i,j}^t$ is a number in $\big[0,\frac{2}{n^2} \big]$. Hence,
\begin{align}
\Ex{\left.\Delta\Gamma^{t+1} \,\right|\, y^t} 
 &\leq h(y^t) + \sum_{i = 1}^n p_i \cdot f(y_i^t) + \sum_{(i,j) \in R^t} \gamma_{i,j}^t \cdot \left( f(y_i^t) - f(y_j^t) \right) \notag \\
 & \qquad \qquad \qquad + \sum_{(i,j) \in [n] \times [n] \colon y_i^t = y_j^t} \gamma_{i,j}^t \cdot \left( f(y_i^t) - f(y_j^t) \right) \notag \\
 &= h(y^t) + \sum_{i = 1}^n p_i \cdot f(y_i^t) + \sum_{(i,j) \in R^t} \gamma_{i,j}^t \cdot \left( f(y_i^t) - f(y_j^t) \right) \notag \\
& \leq - \frac{\gamma}{48n} \cdot \Gamma^t + c + \sum_{(i,j) \in R^t} \gamma_{i,j}^t \cdot \left( f(y_i^t) - f(y_j^t) \right), \label{eq:base_case_change_1}
\end{align}
using in the last inequality that by \cref{lem:ptw_gamma_drop_two_choice} there exists such a constant $c > 0$ for the \TwoChoice process without noise and for the same $\gamma$ (since $\gamma < \frac{1}{6 \cdot 12}$).

For any pair of indices $(i, j) \in R^t$, we define
\begin{align*}
\xi_{i,j}^t :=\  & \gamma_{i,j}^t \cdot \left( f(y_i^t) - f(y_j^t) \right) \\
 \leq \  & \frac{2}{n^2} \cdot \left( (\gamma + \gamma^2) \cdot e^{\gamma y_i^t} + (-\gamma + \gamma^2) \cdot e^{-\gamma y_i^t} - (\gamma + \gamma^2) \cdot e^{\gamma y_j^t} - (-\gamma + \gamma^2) \cdot e^{-\gamma y_j^t} \right),
\end{align*}
and proceed to upper bound $\xi_{i,j}^t$, using the following lemma, which is based on a case distinction and Taylor estimates for $\exp(\cdot)$.
\begin{lem}\label{lem:gamma_reallocation}
For any pair of indices $(i, j) \in R^t$, we have $
\xi_{i,j}^t \leq \frac{\gamma}{96n^2} \cdot (\Gamma_i^t + \Gamma_j^t) + \frac{8}{n^2}.
$
\end{lem}
\begin{proof}[Proof of \cref{lem:gamma_reallocation}]
Recall that for any $(i, j) \in R^t$ we have that $y_j^t < y_i^t \leq y_j^t + g$. So, now we consider the following three disjoint cases:

\textbf{Case 1 [$y_i^t > g$]:} In this case, we also have that $y_j^t > 0$, so
\begin{align*}
\xi_{i,j}^t 
 & \stackrel{(a)}{\leq} \frac{2}{n^2} \cdot \left( (\gamma + \gamma^2) \cdot e^{\gamma y_i^t} - (\gamma + \gamma^2) \cdot e^{\gamma y_j^t} + 1 \right) \\
 & \stackrel{(b)}{\leq} \frac{2}{n^2} \cdot \left( (\gamma + \gamma^2) \cdot e^{\gamma y_i^t} - (\gamma + \gamma^2) \cdot e^{\gamma (y_i^t - g)} + 1 \right) \\
 & = \frac{2}{n^2} \cdot \left( (\gamma + \gamma^2) \cdot e^{\gamma y_i^t} \cdot (1 - e^{-\gamma g}) + 1 \right) \\
 & \leq \frac{2}{n^2} \cdot \Big( (\gamma + \gamma^2) \cdot \Gamma_i^t \cdot (1 - e^{-\gamma g}) + 1 \Big),
\end{align*}
using in $(a)$ that $0 \leq \gamma-\gamma^2 \leq 1$ and $e^{-\gamma y_k^t} \leq 1$ for any bin $k \in [n]$ with $y_k^t > 0$, and in $(b)$ that $y_i^t \leq y_j^t + g$.

\textbf{Case 2 [$y_j^t < -g$]:} In this case, we also have that $y_i^t < 0$, 
\begin{align*}
\xi_{i,j}^t 
 & \stackrel{(a)}{\leq} \frac{2}{n^2} \cdot \left( (-\gamma + \gamma^2) \cdot e^{-\gamma y_i^t} - (-\gamma + \gamma^2) \cdot e^{-\gamma y_j^t} + 1 \right) \\
 & \stackrel{(b)}{\leq} \frac{2}{n^2} \cdot \left( (-\gamma + \gamma^2) \cdot e^{-\gamma (y_j^t + g)} - ( -\gamma + \gamma^2) \cdot e^{-\gamma y_j^t} + 1 \right) \\
 & = \frac{2}{n^2} \cdot \left( (-\gamma + \gamma^2) \cdot e^{-\gamma y_j^t} \cdot (e^{-\gamma g} - 1) + 1 \right) \\
 & \leq \frac{2}{n^2} \cdot \Big( (\gamma - \gamma^2) \cdot \Gamma_j^t \cdot (1 - e^{-\gamma g}) + 1 \Big),
\end{align*}
using in $(a)$ that $\gamma+\gamma^2 \leq 1$ and that $e^{\gamma y_k^t} \leq 1$ for any bin $k \in [n]$ with $y_k^t < 0$ and in $(b)$ that $y_i^t \leq y_j^t + g$ and $\gamma^2 \leq \gamma$.

\textbf{Case 3 [$\max\{|y_i^t|, |y_j^t|\} \leq g$]:} In this case, we have that\
\[
\xi_{i,j}^t \leq \frac{2}{n^2} \cdot (\gamma + \gamma^2) \cdot (\Gamma_i^t + \Gamma_j^t) \leq \frac{2}{n^2} \cdot \frac{1}{2} \cdot \left( 2 \cdot 2 \cdot e^{\gamma g} \right) \leq \frac{2}{n^2} \cdot 4,
\]
using that $\gamma + \gamma^2 \leq \frac{1}{2}$ (since $\gamma < \frac{1}{6 \cdot 12}$) and by the definition of $\gamma$, we have $e^{\gamma g} = e^{ - \log(1-\frac{1}{8 \cdot 48} )} < 2$.
\bigskip

Combining the upper bounds for the three cases, we have that for any $(i,j) \in R^t$,
\begin{align*}
 \xi_{i,j}^t & \leq \frac{2}{n^2} \cdot \Big( (\gamma + \gamma^2) \cdot \Gamma_i^t \cdot (1 - e^{-\gamma g}) + (\gamma - \gamma^2) \cdot \Gamma_j^t \cdot (1 - e^{-\gamma g}) + 4 \Big)  \\
 & \stackrel{(a)}{\leq} \frac{2}{n^2} \cdot \Big( 2\gamma \cdot \Gamma_i^t \cdot (1 - e^{-\gamma g}) + 2\gamma \cdot \Gamma_j^t \cdot (1 - e^{-\gamma g}) + 4 \Big)  \\
 & = \frac{4\gamma}{n^2} \cdot (\Gamma_i^t + \Gamma_j^t) \cdot (1 - e^{-\gamma g}) + \frac{8}{n^2} \\
 & \stackrel{(b)}{=} \frac{4\gamma}{n^2} \cdot (\Gamma_i^t + \Gamma_j^t) \cdot \frac{1}{8 \cdot 48} + \frac{8}{n^2} \\
 & = \frac{\gamma}{96n^2} \cdot (\Gamma_i^t + \Gamma_j^t) + \frac{8}{n^2},
\end{align*}
using in $(a)$ that $\gamma^2 \leq \gamma$ since $\gamma \leq 1$ and in $(b)$ that $\gamma = - \log(1 - \frac{1}{8 \cdot 48})/g$.
\end{proof}

We continue with the proof of \cref{thm:g_adv_warm_up_gap}.
By \cref{lem:gamma_reallocation} and \cref{eq:base_case_change_1}, we have
\begin{align}
\Ex{\left.\Delta\Gamma^{t+1} \,\right|\, y^t} & \leq - \frac{\gamma}{48n} \cdot \Gamma^t + c + \sum_{(i,j) \in R^t} \Big(\frac{\gamma}{96n^2} \cdot (\Gamma_i^t + \Gamma_j^t) + \frac{8}{n^2} \Big) \notag \\
 &\stackrel{(a)}{\leq} - \frac{\gamma}{48n} \cdot \Gamma^t + c +  \frac{\gamma}{96n} \cdot \Gamma^t + 4 \notag\\
 & = - \frac{\gamma}{96n} \cdot \Gamma^t + c_1, \label{eq:base_case_change_2}
\end{align}
for $c_1 := c + 4 \geq 1$, where $(a)$ holds since if $(i,j) \in R^t$ then $(j,i) \not\in R^t$, so every bin $k \in [n]$ appears at most $n$ times in $R^t$. This concludes the proof of the first statement.

\textit{Second statement.} By \cref{lem:geometric_arithmetic}~$(ii)$ (for $a = 1 - \frac{\gamma}{96n}$ and $b = c_1$), since $\Gamma^0 = 2n \leq \frac{96c_1}{\gamma} \cdot n$ (as $c_1 \geq 1$ and $\gamma \leq 1$) and \cref{eq:base_case_change_2} hold, it follows that 
\begin{align*}
\Ex{\Gamma^{t}} & \leq \frac{96c_1}{\gamma}\cdot n =: c_2 ng.
\end{align*}

\textit{Third statement.} Using Markov's inequality, for any step $t \geq 0$,
\[
\Pro{\Gamma^t \leq c_2 \cdot (ng)^{15}} \geq 1 - (ng)^{-14}.
\]
When the event $\big\{ \Gamma^t \leq c_2 \cdot (ng)^{15} \big\}$ holds, we have that\[
\max_{i \in [n]} \left| y_i^t \right| \leq \frac{1}{\gamma} \cdot \left( \log c_2 + 15 \log( ng ) \right) \leq \frac{16 \log( ng )}{\gamma} =: c_3 g \log(ng),
\]
for sufficiently large $n$ and for the constant 
\begin{align} \label{eq:g_adv_c3_def}
c_3 := \frac{16}{\gamma g} = \frac{16}{- \log(1 - \frac{1}{8 \cdot 48})} \geq 2.
\end{align}
Therefore we conclude that,
\[ %
\Pro{\max_{i \in [n]} \left| y_i^t \right| \leq c_3 g \log (ng)} \geq 1 - (ng)^{-14} . \qedhere
\]
\end{proof}

Next we will also state a simple corollary that starting with a ``small'' gap, in any future step, \Whp~the gap will be small. This corollary will be used in obtaining the tighter $\Oh(g + \log n)$ gap bound in \cref{sec:g_adv_g_plus_logn_bound}. We defer its proof to \cref{sec:g_adv_gap_remains_small_proof}.

\newcommand{\CorGAdvGapRemainsSmall}{
Consider the \GAdvComp setting for any $g\geq 1$. Then, for any steps $t_0 \geq 0$ and $t_1 \geq t_0$, we have that
\[
\Pro{\left. \max_{i \in [n]} \left| y_i^{t_1} \right| \leq 2 g (\log(ng))^2 \,\,\right|\,\, \mathfrak{F}^{t_0}, \max_{i \in [n]} \left| y_i^{t_0} \right| \leq g (\log(ng))^2 } \geq 1 - (ng)^{14}.
\]
}
\begin{cor} \label{cor:g_adv_gap_remains_small}
\CorGAdvGapRemainsSmall
\end{cor}

\section{Upper Bound of \texorpdfstring{$\Oh(g + \log n)$}{O(g + log n)} for \texorpdfstring{$g$-\textsc{Adv-Comp}}{g-Adv-Comp}} \label{sec:g_adv_g_plus_logn_bound}

In this section we give the proof of the $\Oh(g + \log n)$ gap bound, as stated in the theorem below. For $g =\Omega(\log n)$, this matches the lower bound for the \GMyopicComp process in \cref{pro:g_myopic_g_lower_bound} up to multiplicative constants.

\newcommand{\GBoundedLognGap}{
Consider the \GAdvComp setting for any $g\geq1$, the constant $\kappa \geq \frac{1}{\alpha}$ defined in \cref{eq:g_adv_kappa_def} in \cref{lem:g_adv_good_gap_after_good_lambda} and $\alpha = \frac{1}{18}$. Then, for any step $m \geq 0$,
\[
\Pro{ \max_{i \in [n]} \left| y_i^m \right| \leq \kappa \cdot (g + \log n)} \geq 1 - 2 \cdot (ng)^{-9}.
\]
}

\newcommand{\GBoundedLognGapSimplified}{
Consider the \GAdvComp setting for any $g\geq1$. Then, there exists a constant $\kappa > 0$, such that for any step $m \geq 0$,
\[
\Pro{ \max_{i \in [n]} \left| y_i^m \right| \leq \kappa \cdot (g + \log n)} \geq 1 - 2 \cdot (ng)^{-9}.
\]
}

{\renewcommand{\thethm}{\ref{thm:g_adv_g_plus_logn_gap}}
	\begin{thm}[\textbf{Simplified version, page~\pageref{thm:g_adv_g_plus_logn_gap}}]
\GBoundedLognGapSimplified
	\end{thm} }
	\addtocounter{thm}{-1}

\subsection{Proof Outline of Theorem~\ref{thm:g_adv_g_plus_logn_gap}}

The proof of this theorem is considerably more involved than that of \cref{thm:g_adv_warm_up_gap}, making use of an interplay between a variant of the hyperbolic cosine, the absolute value and the quadratic potential (to be defined below). 

First, we define the \emph{hyperbolic cosine potential with an offset}, denoted by $\Lambda$, as\begin{align} \label{eq:lambda_def}
\Lambda^t := \Lambda^t(\alpha, c_4g) := \sum_{i = 1}^n \Lambda_i^t := \sum_{i = 1}^n \left[ e^{\alpha \cdot (y_i^t - c_4g)^+} + e^{\alpha \cdot (-y_i^t -c_4g)^+} \right],
\end{align}
for smoothing parameter $\alpha:=\frac{1}{18}$ and offset $c_4g$ with $c_4 := 730$, and recalling that $u^+ := \max\{u, 0\}$. Compared to $\Gamma$ used in \cref{sec:g_adv_warm_up}, the potential $\Lambda$ has a larger, i.e., constant smoothing parameter $\alpha$ at the cost of an offset of $\Theta(g)$. This means that if at some step $t \geq 0$, $\Lambda^t = \Oh(n)$ holds, then we can deduce a stronger upper bound of $\Gap(t) = \Oh(g + \log n)$.

However, we are
not able to show that $\Lambda$ drops in expectation in every step. More specifically, we only show that $\Lambda^t$ drops in expectation at step $t$ when the \textit{absolute value potential} $\Delta^t := \sum_{i = 1}^n \left|y_i^t\right|$ satisfies the condition $\Delta^t \leq Dng$ for $D := 365$.  This condition implies that there is at most a constant fraction of bins $i \in [n]$ with $\big|y_i^t\big| \geq \frac{3}{2}Dg$ and hence there is a bias to place away from bins $j \in [n]$ with $\big|y_j^t \big| \geq 2Dg = c_4g$, i.e., the ones with load above the offset of $\Lambda$. 
More specifically, in such \textit{good} steps $t$ we show that $\Lambda^t$ satisfies a drop inequality, meaning that when large it drops by a multiplicative factor over the next step in expectation.

{\renewcommand{\thelem}{\ref{lem:g_adv_good_step_drop}}
	\begin{lem}[\textbf{Simplified version, page~\pageref{lem:g_adv_good_step_drop}}]
Consider the \GAdvComp setting for any $g\geq1$ and let $\eps := \frac{1}{12}$. Then, for any step $t\geq 0$,  
\[
\Ex{\left. \Lambda^{t+1} \,\right|\, \mathfrak{F}^t, \Delta^t \leq Dng} \leq \Lambda^t \cdot \Big( 1 -\frac{2\alpha\eps}{n}\Big) + 18\alpha.
\]
	\end{lem} }
	\addtocounter{lem}{-1}

In order to show that there are many steps with $\Delta^t \leq Dng$, we relate the expected change of the \textit{quadratic potential} $\Upsilon^t := \sum_{i =1}^n (y_i^t)^2$ to $\Delta^t$. More specifically, by upper bounding the contribution of each pair $(i, j) \in R^t$ (defined in \cref{def:rt_definition}), similarly to the proof of \cref{thm:g_adv_warm_up_gap}, we show that:

\newcommand{\GAdvQuadratic}{
Consider the \GAdvComp setting for any $g \geq 1$. Then, for any step $t \geq 0$,
\[
\Ex{\left. \Delta\Upsilon^{t+1} \,\right\vert\, y^t} \leq -\frac{\Delta^t}{n} + 2g + 1.
\]
}

{\renewcommand{\thelem}{\ref{lem:g_adv_quadratic}}
	\begin{lem}[\textbf{Restated, page~\pageref{lem:g_adv_quadratic}}]
\GAdvQuadratic
	\end{lem} }
	\addtocounter{lem}{-1}
This allows us to prove that in any sufficiently long interval, there is a large constant fraction of good steps, i.e., with $\Delta^t \leq Dng$ (\cref{sec:g_adv_many_good_rounds}). Thus using an \textit{adjusted exponential function} based on $\Lambda$ (see~\cref{sec:g_adv_adjusted_exp_potential} for the definition), we show that \Whp~$\Lambda^s = \Oh(n)$ for some $s \in [m - \Theta(ng \cdot (\log(ng))^2), m]$ (\textit{recovery phase}) and then \Whp~$\Lambda^t = \Oh(n)$ once every $\Oh(n \cdot (g + \log n))$ steps (\textit{stabilization phase}). This implies that \Whp~$\Gap(m) = \Oh(g + \log n)$.
In more detail:
\begin{enumerate}
  \item (\textbf{Base case}) For any $t \geq 0$, \Whp~$\Gap(t) = \Oh(g \log (ng))$. (This follows from~\cref{thm:g_adv_warm_up_gap})
  \item (\textbf{Recovery}) Starting with $\Gap(t_0) = \Oh(g \log (ng))$ for $t_0 = m - \Theta(ng \cdot (\log(ng))^2)$, \Whp~there exists a step $t_1 \in [t_0, m]$ such that $\Lambda^{t_1} = \Oh(n)$ (\cref{lem:g_adv_recovery}).
  \item (\textbf{Stabilization}) Starting with $\Lambda^{t_1} = \Oh(n)$, \Whp~there exists a step $s \in (t_1, t_1 + \Theta(n \cdot (g + \log n))]$ such that $\Lambda^s = \Oh(n)$ (\cref{lem:g_adv_stabilization}).
  \item (\textbf{Gap deduction}) Hence \Whp~there is a step $t \in [m, m + \Theta(n \cdot (g + \log n))]$ with $\Gap(t) = \Oh(g + \log n)$, which by smoothness of the gap, implies $\Gap(m) = \Oh(g + \log n)$ (\cref{sec:completing_the_proof_logn_plus_g}).
\end{enumerate}

\subsection{Absolute Value and Quadratic Potentials} %

Recall that the \textit{absolute value potential} is defined as
\begin{align} \label{eq:abs_def}
\Delta^t := \sum_{i = 1}^n \left| y_i^t \right|,
\end{align}
and the \textit{quadratic potential} is defined as \begin{align} \label{eq:quad_def}
\Upsilon^t := \sum_{i = 1}^n (y_i^t)^2.
\end{align}
We will upper bound the expected change $\Ex{\left.\Delta\Upsilon^{t+1} \,\right|\, y^t}$ in the \GAdvComp setting by relating it to the change of the quadratic potential for \TwoChoice, starting with the same load vector $y^t$ at step $t$.

We will first analyze the expected change of the quadratic potential for the \TwoChoice process without noise. We will make use of the following general lemma (proven in \cref{lem:g_adv_general_quadratic_proof}), which provides a formula for the change of the quadratic potential:

\newcommand{\GAdvGeneralQuadratic}{
Consider any allocation process with probability allocation vector $r^t$ at step $t \geq 0$. Then, for any step $t \geq 0$, $(i)$ it holds that
\[
\Ex{\left. \Delta\Upsilon^{t+1} \,\right\vert\, y^t} = \sum_{i = 1}^n 2 \cdot r_i^t \cdot y_i^t + 1 - \frac{1}{n} \leq \sum_{i = 1}^n 2 \cdot r_i^t \cdot y_i^t
+ 1,
\]
and $(ii)$ it holds that
\[
\left|\Delta\Upsilon^{t+1}\right| \leq 4 \cdot \max_{i \in [n]} \left| y_i^t \right| + 2.
\]
}
\begin{lem} \label{lem:g_adv_general_quadratic} 
\GAdvGeneralQuadratic
\end{lem}

We now use the general formula in \cref{lem:g_adv_general_quadratic}~$(i)$ to  obtain an expression for the expected change of the quadratic potential for \TwoChoice without noise.

\begin{lem}[\textbf{cf.~\cite[Lemma 6.2]{LSS21}}] \label{lem:quadratic_two_choice}
Consider the \TwoChoice process without noise with probability allocation vector $p$. Then, it holds that for any step $t \geq 0$,
\[
\Ex{\left. \Delta\Upsilon^{t+1} \,\right\vert\, y^t} 
  \leq \sum_{i = 1}^n 2 \cdot p_i \cdot y_i^t + 1 \leq - \frac{\Delta^t}{n} + 1.
\]
\end{lem}
\begin{proof}
Applying \cref{lem:g_adv_general_quadratic}~$(i)$ to the probability allocation vector $p$ yields\[
\Ex{\left.\Delta\Upsilon^{t+1} \,\right\vert\, y^t} \leq \sum_{i = 1}^n 2 \cdot p_i \cdot y_i^t + 1.
\]
Recall that $B_+^t := \{ i \in [n] \colon y_i^t \geq 0 \}$ is the set of overloaded bins at step $t$ and $B_-^t := \{ i \in [n] \colon y_i^t< 0  \}$, the set of underloaded bins. The \TwoChoice process allocates a ball into the set of overloaded bins with probability $|B_+^t|^2/n^2$, and thus the average allocation probability across overloaded bins is $p_{+}^t = |B_+^t|/n^2$. Consequently, \TwoChoice allocates to the set of underloaded bins with probability $1 - |B_+^t|^2/n^2$, and thus the average allocation probability across underloaded bins is 
\[
 p_-^t = \frac{1}{|B_-^t|} \cdot \left( 1 - \frac{|B_+^t|^2}{n^2} \right) =  \frac{1}{n - |B_+^t|} \cdot \frac{(n+|B_+^t|) \cdot (n-|B_+^t|)}{n^2}
 = \frac{1}{n} + \frac{|B_+^t|}{n^2}.
\]
By splitting the sum $\sum_{i = 1}^n 2 \cdot p_i \cdot y_i^t$ into underloaded and overloaded bins, we get
\[
\sum_{i = 1}^n 2 \cdot p_i \cdot y_i^t = \sum_{i \in B_+^t} 2 \cdot p_i \cdot y_i^t + \sum_{i \in B_-^t} 2 \cdot p_i \cdot y_i^t.
\]
Since $p_i$ is non-decreasing, we have $\sum_{i=1}^j p_i \leq \sum_{i=1}^j p_{+}^t$ for all $1 \leq j \leq |B_{+}^t|$. Further, since $y_i^t$ is non-increasing over the overloaded bins, by \cref{lem:quasilem} we have
\[
\sum_{i \in B_+^t} 2 \cdot p_i \cdot y_i^t \leq \sum_{i \in B_+^t} 2 \cdot p_+^t \cdot y_i^t = 2 \cdot p_+^t \cdot \sum_{i \in B_+^t} y_i^t = \frac{|B_+^t|}{n^2} \cdot \Delta^t,
\]
since $\sum_{i \in B_-^t} y_i^t = -\sum_{i \in B_+^t} y_i^t$ and thus $\sum_{i \in B_+^t} y_i^t = \frac{1}{2} \Delta^t$. Analogously, since $y_i^t$ is non-increasing over the underloaded bins,
\[
\sum_{i \in B_-^t} 2 \cdot p_i \cdot y_i^t \leq \sum_{i \in B_-^t} 2 \cdot p_-^t \cdot y_i^t = 2 \cdot p_-^t \cdot \sum_{i \in B_-^t} y_i^t =  \left( \frac{1}{n} + \frac{|B_+^t|}{n^2} \right) \cdot \Delta^t.
\]
Combining these we get
\[
\Ex{\left.\Delta\Upsilon^{t+1} \,\right|\, y^t} \leq \sum_{i = 1}^n 2 \cdot p_i \cdot y_i^t + 1
  \leq  \frac{|B_+^t|}{n^2} \cdot \Delta^t -  \left( \frac{1}{n} + \frac{|B_+^t|}{n^2} \right) \cdot \Delta^t + 1 = - \frac{\Delta^t}{n} + 1. \qedhere
\]
\end{proof}

Now we relate the change of the quadratic potential for the \GAdvComp setting to the change of the quadratic potential for \TwoChoice without noise, using that the adversary can determine (and possibly revert) a load comparison between $y_i^t$ and $y_j^t$ only if $|y_i^t - y_j^t| \leq g$. %

\begin{lem}\label{lem:g_adv_quadratic}
\GAdvQuadratic
\end{lem}

\begin{proof}
By \cref{lem:g_adv_general_quadratic}~$(i)$, for the \GAdvComp probability allocation vector $q^t$  we have,
\[
\Ex{\left.\Delta\Upsilon^{t+1} \,\right\vert\, y^t} \leq \sum_{i = 1}^n 2 \cdot q_i^t \cdot y_i^t+ 1.
\]
This probability vector $q^t$ is obtained from the probability allocation vector $p$ of \TwoChoice without noise by moving a probability of up to $\frac{2}{n^2}$ from any bin $j$ to a bin $i$ with $y_j^t < y_i^t \leq y_j^t + g$. Recalling that $R^t:=\{ (i, j) \in [n] \times [n] \colon y_j^t < y_i^t \leq y_j^t + g  \}$,
\begin{align*}
\Ex{\left. \Delta\Upsilon^{t+1} \,\right\vert\, y^t} 
 & \leq \sum_{i = 1}^n 2 \cdot p_i \cdot y_i^t + 1 + 2 \cdot \sum_{(i, j) \in R^t} \frac{2}{n^2} \cdot (y_i^t - y_j^t) \\
 & \leq \sum_{i = 1}^n 2 \cdot p_i \cdot y_i^t + 1 + 2 \cdot \sum_{(i, j) \in R^t} \frac{2}{n^2} \cdot g \\
 & \leq \sum_{i = 1}^n 2 \cdot p_i \cdot y_i^t + 1 + 2g,
\end{align*}
using that $|R^t| < \frac{1}{2} n^2$. Hence, using \cref{lem:quadratic_two_choice}, we conclude that\[
\Ex{\left. \Delta\Upsilon^{t+1} \,\right\vert\, y^t} \leq -\frac{\Delta^t}{n} + 2g + 1. \qedhere
\]
\end{proof}

\subsection{Constant Fraction of Good Steps} \label{sec:g_adv_many_good_rounds}

We define a step $s \geq 0$ to be a \textit{good step} if $\mathcal{G}^s := \{ \Delta^s \leq Dng \}$ holds, for $D := 365$. Further, $G_{t_0}^{t_1} := G_{t_0}^{t_1}(D)$ denotes the number of good steps in $[t_0, t_1]$. Later, in \cref{sec:g_adv_exp_potential} we will show that in a good step, the exponential potential $\Lambda$ with any sufficiently small constant $\alpha$ drops in expectation.

In the following lemma we show that at least a constant fraction $r$ of the steps are good in a sufficiently long interval. We will apply this lemma with two different values for $T$: $(i)$ in the recovery phase, to prove that there exists a step $s \in [m-\Theta(ng \cdot (\log(ng))^2), m]$ with $\Lambda^s = \Oh(n)$ and $(ii)$ in the stabilization phase, to prove that every $\Oh(n \cdot (g + \log n))$ steps there exists a step $s$ with $\Lambda^s = \Oh(n)$. In the analysis below, we pick $r := \frac{6}{6+\eps}$, where $\eps:=\frac{1}{12}$. %

\begin{lem} \label{lem:g_adv_many_good_steps}
Consider the \GAdvComp setting for any $g \geq 1$ and let $r := \frac{6}{6+\eps}$ and $D := 365$. Then, for any constant $\hat{c} \geq 1$ and any $T \in [ng^2, n^2 g^3/\hat{c}]$, we have for any steps $t_0 \geq 0$ and $t_1 := t_0 + \hat{c} \cdot T \cdot g^{-1} - 1$,
\[
\Pro{ G_{t_0}^{t_1}(D) \geq r\cdot (t_1 - t_0 + 1) \,\left|\, \mathfrak{F}^{t_0}, \Upsilon^{t_0} \leq T, \, \max_{i \in [n]} \left| y_i^{t_0} \right| \leq g (\log(ng))^2 \right.} \geq 1 - 2\cdot (ng)^{-12}.
\]
\end{lem}
\begin{proof}
We define the sequence $(Z^t)_{t \geq t_0}$ with $Z^{t_0} := \Upsilon^{t_0}$ and for any $t > t_0$,  
\[
 Z^{t} := \Upsilon^{t} + \sum_{s=t_0}^{t-1} \left(\frac{\Delta^{s}}{n} - 2g - 1\right).
\]
This sequence forms a super-martingale since by~\cref{lem:g_adv_quadratic},\begin{align*}
\Ex{\left. Z^{t+1} \,\right|\, \mathfrak{F}^{t}} 
  & = \Ex{\left. \Upsilon^{t+1} + \sum_{s=t_0}^{t} \left(\frac{\Delta^{s}}{n} - 2g - 1\right) ~\,\right|\,~ \mathfrak{F}^{t}} \\
  & \leq \Upsilon^{t} - \frac{\Delta^{t}}{n} + 2g + 1 + \sum_{s=t_0}^{t} \left(\frac{\Delta^{s}}{n} - 2g - 1\right) \\
  & = \Upsilon^t + \sum_{s=t_0}^{t-1} \left(\frac{\Delta^{s}}{n} - 2g - 1\right) 
  = Z^{t}.
\end{align*}

Further, let $\tau:=\inf\{ t \geq t_0 \colon \max_{i \in [n]} |y_i^{t}| > 2 g (\log(ng))^2 \}$ and consider the stopped random variable
\[
 \tilde{Z}^{t} := Z^{t \wedge \tau},
\]
which is then also a super-martingale. Applying 
\cref{cor:g_adv_gap_remains_small} and the union bound over steps $[t_0, t_1]$, we get%
\begin{align}\label{eq:g_adv_tau_small_whp}
 \Pro{ \tau \leq t_1 \,\left|\, \, \mathfrak{F}^{t_0}, \Upsilon^{t_0} \leq T, \, \max_{i \in [n]} \left| y_i^{t_0} \right| \leq g (\log(ng))^2 \right.} 
  & \leq (\hat{c} \cdot T \cdot g^{-1}) \cdot (ng)^{-14}
  \leq (ng)^{-12},
\end{align}
using that $T \leq n^2g^3/\hat{c}$. This means that the maximum absolute normalized load does not increase above $2 g (\log(ng))^2$ in any of the steps in $[t_0,t_1]$ \Whp

 To prove concentration of $\tilde{Z}^{t_1+1}$, we will now derive an upper bound on the difference $| \tilde{Z}^{t+1} - \tilde{Z}^{t}|$: 
\medskip 

\noindent\textbf{Case 1 [$t \geq \tau$]:}
In this case, $\tilde{Z}^{t+1} = Z^{(t+1) \wedge \tau} = Z^{\tau}$, and similarly, $\tilde{Z}^{t} = Z^{t \wedge \tau} = Z^{\tau}$, so $| \tilde{Z}^{t+1} - \tilde{Z}^{t}|=0$.
\medskip

\noindent\textbf{Case 2 [$t < \tau$]:} In this case, we have $\max_{i \in [n]} |y_i^t| \leq 2 g (\log(ng))^2$ and by \cref{lem:g_adv_general_quadratic}~$(ii)$, we have that $|\Delta\Upsilon^{t+1}| \leq 8c_3 g (\log(ng))^2 + 2$. This implies that
\[
| \tilde{Z}^{t+1} - \tilde{Z}^{t}| \leq |\Delta\Upsilon^{t+1}| + \bigg|\frac{\Delta^t}{n} - 2g - 1\bigg| \leq 8 g (\log(ng))^2 + 2 + (2g (\log(ng))^2 - 2g - 1) \leq 10 g (\log(ng))^2.
\]

\medskip
Combining the two cases above, we conclude that for all $t \geq t_0$,
\[
 | \tilde{Z}^{t+1} - \tilde{Z}^{t}| \leq 10 g (\log(ng))^2.
\]
Using Azuma's inequality for super-martingales (\cref{lem:azuma}) for $\lambda = T$ and $a_i = 10 g (\log(ng))^2$,
\begin{align*}
 & \Pro{ \tilde{Z}^{t_1+1} - \tilde{Z}^{t_0} \geq T \,\left|\, \, \mathfrak{F}^{t_0}, \Upsilon^{t_0} \leq T, \, \max_{i \in [n]} \left| y_i^{t_0} \right| \leq g (\log(ng))^2 \right. } \\
 & \qquad \leq \exp\left( - \frac{ T^2 }{ 2 \cdot \sum_{t=t_0}^{t_1} (10 g (\log(ng))^2)^2  } \right) \\ 
 & \qquad = \exp\left( - \frac{ T^2 }{ \hat{c} \cdot T \cdot g^{-1} \cdot 200 \cdot g^2 \cdot (\log(ng))^4  } \right) \\
 & \qquad = \exp\left( - \frac{ T }{ 200 \cdot \hat{c} \cdot g \cdot (\log(ng))^4  } \right) \\
 & \qquad \stackrel{(a)}{\leq} \exp\left( - \frac{ng^2}{ 200 \cdot \hat{c} \cdot g \cdot (\log(ng))^4  } \right)
  = (ng)^{-\omega(1)},
\end{align*}
where in $(a)$ we used that $T \geq ng^2$.
Hence, we conclude that \[\Pro{ \tilde{Z}^{t_1+1}  < \tilde{Z}^{t_0} + T \,\left|\, \, \mathfrak{F}^{t_0}, \Upsilon^{t_0} \leq T, \, \max_{i \in [n]} \left| y_i^{t_0} \right| \leq g (\log(ng))^2 \right.} \geq 1 - (ng)^{-\omega(1)}.\]
Thus by taking the union bound with \cref{eq:g_adv_tau_small_whp} we have  \[\Pro{ Z^{t_1+1} < Z^{t_0} + T \,\left|\, \, \mathfrak{F}^{t_0}, \Upsilon^{t_0} \leq T, \, \max_{i \in [n]} \left| y_i^{t_0} \right| \leq g (\log(ng))^2 \right.} \geq 1 - 2 \cdot (ng)^{-12}.
\]
For the sake of a contradiction, assume now that more than an $(1-r)$ fraction of the steps $t \in [t_0,t_1]$ satisfy $\Delta^t > Dng$. This implies that
\begin{align} \label{eq:g_adv_sum_of_abs_value}
\sum_{t=t_0}^{t_1} \frac{\Delta^{t}}{n}
 > Dg \cdot (1 - r) \cdot (t_1 - t_0 + 1) = D \cdot (1-r) \cdot \hat{c} \cdot T,
\end{align}
using that $t_1 - t_0 + 1 = \hat{c} \cdot T \cdot g^{-1}$.
When $\{ Z^{t_1+1} < Z^{t_0} + T \}$ and $\{ \Upsilon^{t_0} \leq T \}$ hold, then we have
\[
\Upsilon^{t_1+1} + \sum_{t=t_0}^{t_1} \frac{\Delta^{t}}{n} - (2g + 1) \cdot (t_1 - t_0 + 1) < \Upsilon^{t_0} + T \leq 2T.
\]
By rearranging this leads to a contradiction as
\begin{align*}
0 \leq \Upsilon^{t_1+1} 
 & < 2T - \sum_{t=t_0}^{t_1} \frac{\Delta^{t}}{n} + (2g + 1) \cdot (t_1 - t_0 + 1) \\
 & \stackrel{(a)}{\leq} 2T - \sum_{t=t_0}^{t_1} \frac{\Delta^{t}}{n} + \hat{c} (2g + 1) \cdot T \cdot g^{-1} \\
 & \stackrel{(b)}{\leq} - \sum_{t=t_0}^{t_1} \frac{\Delta^{t}}{n} + 5\hat{c} \cdot T \\
 & \!\!\stackrel{(\text{\ref{eq:g_adv_sum_of_abs_value}})}{<} -D \cdot (1-r) \cdot \hat{c} \cdot T + 5\hat{c} \cdot T \\
 & \stackrel{(c)}{=} 0
\end{align*}
using in $(a)$ that $t_1 - t_0 + 1 = \hat{c} \cdot T \cdot g^{-1}$, in $(b)$ that $\hat{c} \geq 1$ and $g \geq 1$, and in $(c)$ that $D = \frac{5}{1 - r} = 365$ (as $r = \frac{6}{6 + 1/12}$).

We conclude that when $\{ Z^{t_1+1} < Z^{t_0} + T \}$ holds, then at least an $r$ fraction of the steps $t \in [t_0,t_1]$ satisfy $\Delta^{t} \leq Dng$, and thus,
\[
\Pro{G_{t_0}^{t_1} \geq r \cdot (t_1 - t_0 + 1) \,\left|\, \, \mathfrak{F}^{t_0}, \Upsilon^{t_0} \leq T, \, \max_{i \in [n]} \left| y_i^{t_0} \right| \leq g (\log(ng))^2 \right.} \geq 1 - 2\cdot (ng)^{-12}. \qedhere
\]
\end{proof}

The following lemma provides two ways of upper bounding the quadratic potential using the exponential potential $\Lambda$. These will be used in the recovery and stabilization lemmas, to obtain the starting point condition ($\Upsilon^{t_0} \leq T$) for \cref{lem:g_adv_many_good_steps}. We defer its proof to \cref{sec:g_adv_bounds_on_quadratic_proof}.

\newcommand{\GAdvBoundsOnQuadratic}{
Consider the potential $\Lambda := \Lambda(\alpha, c_4g)$ for any constant $\alpha \in (0, 1)$, any $g \geq 1$ and any constant $c_4 > 0$. Then~$(i)$ for any constant $\hat{c} > 0$, there exists a constant $c_s := c_s(\alpha, c_4, \hat{c}) \geq 1$, such that for any step $t \geq 0$ with $\Lambda^t \leq \hat{c} \cdot n$, 
\[
\Upsilon^t \leq c_s \cdot n g^2.
\]
Furthermore, $(ii)$ there exists a constant $c_r := c_r(\alpha, c_4) \geq 1$, such that for any step $t \geq 0$,
\[
\Upsilon^t \leq c_r \cdot n \cdot \Big(g^2 + (\log \Lambda^t)^2 \Big).
\]
}
\begin{lem}\label{lem:g_adv_bounds_on_quadratic}
\GAdvBoundsOnQuadratic
\end{lem}

\subsection{Exponential Potential} \label{sec:g_adv_exp_potential}

We now prove bounds on the expected change of the $\Lambda$ potential function over one step. Note that these hold for any sufficiently small constant $\alpha > 0$. 

We start with a relatively weak bound which holds at any step.

\begin{lem} \label{lem:g_adv_bad_step_increase}
Consider any allocation process with probability allocation vector $r^t$ such that $\max_{i \in [n]} r_i^t \leq \frac{2}{n}$. Further, consider the potential $\Lambda := \Lambda(\alpha, c_4g)$ for any $\alpha \in \big(0, \frac{1}{2}\big]$, any $g \geq 1$ and any $c_4 > 0$. Then, for any step $t \geq 0$, $(i)$~for every bin $i \in [n]$ it holds that \[
\Ex{\left. \Lambda_i^{t+1} \,\right|\, \mathfrak{F}^t} \leq \Lambda_i^t \cdot \Big( 1 + \frac{3\alpha}{n}\Big).
\]
Furthermore, by aggregating over all bins, $(ii)$~it holds that
\[
\Ex{\left. \Lambda^{t+1} \,\right|\, \mathfrak{F}^t} \leq \Lambda^t \cdot \Big( 1 + \frac{3\alpha}{n}\Big).
\]
\end{lem}
\begin{proof}
Consider an arbitrary bin $i \in [n]$. We upper bound the change of the overloaded component, i.e., $e^{\alpha \cdot (y_i^t - c_4g)^{+}}$, by placing one ball in bin $i \in [n]$ with probability $r_i^t \leq \frac{2}{n}$ and ignoring the change of the average load. Also, we upper bound the change for the underloaded component, i.e.,
$e^{\alpha \cdot (- y_i^t - c_4g)^{+}}$,
 by considering only the change of the average load by $1/n$. Hence,
\begin{align*}
\Ex{\left. \Lambda_i^{t+1} \,\right|\, \mathfrak{F}^t} 
 & \leq e^{\alpha \cdot (y_i^t - c_4g)^+} \cdot ((1-r_i^t) + r_i^t \cdot e^{\alpha}) + e^{\alpha \cdot (-y_i^t - c_4g)^+} \cdot ((1-r_i^t) \cdot e^{\alpha/n} + r_i^t \cdot e^{\alpha/n}) \\
 & = e^{\alpha \cdot (y_i^t - c_4g)^+} \cdot (1 + r_i^t \cdot (e^{\alpha} - 1)) + e^{\alpha \cdot (-y_i^t - c_4g)^+} \cdot e^{\alpha/n} \\
 & \stackrel{(a)}{\leq} e^{\alpha \cdot (y_i^t - c_4g)^+} \cdot \Big(1 + \frac{2}{n} \cdot 1.5\alpha\Big) + e^{\alpha \cdot (-y_i^t - c_4g)^+} \cdot \Big(1 + 1.5\frac{\alpha}{n}\Big) \\
 & \leq \Lambda_i^t \cdot \Big( 1 + \frac{3\alpha}{n}\Big),
\end{align*}
using in $(a)$ that $e^u \leq 1 + 1.5u$ (for any $0 \leq u \leq 0.7$), $\alpha \leq 1/2$ and $r_i^t \leq \frac{2}{n}$. %
\end{proof}

Now we improve this bound for any good step $t$, i.e., when $\Delta^t \leq Dng$ holds.

\begin{lem}\label{lem:g_adv_good_step_drop}
Consider the \GAdvComp setting for any $g \geq 1$ and the potential $\Lambda := \Lambda(\alpha, c_4g)$ for any $\alpha \in \left(0, \frac{1}{18}\right]$, $c_4 := 2D$ and $D := 365$. Then, for any step $t\geq 0$, $(i)$~for $\eps := \frac{1}{12}$, it holds that 
\[
\Ex{\left. \Lambda^{t+1} \,\right|\, \mathfrak{F}^t, \Delta^t \leq Dng} \leq \Lambda^t \cdot \Big( 1 -\frac{2\alpha\eps}{n}\Big) + 18\alpha.
\]
Furthermore, this also implies that $(ii)$~for $c := \frac{18}{\eps}$, it holds that
\[
\Ex{\left. \Lambda^{t+1} \,\right|\, \mathfrak{F}^t, \Delta^t \leq Dng, \Lambda^t > cn} \leq \Lambda^t \cdot \Big( 1 -\frac{\alpha\eps}{n}\Big).
\]
\end{lem}

\begin{proof}
\textit{First statement.} Consider an arbitrary step $t \geq 0$ with $\Delta^t \leq Dng$. We bound the expected change of $\Lambda$ over one step, by considering the following cases for each bin $i \in [n]$:

\textbf{Case 1 [$y_i^t \in (-c_4g - 2, c_4g + 2)$]:} Using \cref{lem:g_adv_bad_step_increase}~$(i)$,\begin{align*}
\Ex{\left. \Lambda_i^{t+1} \,\right|\, \mathfrak{F}^t} 
 & \leq \Lambda_i^t \cdot \Big( 1 + \frac{3\alpha}{n}\Big) \\
 & = \Lambda_i^t \cdot \Big( 1 - \frac{\alpha}{6n}\Big) + \Lambda_i^t \cdot \Big( \frac{\alpha}{6n} + \frac{3\alpha}{n}  \Big)\\
 & \stackrel{(a)}{\leq} \Lambda_i^t \cdot \Big( 1 -\frac{\alpha}{6n}\Big) + (2 \cdot e^{2\alpha}) \cdot \Big(\frac{1}{6} + 3\Big) \cdot \frac{\alpha}{n} \\
& \stackrel{(b)}{\leq} \Lambda_i^t \cdot \Big( 1 -\frac{\alpha}{6n}\Big) + \frac{18\alpha}{n},
\end{align*}
using in $(a)$ that $\Lambda_i^t \leq e^{2\alpha} + 1 \leq 2 \cdot e^{2\alpha}$ by the assumption that $y_i^t \in (-c_4g - 2, c_4g + 2)$ and in $(b)$ that $(2 \cdot e^{2\alpha}) \cdot (\frac{1}{6} + 3) \leq 18$, since $\alpha \leq 1/4$.

\textbf{Case 2 [$y_i^t \geq c_4g + 2$]:} By the condition $\Delta^t \leq Dng$,  the number of bins $j$ with $y_j^t \geq \frac{3}{2} Dg$ is at most $\frac{1}{2} \Delta^t \cdot \frac{2}{3Dg} \leq \frac{Dng}{3Dg} = \frac{n}{3}$ (see~\cref{fig:g_adv_overload_bias}). We can allocate to bin $i \in [n]$ with $y_i^t \geq c_4g + 2 = 2Dg + 2 \geq \frac{3}{2} Dg + g$ only if we sample bin $i$ and a bin $j$ with $y_j^t \geq \frac{3}{2}Dg$. Hence, \[
q_i^t \leq 2 \cdot \frac{1}{n} \cdot \frac{1}{3} = \frac{2}{3n}.
\]

\begin{figure}[H]
    \centering

\footnotesize{
 \begin{tikzpicture}[level/.style={},decoration={brace,mirror,amplitude=7}]
      \fill[red!50!white] (0, 0) rectangle (2, 0.2);
      \fill[black!40!white] (2, 0) rectangle (4, 0.2);
      \fill[green!80!black] (4, 0) rectangle (12, 0.2);
      
      \draw[very thick] (0,0) -- (12, 0);
      \draw[very thick] (0,0.5) -- (0,-0.5);
      \draw[very thick] (12,0.5) -- (12,-0.5);
      \node at (1,0.5) {$\geq 2Dg$};
      \node at (3,0.5) {$[\frac{3}{2}Dg, 2Dg)$};
      \node at (8,0.5) {$< \frac{3}{2}Dg$};
      
      \draw[dashed, very thick] (2,0.9) -- (2,-0.5);
      \draw[dashed, very thick] (4,0.9) -- (4,-0.5);
      
      \draw [decorate] (0.1,-0.4) --node[below=2.5mm]{$\leq n/3$} (3.9,-0.4);
      \draw [decorate] (4.1,-0.4) --node[below=2.5mm]{$\geq 2n/3$} (11.9,-0.4);

      \draw[->, thick] (0.5, -0.7) -- (0.5, -0.02);
      \node[anchor=north] at (0.5, -0.7) {$i$};

      \draw[->, thick] (3.2, -0.7) -- (3.2, -0.02);
      \node[anchor=north] at (3.2, -0.7) {$j$};
      
    \end{tikzpicture}}
    \caption{In any step $t$ with $\Delta^t \leq Dng$ there are at most $\frac{n}{3}$ bins with $y_i^t \geq \frac{3}{2} Dg$. In the \GAdvComp, we can distinguish between the red and green bins, so the probability to allocate to a red bin is at most $\frac{2}{3n}$.}
    \label{fig:g_adv_overload_bias}
\end{figure}

By assumption, bin $i$ deterministically satisfies $y_i^{t+1} \geq c_4g + 2 - 1/n > c_4g$, so
\begin{align*}
& \Ex{\left. \Lambda_i^{t+1} \,\right|\, \mathfrak{F}^t, \Delta^t \leq Dng} \\
 & \qquad = e^{\alpha \cdot (y_i^t - c_4g)^+} \cdot \Big((1-q_i^t) \cdot e^{-\alpha/n}+ q_i^t \cdot e^{\alpha (1 - 1/n)} \Big) + 1 \\
 & \qquad \stackrel{(a)}{\leq} e^{\alpha \cdot (y_i^t - c_4g)^+} \cdot \left( (1- q_i^t) \cdot \Big( 1 - \frac{\alpha}{n} + \frac{\alpha^2}{n^2} \Big) + q_i^t \cdot \Big( 1 + \alpha \cdot \Big(1 - \frac{1}{n}\Big) + \alpha^2 \Big)\right) + 1 \\
 & \qquad = e^{\alpha \cdot (y_i^t - c_4g)^+} \cdot \Big( 1 + \alpha \cdot \Big(q_i^t - \frac{1}{n} \Big) + (1-q_i^t) \cdot \frac{\alpha^2}{n^2} + q_i^t \cdot \alpha^2 \Big) + 1 \\
 & \qquad \stackrel{(b)}{\leq} e^{\alpha \cdot (y_i^t - c_4g)^+} \cdot \Big( 1 - \frac{\alpha}{3n} + \frac{4    \alpha^2}{3n} \Big) + 1 \\
 & \qquad \stackrel{(c)}{\leq} e^{\alpha \cdot (y_i^t - c_4g)^+} \cdot \Big( 1 - \frac{\alpha}{6n}\Big) + 1 \\
 & \qquad = e^{\alpha \cdot (y_i^t - c_4g)^+} \cdot \Big( 1 - \frac{\alpha}{6n}\Big) + 1 \cdot \Big(1 - \frac{\alpha}{6n}\Big) + \frac{\alpha}{6n} \\
 & \qquad = \Lambda_i^t \cdot \Big( 1 - \frac{\alpha}{6n} \Big) + \frac{\alpha}{6n},
\end{align*}
using in $(a)$ that $e^u \leq 1 + u + u^2$ for $u < 1.75$, $\alpha \leq 1$ and $(1 - 1/n)^2 \leq 1$, in $(b)$ that $q_i^t \leq \frac{2}{3n}$ and $(1 - q_i^t) \cdot \frac{\alpha^2}{n^2} \leq \frac{2\alpha^2}{3n}$ for $n \geq 2$, and in $(c)$ that $\alpha \leq \frac{1}{8}$, so $\frac{4\alpha^2}{3n} \leq \frac{\alpha}{6n}$.

\textbf{Case 3 [$y_i^t \leq -c_4g - 2$]:} 
The number of bins $j$ with $y_j^t \leq -\frac{3}{2} Dg$ is at most $\frac{1}{2}\Delta^t \cdot \frac{2}{3Dg} \leq \frac{Dng}{3Dg} = \frac{n}{3}$ and the number of bins $j$ with $y_j^t > -\frac{3}{2} Dg$ is at least $\frac{2n}{3}$ (see~\cref{fig:g_adv_underload_bias}). Similarly to Case 2, we can allocate to a bin $i \in [n]$ with load $y_i^t \leq -c_4g - 2$ only if we sample $i$ and a bin $j$ with $y_j^t > -\frac{3}{2} Dg$. Hence, \[
q_i^t \geq 2 \cdot \frac{1}{n} \cdot \frac{2}{3} = \frac{4}{3n}.
\]

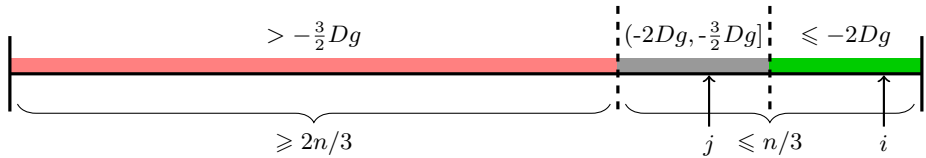
\begin{figure}[H]
    \centering
\footnotesize{    
 \begin{tikzpicture}[level/.style={},decoration={brace,mirror,amplitude=7}]
      \fill[red!50!white] (0, 0) rectangle (8, 0.2);
      \fill[black!40!white] (8, 0) rectangle (10, 0.2);
      \fill[green!80!black] (10, 0) rectangle (12, 0.2);
      
      \draw[very thick] (0,0) -- (12, 0);
      \draw[very thick] (0,0.5) -- (0,-0.5);
      \draw[very thick] (12,0.5) -- (12,-0.5);
      \node at (4,0.5) {$> -\frac{3}{2} Dg$};
      \node at (9,0.5) {$(\text{-}2Dg, \text{-}\frac{3}{2} Dg]$};
      \node at (11,0.5) {$\leq -2Dg$};
      
      \draw[dashed, very thick] (8,0.9) -- (8,-0.5);
      \draw[dashed, very thick] (10,0.9) -- (10,-0.5);
      
      \draw [decorate] (0.1,-0.4) --node[below=2.5mm]{$\geq 2n/3$} (7.9,-0.4);
      \draw [decorate] (8.1,-0.4) --node[below=2.5mm]{$\leq n/3$} (11.9,-0.4);
      
      \draw[->, thick] (11.5, -0.7) -- (11.5, -0.02);
      \node[anchor=north] at (11.5, -0.7) {$i$};

      \draw[->, thick] (9.2, -0.7) -- (9.2, -0.02);
      \node[anchor=north] at (9.2, -0.7) {$j$};
      
    \end{tikzpicture}}
    \vspace{-0.2cm}
    \caption{In any step $t$ with $\Delta^t \leq Dng$ there are at most $n/3$ bins with $y_i^t \leq -\frac{3}{2} Dg$ and at least $2n/3$ bins with $y_i^t > -\frac{3}{2} Dg$. In the \GAdvComp, we can distinguish between the red and green bins, so the probability to allocate to a green bin is at least $\frac{4}{3n}$.}
    \label{fig:g_adv_underload_bias}
\end{figure}

By assumption, bin $i$ deterministically satisfies $y_i^{t+1} \leq -c_4g - 2 + 1 < -c_4g$, so
\begin{align*}
& \Ex{\left. \Lambda_i^{t+1} \,\right|\, \mathfrak{F}^t, \Delta^t \leq Dng} \\
 & \qquad = 1 + e^{\alpha (- y_i^t - c_4g)^+} \cdot \Big((1-q_i^t) \cdot e^{\alpha/n} + q_i^t \cdot e^{-\alpha(1 - 1/n)}\Big) \\
 & \qquad \stackrel{(a)}{\leq} 1 + e^{\alpha (- y_i^t - c_4g)^+} \cdot \left( (1 - q_i^t) \cdot \Big( 1 + \frac{\alpha}{n} + \frac{\alpha^2}{n^2}\Big) + q_i^t \cdot \Big( 1 - \alpha \cdot \Big( 1 - \frac{1}{n} \Big) + \alpha^2 \Big) \right) \\
 & \qquad = 1 + e^{\alpha (- y_i^t - c_4g)^+} \cdot \Big( 1 + \alpha \cdot \Big( \frac{1}{n} - q_i^t\Big) + (1-q_i^t) \cdot \frac{\alpha^2}{n^2} + q_i^t \cdot \alpha^2 \Big) \\
 & \qquad \stackrel{(b)}{\leq} 1 + e^{\alpha (- y_i^t - c_4g)^+} \cdot \Big( 1 - \frac{\alpha}{3n} + \frac{5\alpha^2}{2n} \Big) \\
 & \qquad \stackrel{(c)}{\leq} 1 + e^{\alpha (- y_i^t - c_4g)^+} \cdot \Big( 1 - \frac{\alpha}{6n} \Big) \\
 & \qquad = \Big(1 - \frac{\alpha}{6n}\Big) + e^{\alpha (- y_i^t - c_4g)^+} \cdot \Big( 1 - \frac{\alpha}{6n} \Big) + \frac{\alpha}{6n} \\
 & \qquad = \Lambda_i^t \cdot \Big( 1 - \frac{\alpha}{6n} \Big) + \frac{\alpha}{6n},
\end{align*}
using in $(a)$ that $e^u \leq 1 + u + u^2$ for $u < 1.75$, $\alpha \leq 1$ and $(1-1/n)^2 \leq 1$, in $(b)$ that $q_i^t \in [\frac{4}{3n}, \frac{2}{n}]$ and $(1-q_i^t) \cdot \frac{\alpha^2}{n^2} \leq \frac{\alpha^2}{2n}$ for $n \geq 2$ and in $(c)$ that $\alpha \leq \frac{1}{15}$, so $\frac{5\alpha^2}{2n} \leq \frac{\alpha}{6n}$.

Combining these three cases and letting $\eps := \frac{1}{12}$, we conclude that \[
\Ex{\left. \Lambda^{t+1} \,\right|\, \mathfrak{F}^t, \Delta^t \leq Dng} \leq \sum_{i = 1}^n \left(\Lambda_i^t \cdot \Big( 1 - \frac{\alpha}{6n} \Big) + \frac{18\alpha}{n}\right) = \Lambda^t \cdot \Big( 1 - \frac{2\alpha\eps}{n} \Big) + 18\alpha.
\]
\textit{Second statement.} Letting $c := \frac{18}{\eps} = 18 \cdot 12$, it follows that
\begin{align*}
\Ex{\left. \Lambda^{t+1} \,\right|\, \mathfrak{F}^t, \Delta^t \leq Dng, \Lambda^t > cn} 
 & \leq \Lambda^t \cdot \Big( 1 -\frac{2\alpha\eps}{n}\Big) + 18\alpha \\
 & = \Lambda^t \cdot \Big( 1 -\frac{\alpha\eps}{n}\Big) - \Lambda^t \cdot \frac{\alpha\eps}{n} + 18\alpha \\
 & \leq \Lambda^t \cdot \Big( 1 -\frac{\alpha\eps}{n}\Big). \qedhere
\end{align*}
\end{proof}

\subsection{Adjusted Exponential Potential} \label{sec:g_adv_adjusted_exp_potential}

In \cref{lem:g_adv_good_step_drop}, we proved that in a good step $t$ with $\Lambda^t > cn$ (for $c := 18 \cdot 12$), the potential drops in  expectation by a multiplicative factor. Our goal will be to show that \Whp~$\Lambda^t \leq cn$ at a single step (\textbf{recovery}) and then show that it becomes small at least once every $\Oh(n \cdot (g + \log n))$ steps (\textbf{stabilization}). Since we do not have an expected drop in every step, but only at a constant fraction $r$ of the steps,  we will define an \textit{adjusted exponential potential function}. 
First, for any step $t_0$, 
and any step $s \geq t_0$, we define the following event:\begin{equation*}%
\mathcal{E}_{t_0}^{s} := \bigcap_{t \in [t_0, s]} \left\{\Lambda^t > c n \right\}.  \end{equation*}

Next, we define the sequence $(\tilde{\Lambda}_{t_0}^{s})_{s \geq t_0} := (\tilde{\Lambda}_{t_0}^{s})_{s \geq t_0}(\alpha, c_4g, \eps)$ as $\tilde{\Lambda}_{t_0}^{t_0} := \Lambda^{t_0}(\alpha, c_4g)$ and, for any $s > t_0$, %
\begin{equation}\label{eq:tilde_lambda}
\tilde{\Lambda}_{t_0}^s := \Lambda^s(\alpha, c_4g) \cdot \mathbf{1}_{\mathcal{E}_{t_0}^{s-1}} \cdot \exp\bigg( - \frac{3\alpha }{n} \cdot B_{t_0}^{s-1} \bigg) \cdot \exp\bigg( + \frac{\alpha \eps}{n} \cdot G_{t_0}^{s-1} \bigg),
\end{equation}
recalling that $G_{t_0}^{t_1}$ is the number of good steps in $[t_0, t_1]$, i.e., steps where the event $\mathcal{G}^s := \{ \Delta^s \leq Dng \}$ holds, and $B_{t_0}^{t_1} := [t_0, t_1] \setminus G_{t_0}^{t_1}$ is the number of bad steps in $[t_0, t_1]$.

By its definition and \cref{lem:g_adv_bad_step_increase,lem:g_adv_good_step_drop}, it follows that the sequence $(\tilde{\Lambda}_{t_0}^{s})_{s \geq t_0}$ is a super-martingale. A proof is given in \cref{sec:g_adv_lambda_tilde_is_supermartingale}.

\newcommand{\LambdaTildeSuperMartingale}{
Consider the \GAdvComp setting for any $g \geq 1$ and the sequence $(\tilde{\Lambda}_{t_0}^{s})_{s \geq t_0} := (\tilde{\Lambda}_{t_0}^{s})_{s \geq t_0}(\alpha, c_4g, \eps)$ for any starting step $t_0 \geq 0$, any $\alpha \in (0, \frac{1}{18}]$ and $\eps, c_4 > 0$ as defined in \cref{lem:g_adv_good_step_drop}. Then, for any step $s \geq t_0$,
\[
 \Ex{\left. \tilde{\Lambda}_{t_0}^{s+1} \,\right|\, \mathfrak{F}^s} \leq \tilde{\Lambda}_{t_0}^{s}.
\]
}

\begin{lem}[\textbf{cf.~\cite[Lemma 9.1]{LSS21}}]\label{lem:g_adv_lambda_tilde_is_supermartingale}
\LambdaTildeSuperMartingale
\end{lem}

\subsection{Recovery and Stabilization}

We are now ready to prove the recovery lemma, that is the $\Lambda$ potential becomes small every $\Theta(ng (\log(ng))^2)$ steps.

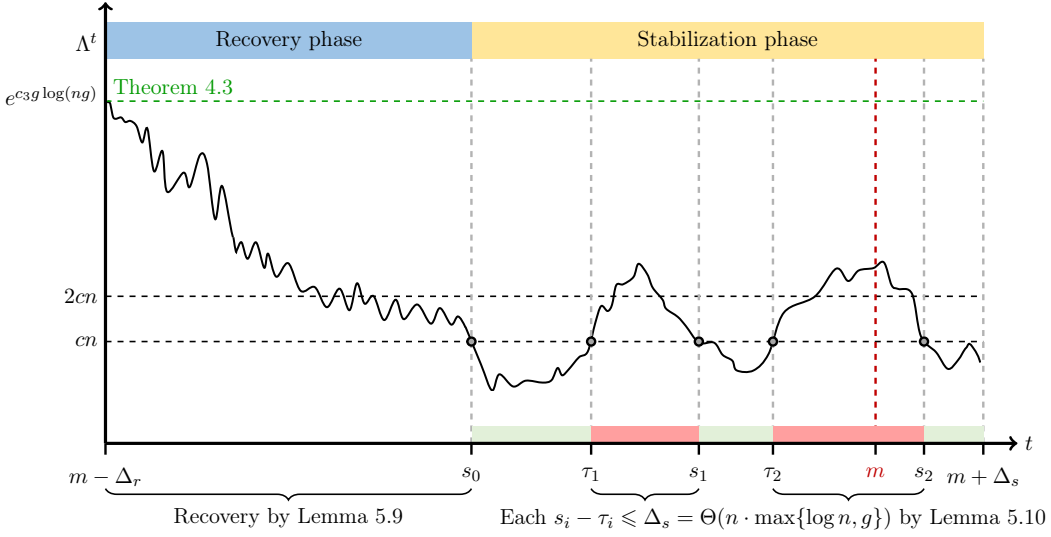
\begin{figure}
    \centering
\scalebox{0.75}{
    \begin{tikzpicture}[
  IntersectionPoint/.style={circle, draw=black, very thick, fill=black!35!white, inner sep=0.05cm}
]

\definecolor{MyBlue}{HTML}{9DC3E6}
\definecolor{MyYellow}{HTML}{FFE699}
\definecolor{MyGreen}{HTML}{E2F0D9}
\definecolor{MyRed}{HTML}{FF9F9F}
\definecolor{MyDarkRed}{HTML}{C00000}

\node[anchor=south west] (plt) at (-0.1, 0) {\includegraphics{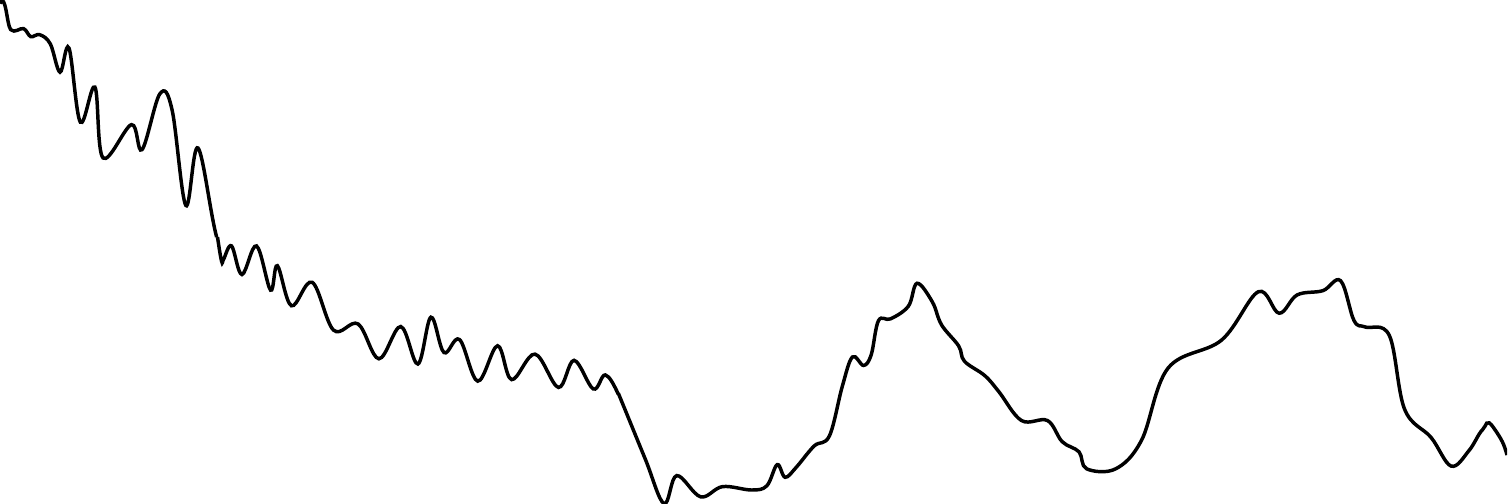}};

\def\xEnd{16}
\def\xLast{15.40}
\def\yLast{6}
\def\cn{1}
\def\yBottom{-0.8}

\node[anchor=south west, inner sep=0.15cm, fill=MyBlue, rectangle,minimum width=6.42cm] at (0, \yLast) {Recovery phase};
\node[anchor=south west, inner sep=0.15cm, fill=MyYellow, rectangle,minimum width=\xLast cm - 6.42cm] at (6.42, \yLast) {Stabilization phase};

\draw[dashed, thick] (0,1) -- (\xLast, 1);
\draw[dashed, thick] (0,1.8) -- (\xLast, 1.8);
\draw[dashed, thick, green!60!black] (0,5.25) -- (\xLast, 5.25);

\node[anchor=south west, green!60!black] at (0,5.25) {\cref{thm:g_adv_warm_up_gap}};

\node[anchor=east] at (0, \cn) {$cn$};
\node[anchor=east] at (0, 1.8) {$2cn$};
\node[anchor=east] at (0, 5.35) {$e^{c_3 g \log(ng)}$};
\node[anchor=east] at (0, 6.3) {$\Lambda^t$};

\node[anchor=west] at (\xEnd, \yBottom) {$t$};

\def\tA{6.42}
\def\tB{8.52}
\def\tC{10.41}
\def\tD{11.71}
\def\tM{13.51}
\def\tE{14.36}

\newcommand{\drawLine}[3]{
\draw[dashed, very thick, #3] (#1, \yBottom) -- (#1, \yLast);
\draw[very thick] (#1, \yBottom) -- (#1, \yBottom -0.2);
\node[anchor=north] at (#1, \yBottom -0.3) {#2};}

\newcommand{\drawPoint}[3]{
\drawLine{#1}{#2}{#3}
\node[IntersectionPoint] at (#1, \cn) {};}

\draw[very thick] (0, \yBottom) -- (0, \yBottom -0.2);
\node[anchor=north] at (0, \yBottom -0.3) {$m - \Delta_r$};

\drawPoint{\tA}{$s_0$}{black!30!white};
\drawPoint{\tB}{$\tau_1$}{black!30!white};
\drawPoint{\tC}{$s_1$}{black!30!white};
\drawPoint{\tD}{$\tau_2$}{black!30!white};
\drawPoint{\tE}{$s_2$}{black!30!white};
\drawLine{\tM}{\textcolor{MyDarkRed}{$m$}}{MyDarkRed};
\drawLine{\xLast}{$m + \Delta_s$}{black!30!white};

\newcommand{\drawRegionRect}[3]{
\node[anchor=south west, rectangle, fill=#3, minimum width=#2 cm- #1 cm, minimum height=0.3cm] at (#1, \yBottom) {};}

\drawRegionRect{\tA}{\tB}{MyGreen}
\drawRegionRect{\tB}{\tC}{MyRed}
\drawRegionRect{\tC}{\tD}{MyGreen}
\drawRegionRect{\tD}{\tE}{MyRed}
\drawRegionRect{\tE}{\xLast}{MyGreen}

\newcommand{\drawBrace}[4]{
\draw [
    thick,
    decoration={
        brace,
        raise=0.5cm,
        amplitude=0.2cm
    },
    decorate
] (#2, \yBottom - 0.3) -- (#1, \yBottom - 0.3) 
node [anchor=north,yshift=-0.7cm,#4] {#3}; }

\drawBrace{0}{\tA}{Recovery by \cref{lem:g_adv_recovery}}{pos=0.5};
\drawBrace{\tB}{\tC}{}{};
\drawBrace{\tD}{\tE}{Each $s_i - \tau_i \leq \Delta_s = \Theta(n \cdot \max\{ \log n, g \})$ by \cref{lem:g_adv_stabilization}}{};

\draw[->, ultra thick] (0,\yBottom) -- (0, 7);
\draw[->, ultra thick] (0,\yBottom) -- (\xEnd, \yBottom);

\end{tikzpicture}}
    \caption{Visualization of the recovery (\cref{lem:g_adv_recovery}) and stabilization (\cref{lem:g_adv_stabilization}) phases. Note that the red intervals \Whp~will have length $\leq \Delta_s$, implying $\Gap(m) = \Oh(g + \log n)$.}
    \label{fig:g_adv_recovery_and_stabilization}
\end{figure}

\begin{lem}[\textbf{Recovery}] \label{lem:g_adv_recovery}
Consider the \GAdvComp setting for any $g \geq 1$ and the potential $\Lambda := \Lambda(\alpha, c_4g)$ with $\alpha = \frac{1}{18}$, and $c_4 > 0$ as defined in \cref{lem:g_adv_good_step_drop}. Further, let the constants $c, \eps > 0$ be as defined in \cref{lem:g_adv_good_step_drop}, $c_r := c_r(\alpha, c_4) \geq 1$ as in \cref{lem:g_adv_bounds_on_quadratic}~$(ii)$, $r \in (0, 1)$ as in \cref{lem:g_adv_many_good_steps} and $c_3 \geq 2$ as in \cref{thm:g_adv_warm_up_gap}. Then, for any step $t_0 \geq 0$, $(i)$~for $\Delta_r := \Delta_r(g) := \frac{60 c_3^2 c_r}{\alpha\eps r} \cdot n g \cdot (\log (ng))^2$, it holds that
\[
\Pro{ \left. \bigcup_{t \in [t_0, t_0 + \Delta_r]} \{ \Lambda^t \leq cn \} \; \right| \; \mathfrak{F}^{t_0}, \max_{i \in [n]} \left| y_i^{t_0} \right| \leq c_3g \log(ng) }\geq 1 - 3 \cdot (ng)^{-12}.
\]
Further, $(ii)$~it holds that,
\[
 \Pro{ \bigcup_{t \in [t_0,t_0 + \Delta_r]} \{ \Lambda^t \leq cn \}} \geq 1 - (ng)^{-11}.
\]
\end{lem}

\begin{proof}
\textit{First statement.} Consider an arbitrary step $t_0$ with  $\max_{i \in [n]} |y_i^{t_0}| \leq c_3 g \log (ng)$. Our aim is to show that \Whp~$\Lambda^t \leq cn$ for some step $t \in [t_0, t_0 + \Delta_r]$. We will do this by first showing that \Whp~there is a significant number of good steps, i.e., $G_{t_0}^{t_1} \geq r \cdot \Delta_r$, and when this happens \Whp~$\tilde{\Lambda}_{t_0}^{t_1} = 0$, which implies the conclusion.

We start by upper bounding $\Lambda^{t_0}$ as follows, \[
\Lambda^{t_0} \leq 2n \cdot e^{\alpha c_3 g \log(ng)} %
\leq e^{c_3 g \log(ng)} =: \lambda,
\]
since $\alpha \leq \frac{1}{18}$ and $c_3 \geq 2$. Hence, by \cref{lem:g_adv_bounds_on_quadratic}~$(ii)$, there exists a constant $c_r := c_r(\alpha, c_4)$, such that
\[
\Upsilon^{t_0} \leq c_r \cdot n \cdot (g^2 + (\log \Lambda^{t_0})^2) \leq  2c_r \cdot n \cdot (c_3 g \log (ng))^2 =: T.
\]
Let $t_1 :=t_0+\Delta_r $. Applying \cref{lem:g_adv_many_good_steps} with $T = 2c_r \cdot n \cdot (c_3 g \log (ng))^2 = o(n^2g^3)$ (and $\geq ng^2$) and $\hat{c} = \frac{\Delta_r \cdot g}{T} \geq \frac{30}{\alpha\eps r} \geq 1$ as $\alpha, \eps, r \leq 1$, we get
\begin{align}\label{eq:rec_many_good_quantiles_whp_new}
 & \Pro{ G_{t_0}^{t_1-1} \geq r \cdot \Delta_r \;\left|\; \mathfrak{F}^{t_0}, \Lambda^{t_0} \leq \lambda, \max_{i \in [n]} \left| y_i^{t_0} \right| \leq c_3 g \log(ng) \right.} \notag \\
 & \quad \geq \Pro{ G_{t_0}^{t_1-1} \geq r \cdot \Delta_r \;\left|\; \mathfrak{F}^{t_0}, \Upsilon^{t_0} \leq T, \max_{i \in [n]} \left| y_i^{t_0} \right| \leq g (\log(ng))^2 \right.} \notag \\
 & \quad \geq 1 - 2 \cdot (ng)^{-12}.
\end{align}
By \cref{lem:g_adv_lambda_tilde_is_supermartingale}, $(\tilde{\Lambda}_{t_{0}}^{t})_{t \geq t_{0}}$ is a super-martingale, so $\ex{\tilde{\Lambda}_{t_0}^{t_1}\mid \mathfrak{F}^{t_0}} \leq \tilde{\Lambda}_{t_0}^{t_0} =  \Lambda^{t_0}$. Hence, using Markov's inequality we get $\Pro{\tilde{\Lambda}_{t_0}^{t_1} >\Lambda^{t_0} \cdot (ng)^{12}\,\big|\, \mathfrak{F}^{t_0}, \Lambda^{t_0} \leq \lambda } \leq  (ng)^{-12}$. Thus, by the definition of $\tilde{\Lambda}_{t_0}^{t_1}$ in \cref{eq:tilde_lambda}, we have 
\begin{equation} \label{eq:rec_supermartingale_markov_new}
\Pro{\Lambda^{t_1} \cdot \mathbf{1}_{\mathcal{E}_{t_0}^{t_1-1}} \leq \Lambda^{t_0} \cdot (ng)^{12} \cdot \exp\left(  \frac{3 \alpha}{n} \cdot B_{t_0}^{t_1-1}  -\frac{\alpha \eps}{n} \cdot G_{t_0}^{t_1-1} \right) \, \Bigg| \, \mathfrak{F}^{t_0}, \Lambda^{t_0} \leq \lambda } \geq  1 - (ng)^{-12}.    
\end{equation}
Further, if in addition to the two events $\{\tilde{\Lambda}_{t_0}^{t_1} \leq \Lambda^{t_0} \cdot (ng)^{12}\}$ and $\{\Lambda^{t_0} \leq \lambda \}$, also the event $\{G_{t_0}^{t_1-1} \geq r \cdot \Delta_r \}$ holds, then 
\begin{align*}
\Lambda^{t_1} \cdot \mathbf{1}_{\mathcal{E}_{t_0}^{t_1-1}} & \leq \Lambda^{t_0} \cdot (ng)^{12} \cdot \exp\bigg(  \frac{3 \alpha}{n} \cdot B_{t_0}^{t_1-1}  -\frac{\alpha \eps}{n} \cdot G_{t_0}^{t_1-1} \bigg) \\
 & \leq e^{c_3 g \log(ng)} \cdot (ng)^{12} \cdot \exp\left( \frac{3 \alpha}{n} \cdot (1 - r) \cdot \Delta_r - \frac{\alpha \eps}{n} \cdot r \cdot \Delta_r \right) \\
 & \stackrel{(a)}{=} e^{c_3 g \log(ng)} \cdot (ng)^{12} \cdot \exp\bigg( \frac{\alpha \eps}{n} \cdot \frac{r}{2} \cdot \Delta_r - \frac{\alpha \eps}{n} \cdot r \cdot \Delta_r \bigg) \\
 & = e^{c_3 g \log(ng)} \cdot (ng)^{12} \cdot \exp\bigg( - \frac{ \alpha \eps}{n} \cdot \frac{r}{2} \cdot \Delta_r \bigg) \\
 & = e^{c_3 g \log(ng)} \cdot (ng)^{12} \cdot \exp\bigg( - \frac{\alpha \eps}{n} \cdot \frac{r}{2} \cdot  \frac{60 c_3^2 c_r}{\alpha \eps r} \cdot ng \cdot (\log (ng))^2 \bigg) \\ 
 & \stackrel{(b)}{\leq} e^{c_3 g \log(ng)} \cdot (ng)^{12} \cdot \exp\left( - 30 c_3 g\log (ng) \right) \\
 & \leq 1,
\end{align*}
where we used in $(a)$ that $r = \frac{6}{6 + \eps}$ implies $\frac{3\alpha}{n} \cdot (1-r) = \frac{3\alpha}{n} \cdot \frac{\eps}{6 + \eps} =  \frac{\alpha\eps}{n} \cdot \frac{r}{2}$, in $(b)$ that $g \geq 1$, $c_r \geq 1$ and $c_3 \geq 1$.
By the definition of $\Lambda$, we have that $\Lambda^{t_1} \geq n$  holds deterministically, and so we can deduce from the above inequality that $\mathbf{1}_{\mathcal{E}_{t_0}^{t_1-1}}=0$, that is,
\[\Pro{ \neg \mathcal{E}_{t_0}^{t_1 -1} \; \Bigg|\; \mathfrak{F}^{t_0},  \;\;\; \tilde{\Lambda}_{t_0}^{t_1} \leq \Lambda^{t_0} \cdot (ng)^{12}, \;\;\; \Lambda^{t_0} \leq \lambda,  \;\; \; G_{t_0}^{t_1-1} \geq r \cdot \Delta_r } =1.\] 
Recalling the definition of $\mathcal{E}_{t_0}^{t_1-1} := \bigcap_{t \in [t_0, t_1-1]} \{ \Lambda^t > cn \}$ and taking the union bound over \cref{eq:rec_many_good_quantiles_whp_new} and \cref{eq:rec_supermartingale_markov_new} yields 
\begin{align} \label{eq:g_adv_conditional_recovery}
\Pro{ \left. \bigcup_{t \in [t_0, t_0 + \Delta_r]} \{ \Lambda^t \leq cn \} \; \right| \; \mathfrak{F}^{t_0}, \Lambda^{t_0} \leq \lambda  }\geq 1 - 2 \cdot (ng)^{-12} -(ng)^{-12} = 1 - 3 \cdot (ng)^{-12}.
\end{align}
\textit{Second statement.} Using \cref{thm:g_adv_warm_up_gap}~$(iii)$, for the constant $c_3 \geq 2$, it holds that,\begin{align*}
\Pro{\max_{i \in [n]} \left| y_i^{t_0} \right| \leq c_3 g \log (ng)} \geq 1 - (ng)^{-14}.
\end{align*}
Hence, combining with \cref{eq:g_adv_conditional_recovery}, we conclude that
\[
\Pro{ \bigcup_{t \in [t_0, t_0 + \Delta_r]} \{ \Lambda^t \leq cn \}} 
\geq \left(1 - 3 \cdot (ng)^{-12} \right) \cdot \left(1 - (ng)^{-14} \right) \geq 1 - (ng)^{-11}. \qedhere
\]
\end{proof}

The derivation of the lemma below is similar to that of \cref{lem:g_adv_recovery}, with the main difference being the tighter condition that $\Lambda^{t_0} \leq 2 cn$, which allows us to choose a slightly shorter time window of $\Theta(n \cdot \max\{\log n, g\})$ steps (see \cref{fig:g_adv_recovery_and_stabilization}). We defer the proof for \cref{sec:g_adv_stabilization_proof}.

\newcommand{\GAdvStabilization}{
Consider the \GAdvComp setting for any $g \geq 1$ and the potential $\Lambda := \Lambda(\alpha, c_4g)$ with $\alpha = \frac{1}{18}$, and $c_4 > 0$ as defined in \cref{lem:g_adv_good_step_drop}. Further, let the constants $c, \eps > 0$ be as defined in \cref{lem:g_adv_good_step_drop}, $c_s := c_s(\alpha, c_4, 2c) \geq 1$ as in \cref{lem:g_adv_bounds_on_quadratic}~$(i)$ and $r \in (0, 1)$ as in \cref{lem:g_adv_many_good_steps}. Then, for $\Delta_s := \Delta_s(g) := \frac{60 c_s}{\alpha \eps r} \cdot n \cdot \max\{\log n, g\}$, we have that for any step $t_0 \geq 0$,
\[
 \Pro{ \left. \bigcup_{t \in [t_0,t_0 + \Delta_s]} \left\{\Lambda^{t} \leq cn \right\} ~\right\vert~ \mathfrak{F}^{t_0}, \Lambda^{t_0} \leq 2cn} \geq 1 -  (ng)^{-11}.
\] 
}
\begin{lem}[\textbf{Stabilization}] \label{lem:g_adv_stabilization}
\GAdvStabilization
\end{lem}

\subsection{Completing the Proof of Theorem~\ref{thm:g_adv_g_plus_logn_gap}} \label{sec:completing_the_proof_logn_plus_g}

We will now prove that starting with $\Lambda^{t_0} \leq 2cn$ implies that for any step $t_1 \in [t_0, t_0 + (ng)^2]$, we have \Whp~$\Gap(t_1) = \Oh(g + \log n)$.
\begin{lem} \label{lem:g_adv_good_gap_after_good_lambda}
Consider the \GAdvComp setting for any $g \geq 1$, the potential $\Lambda := \Lambda(\alpha, c_4g)$ with $\alpha = \frac{1}{18}$, $c_4 > 0$ as defined in \cref{lem:g_adv_good_step_drop} and $\Delta_s > 0$ as defined in \cref{lem:g_adv_stabilization}. Then, there exists a constant $\kappa \geq  \frac{1}{\alpha}$ such that for any steps $t_0 \geq 0$ and $t_1 \in (t_0, t_0+ (ng)^2]$,
\[
\Pro{\max_{i \in [n] } \left| y_i^{t_1} \right| \leq \kappa \cdot (g + \log n) \,\, \Big| \,\, \mathfrak{F}^{t_0}, \Lambda^{t_0} \leq cn} \geq 1 - (ng)^{-9}.
\]
\end{lem}
\begin{proof} 
Consider any step $t_0$ with $\Lambda^{t_0} \leq cn$. We define the event
\[ \mathcal{M}_{t_0}^{t_1} := \left\{\text{for all }t\in [t_0, t_1 ]\text{ there exists } s\in [t, t+\Delta_s  ] \text{ such that }\Lambda^s \leq cn\right\},\]	
that is, if $\mathcal{M}_{t_0}^{t_1}$ holds then we have $\Lambda^s \leq cn $ at least once every $\Delta_s := \frac{60 c_s}{\alpha \eps r} \cdot n \cdot \max\{\log n, g\}$ steps. 

Assume now that $\mathcal{M}_{t_0}^{t_1}$ holds. We will show that \[\max_{i \in [n]} \big| y_i^{t_1} \big| \leq \kappa \cdot (g + \log n).
\]
Choosing $t=t_1$, implies that there exists $s \in [t_1,t_1+\Delta_s]$ such that $\Lambda^s \leq cn$, which in turn implies by definition of $\Lambda$ that $\max_{i \in [n]} |y_i^s| \leq \frac{1}{\alpha} \cdot \log (cn) + c_4g < \frac{2}{\alpha} \cdot \log n + c_4g$. Clearly, any $y_i^t$ can decrease by at most $1/n$ in each step, and from this it follows that if $\mathcal{M}_{t_0}^{t_1}$ holds, then 
\[
\max_{i \in [n]} y_i^{t_1} \leq \max_{i \in [n]} \left|y_i^s \right| + \frac{\Delta_s}{n} \leq \kappa \cdot (g + \log n), 
\]
for the constant 
\begin{align} \label{eq:g_adv_kappa_def}
\kappa := \frac{2}{\alpha} + c_4 + \frac{\Delta_s}{n \cdot \max\{\log n, g\}} = \frac{2}{\alpha} + c_4 + \frac{60c_s}{\alpha\eps r} > 0.
\end{align}

If $t_1 \geq t_0 + \Delta_s$ and $\mathcal{M}_{t_0}^{t_1}$ holds, then choosing $t = t_1 - \Delta_s$, there exists $s \in [t_1 - \Delta_s, t_1]$ such that $\Lambda^s \leq cn$. (In case $t_1 < t_0 + \Delta_s$, then we arrive at the same conclusion by choosing $s = t_0$ and using the precondition $\Lambda^{t_0} \leq cn$). %
This in turn implies
$
 \max_{i \in [n]} | y_i^s  | \leq \frac{2}{\alpha} \cdot \log n + c_4g.
$
Hence 
\[
\min_{i \in [n]} y_i^{t_1} \geq -\max_{i \in [n]} \left|y_i^{s}\right| -  \frac{\Delta_s}{n} 
\geq -\kappa \cdot (g + \log n).
\]
 Hence, $\mathcal{M}_{t_0}^{t_1}$ together with the precondition on $\Lambda^{t_0} \leq cn$ implies that $\max_{i \in [n]} \big| y_i^{t_1} \big| \leq \kappa \cdot (g + \log n)$. It remains to bound $\Pro{\left. \neg \mathcal{M}_{t_0}^{t_1} \,\right|\, \mathfrak{F}^{t_0}, \Lambda^{t_0} \leq cn}$.

Note that if for some step $j_1$ we have that $\Lambda^{j_1} \leq cn$ and for some $j_2 \geq j_1$ that $\Lambda^{j_2} > 2cn$, then there must exist $j \in (j_1, j_2)$ such that $\Lambda^j \in (cn, 2cn]$, since for every $t \geq 0$ it holds that $\Lambda^{t+1} \leq \Lambda^t \cdot e^{\alpha } \leq 2 \Lambda^t$, as $\alpha \leq 1/2$. 
Let $t_0 < \tau_1 < \tau_2<\cdots $ and $t_0 =: s_0<s_1<\cdots $ be two interlaced sequences defined recursively for $i\geq 1$ by \[\tau_i := \inf\left\{\tau >  s_{i-1}: \Lambda^{\tau} \in (cn, 2cn] \right\}\qquad\text{and} \qquad s_i := \inf\left\{s> \tau_i : \Lambda^s \leq cn\right\}. \] 
  Thus we have
  \[
  t_0 = s_0 < \tau_1 < s_1 < \tau_2 <s_2 < \cdots, 
  \]
  and since $\tau_i > \tau_{i-1}$ we have $ \tau_{t_1 - t_0}\geq t_1$. Therefore, if the event $\cap_{i=1}^{t_1 - t_0}\{s_i-\tau_i\leq \Delta_s \} $ holds, then also $ \mathcal{M}_{t_0}^{t_1}$ holds.

   Recall that by \cref{lem:g_adv_stabilization} we have for any $i=1,2,\ldots, t_1 - t_0$ and any $\tau = t_0 + 1, \ldots , t_1$ \[
 \Pro{ \left.\bigcup_{t \in [\tau_i,\tau_i + \Delta_s]} \left\{\Lambda^{t} \leq c n \right\} ~\right|~ \mathfrak{F}^{\tau} , \; \Lambda^{\tau} \in (cn, 2cn], \tau_i = \tau } \geq  1 - (ng)^{-11},
  \] and by negating and the definition of $s_i$,
  \[
  \Pro{s_i-\tau_i> \Delta_s \, \Big| \, \mathfrak{F}^{\tau}, \Lambda^{\tau} \in (cn, 2cn], \tau_i = \tau } \leq (ng)^{-11}.
  \]
  Since the above bound holds for any $i \geq 1$ and $\mathfrak{F}^{\tau}$, with $\tau_i= \tau$, it follows by the union bound over all $i=1,2,\ldots,t_1 - t_0$, as $t_1 - t_0 \leq (ng)^2$,
  \[
  \Pro{\left. \neg \mathcal{M}_{t_0}^{t_1} \,\right|\, \mathfrak{F}^{t_0}, \Lambda^{t_0} \leq cn }\leq (t_1 - t_0) \cdot (ng)^{-11} \leq (ng)^{-9}. \qedhere\]
\end{proof}

Finally, we deduce that for the \GAdvComp setting, for an arbitrary step $m$ \Whp~$\Gap(m) = \Oh(g + \log n)$.

\begin{thm} \label{thm:g_adv_g_plus_logn_gap}
\GBoundedLognGap
\end{thm}

\begin{proof}
Consider an arbitrary step $m \geq 0$ and recall that $\Delta_r := \frac{60 c_3^2 c_r}{\alpha\eps r} \cdot n g \cdot (\log (ng))^2$. If $m < \Delta_r$, then the claim follows by \cref{lem:g_adv_good_gap_after_good_lambda} as $\Lambda^{0} = 2n \leq cn$ and $\Delta_r < (ng)^2$.

Otherwise, let $t_0 := m - \Delta_r$. %
Firstly, by the recovery lemma (\cref{lem:g_adv_recovery}~$(ii)$), we get
\begin{align} \label{eq:g_adv_recovery_lambda}
 \Pro{ \bigcup_{t \in [t_0,t_0 + \Delta_r]} \left\{\Lambda^{t} \leq cn \right\}} \geq 1 - (ng)^{-11}.
\end{align}
Hence for $\tau:=\inf\{ s \geq t_0 \colon \Lambda^s \leq cn \}$ we have $\Pro{ \tau \leq m } \geq 1-(ng)^{-11}$, as $t_0 + \Delta_r = m$.

Secondly, using \cref{lem:g_adv_good_gap_after_good_lambda}, there exists a constant $\kappa := \kappa(\alpha, \eps) > 0$ such that for any step $s \in [t_0, m]$,
\begin{align} \label{eq:g_adv_stabilisation_minmax_gap}
\Pro{\max_{i \in [n]} \left| y_i^m \right| \leq \kappa \cdot (g + \log n) ~\Big|~ \mathfrak{F}^s , \Lambda^s \leq cn} \geq 1 - (ng)^{-9}.
\end{align}
Combining the two inequalities from above, we conclude the proof
\begin{align*}
    \Pro { \max_{i \in [n]} \left| y_i^m \right| \leq \kappa \cdot (g + \log n)} &\geq \sum_{s=t_0}^{m} \Pro{ \max_{i \in [n]} \left| y_i^m \right| \leq \kappa \cdot (g + \log n)  ~\Big|~ \tau = s } \cdot \Pro{ \tau = s} \\
    &\geq \sum_{s=t_0}^{m} \Pro{ \max_{i \in [n]} \left| y_i^m \right| \leq \kappa \cdot (g + \log n)  ~\Big|~ \mathfrak{F}^{s}, \Lambda^{s} \leq cn } \cdot \Pro{ \tau = s} \\
    & \!\!\!\stackrel{(\text{\ref{eq:g_adv_stabilisation_minmax_gap}})}{\geq} \left( 1-(ng)^{-9} \right) \cdot \Pro{ \tau \leq m} \\
    & \!\!\!\stackrel{(\text{\ref{eq:g_adv_recovery_lambda}})}{\geq} \left(1 - (ng)^{-9} \right) \cdot \left(1 - (ng)^{-11} \right) \geq 1 - 2 \cdot (ng)^{-9}.\qedhere
\end{align*}
\end{proof}

\section{Upper Bound of \texorpdfstring{$\Oh(\frac{g}{\log g} \cdot \log \log n)$}{O(g/log g * log log n)} for  \texorpdfstring{$g$-\textsc{Adv-Comp}}{g-Adv-Comp} with \texorpdfstring{$g \leq \log n$}{g <= log n}: Outline}\label{sec:g_adv_upper_bound_for_small_g_outline}

In this section, we will outline the proof for the upper bound on the gap of $\Oh( \frac{g}{\log g} \cdot \log \log n)$ for the \GAdvComp setting with $1 < g \leq \log n$ (which also implies for $g \in \{0,1\}$ the $\Oh(\log \log n)$ gap bound by monotonicity). This matches the lower bound for the \GMyopicComp process proven later in \cref{thm:g_myopic_layered_induction_lower_bound} up to multiplicative constants. The upper bound proof will be completed in \cref{sec:g_adv_base_case,sec:super_exponential_potentials,sec:g_adv_layered_induction}.

\newcommand{\GBoundedStrongestBound}{
Consider the \GAdvComp setting for any $g \in (1, \log n]$. Then, there exists a constant $\tilde{\kappa} > 0$ such that for any step $m \geq 0$,
\[
\Pro{\Gap(m) \leq \tilde{\kappa} \cdot \frac{g}{\log g} \cdot \log \log n} \geq 1 - n^{-3}.
\]
}

{\renewcommand{\thethm}{\ref{thm:g_adv_strongest_bound}}
	\begin{thm}[\textbf{Restated, page~\pageref{thm:g_adv_strongest_bound}}]
\GBoundedStrongestBound
	\end{thm} }
	\addtocounter{thm}{-1}

\subsection{Definitions of Super-Exponential Potential Functions} \label{sec:super_exp_potentials_def}

The proof of this theorem employs some kind of layered induction over $k$ different, super-exponential potential functions. We will be using the following definition of super-exponential potentials.

The \emph{super-exponential potential function} with smoothing parameter $\phi \geq 1$ and integer offset $z := z(n) > 0$ is defined at any step $t \geq 0$ as
\begin{align} \label{defi:super_exponential_potential}
\Phi^t := \Phi^t(\phi, z) := \sum_{i = 1}^n \Phi_i^t := \sum_{i = 1}^n e^{\phi \cdot (y_i^t - z)^+},
\end{align}
where $u^+ := \max\{u, 0\}$.

These are similar to the potential functions used in \cite[Section 6]{LS21}, but with two extensions: $(i)$~the potential functions are ``smoother'', e.g., allowing us to deduce gap bounds between $\sqrt{\log n}$ and $\log n$, and $(ii)$~the concentration of the base potential in the layered induction follows from strengthening the stabilization theorem in \cref{sec:g_adv_g_plus_logn_bound}.

Let \begin{align} \label{eq:g_adv_alpha_1_def}
\alpha_1 := \frac{1}{6 \kappa} \leq \frac{1}{6 \cdot 18},
\end{align}
for $\kappa \geq \frac{1}{\alpha} = 18 > 0$ the constant in \cref{eq:g_adv_kappa_def} in  \cref{lem:g_adv_good_gap_after_good_lambda} and 
\begin{align} \label{eq:g_adv_alpha_2_def}
\alpha_2 := \frac{\alpha_1}{84} \leq \frac{1}{84 \cdot 6 \cdot 18}.
\end{align}
We define the function \label{sec:g_adv_k_def}\[
f(k) := (\alpha_1 \log n)^{1/k} = e^{\frac{1}{k} \cdot \log (\alpha_1 \log n)},
\]
which is monotone decreasing in $k > 0$, and for $k=1$, $f(1)= \alpha_1 \log n$. This implies that for every $1 < g < \alpha_1 \log n$, there exists a unique integer $k := k(g) \geq 2$ satisfying,
\[
  (\alpha_1 \log n)^{1/k} \leq g < (\alpha_1 \log n)^{1/(k-1)}.
\]
This definition implies that $k = \Theta\big(\frac{\log \log n}{\log g}\big)$ and that $k = \Oh(\log \log n)$, since $g > 1$.

Keeping in mind the previous inequality, we will be making the slightly stronger assumption for $g = \Omega(1)$ (see \cref{clm:tilde_g_justification}) that
\begin{align} \label{eq:tilde_g_stronger_assumption}
(\alpha_1 \cdot (\log n))^{1/k} \leq g < \left(\frac{\alpha_2}{4} \cdot (\log n)\right)^{1/(k-1)}.
\end{align}
For any $g$ satisfying $\big(\frac{\alpha_2}{4} \cdot (\log n)\big)^{1/k} \leq g < (\alpha_1 \cdot (\log n))^{1/k}$, we will obtain the stated $\Oh\big( \frac{g}{\log g} \cdot \log \log n\big)$ bound by analyzing the $\tilde{g}$-\AdvComp setting for $\tilde{g} = (\alpha_1 \cdot (\log n))^{1/k} > g$, since
\[
\frac{\tilde{g}}{\log \tilde{g}} 
 \leq \frac{\tilde{g}}{\log g} 
  = \frac{\tilde{g}}{g} \cdot \frac{g}{\log g} 
 \leq \left( \frac{4\alpha_1}{\alpha_2} \right)^{1/k} \cdot \frac{g}{\log g} = \Oh\left(\frac{g}{\log g}\right).
\]

We will now define the super-exponential potential functions $\Phi_0, \ldots, \Phi_{k-1}$. The base potential function $\Phi_0$ is just an exponential potential (i.e., has a constant smoothing parameter) defined as 
\begin{align} \label{eq:g_adv_phi_0_def}
\Phi_0^s := \Phi_0^s(\alpha_2, z_0) := \sum_{i = 1}^n \Phi_{0, i}^s := \sum_{i = 1}^n \exp\Big(\alpha_2 \cdot (y_i^s - z_0)^+ \Big),
\end{align}
where $\alpha_2 := \frac{\alpha_1}{84}$ and $z_0:=c_5 \cdot g$ for some sufficiently large constant integer $c_5 > 0$ (to be defined in \cref{eq:g_adv_c5_def} in \cref{lem:g_adv_base_case_for_modified_process}). 
Further, we define for any integer $1 \leq j \leq k-1$,
\begin{align} \label{eq:g_adv_phi_j_def}
\Phi_j^s := \Phi_j^s(\alpha_2 \cdot (\log n) \cdot g^{j - k}, z_j) := \sum_{i = 1}^n \Phi_{j, i}^s := \sum_{i = 1}^n \exp\Big(\alpha_2 \cdot (\log n) \cdot g^{j - k} \cdot (y_i^s - z_j)^+ \Big),
\end{align}
where the offsets are given by 
\begin{align}
z_j := c_5 \cdot g + \left\lceil \frac{4}{\alpha_2} \right\rceil \cdot j \cdot g. \label{eq:g_adv_offsets}
\end{align}

Note that for most choices of $g$ and $j$, the smoothing parameter is $\omega(1)$, motivating the term ``super-exponential''. For $g = (\log n)^{1/k}$, the potential functions match in form those defined in \cite{LS21}.

As concrete example, consider $g = (\log n)^{5/12}$, for which we have $k = 3$ and the potential functions are:
\begin{align*}
\Phi_0^s & = \sum_{i = 1}^n \exp\Big(\alpha_2 \cdot (y_i^s - z_0)^+ \Big), \\
\Phi_1^s & = \sum_{i = 1}^n \exp\Big(\alpha_2 \cdot (\log n)^{2/12} \cdot (y_i^s - z_1)^+ \Big), \\
\Phi_2^s & = \sum_{i = 1}^n \exp\Big(\alpha_2 \cdot (\log n)^{7/12} \cdot (y_i^s - z_2)^+ \Big).
\end{align*}
When each of these potentials is $\Oh(n)$, this implies increasingly stronger bounds on the gap: $\Oh(g + \log n)$, $\Oh(2g + (\log n)^{10/12})$ and finally $\Oh( 3g + (\log n)^{5/12} ) = \Oh(g)$.

In general, we will employ this series of potential functions $(\Phi_j)_{j = 0}^{k-1}$ to analyze the process over the time-interval $[m - n \log^5 n, m]$, with the goal of eventually establishing an $\Oh(g \cdot k)$ gap at time $m$. By definition of $k$, this will imply that $\Gap(m) = \Oh\big(\frac{g}{\log g} \cdot \log \log n\big)$. %

\subsection{The Layered Induction Argument}

The next lemma (\cref{lem:new_inductive_step}) formalizes the inductive argument outlined in \cref{sec:super_exp_potentials_def}. It shows that if for all steps $s$ within some suitable time-interval, the number of bins with load at least $z_{j}$ is small, then the number of bins with load at least $z_{j+1}$ is even smaller. This ``even smaller'' is captured by the (non-constant) base of $\Phi_j$, which increases in $j$; however, this comes at the cost of reducing the time-interval slightly by an additive $2n \log^4 n$ term.
Finally, for $j=k-1$, we conclude that at step $m$, there are no bins with load at least $z_k := c_5 \cdot g + \big\lceil\frac{4}{\alpha_2}\big\rceil \cdot k \cdot g$, implying that  $\Gap(m) = \Oh(\frac{g}{\log g} \cdot \log \log n)$ and establishing \cref{thm:g_adv_strongest_bound}.  %

{\renewcommand{\thelem}{\ref{lem:new_inductive_step}}
	\begin{lem}[\textbf{Induction step -- Simplified version, page~\pageref{lem:new_inductive_step}}]
Consider the \GAdvComp setting for any $g= \Omega(1)$ satisfying $(\alpha_1 \log n)^{1/k} < g \leq (\frac{\alpha_2}{4} \log n)^{1/(k-1)}$ for some integer $k \geq 2$, for $\alpha_1, \alpha_2 > 0$ defined in \cref{eq:g_adv_alpha_1_def} and \cref{eq:g_adv_alpha_2_def}. Then, for any integer $1 \leq j \leq k-1$ and any step $m \geq 0$, if it holds that
\[
  \Pro{\bigcap_{s \in [m - 2n(k - j + 1) \cdot \log^4 n,m]} \{ \Phi_{j-1}^{s} = \Oh(n) \} } \geq 1 - \frac{(\log n)^{8(j-1)}}{n^4},
\]
then it also follows that
\[
  \Pro{ \bigcap_{s \in [m - 2n(k - j) \cdot \log^4 n, m]} \{ \Phi_{j}^{s} = \Oh(n) \} } \geq 1 - \frac{(\log n)^{8j}}{n^4}.
\]
	\end{lem} }
	\addtocounter{thm}{-1}

\paragraph{Base case.} The base case follows by strengthening the stabilization lemma for a variant of the $\Lambda$ potential used in \cref{sec:g_adv_g_plus_logn_bound}, so that the potential becomes $\Oh(n)$ every $\Oh(ng)$ steps (instead of $\Lambda$ becoming $\Oh(n)$ every $\Oh(n \cdot (g + \log n))$ steps which was proven in \cref{lem:g_adv_stabilization}). By making the constant $c_5 > 0$ sufficiently large in the offset $c_5g$ of the $\Phi_0$ potential, we obtain concentration for $\Phi_0$ and establish the base case for the layered induction.

{\renewcommand{\thethm}{\ref{thm:g_adv_strong_base_case}}
	\begin{thm}[\textbf{Base case -- Simplified version, page~\pageref{thm:g_adv_strong_base_case}}]
Consider the \GAdvComp setting for any $g \in [1, \Oh(\log n)]$. Then, for any step $m \geq 0$, \[
\Pro{\bigcap_{s\in [m - n\log^5 n, m]} \left\{ \Phi_0^s = \Oh(n) \right\} } \geq 1 - n^{-4}.
\]
	\end{thm} }
	\addtocounter{thm}{-1}

\subsection{Analysis of Super-Exponential Potential Functions}

In order to prove the induction step (\cref{lem:new_inductive_step}), we will need to prove concentration for super-exponential potentials. We will now outline the proof of this and the challenges involved.

For super-exponential potentials, unlike the hyperbolic cosine potential, there exist load configurations, where a super-exponential potential may increase in expectation. We will show that in each step where the probability to allocate to a bin with load at least $z-1$ is sufficiently small, the potential function $\Phi := \Phi(\phi, z)$ drops in expectation over one step. More specifically, we show that the following event is sufficient:
\[
\mathcal{K}^s := \mathcal{K}_{\phi, z}^s(q^s) := \left\{ \forall i \in [n] \colon\  y_i^s \geq z-1 \ \ \Rightarrow \ \ q_i^s \leq \frac{1}{n} \cdot e^{-\phi}\right\},
\]
where $q^s$ is the probability allocation vector used by the process at step $s$.

We prove the following drop inequality in \cref{sec:super_exponential_potentials}.

\newcommand{\GeneralDropInequality}{
Consider any allocation process and any super-exponential potential $\Phi := \Phi(\phi, z)$ with $\phi \in [4, n]$. For any step $s \geq 0$ where $\mathcal{K}^s$ holds, we have that
\[
\Ex{\left. \Phi^{s+1} \,\right|\, \mathfrak{F}^s, \mathcal{K}^{s} } \leq \Phi^{s} \cdot \left(1 - \frac{1}{n} \right) + 2.
\]
}

{\renewcommand{\thelem}{\ref{lem:general_drop_superexponential}}
	\begin{lem}[\textbf{Restated, page~\pageref{lem:general_drop_superexponential}}]
\GeneralDropInequality
	\end{lem} }
	\addtocounter{lem}{-1}

Using this drop inequality we are able to prove the following concentration lemma for super-exponential potentials.

\newcommand{\PreconditionOne}{
\Ex{\left. \Phi_1^{s+1} \,\right|\, \mathfrak{F}^s, \mathcal{K}^{s} } \leq \Phi_1^{s} \cdot \Big(1 - \frac{1}{n} \Big) + 2}
\newcommand{\PreconditionTwo}{
\Ex{\left. \Phi_2^{s+1} \,\right|\, \mathfrak{F}^s, \mathcal{K}^{s} } \leq \Phi_2^{s} \cdot \Big(1 - \frac{1}{n} \Big) + 2}
\newcommand{\PreconditionThree}{
\Pro{\left\{ \Gap(t - 2n \log^4 n) \leq \log^2 n \right\} \cap \bigcap_{s \in [t - 2n \log^4 n, \tilde{t}]} \mathcal{K}^s} \geq 1 - P}

\newcommand{\ThmSuperExponentialPotentialConcentrationWithLabels}[1]{
Consider any allocation process for which there exist super-exponential potential functions $\Phi_1 := \Phi_1(\phi_1, z)$ and $\Phi_2 := \Phi_2(\phi_2, z)$ with integer offset $z := z(n) > 0$ and smoothing parameters $\phi_1, \phi_2 \in (0, (\log n)/6]$ with $\phi_2 \leq \frac{\phi_1}{84}$, such that they satisfy for any step $s \geq 0$,
\ifthenelse{\equal{#1}{1}}{
\begin{align}\label{eq:phi_1_drop_precondition}
\PreconditionOne,
\end{align}}{
\begin{align*}
\PreconditionOne,
\end{align*}}
and
\ifthenelse{\equal{#1}{1}}{
\begin{align}\label{eq:phi_2_drop_precondition}
\PreconditionTwo,
\end{align}}{
\begin{align*}
\PreconditionTwo,
\end{align*}}
where $\mathcal{K}^s := \mathcal{K}_{\phi_1, z}^s$. Further, let $P \in [n^{-4}, 1]$. Then, for any steps $t \geq 0$  and $\tilde{t} \in [t, t + n \log^5 n]$, which satisfy
\ifthenelse{\equal{#1}{1}}{
\begin{align}\label{eq:gap_and_condition_k_precondition}
\PreconditionThree,
\end{align}}{
\begin{align*}
\PreconditionThree,
\end{align*}}
they must also satisfy
\[
\Pro{\bigcap_{s \in [t, \tilde{t}]} \left\{ \Phi_2^s \leq 8n \right\}} \geq 1 - (\log^{8} n) \cdot P.
\]
}
{\renewcommand{\thethm}{\ref{thm:super_exp_potential_concentration}}
	\begin{thm}[\textbf{Restated, page~\pageref{thm:super_exp_potential_concentration}}]
\ThmSuperExponentialPotentialConcentrationWithLabels{0}
	\end{thm} }
	\addtocounter{thm}{-1}

The statement of this theorem concerns steps in $[t - 2n \log^4 n, \tilde{t}]$ with $\tilde{t} \in [t, t + n \log^5 n]$. The interval $[t, \tilde{t}]$ is the \textit{stabilization interval}, i.e., the interval where we want to show that $\Phi_2^s \leq 8n$ for every $s \in [t, \tilde{t}]$. The interval $[t - 2n\log^4 n, t]$ is the \textit{recovery} interval where we will show that \Whp~$\Phi_2$ becomes $\Oh(n)$ at least \textit{once},
provided we start with a ``good'' $\Oh(\log^2 n)$ gap at step $t - 2n \log^4 n$. For both the recovery and stabilization intervals we will condition that the event $\mathcal{K}$ holds at every step.

\paragraph{Using \cref{thm:super_exp_potential_concentration} to prove \cref{lem:new_inductive_step}.} In the concentration lemma (\cref{thm:super_exp_potential_concentration}), the potentials $(\Phi_j)_{j = 0}^{k-1}$ will take the role of $\Phi_2$ and the role of $\Phi_1$ will be taken by the potentials $(\Psi_j)_{j = 0}^{k-1},$ defined as $(\Phi_j)_{j = 0}^{k-1}$ but with a smoothing parameter $\alpha_1$ that is a constant factor larger, i.e., $\alpha_1 := 84 \cdot \alpha_2$. More specifically, \begin{align} \label{eq:g_adv_psi_0_def}
\Psi_0 := \Psi_0(\alpha_1, z_0)
  := \sum_{i = 1}^n \Psi_{0, i}^s := \sum_{i = 1}^n \exp\Big(\alpha_1 \cdot (y_i^s - z_0)^+ \Big),
\end{align}
and for $1 \leq j \leq k - 1$, \begin{align} \label{eq:g_adv_psi_j_def}
\Psi_j^s := \Psi_j^s(\alpha_1 \cdot (\log n) \cdot g^{j-k}, z_j) := \sum_{i=1}^n \Psi_{j,i}^s := \sum_{i=1}^n \exp\left(  \alpha_1 \cdot (\log n) \cdot g^{j-k} \cdot (y_i^s - z_j)^{+} \right).
\end{align}

For the induction step (\cref{lem:new_inductive_step}), it trivially follows that when $\Phi_{j-1}^t = \Oh(n)$, then we also have that $\Gap(t) \leq \log^2 n$ and the following lemma establishes that the event $\mathcal{K}_{\psi_j, z_j}^t$ (associated with the drop of $\Phi_j$ and $\Psi_j$) also holds.

{\renewcommand{\thelem}{\ref{lem:gadv_precondition_satisfied}}
	\begin{lem}[Simplified version, page~\pageref{lem:gadv_precondition_satisfied}]
For any integer $1 \leq j \leq k - 1$ and any step $s \geq 0$ such that $\Phi_{j-1}^s = \Oh(n)$ holds, then also $\mathcal{K}_{\psi_j, z_j}^s$ holds.
	\end{lem} }
	\addtocounter{lem}{-1}

Hence, combining this lemma (\cref{lem:gadv_precondition_satisfied}) with the super-exponential potential drop (\cref{lem:general_drop_superexponential}), the preconditions of super-exponential concentration lemma (\cref{thm:super_exp_potential_concentration}) are satisfied and so layered induction step (\cref{lem:new_inductive_step}) follows.

\paragraph{Road map.} The proof of \cref{thm:g_adv_strongest_bound} is split into three parts. In \cref{sec:g_adv_base_case}, we strengthen the analysis of \cref{sec:g_adv_g_plus_logn_bound} to establish the base case (\cref{thm:g_adv_strong_base_case}). In \cref{sec:super_exponential_potentials}, we prove the general concentration for super-exponential potentials (\cref{thm:super_exp_potential_concentration}). Finally, in \cref{sec:g_adv_layered_induction}, we use this concentration theorem to prove the layered induction step (\cref{lem:new_inductive_step}) and complete the proof for the bound on the gap (\cref{thm:g_adv_strongest_bound}).

\section{Upper Bound of \texorpdfstring{$\Oh(\frac{g}{\log g} \cdot \log \log n)$}{O(g/log g * log log n)} for  \texorpdfstring{$g$-\textsc{Adv-Comp}}{g-Adv-Comp} with \texorpdfstring{$g \leq \log n$}{g <= log n}: Base Case} \label{sec:g_adv_base_case}

\newcommand{\V}{V}
\newcommand{\tV}{\tilde{V}}

In this section, we will obtain for the \GAdvComp setting a stronger guarantee for a variant of the $\Lambda$ potential used in \cref{sec:g_adv_g_plus_logn_bound}. The precise upper bound that we need on $g$ is $g \leq c_6 \log n$, where $c_6 > 0$ is a sufficiently small constant defined as
\begin{align} \label{eq:g_adv_c6_def}
c_6 := \frac{r}{9 \cdot 20 \cdot \tilde{c}_s \cdot \log(2c e^{2\alpha_1})} \leq 1,
\end{align}
where $r \in (0, 1)$ is as defined in \cref{lem:g_adv_many_good_steps}, $\tilde{c}_s := \tilde{c}_s(\alpha_1, c_4, e^{2\alpha_1}c) \geq 1$ as defined in
\cref{lem:g_adv_bounds_on_quadratic} (for $c_4 := 730$), $c > 0$ as defined in \cref{lem:g_adv_good_step_drop} and $\alpha_1 := \frac{1}{6 \kappa} \leq \frac{\alpha}{6}$, for $\kappa > 0$ the constant defined in \cref{eq:g_adv_kappa_def} and $\alpha := \frac{1}{18}$ used in \cref{sec:g_adv_g_plus_logn_bound}.

In \cref{sec:g_adv_g_plus_logn_bound}, we showed that \Whp~$\Lambda^t = \Oh(n)$ at least \textit{once} every $\Oh(n \cdot (g + \log n))$ steps. Here, we will show that \Whp~for any step $t$, we have for \textit{all} steps $s \in [t, t + n \log^5 n]$ that $\Psi_0^s = \Oh(n)$. This will serve as the base case for the layered induction in \cref{sec:g_adv_layered_induction}.

We start by defining the potential function $\V := \V(\alpha_1, c_4g)$ which is a variant of the $\Lambda := \Lambda(\alpha, c_4g)$ potential function (defined in \cref{eq:lambda_def}), with the same offset $c_4g = 2Dg = 730g$, but with a smaller smoothing parameter $\alpha_1 \leq \frac{\alpha}{6}$, \begin{align} \label{eq:g_adv_v_def}
\V^t := \V^t(\alpha_1, c_4g) := \sum_{i = 1}^n \V_i^t := \sum_{i = 1}^n \left[ e^{\alpha_1 (y_i^t - c_4g)^+} + e^{\alpha_1 (- y_i^t - c_4g)^+} \right].
\end{align}

In \cref{sec:g_adv_g_plus_logn_bound}, we proved that \Whp~every $\Oh(n \cdot (g + \log n) )$ steps the potential $\Lambda$ satisfies $\Lambda^s \leq cn$. In this section, we will strengthen this to show that every $\Oh(ng)$ steps (for $g = \Oh(\log n)$) the potential $\V$ satisfies $\V^s \leq e^{\Oh(\alpha_1 g)} \cdot n$. By \cref{lem:g_adv_v_psi0_relation}, this implies the base case of the layered induction in \cref{sec:g_adv_layered_induction}, i.e., that for \emph{all} steps $s \in [m - n\log^5 n, m]$, $\Psi_0^s \leq Cn$ for $C := 2e^{2\alpha_1} \cdot c + 1$ and recalling that $\Psi_0 := \Psi_0(\alpha_1, c_5g)$ for some sufficiently large constant $c_5 > 0$ (to be defined in \cref{eq:g_adv_c5_def} in \cref{lem:g_adv_base_case_for_modified_process}) is given by
\[
\Psi_0^t = \sum_{i = 1}^n \exp\Big(\alpha_1 \cdot (y_i^s - z_0)^+ \Big) = \sum_{i = 1}^n \exp\Big(\alpha_1 \cdot (y_i^s - c_5g)^+ \Big).
\]

The proof follows along the lines of \cref{thm:g_adv_g_plus_logn_gap} in \cref{sec:g_adv_g_plus_logn_bound}, but it further conditions on the gap being $\Oh(g + \log n)$ at every step of the analysis. In particular, by conditioning on $\max_{i \in [n]} |y_i^t| \leq \kappa \cdot (g + \log n)$, we obtain that $|\Delta\V^{t+1}| = \Oh(n^{1/3})$ (\cref{lem:g_adv_v_tilde_lipschitz}), which allows us to apply Azuma's inequality (\cref{lem:azuma}) to deduce that \Whp~$\V$ remains small. This bounded difference condition is similar to the one used in~\cite[Proof of~Theorem 5.3]{LS21}.

\subsection{A Modified Process} \label{sec:g_adv_modified_process}

Let $\mathcal{P}$ be the process in the \GAdvComp setting (with arbitrary $1 \leq g \leq c_6 \log n$) that we want to analyze. We would like to condition on the event that $\mathcal{P}$ satisfies $\max_{i \in [n]} |y_i^t| \leq \kappa \cdot (g + \log n)$, for every step $t$ in an interval of $2n \log^5 n$ steps, which holds \Whp, as implied by \cref{thm:g_adv_g_plus_logn_gap}. 

We implement this conditioning by defining a \textit{modified process} $\mathcal{Q}_{g, r_0} := \mathcal{Q}_{g, r_0}(\mathcal{P})$ for the same $g$ and some arbitrary step $r_0$. Consider the stopping time $\sigma := \inf\{ s \geq r_0 : \max_{i \in [n]} |y_i^s| > \kappa \cdot (g + \log n) \}$, then the process $\mathcal{Q}_{g, r_0}$ is defined so that
\begin{itemize}
  \item in steps $s \in [0, \sigma)$ makes the same allocations as $\mathcal{P}$, and
  \item in steps $s \in [\sigma, \infty)$ allocates to the currently least loaded bin, i.e., it uses the probability allocation vector $r^s = (0, 0, \ldots, 0, 1)$. 
\end{itemize}

Let $y_{\mathcal{Q}}^s$ be the normalized load vector of $\mathcal{Q}_{g, r_0}$ at step $s \geq 0$. By \cref{thm:g_adv_g_plus_logn_gap}, it follows that for any interval $[r_0, m]$ with $m - r_0 \leq n^2$, with high probability the two processes agree\begin{align} \label{eq:modified_process_agrees_with_gadv}
\Pro{\bigcap_{s \in [r_0, m]} \left\{ y_{\mathcal{Q}}^s = y^s \right\} } 
 & \geq \Pro{\bigcap_{s \in [r_0, m]} \left\{ \max_{i \in [n]} \left| y_i^s \right| \leq \kappa \cdot (g + \log n) \right\} } \notag \\
 & \geq 1 - 2 \cdot (ng)^{-9} \cdot n^2 \geq 1 - 2n^{-7}.
\end{align}
The process $\mathcal{Q}_{g, r_0}$ is defined in a way to satisfy the following property:
\begin{itemize}
  \item (\textbf{Property 1}) \label{g_adv_modified_process_property_1} The $\mathcal{Q}_{g, r_0}$ process satisfies the drop inequalities for the potential functions $\Lambda_{\mathcal{Q}}$, $\V_{\mathcal{Q}}$ and $\Upsilon_{\mathcal{Q}}$ (Lemmas~\ref{lem:g_adv_quadratic}, \ref{lem:g_adv_bad_step_increase}~$(ii)$ and~\ref{lem:g_adv_good_step_drop}) for any step $s \geq 0$.  This holds because for any step $s < \sigma$, the process follows $\mathcal{P}$ and so it is an instance of the \GAdvComp setting. For any step $s \geq \sigma$, the process allocates to the currently least loaded bin and therefore minimizes the potential $\Lambda_{\mathcal{Q}}^{s+1}$ given any $\mathfrak{F}^s$,
  which means that $\Lambda_{\mathcal{Q}}^{s+1} \leq \Ex{\Lambda^{s+1} \mid \mathfrak{F}^s}$ and so it trivially satisfies any drop inequality (and similarly for $\V_{\mathcal{Q}}$ and $\Upsilon_{\mathcal{Q}}$).
\end{itemize}
Further, we define the event that the maximum normalized load in absolute value is small at step $r_0$ as,\begin{align} \label{eq:g_adv_z_event_def}
\mathcal{Z}^{r_0} := \left\{  \max_{i \in [n]} \left| y_{\mathcal{Q}, i}^{r_0} \right| \leq \min\left\{ \kappa \cdot (g + \log n),  c_3 g \log(ng)\right\} \right\},
\end{align}
where $c_3 \geq 2$ is the constant defined in \cref{eq:g_adv_c3_def}. We are primarily interested in the $\kappa \cdot (g + \log n)$ bound on the gap and the second bound is only needed for very small values of $g = \Oh(1)$.
When the event $\mathcal{Z}^{r_0}$ holds, then the process $\mathcal{Q}_{g, r_0}$ also satisfies the following property (which ``implements'' the conditioning that the gap is $\Oh(g + \log n)$):
\begin{itemize}
  \item (\textbf{Property 2}) For any step $s \geq r_0$, it follows that \[
    \max_{i \in [n]} \left|y_i^s\right| \leq \kappa \cdot (g + \log n) + 1 \leq 2\kappa \log n,
  \]
  using that $g \leq c_6 \log n$ with $c_6 \leq \frac{1}{4}$ by \cref{eq:g_adv_c6_def}.
  At any step $s \in [r_0, \sigma)$, this holds by the definition of $\sigma$. For any step $s \geq \sigma$, a ball will never be allocated to a bin with $y_i^s > 0$ and in every $n$ steps the at most $n$ bins with load equal to the minimum load (at step $s$) will receive at least one ball each. Hence, over any $n$ steps the maximum absolute normalized load does not increase and in the steps in between this can be larger by at most $1$.

\end{itemize}

\subsection{Preliminaries} 

We now define the adjusted potential $\tV$, analogously to $\tilde{\Lambda}$ in \cref{eq:tilde_lambda}. Note that \cref{lem:g_adv_good_step_drop} with constants $\eps = \frac{1}{12}, c = 12 \cdot 18$ also applies to the potential $\V$, since $\V$ has the same form as $\Lambda$ but a smaller smoothing parameter $\alpha_1 \leq \alpha$. Next, we define the sequence $(\tV_{t_0}^s)_{s \geq t_0} := (\tV_{t_0}^s)_{s \geq t_0}(\alpha_1, c_4g, \eps)$ as $\tV_{t_0}^{t_0} := \V^{t_0}(\alpha_1, c_4g)$ and, for any $s > t_0$,
\begin{equation} \label{eq:g_adv_tilde_v_def}
\tV_{t_0}^s := \V^s(\alpha_1, c_4g) \cdot \mathbf{1}_{\tilde{\mathcal{E}}_{t_0}^{s-1}} \cdot \exp\bigg( - \frac{3\alpha_1 }{n} \cdot B_{t_0}^{s-1} \bigg) \cdot \exp\bigg( + \frac{\alpha_1 \eps}{n} \cdot G_{t_0}^{s-1} \bigg), 
\end{equation}
where $G_t^{s-1}$ (and $B_t^{s-1}$) is the number of good (bad) steps in $[t_0, s-1]$ (as defined in~\cref{sec:g_adv_adjusted_exp_potential}) and
\begin{equation*}%
\tilde{\mathcal{E}}_{t_0}^{s} := \tilde{\mathcal{E}}_{t_0}^{s}(V, c) := \bigcap_{t \in [t_0, s]} \left\{\V^t > cn \right\}.  \end{equation*}
Similarly, to $\tilde{\Lambda}$ in \cref{sec:g_adv_g_plus_logn_bound}, we have that $\tV$ is a super-martingale.

\begin{lem}[\textbf{cf.~\cref{lem:g_adv_lambda_tilde_is_supermartingale}}]\label{lem:g_adv_v_tilde_is_supermartingale}
Consider the $\mathcal{Q}_{g, r_0}$ process for any $g \geq 1$, any step $r_0 \geq 0$, the sequence $(\tV_{t_0}^s)_{s \geq t_0} := (\tV_{t_0}^s)_{s \geq t_0}(\alpha_1, c_4g, \eps)$ for any $t_0 \geq r_0$ with $\alpha_1 > 0$ as defined in \cref{eq:g_adv_alpha_1_def} and $\eps, c_4 > 0$ as defined in \cref{lem:g_adv_good_step_drop}. For any step $s \geq t_0$, we have that,
\[
 \ex{\tV_{t_0}^{s+1} \mid \mathfrak{F}^s} \leq \tV_{t_0}^{s}.
\]
\end{lem}
\begin{proof}
The proof is similar to that of~\cref{lem:g_adv_lambda_tilde_is_supermartingale}, by substituting $\Lambda$ with $\V$ and $\tilde{\Lambda}$ with $\tV$. The drop inequalities follow from \cref{lem:g_adv_bad_step_increase} and \cref{lem:g_adv_good_step_drop}, since $\V$ has the same form as $\Lambda$ and a smaller smoothing parameter $\alpha_1 \leq \alpha$. The process $\mathcal{Q}_{g, r_0}$ also satisfies the drop inequalities by Property 1 (see \cref{g_adv_modified_process_property_1}).
\end{proof}

We defer the proofs of the next two lemmas to \cref{sec:g_adv_v_smoothness_proof}. The first one is a simple smoothness argument for the potential $\V$ defined in \cref{eq:g_adv_v_def}. 

\newcommand{\GAdvVSmoothness}{
Consider the potential $\V := \V(\alpha_1, c_4g)$ for any $\alpha_1 > 0$, any $c_4 > 0$ and any $g \geq 1$. Then, $(i)$~for any step $t \geq 0$, we have that \[
e^{-\alpha_1} \cdot \V^t \leq \V^{t+1} \leq e^{\alpha_1} \cdot \V^t.
\]
Further, $(ii)$~for any $\hat{c} > 0$, for any integer $T > 0$ and any step $t \geq 0$, for which there exist steps $s_0 \in [t - T, t]$ and $s_1 \in [t, t + T]$, such that $\V^{s_0} \leq \hat{c} n$ and $\V^{s_1} \leq \hat{c} n$, we have that \[
\V^t \leq e^{\alpha_1 \frac{T}{n}} \cdot 2\hat{c} n.
\]
}
\begin{lem} \label{lem:g_adv_v_smoothness}
\GAdvVSmoothness
\end{lem}

The following lemma shows that by choosing a large enough offset $c_5 > 0$ in the potential $\Psi_0 := \Psi_0(\alpha_1, c_5g)$ (defined in \cref{eq:g_adv_psi_0_def}), when $\V^t = e^{\Oh(\alpha_1 g)} \cdot cn$, then $\Psi_0^t = \Oh(n)$.

\newcommand{\GAdvVPsiRelation}{
Consider any $c, \hat{c} > 0$ and the potential $\V := \V(\alpha_1, c_4g)$ for any $\alpha_1 > 0$, any $c_4 > 0$ and any $g \geq 1$. Further, consider the potential $\Psi_0 := \Psi_0(\alpha_1, c_5g)$ with offset $c_5 := 2 \cdot \max\{c_4, \hat{c} \}$ and $C := 2e^{2\alpha_1} \cdot c + 1$. Then, for any step $t \geq 0$ with
$
\V^t \leq e^{\alpha_1 \cdot \hat{c} \cdot g} \cdot 2e^{2\alpha_1} cn
$, it holds that
$
\Psi_0^t \leq Cn.
$
}

\begin{lem} \label{lem:g_adv_v_psi0_relation}
\GAdvVPsiRelation
\end{lem}

In comparison to \cref{sec:g_adv_g_plus_logn_bound}, where we proved stabilization over an interval of length $\Omega(n \cdot \max\{\log n, g\})$, here we will be using a shorter interval of length \begin{align} \label{eq:g_adv_tilde_delta_s_def}
\tilde{\Delta}_s := \frac{20 \cdot \tilde{c}_s \cdot \log(2c e^{2\alpha_1})}{\alpha_1 \eps r} \cdot ng,
\end{align}
where constants $\eps := \frac{1}{12}, r := \frac{6}{6 + \eps} > 0$ are as defined in \cref{sec:g_adv_g_plus_logn_bound} and  $\tilde{c}_s := \tilde{c}_s(\alpha_1, c_4, e^{2\alpha_1}c) \geq 1$ is defined in \cref{lem:g_adv_bounds_on_quadratic}.

We now prove the bounded difference condition for the $\tV$ potential. This follows from the second property of $\mathcal{Q}_{g, r_0}$ that the maximum normalized load in absolute value is $\Oh(g + \log n)$ for any step $s \geq r_0$.

\begin{lem} \label{lem:g_adv_v_tilde_lipschitz}
Consider the $\mathcal{Q}_{g, r_0}$ process for any $g \in [1, c_6 \log n]$ with $c_6 > 0$ as defined in \cref{eq:g_adv_c6_def}, any step $r_0 \geq 0$, and $\mathcal{Z}^{r_0}$ as defined in \cref{eq:g_adv_z_event_def}. Further, consider the sequence $(\tV_{t_0}^s)_{s \geq t_0} := (\tV_{t_0}^s)_{s \geq t_0}(\alpha_1, c_4g, \eps)$ for any step $t_0 \geq r_0$ with $\alpha_1 > 0$ as defined in \cref{eq:g_adv_alpha_1_def} and $\eps, c_4 > 0$ as defined in \cref{lem:g_adv_good_step_drop}.
Then, for any step $s \geq t_0 \geq r_0$ we have that $\tilde{V}_{t_0}^{s+1} = 0$ or \[
\left( \left| \Delta\tV_{t_0}^{s+1} \right| \mid \mathcal{Z}^{r_0}, \mathfrak{F}^s \right) \leq 16 \cdot e^{\alpha_1 \eps \cdot \frac{s - t_0}{n}} \cdot n^{1/3}.
\]
\end{lem}
\begin{proof}
Consider an arbitrary step $s \geq r_0$ and assume that the event $\mathcal{Z}^{r_0}$ holds. By Property 2 of the $\mathcal{Q}_{g, r_0}$ process (see \cref{sec:g_adv_modified_process}), we have that
\[
\max_{i \in [n]} \left| y_i^s \right| \leq \kappa \cdot (g + \log n) + 1,
\]
which also implies for $c_4 := 730 > 0$, since $g \leq c_6 \log n \leq \log n$ and $\kappa > 1$ that
\begin{align} \label{eq:g_adv_v_normalised_load_bound}
\max_{i \in [n]} \left\{ (y_i^s - c_4g)^+, (-y_i^s - c_4g)^+ \right\} \leq \kappa \cdot (g + \log n) + 1 - c_4 g \leq 2\kappa \log n.
\end{align}

We will now show that $|\Delta\V^{s+1}| \leq 5 n^{1/3}$ and then use this to bound $|\Delta\tV^{s+1}|$. By \cref{eq:g_adv_v_normalised_load_bound} for any bin $i \in [n]$,  \begin{align*}
\V_i^s \leq 2 \cdot e^{2 \alpha_1 \kappa \log n } = 2 \cdot e^{\frac{1}{3} \cdot \log n } = 2 n^{1/3}, %
\end{align*}
using that $\alpha_1 := \frac{1}{6\kappa}$.
Hence, by aggregating over all bins, $\V^s \leq 2 n^{4/3}$. Furthermore, if the ball at step $s+1$ is allocated to bin $j \in [n]$, then
\[
\Delta\V^{s+1} \leq e^{\alpha_1/n} \cdot \V^s + e^{\alpha_1} \cdot \V_j^s - \V^s \leq \frac{2\alpha_1}{n} \V^s + 2 \cdot 2n^{1/3} \leq \frac{2\alpha_1}{n} \cdot (2n^{4/3}) + 4n^{1/3} \leq 5n^{1/3},
\]
using that $e^{\alpha_1/n} \leq 1 + 2 \cdot \frac{\alpha_1}{n}$ and $e^{\alpha_1} \leq 2$, which both hold as $\alpha_1 \leq 1/4$. Similarly,
\[
\Delta\V^{s+1} \geq e^{-\alpha_1/n} \cdot \V^s - e^{\alpha_1} \cdot \V_j^s - \V^s \geq -\frac{\alpha_1}{n} \V^s - 2\cdot 2n^{1/3} \geq -\frac{\alpha_1}{n} \cdot (2n^{4/3}) - 4n^{1/3} \geq -5n^{1/3},
\]
using that $e^{-\alpha_1/n} \geq 1 - \frac{\alpha_1}{n}$ and $e^{\alpha_1} \leq 2$,  as $\alpha_1 \leq 1/4$.

Now, we turn to upper bounding $|\Delta\tV_{t_0}^{s+1}|$ by proving lower and upper bounds on $\tV_{t_0}^{s+1}$. If $\tV_{t_0}^{s+1} = 0$, then the conclusion follows. Otherwise, since $\tV_{t_0}^{s+1} > 0$, we have that $\mathbf{1}_{\tilde{\mathcal{E}}_{t_0}^s} = \mathbf{1}_{\tilde{\mathcal{E}}_{t_0}^{s-1}} = 1$, so by definition of $\tV$ in \cref{eq:g_adv_tilde_v_def},
\begin{align} \label{eq:g_adv_tilde_v_s_bound}
\tV_{t_0}^{s} = \V^{s} \cdot \exp\Big( - \frac{3\alpha_1 }{n} \cdot B_{t_0}^{s-1} \Big) \cdot \exp\Big( + \frac{\alpha_1 \eps}{n} \cdot G_{t_0}^{s-1} \Big) \leq 2 \cdot e^{\alpha_1 \eps \cdot \frac{s - t_0}{n}}\cdot n^{4/3},
\end{align}
using that $G_{t_0}^{s-1} \leq s - t_0$ and $\V^s \leq 2n^{4/3}$.

Now, we upper bound $\tV_{t_0}^{s+1}$, recalling that $\mathbf{1}_{\tilde{\mathcal{E}}_{t_0}^s} = 1$,
\begin{align} 
\tV_{t_0}^{s+1}
 & = \V^{s+1} \cdot \exp\Big( - \frac{3\alpha_1 }{n} \cdot B_{t_0}^{s} \Big) \cdot \exp\Big( + \frac{\alpha_1 \eps}{n} \cdot G_{t_0}^{s} \Big) \notag \\
 & \leq (\V^{s} + 5 n^{1/3}) \cdot \exp\Big( - \frac{3\alpha_1 }{n} \cdot B_{t_0}^{s} \Big) \cdot \exp\Big( + \frac{\alpha_1 \eps}{n} \cdot G_{t_0}^{s} \Big) \notag \\
 & \leq (\V^{s} + 5 n^{1/3}) \cdot \exp\Big( - \frac{3\alpha_1 }{n} \cdot B_{t_0}^{s-1} \Big) \cdot \exp\Big( + \frac{\alpha_1 \eps}{n} \cdot G_{t_0}^{s-1} \Big) \cdot \exp\Big( \frac{\alpha_1\eps}{n}\Big) \notag \\
 & = \tV_{t_0}^s \cdot \exp\Big( \frac{\alpha_1\eps}{n}\Big) + 5 \cdot n^{1/3} \cdot \exp\Big( - \frac{3\alpha_1 }{n} \cdot B_{t_0}^{s-1} \Big) \cdot \exp\Big( + \frac{\alpha_1 \eps}{n} \cdot G_{t_0}^{s-1} \Big) \cdot \exp\Big( \frac{\alpha_1\eps}{n}\Big) \notag \\
 & \stackrel{(a)}{\leq} \tV_{t_0}^s \cdot \Big( 1 + \frac{3\alpha_1}{n}\Big) + 5 \cdot n^{1/3} \cdot e^{\alpha_1 \eps \cdot \frac{s - t_0}{n} } \cdot 2 \notag \\
 & \stackrel{(b)}{\leq} \tV_{t_0}^s + (2 \cdot e^{\alpha_1 \eps \cdot \frac{s - t_0}{n}}\cdot n^{4/3}) \cdot \frac{3\alpha_1}{n} + 10 \cdot n^{1/3} \cdot e^{\alpha_1 \eps \cdot \frac{s - t_0}{n}} \notag \\
 & \leq \tV_{t_0}^s  + 16 \cdot e^{\alpha_1 \eps \cdot \frac{s - t_0}{n}} \cdot n^{1/3}, \label{eq:tilde_v_ub_1}
\end{align}
using in $(a)$ that $e^{\alpha_1\eps/n} \leq 1 + \frac{3\alpha_1}{n}$ as $\alpha_1 \leq 1$ and $\eps = \frac{1}{12}$, 
$G_{t_0}^{s-1} \leq s - t_0$,
and $e^{\alpha_1\eps/n} \leq 2$ and in $(b)$ using \cref{eq:g_adv_tilde_v_s_bound}.

Similarly, we lower bound $\tV_{t_0}^{s+1}$,%
\begin{align} 
\tV_{t_0}^{s+1}
 & = \V^{s+1} \cdot \exp\Big( - \frac{3\alpha_1 }{n} \cdot B_{t_0}^{s} \Big) \cdot \exp\Big( + \frac{\alpha_1 \eps}{n} \cdot G_{t_0}^{s} \Big) \notag \\
 & \geq (\V^{s} - 5 n^{1/3}) \cdot \exp\Big( - \frac{3\alpha_1 }{n} \cdot B_{t_0}^{s} \Big) \cdot \exp\Big( + \frac{\alpha_1 \eps}{n} \cdot G_{t_0}^{s} \Big) \notag \\
 & \stackrel{(a)}{\geq} (\V^{s} - 5 n^{1/3}) \cdot \exp\Big( - \frac{3\alpha_1 }{n} \cdot B_{t_0}^{s-1} \Big) \cdot \exp\Big( + \frac{\alpha_1 \eps}{n} \cdot G_{t_0}^{s-1} \Big) \cdot \exp\Big(-\frac{3\alpha_1}{n}\Big) \notag \\
 & = \tV_{t_0}^s \cdot \exp\Big( -\frac{3\alpha_1}{n}\Big) - 5 \cdot n^{1/3} \cdot \exp\Big( - \frac{3\alpha_1}{n} \cdot B_{t_0}^{s-1} \Big) \cdot \exp\Big( + \frac{\alpha_1 \eps}{n} \cdot G_{t_0}^{s-1} \Big) \cdot \exp\Big( -\frac{3\alpha_1}{n}\Big) \notag \\
 & \stackrel{(b)}{\geq} \tV_{t_0}^s \cdot \Big( 1 - \frac{3\alpha_1}{n}\Big) - 5 \cdot n^{1/3} \cdot e^{\alpha_1 \eps \cdot \frac{s - t_0}{n}} \notag \\
 & \stackrel{(c)}{\geq} \tV_{t_0}^s - (2 \cdot e^{\alpha_1 \eps \cdot \frac{s - t_0}{n}}\cdot n^{4/3}) \cdot \frac{3\alpha_1}{n} - 5 \cdot n^{1/3} \cdot e^{\alpha_1 \eps \cdot \frac{s - t_0}{n}} \notag \\
 & \geq \tV_{t_0}^s - 16 \cdot e^{\alpha_1 \eps \cdot \frac{s - t_0}{n}} \cdot n^{1/3}, \label{eq:tilde_v_ub_2}
\end{align}
using in $(a)$ that $V^s \geq 2n \geq 5n^{1/3}$ holds deterministically, in $(b)$ that $e^{-3\alpha_1/n} \geq 1 - \frac{3\alpha_1}{n}$ and $G_{t_0}^{s-1}  \leq s - t_0$ and in $(c)$ using \cref{eq:g_adv_tilde_v_s_bound}.

Hence, combining the two upper bounds in \cref{eq:tilde_v_ub_1} and \cref{eq:tilde_v_ub_2}, we conclude that $|\Delta\tV_{t_0}^{s+1}| \leq 16 \cdot e^{\alpha_1 \eps \cdot \frac{s - t_0}{n}} \cdot n^{1/3}$.
\end{proof}

\subsection{Strong Stabilization}

We will now prove the following slightly stronger version of \cref{lem:g_adv_stabilization}, meaning that stabilization is over intervals of length $\Theta(ng)$ instead of $\Theta(n \cdot (g + \log n))$.

\begin{lem}[\textbf{Strong Stabilization}] \label{lem:g_adv_strong_stabilization}
Consider the $\mathcal{Q}_{g, r_0}$ process for any $g \in [1, c_6 \log n]$ for $c_6 > 0$ as defined in \cref{eq:g_adv_c6_def}, any step $r_0 \geq 0$ and $\mathcal{Z}^{r_0}$ as defined in \cref{eq:g_adv_z_event_def}. Then, for the potential $\V := \V(\alpha_1, c_4g)$ with $\alpha_1 > 0$ as defined in \cref{eq:g_adv_alpha_1_def}, $c_4, c > 0$ as defined in \cref{lem:g_adv_good_step_drop} and $\tilde{\Delta}_s> 0$ as defined in \cref{eq:g_adv_tilde_delta_s_def}, it holds that for any step $t_0 \geq r_0$,
\[
\Pro{ \left. \bigcup_{s \in [t_0, t_0 + \tilde{\Delta}_s]} \left\{ \V^{s} \leq e^{\alpha_1} cn \right\} ~\,  \right\vert~ \mathcal{Z}^{r_0}, \mathfrak{F}^{t_0}, e^{\alpha_1} cn < \V^{t_0} \leq e^{2\alpha_1} cn } \geq 1 - n^{-11}.
\]
\end{lem}
\begin{proof}
The proof of this lemma proceeds similarly to that of \cref{lem:g_adv_stabilization}, but we will apply Azuma's inequality for $\tilde{V}_{t_0}^s$ instead of Markov's inequality. However, we cannot directly apply concentration to $\tilde{V}_{t_0}^s$ because the bounded difference condition (\cref{lem:g_adv_v_tilde_lipschitz}) holds only when $\tilde{V}_{t_0}^s$ is positive. So instead we apply it to a stopped random variable $X_{t_0}^s$ to be defined in a way that ensures it is always positive.

Let $t_1 :=t_0 + \tilde{\Delta}_s$. We define the stopping time $\tau := \inf \{ s \geq t_0 \colon \V^s \leq e^{\alpha_1} cn \}$ and for any $s \in [t_0, t_1]$,
\[
  X_{t_0}^{s} := \tV_{t_0}^{s \wedge \tau}.
\]
We will now verify that $X_{t_0}^s > 0$ for all $s \in [t_0, t_1]$. Firstly, consider any $s < \tau$. Since $\V^s > e^{\alpha_1} cn$, by \cref{lem:g_adv_v_smoothness}~$(i)$, we have that $\V^{s+1} \geq V^s \cdot e^{-\alpha_1} >  cn$ and hence $X_{t_0}^{s+1} = \tV_{t_0}^{s+1} > 0$. Secondly, for any $s \geq \tau$, it trivially holds that $X_{t_0}^{s+1} = X_{t_0}^s > 0$.

We proceed to verify the preconditions of Azuma's inequality for super-martingales (\cref{lem:azuma}) for the sequence $(X_{t_0}^{s})_{s \in [t_0, t_1]}$. Firstly, using \cref{lem:g_adv_v_tilde_is_supermartingale} it forms a super-martingale, i.e., that $\ex{X_{t_0}^{s} \mid \mathfrak{F}^{s-1}} \leq X_{t_0}^{s - 1}$. %
Secondly, by~\cref{lem:g_adv_v_tilde_lipschitz}, since $X_{t_0}^{s} > 0$, for any filtration $\mathfrak{F}^{s-1}$ where $\mathcal{Z}^{r_0}$ holds, we have that
\begin{align*}
\left( \bigl| X_{t_0}^{s} - X_{t_0}^{s-1} \bigr| \mid \mathcal{Z}^{r_0}, \mathfrak{F}^{s-1} \right)
 & \leq \left(\bigl| \tV_{t_0}^{s} - \tV_{t_0}^{s-1} \bigr| \mid \mathcal{Z}^{r_0}, \mathfrak{F}^{s-1} \right) \\
 & \leq 16 \cdot e^{\alpha_1 \eps \cdot \frac{\tilde{\Delta}_s}{n}} \cdot n^{1/3} \\ 
 & \stackrel{(a)}{=} (16 \cdot e^{20 \cdot \tilde{c}_s \cdot r^{-1} \cdot \log(2c e^{2\alpha_1}) g}) \cdot n^{1/3} \\
 & \stackrel{(b)}{\leq} 16 \cdot e^{20 \cdot \tilde{c}_s \cdot r^{-1} \cdot \log(2c e^{2\alpha_1}) \cdot c_6 \log n} \cdot n^{1/3} \\ 
 & = 16 \cdot n^{4/9},
\end{align*}
using in $(a)$ that $\tilde{\Delta}_s := \frac{20 \cdot \tilde{c}_s \cdot \log(2c e^{2\alpha_1})}{\alpha_1 \eps r} \cdot ng$ and in $(b)$ that $g \leq c_6 \log n$ and $c_6 := \frac{r}{9 \cdot 20 \cdot \tilde{c}_s \cdot \log(2c e^{2\alpha_1})}$. %

Hence, applying \cref{lem:azuma} for $\lambda := n$, $N := \tilde{\Delta}_s$ and $a_i := 16 \cdot n^{4/9}$, 
\begin{align} \label{eq:x_concentration_ineq}
\Pro{\left. X_{t_0}^{t_1} \geq X_{t_0}^{t_0} + n ~\,  \right\vert~ \mathcal{Z}^{r_0}, \mathfrak{F}^{t_0}, e^{\alpha_1} cn < \V^{t_0} \leq e^{2\alpha_1} cn} 
 & \leq \exp\left( - \frac{n^2}{2 \cdot \tilde{\Delta}_s \cdot (16 \cdot n^{4/9})^2 } \right) \notag \\
 & \leq n^{-\omega(1)}.
\end{align}

As we condition on $\{\V^{t_0} \leq e^{2\alpha_1} cn \}$, we have that \[
\max_{i \in [n]} \left| y_i^{t_0} \right| 
  \leq c_4g + \frac{\log (e^{2\alpha_1} cn)}{\alpha_1} 
  \leq g (\log(ng))^2, 
\]
for sufficiently large $n$, using that $\alpha_1, c_5, c > 0$ are constants. Further, by \cref{lem:g_adv_bounds_on_quadratic}~$(i)$ (since $V$ has the same form as $\Lambda$), we have that $\Upsilon^{t_0} \leq \tilde{c}_s n g^2$, for some constant $\tilde{c}_s := \tilde{c}_s(\alpha_1, c_4, e^{2\alpha_1}c) \geq 1$.

Applying \cref{lem:g_adv_many_good_steps} for $T := \tilde{c}_s n g^2$ (since $T \in [ng^2, o(n^2g^3)]$) and $\hat{c} := \frac{\tilde{\Delta}_s \cdot g}{T} \geq \frac{20}{\alpha_1\eps r} \geq 1$ (since $\alpha_1, \eps, r \leq 1$),
\begin{align}
& \Pro{ G_{t_0}^{t_1-1} \geq r \cdot \tilde{\Delta}_s \;\Big|\; \mathcal{Z}^{r_0}, \mathfrak{F}^{t_0}, e^{\alpha_1} cn < \V^{t_0} \leq e^{2\alpha_1} cn} \notag \\
 &\qquad \geq \Pro{ G_{t_0}^{t_1-1} \geq r \cdot \tilde{\Delta}_s \;\Big|\; \mathfrak{F}^{t_0}, \Upsilon^{t_0} \leq T, \, \max_{i \in [n]} \left| y_i^{t_0} \right| \leq g (\log(ng))^2 } \notag \\ 
 &\qquad \geq 1 - 2 \cdot n^{-12} . \label{eq:g_adv_many_good_steps_strong_stabilisation}
\end{align}
Taking the union bound over \cref{eq:x_concentration_ineq} and \cref{eq:g_adv_many_good_steps_strong_stabilisation}, we have
\[
\Pro{ \left\{ X_{t_0}^{t_1} < X_{t_0}^{t_0} + n \right\} \cap \left\{ G_{t_0}^{t_1-1} \geq r \cdot \tilde{\Delta}_s \right\} \;\Big|\; \mathcal{Z}^{r_0}, \mathfrak{F}^{t_0}, e^{\alpha_1} cn < \V^{t_0} \leq e^{2\alpha_1} cn} \geq 1 - n^{-11}.
\]

Assume that $\{ X_{t_0}^{t_1} < X_{t_0}^{t_0} + n \}$ and $\{ G_{t_0}^{t_1-1} \geq r \cdot \tilde{\Delta}_s \}$ hold. Then, we consider two cases based on whether the stopping was reached. 

\textbf{Case 1 [$\tau \leq t_1$]:} Here, clearly there is an $s \in [t_0, t_1]$, namely $s = \tau$, such that $\V^s \leq  e^{\alpha_1} c n$ and the conclusion follows.

\textbf{Case 2 [$\tau > t_1$]:} Here, using that $\{\tau > t_1\}$ and $\{X_{t_0}^{t_1} < X_{t_0}^{t_0} + n \}$ both hold, it follows that 
\[
\tV_{t_0}^{t_1} < \tV_{t_0}^{t_0} + n,
\]
and so, by definition of $\tV_{t_0}^{t_1}$ (\cref{eq:g_adv_tilde_v_def}),
\[
\V^{t_1} \cdot \mathbf{1}_{\tilde{\mathcal{E}}_{t_0}^{t_1-1}} \cdot \exp\bigg( - \frac{3\alpha_1 }{n} \cdot B_{t_0}^{t_1-1} \bigg) \cdot \exp\bigg( + \frac{\alpha_1 \eps}{n} \cdot G_{t_0}^{t_1-1} \bigg) < \V^{t_0} + n.
\]
By re-arranging and using that $\{ G_{t_0}^{t_1-1} \geq r \cdot \tilde{\Delta}_s \}$ holds, we have that
\begin{align*}
\V^{t_1} \cdot \mathbf{1}_{\tilde{\mathcal{E}}_{t_0}^{t_1-1}} 
& < (\V^{t_0} + n) \cdot \exp\bigg(  \frac{3 \alpha_1}{n} \cdot B_{t_0}^{t_1-1}  -\frac{\alpha_1 \eps}{n} \cdot G_{t_0}^{t_1-1} \bigg) \\
 & \leq (e^{2\alpha_1} cn + n) \cdot \exp\bigg( \frac{3 \alpha_1}{n} \cdot (1 - r) \cdot \tilde{\Delta}_s - \frac{\alpha_1 \eps}{n} \cdot r \cdot \tilde{\Delta}_s \bigg) \\
 &  \stackrel{(a)}{=} 2e^{2\alpha_1} \cdot cn \cdot \exp\bigg( \frac{\alpha_1 \eps}{n} \cdot \frac{r}{2} \cdot \tilde{\Delta}_s - \frac{\alpha_1 \eps}{n} \cdot r \cdot \tilde{\Delta}_s \bigg) \\
 &= 2e^{2\alpha_1} \cdot cn \cdot \exp\bigg( - \frac{ \alpha_1 \eps}{n} \cdot \frac{r}{2} \cdot \tilde{\Delta}_s \bigg) \\
 & = 2e^{2\alpha_1} \cdot cn \cdot \exp\bigg( - \frac{\alpha_1 \eps}{n} \cdot \frac{r}{2} \cdot  \frac{20 \cdot \tilde{c}_s \cdot \log (2ce^{2\alpha_1})}{ \alpha_1 \eps r} \cdot ng \bigg) \\ 
 & \stackrel{(b)}{\leq} 2e^{2\alpha_1} \cdot cn \cdot \exp\left( - 10 \log(2ce^{2\alpha_1}) \right) \\
 & \stackrel{(c)}{\leq} n/2,
\end{align*}
where we used in $(a)$ that $r = \frac{6}{6 + \eps}$ implies $\frac{3\alpha_1}{n} \cdot (1-r) = \frac{3\alpha_1}{n} \cdot \frac{\eps}{6 + \eps} = \frac{\alpha_1\eps}{n} \cdot \frac{r}{2}$ and $c \geq 1$, in $(b)$ that $g \geq 1$, $c \geq 1$ and $\tilde{c}_s \geq 1$ and in $(c)$ that $c = 18 \cdot 12$. Since deterministically we have that $\V^{t_1} \geq n$, it must be that $\mathbf{1}_{\tilde{\mathcal{E}}_{t_0}^{t_1-1}} = 0$, implying that there exists $s \in [t_0, t_1)$ such that $\V^s \leq cn \leq e^{\alpha_1} cn$.
\end{proof}

We will now show that the potential $\V$ becomes small every $\Theta(ng)$ steps. The proof, given in \cref{sec:g_adv_good_v_every_ng_steps_proof}, proceeds similarly to the proof of \cref{lem:g_adv_good_gap_after_good_lambda}.

\newcommand{\GAdvGoodVEveryNgRounds}{
Consider the $\mathcal{Q}_{g, r_0}$ process for any $g \in [1, c_6 \log n]$ for $c_6 > 0$ as defined in \cref{eq:g_adv_c6_def} and any step $r_0 \geq 0$. Then, for the potential $V := V(\alpha_1, c_4g)$ with $\alpha_1$ as defined in \cref{eq:g_adv_alpha_1_def}, $c_4, c > 0$ as defined in \cref{lem:g_adv_good_step_drop} and $\tilde{\Delta}_s> 0$ as defined in \cref{eq:g_adv_tilde_delta_s_def}, it holds that for any step $t_0 \geq r_0$ and $t_1$ such that $t_0 < t_1 \leq t_0 + 2n\log^5 n$,
\[
\Pro{\left.\bigcap_{t \in [t_0, t_1]}\bigcup_{s \in [t, t + \tilde{\Delta}_s]} \left\{ \V^{s} \leq e^{2\alpha_1} cn \right\} ~\right|~ \mathcal{Z}^{r_0}, \mathfrak{F}^{t_0}, \V^{t_0} \leq c n } \geq 1 - n^{-9},
\]
}

\begin{lem} \label{lem:g_adv_good_v_every_ng_steps}
\GAdvGoodVEveryNgRounds
\end{lem}

\subsection{Completing the Proof of Theorem~\ref{thm:g_adv_strong_base_case}}

In this section we will complete the base case proof using the stronger stabilization for the $\V$ potential (see \cref{lem:g_adv_strong_stabilization}) and upper bounding $\Psi_0$ using $\V$ (see \cref{lem:g_adv_v_psi0_relation}). We will first prove the result for the modified process and then relate the results to the \GAdvComp setting.

\begin{lem} \label{lem:g_adv_base_case_for_modified_process}
Consider the $\mathcal{Q}_{g, r_0}$ process for any $g \in [1, c_6 \log n]$ and any step $r_0 \geq - \Delta_r - n \log^5 n$, where $c_6 > 0$ is as defined in \cref{eq:g_adv_c6_def} and $\Delta_r := \Delta_r(g) > 0$ is as in \cref{lem:g_adv_recovery}. Further, let $c > 0$ be as defined in \cref{lem:g_adv_good_step_drop} and $\alpha_1 > 0$ as in \cref{eq:g_adv_alpha_1_def}, then for the constant $C := 2e^{2\alpha_1} \cdot c + 1 \geq 8$ and the potential $\Psi_0 := \Psi_0(\alpha_1, c_5g)$ with the constant integer $c_5 > 0$ (to be defined in \cref{eq:g_adv_c5_def}), we have that, \[
\Pro{\left. \bigcap_{s\in [r_0 + \Delta_r, r_0 + \Delta_r + n \log^5 n]} \left\{ \Psi_0^s \leq Cn \right\} \,\, \right|\,\, \mathfrak{F}^{r_0}, \mathcal{Z}^{r_0} } \geq 1 - n^{-8}.
\]
\end{lem}
\begin{proof}
This proof proceeds similarly to that of \cref{thm:g_adv_g_plus_logn_gap} having a recovery and a stabilization phase to show that the potential $\V := \V(\alpha_1, c_4g)$ stabilizes at $\leq e^{\Oh(\alpha_1 g)} \cdot cn$. By \cref{lem:g_adv_v_psi0_relation}, this implies that $\Psi_0 := \Psi_0(\alpha_1, c_5g)$ stabilizes at $\leq Cn$ for sufficiently large constants $c_5, C > 0$.

For the recovery phase, we will use the potential function $\Lambda := \Lambda(\alpha, c_4g)$ defined in \cref{eq:lambda_def}. More specifically,  we have
\begin{align*}
 & \Pro{ \left. \bigcup_{t \in [r_0, r_0 + \Delta_r]} \left\{\V^{t} \leq cn \right\} \,\, \right|\,\, \mathfrak{F}^{r_0}, \mathcal{Z}^{r_0} } \\
& \qquad \stackrel{(a)}{\geq} \Pro{ \left. \bigcup_{t \in [r_0, r_0 + \Delta_r]} \left\{\Lambda^{t} \leq cn \right\} \,\, \right|\,\, \mathfrak{F}^{r_0}, \mathcal{Z}^{r_0}} \\
 & \qquad \stackrel{(b)}{\geq} \Pro{ \left. \bigcup_{t \in [r_0, r_0 + \Delta_r]} \left\{\Lambda^{t} \leq cn \right\} \,\, \right|\,\, \mathfrak{F}^{r_0}, \max_{i \in [n]} \left| y_i^{r_0}\right| \leq c_3 g \log(ng) } \\
 & \qquad \stackrel{(c)}{\geq} 1 - 3 \cdot (ng)^{-12} \geq 1 - n^{-11},
\end{align*}
using in $(a)$ that $\V^t \leq \Lambda^t$ for any step $t \geq 0$, as $\V$ and $\Lambda$ have the same form, but $\V$ has a smoothing parameter $\alpha_1 < \alpha$, in $(b)$ recalling that $\mathcal{Z}^{r_0} := \{ \max_{i \in [n]} \left| y_i^{r_0} \right| \leq \min\{ \kappa \cdot (g + \log n), c_3 g \log(ng) \} \}$ for $c_3 \geq 2$ the constant defined in \cref{eq:g_adv_c3_def} and in $(c)$ applying \cref{lem:g_adv_recovery}~$(i)$.

Therefore, for the stopping time  $\tau := \inf \{ t \geq r_0 \colon \V^{t} \leq cn \}$, it holds that
\begin{equation} \label{eq:g_adv_stopping_time_for_v_recovery}
 \Pro{\tau \leq r_0 + \Delta_r\,\, \left|\,\, \mathfrak{F}^{r_0}, \mathcal{Z}^{r_0} \right.} \geq 1 - n^{-11}.
\end{equation}

Consider any $t_0 \in [r_0, r_0 + \Delta_r]$ ($t_0$ will play the role of a concrete value of $\tau$). By \cref{lem:g_adv_good_v_every_ng_steps} (for $t_0 := t_0$ and $t_1 := r_0 + \Delta_r + n \log^5 n$, since $t_1 - t_0 \leq n \log^5 n + \Delta_r \leq 2n \log^5 n$), we have that
\[
\Pro{\left.\bigcap_{t \in [t_0, r_0 + \Delta_r + n \log^5 n]}\bigcup_{s \in [t, t + \tilde{\Delta}_s]} \left\{ \V^{s} \leq e^{2\alpha_1} cn \right\} ~\right|~ \mathcal{Z}^{r_0}, \mathfrak{F}^{t_0}, \V^{t_0} \leq c n } \geq 1 - n^{-9}.
\]
When the above event holds, then for every $t \in [t_0, r_0 + \Delta_r + n \log^5 n]$, there exists $s_0 \in [t - \tilde{\Delta}_s, t]$ and $s_1 \in [t, t + \tilde{\Delta}_s]$ such that $\V^{s_0} \leq e^{2\alpha_1} cn$ and $\V^{s_1} \leq e^{2\alpha_1} cn$, using that for any $t \in [t_0, t_0 + \tilde{\Delta}_s]$ we can set $s_0 = t_0$ since by the conditioning $V^{t_0} \leq cn \leq e^{2\alpha_1}cn$. So, by \cref{lem:g_adv_v_smoothness}~$(ii)$ (for $\hat{c} := e^{2\alpha_1} c$ and $\Delta := \tilde{\Delta}_s$), it follows that
\[
\Pro{\left. \bigcap_{t \in [t_0, r_0 + \Delta_r + n \log^5 n]} \left\{ \V^{t} \leq e^{\alpha_1 \frac{\tilde{\Delta}_s}{n}} \cdot 2e^{2\alpha_1}cn \right\} ~\right|~ \mathcal{Z}^{r_0}, \mathfrak{F}^{t_0}, \V^{t_0} \leq c n } \geq 1 - n^{-9}.
\]
Next, adjusting the range of the big intersection using that $t_0 \leq r_0 + \Delta_r$, it follows that
\[
\Pro{\left. \bigcap_{t \in [r_0 + \Delta_r, r_0 + \Delta_r + n \log^5 n]} \left\{ \V^{t} \leq e^{\alpha_1 \frac{\tilde{\Delta}_s}{n}} \cdot 2e^{2\alpha_1}cn \right\} ~\right|~ \mathcal{Z}^{r_0}, \mathfrak{F}^{t_0}, \V^{t_0} \leq c n } \geq 1 - n^{-9}.
\]
By \cref{lem:g_adv_v_psi0_relation}
(for $\hat{c} := \big\lceil \frac{\tilde{\Delta}_s}{ng} \big\rceil$), we conclude that for $\Psi_0 := \Psi_0(\alpha_1, c_5 g)$ with constant integer 
\begin{align} \label{eq:g_adv_c5_def}
c_5 := 2 \cdot \max \left\{ c_4, \left\lceil \frac{\tilde{\Delta}_s}{ng} \right\rceil \right\}
 = 2 \cdot \max \left\{ c_4, \left\lceil \frac{20 \cdot \tilde{c}_s \cdot \log(2c e^{2\alpha_1})}{\alpha_1 \eps r} \right\rceil \right\},
\end{align}
and for the constant $C := 2e^{2\alpha_1} \cdot c + 1$, it holds that
\begin{align} \label{eq:g_adv_psi_0_concentration}
\Pro{\left. \bigcap_{t \in [r_0 + \Delta_r, r_0 + \Delta_r + n \log^5 n]} \left\{ \Psi_0^{t} \leq Cn \right\} \, \right| \, \mathcal{Z}^{r_0}, \mathfrak{F}^{t_0}, \V^{t_0} \leq c n } \geq 1 - n^{-9}.
\end{align}
Finally, 
\begin{align*}
 & \Pro { \left. \bigcap_{t \in [r_0 + \Delta_r, r_0 + \Delta_r + n \log^5 n]} \left\{ \Psi_0^{t} \leq Cn \right\} \,\, \right|\,\, \mathfrak{F}^{r_0}, \mathcal{Z}^{r_0} } \\
  & \qquad \geq \sum_{t_0=r_0}^{r_0 + \Delta_r} \Pro{\left. \bigcap_{t \in [r_0 + \Delta_r, r_0 + \Delta_r + n \log^5 n]} \left\{ \Psi_0^{t} \leq Cn \right\}  ~\right|~ \mathfrak{F}^{r_0},\mathcal{Z}^{r_0}, \tau = t_0 } \cdot \Pro{ \tau = t_0 \,\left|\, \mathfrak{F}^{r_0},\mathcal{Z}^{r_0} \right.}\\
  & \qquad \geq \sum_{t_0=r_0}^{r_0 + \Delta_r} \Pro{ \left. \bigcap_{t \in [r_0 + \Delta_r, m]} \left\{ \Psi_0^{t} \leq Cn \right\}  ~\right|~ \mathcal{Z}^{r_0}, \mathfrak{F}^{t_0}, \V^{t_0} \leq cn } \cdot \Pro{ \tau = t_0 \,\left|\, \mathfrak{F}^{r_0},\mathcal{Z}^{r_0} \right.} \\
  & \qquad \!\!\!\! \stackrel{(\text{\ref{eq:g_adv_psi_0_concentration}})}{\geq} \left( 1 - n^{-9} \right) \cdot \Pro{ \tau \leq r_0 + \Delta_r \,\left|\, \mathfrak{F}^{r_0},\mathcal{Z}^{r_0} \right.} \\
  & \qquad \!\!\!\! \stackrel{(\text{\ref{eq:g_adv_stopping_time_for_v_recovery}})}{\geq} \left(1 - n^{-9} \right) \cdot \left(1 - n^{-11} \right) \geq 1 - n^{-8}, 
\end{align*}
This concludes the claim for $r_0 \geq 0$.

If $r_0 < 0$, then deterministically $\tau = 0$, since $\V^{\tau} = n \leq cn$. Hence, the rest of the proof follows for an interval of length at most $2n \log^5 n$.
\end{proof}

\begin{thm}[\textbf{Base case}] \label{thm:g_adv_strong_base_case}
Consider the \GAdvComp setting for any $g \in [1, c_6 \log n]$, where $c_6 > 0$ is as defined in \cref{eq:g_adv_c6_def}. For constant $C := 2e^{2\alpha_1} \cdot c + 1 \geq 8$ with $\alpha_1 > 0$ as defined in \cref{eq:g_adv_alpha_1_def} and $c > 0$ as in \cref{lem:g_adv_good_step_drop}, and the potential $\Phi_0 := \Phi_0(\alpha_2, c_5g)$ with $\alpha_2 > 0$ as in \cref{eq:g_adv_alpha_2_def}, and constant integer $c_5 > 0$ as in \cref{eq:g_adv_c5_def}, we have that for any step $m \geq 0$, \[
\Pro{\bigcap_{s\in [m - n\log^5 n, m]} \left\{ \Phi_0^s \leq Cn \right\} } \geq 1 - n^{-4}.
\]
\end{thm}
\begin{proof}
Let $\mathcal{P}$ be the original process in the \GAdvComp setting. Consider the modified process $\mathcal{Q}_{g, r_0}$ for $r_0 := m - n \log^5 n - \Delta_r$ where $\Delta_r := \Theta(ng \cdot (\log(ng))^2)$ is as defined in \cref{lem:g_adv_recovery}. Let $y_{\mathcal{Q}}$ be its normalized load vector and $\Psi_{\mathcal{Q}, 0}$ be its $\Psi_{0}$ potential function. Recall from \cref{eq:modified_process_agrees_with_gadv} that the two processes $\mathcal{Q}_{g, r_0}$ and $\mathcal{P}$ agree at every step \Whp, since $m - r_0 \leq 2n \log^5 n$,
\begin{align} \label{eq:modified_process_agrees_with_gadv_new}
\Pro{\bigcap_{s \in [r_0, m]} \left\{ y_{\mathcal{Q}}^s = y^s \right\} } \geq 1 - 2n^{-7}.
\end{align}
Taking the union bound of the conclusions in \cref{thm:g_adv_g_plus_logn_gap} and \cref{thm:g_adv_warm_up_gap}~$(iii)$, we have that 
\begin{align} \label{eq:r0_gap_is_g_plus_logn}
\Pro{\mathcal{Z}^{r_0}} = \Pro{\max_{i \in [n]} \left| y_i^{r_0} \right| \leq 
\min\{ \kappa \cdot (g + \log n), c_3 g \log (ng) \} } \geq 1 - 3 \cdot (ng)^{-9}.
\end{align}
By \cref{lem:g_adv_base_case_for_modified_process} we have that,\[
\Pro{\left. \bigcap_{s\in [m - n\log^5 n, m]} \left\{ \Psi_{\mathcal{Q}, 0}^s \leq Cn \right\} \,\, \right|\,\, \mathfrak{F}^{r_0}, \mathcal{Z}^{r_0} } \geq 1 - n^{-8}.
\]
By combining with \cref{eq:r0_gap_is_g_plus_logn},
\[
\Pro{\bigcap_{s\in [m - n\log^5 n, m]} \left\{ \Psi_{\mathcal{Q}, 0}^s \leq Cn \right\} } \geq \left( 1 - n^{-8} \right) \cdot \left( 1 - 3 \cdot (ng)^{-9} \right) \geq 1 - n^{-7}.
\]
Taking the union bound with \cref{eq:modified_process_agrees_with_gadv_new}, we get that%
\[
\Pro{\bigcap_{s \in [r_0, m]} \left\{ y_{\mathcal{Q}}^s = y^s \right\} \cap \bigcap_{s\in [m - n\log^5 n, m]} \left\{ \Psi_{\mathcal{Q}, 0}^s \leq Cn \right\} } \geq 1 - n^{-7} - 2n^{-7} \geq 1 - n^{-6}.
\]
When this event holds we have that $\Psi_{\mathcal{Q}, 0}^s = \Psi_0^s$ for every step $s \in [r_0, m]$ and hence we can deduce for the original process $\mathcal{P}$ that 
\[
\Pro{\bigcap_{s\in [m - n\log^5 n, m]} \left\{ \Psi_0^s \leq Cn \right\} } \geq 1 - n^{-6}.
\]
Finally, since $\alpha_2 \leq \alpha_1$, we have that $\Phi_0^s \leq \Psi_0^s$ for any step $s \geq 0$ and hence, the conclusion follows.
\end{proof}

\section{Super-Exponential Potential Functions}  \label{sec:super_exponential_potentials}

In this section, we will analyze super-exponential potential functions, i.e., exponential potentials with $\Omega(1)$ smoothing parameters. In \cref{sec:super_exponential_potentials_outline}, we recall the definition of the super-exponential potentials and give an overview of the theorems we prove. In \cref{sec:general_drop_inequality}, we give a sufficient condition for a super-exponential potential to satisfy a drop inequality over one step, and in \cref{sec:super_exponential_potentials_concentration}, we prove the concentration theorem for super-exponential potentials.

\subsection{Outline} \label{sec:super_exponential_potentials_outline}

Recall from \cref{defi:super_exponential_potential} that a super-exponential potential function is defined as 
\[
\Phi^t := \Phi^t(\phi, z) := \sum_{i = 1}^n \Phi_i^t := \sum_{i = 1}^n e^{\phi \cdot (y_i^t - z)^+},
\]
for smoothing parameter $\phi > 0$ and integer offset $z := z(n) > 0$.

Note that if $\Phi^t = \Oh(\poly(n))$ at some step $t \geq 0$, then\[
\Gap(t) = \Oh\left( z + \frac{\log n}{\phi} \right).
\]
There are two differences in the \textit{form} compared to the hyperbolic cosine potential $\Gamma$ that we used in \cref{sec:g_adv_warm_up}: $(i)$ there is no underloaded component, as \Whp, its contribution would be $\omega(n)$ for any process making a constant number of samples in each step\footnote{It follows by a coupon collector's argument that the minimum load is \Whp~$\frac{t}{n}-\Omega(\log n)$, for sufficiently large $t$.} and $(ii)$ there is this $(\ldots)^+$ operation which is not essential, but simplifies some of the derivations.

However, the \textit{main} difference is that there exist load vectors where the potential \textit{could increase} in expectation, even when sufficiently large.\footnote{An example of such a configuration for the \TwoChoice process without noise is one where $n-1$ bins have normalized loads $z + 10/\phi$ and one bin is underloaded.} Here, we will prove a sufficient condition for the potential $\Phi$ to drop in expectation over one step; requiring that the probability to allocate to a bin with normalized load at least $z-1$ is sufficiently small. More specifically, the event that we require is the following:
\[
\mathcal{K}^s := \mathcal{K}_{\phi, z}^s(q^s) := \left\{ \forall i \in [n] \colon\  y_i^s \geq z-1 \ \ \Rightarrow \ \ q_i^s \leq \frac{1}{n} \cdot e^{-\phi}\right\},
\]
where $q^s$ is the probability allocation vector used by the process at step $s$.

{\renewcommand{\thelem}{\ref{lem:general_drop_superexponential}}
	\begin{lem}[Restated, page~\pageref{lem:general_drop_superexponential}]
\GeneralDropInequality
	\end{lem} }
	\addtocounter{lem}{-1}

We will try to establish that this event $\mathcal{K}^s$ holds for a sufficiently long interval and then show that in this interval the potential $\Phi$ becomes small. In the analysis of \GAdvComp, this event will arise from the concentration of the hyperbolic cosine potential or from that of a super-exponential potential with smaller smoothing parameter. Now we are ready to state the main theorem.

{\renewcommand{\thethm}{\ref{thm:super_exp_potential_concentration}}
	\begin{thm}[Restated, page~\pageref{thm:super_exp_potential_concentration}]
\ThmSuperExponentialPotentialConcentrationWithLabels{1}
	\end{thm} }
	\addtocounter{thm}{-1}

The statement of this theorem concerns steps in $[t - 2n \log^4 n, \tilde{t}]$ with $\tilde{t} \in [t, t + n \log^5 n]$. The interval $[t, \tilde{t}]$ is the \textit{stabilization interval}, i.e., the interval where we want to show that $\Phi_2^s \leq 8n$ for every $s \in [t, \tilde{t}]$. The interval $[t - 2n\log^4 n, t]$ is the \textit{recovery} interval where we will show that \Whp~$\Phi_2$ becomes $\Oh(n)$ at least \textit{once},
provided we start with a ``weak'' $\leq \log^2 n$ gap at step $t - 2n \log^4 n$. For both the recovery and stabilization intervals we will condition on the event $\mathcal{K}$ holding at every step.

We will write $\mathcal{H}^{t - 2n \log^4 n}$ as a shorthand for the event $\big\{\!\Gap(t - 2n \log^4 n) \leq \log^2 n \big\}$.

\paragraph{Dealing with negative steps.} In the analysis, we deal with negative (integer) steps, by extending the definition of the process to allocate zero weight balls in steps $t < 0$, i.e., we have that $x^t = 0$. Hence, for any $t < 0$ it also follows deterministically that $\Phi_1^t = \Phi_2^t = n$ and so the potentials trivially satisfy the drop inequalities \cref{eq:phi_1_drop_precondition} and \cref{eq:phi_2_drop_precondition} at these steps.

\subsection{General Drop Inequality}
\label{sec:general_drop_inequality}

We now show that the drop inequality is satisfied in every step $s$ where $\mathcal{K}^s$ holds. 

\begin{lem}[General drop inequality] \label{lem:general_drop_superexponential} 
\GeneralDropInequality
\end{lem}
\begin{proof}
We consider the following three cases for the contribution of a bin $i \in [n]$:

\textbf{Case 1 [$y_i^s < z - 1$]:} The contribution of $i$ will remain $\Phi_{i}^{s+1} = \Phi_{i}^{s} = 1$, even if a ball is allocated to bin $i$. Hence,
\[
 \Ex{\left. \Phi_{i}^{s+1} \,\right\vert\, \mathfrak{F}^s, \mathcal{K}^s} 
 = \Phi_{i}^s
 = \Phi_{i}^s \cdot \left( 1 - \frac{1}{n}\right) + \frac{1}{n}.
\]

\textbf{Case 2 [$y_i^s \in [z -1, z]$]:} By the condition $\mathcal{K}^s$, the probability of allocating a ball to bin $i$ with $y_i^s \geq z - 1$ is $q_i^s \leq \frac{1}{n} \cdot e^{-\phi}$. Hence, the expected contribution of this bin is at most
\begin{align*}
\Ex{\left. \Phi_{i}^{s+1} \,\right\vert\, \mathfrak{F}^s, \mathcal{K}^s} 
 & \leq e^{\phi} \cdot q_i^s + \Phi_{i}^s \cdot (1-q_i^s) \\
 & \leq e^{\phi} \cdot \frac{1}{n} \cdot e^{-\phi} + \Phi_{i}^s \\
 & = \frac{1}{n} + \Phi_{i}^s \\
 & = \Phi_{i}^s \cdot \left( 1 - \frac{1}{n}\right) + \frac{2}{n},
\end{align*}
using in the last equation that $\Phi_{i}^s = 1$.

\textbf{Case 3 [$y_i^s > z$]:} Again, by the condition $\mathcal{K}^s$, the probability of allocating a ball to bin $i$ with $y_i^s > z$ is $q_i^s \leq \frac{1}{n} \cdot e^{-\phi}$. Hence,
\begin{align*}
\Ex{\left. \Phi_{i}^{s+1} \,\right\vert\, \mathfrak{F}^s, \mathcal{K}^s}
 & \stackrel{(a)}{=} \Phi_{i}^s \cdot e^{\phi \cdot (1  - 1/n)} \cdot q_i^s + \Phi_{i}^s \cdot e^{-\phi/n} \cdot (1-q_i^s) \\
 & \leq \Phi_{i}^s \cdot e^{\phi \cdot (1  - 1/n)} \cdot q_i^s + \Phi_{i}^s \cdot e^{-\phi/n} \\
 & = \Phi_{i}^s \cdot e^{-\phi /n} \cdot \Big(1 + e^{\phi} \cdot q_i^s \Big) \\
 & \stackrel{(b)}{\leq} \Phi_{i}^s \cdot \left( 1 - \frac{\phi}{2n} \right) \cdot \Big( 1 + e^{\phi} \cdot \frac{1}{n} \cdot e^{-\phi} \Big) \\
 & \stackrel{(c)}{\leq} \Phi_{i}^s \cdot \left(1 - \frac{2}{n}\right) \cdot \left(1 + \frac{1}{n}\right) \\
 & \leq \Phi_{i}^s \cdot \left( 1 - \frac{1}{n}\right),
\end{align*}
using in $(a)$ that $y_i^s > z$ implies that $y_i^s \geq z + \frac{1}{n}$ since $z$ is an integer, in $(b)$ that $e^u \leq 1 + \frac{1}{2} u$ (for any $-1.5 \leq u < 0$) and that $\phi \leq n$ and in $(c)$ that $\phi \geq 4$.

Aggregating over the three cases, we get the claim:
\begin{align*}
 \Ex{\left. \Phi^{s+1} \,\right|\, \mathfrak{F}^s, \mathcal{K}^s} 
 & = \sum_{i = 1}^n \Ex{\left. \Phi_{i}^{s+1} \,\right|\, \mathfrak{F}^s, \mathcal{K}^s} \leq \sum_{i=1}^n \left(\Phi_{i}^s \cdot \left( 1 - \frac{1}{n} \right) + \frac{2}{n}\right) = \Phi^s \cdot \left( 1 - \frac{1}{n} \right) + 2. \qedhere
\end{align*}
\end{proof}

\subsection{Concentration} \label{sec:super_exponential_potentials_concentration}

\subsubsection{Proof Outline of Theorem~\ref{thm:super_exp_potential_concentration}}

We will now give a summary of the main technical steps in the proof of \cref{thm:super_exp_potential_concentration} (an illustration of the key steps is shown in \cref{fig:super_exp_concentration_proof_outline}). The proof uses an interplay between two instances $\Phi_1 := \Phi_1(\phi_1, z)$ and $\Phi_2 := \Phi_2(\phi_2, z)$ of the super-exponential potential function with $\phi_2 \leq \phi_1/84$, such that in steps $s \geq 0$ when $\Phi_1^s$ is small, then the change of $\Phi_2^s$ is very small.

\paragraph{Recovery.} By the third precondition \cref{eq:gap_and_condition_k_precondition} of \cref{thm:super_exp_potential_concentration}, we start with $\Gap(t - 2n \log^4 n) \leq \log^2 n$, which implies that $\Phi_1^{t - 2n \log^4 n} \leq e^{\frac{1}{2} \cdot \log^4 n}$ (\cref{clm:weak_gap_implies_phi_1_small}). Using the drop inequality for the potential $\Phi_1$ (first precondition \cref{eq:phi_1_drop_precondition}), it follows that $\ex{\Phi_1^{s}} \leq 6n$, for any step $s \in [t - n \log^4 n, \tilde{t}]$ (\cref{lem:newexists_s_st_ex_psi_linear}).
By using Markov's inequality and a union bound, we can deduce that \Whp~$\Phi_1^{s} \leq 6n^{12}$ for all steps $s \in [t - n \log^4 n, \tilde{t}]$. By a simple interplay between two potentials, this implies $\Phi_2^{t - n \log^4 n} \leq n^{7/6}$ (\cref{lem:psi_potential_poly_implies}~$(i)$). Now using a drop inequality for the potential $\Phi_2$ (second precondition \cref{eq:phi_2_drop_precondition}), guarantees that \Whp~$\Phi_2^{r_0} \leq 6n$ for \emph{some single} step $r_0 \in [t - n \log^4 n, t]$ (\cref{lem:recovery_phi_2}). 

\paragraph{Stabilization.} To obtain the stronger statement which holds \emph{for all} steps $s \in [t, \tilde{t}]$, we will use a concentration inequality. The key point is that for any step $r$ with $\Phi_1^{r} \leq 6n^{12}$ the absolute difference $|\Phi_2^{r+1} - \Phi_2^{r}|$ is at most $n^{1/3}$, because $\phi_2 \leq \frac{\phi_1}{84}$ (by preconditions \cref{eq:phi_1_drop_precondition} and \cref{eq:phi_2_drop_precondition}). This is crucial for applying the Azuma's inequality for super-martingales (\cref{lem:azuma}) to $\Phi_2$ which yields that $\Phi_2^s \leq 8n$ for all steps $s \in [t, \tilde{t}]$ using a smoothing argument (\cref{clm:phi_j_does_not_drop_quickly}).

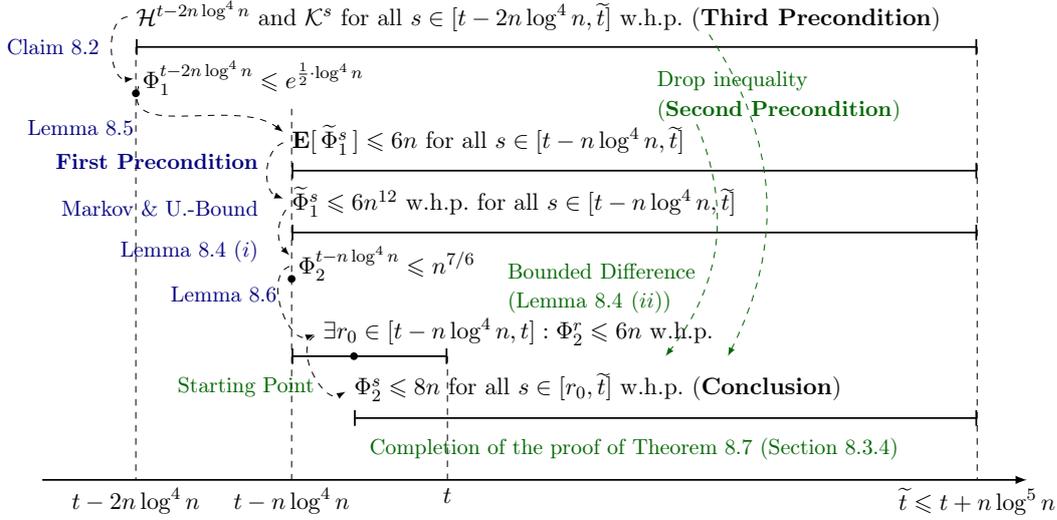
\begin{figure}[H]
\scalebox{0.82}{
\begin{tikzpicture}[
  txt/.style={anchor=west,inner sep=0pt},
  Dot/.style={circle,fill,inner sep=1.25pt},
  implies/.style={-latex,dashed}]

\def\betaO{0}
\def\betaA{1}
\def\betaB{3.5}
\def\betaC{6}
\def\End{14.5}
\def\yO{-8}

\draw[dashed] (\betaA, -1) -- (\betaA, \yO);
\node[anchor=north] at (\betaA, \yO) {$t - 2n \log^4 n$};
\draw[dashed] (\betaB, -2) -- (\betaB, \yO);
\node[anchor=north] at (\betaB, \yO) {$t - n \log^4 n$};
\draw[dashed] (\betaC, -6) -- (\betaC, \yO);
\node[anchor=north] at (\betaC, \yO) {$t$};
\draw[dashed] (\End, -1) -- (\End, \yO);
\node[anchor=north] at (\End, \yO) {$\tilde{t} \leq t + n \log^5 n$};

\node (indstep) at (\betaA,-0.5) {};
\node[txt] at (indstep) {$\mathcal{H}^{t - 2n \log^4 n}$ and $\mathcal{K}^{s}$ for all $s \in [t - 2n \log^4 n, \tilde{t}]$ \Whp~(\textbf{Third Precondition})};
\draw[|-|, thick] (\betaA, -1) -- (\End, -1) ;

\node (PsiSmall) at (\betaA+0.1,-1.5) {};
\node[txt] at (PsiSmall) {$\Phi_1^{t - 2n \log^4 n} \leq e^{\frac{1}{2} \cdot \log^4 n}$};
\node[Dot] at (\betaA, -1.75){};

\node (ExPsiLinear) at (\betaB, -2.5) {};
\node[txt] at (\betaB, -2.5) {$\ex{\tilde{\Phi}_1^{s}} \leq 6n$ for all $s \in [t - n \log^4 n, \tilde{t}]$};
\draw[|-|, thick] (\betaB, -3) -- (\End, -3);

\node (PsiPoly) at (\betaB, -3.5) {};
\node[txt] at (PsiPoly) {$\tilde{\Phi}_1^{s} \leq 6n^{12}$ \Whp~for all $s \in [t - n \log^4 n, \tilde{t}]$};
\draw[|-|, thick] (\betaB, -4) -- (\End, -4);

\node (PhiPoly) at (\betaB+0.1,-4.5) {};
\node[txt] at (\betaB+0.1,-4.5) {$\Phi_2^{t - n \log^4 n} \leq n^{7/6}$};
\node[Dot] at (\betaB, -4.75){};

\node (ExistsPhiLinear) at (\betaB+0.5,-5.6) {};
\node[txt] at (ExistsPhiLinear) {$\exists r_0 \in [t - n \log^4 n, t]: \Phi_2^{r} \leq 6n$ \Whp};
\node[Dot] at (\betaB+1, -6){};
\draw[|-|, thick] (\betaB, -6) -- (\betaC, -6);

\node (PhiLinear) at (\betaB+1,-6.5) {};
\node[txt] at (\betaB+1,-6.5) {$\Phi_2^{s} \leq 8n$ for all $s \in [r_0, \tilde{t}]$ \Whp~(\textbf{Conclusion})};
\draw[|-|, thick] (\betaB+1, -7) -- (\End, -7);

\draw[-latex, thick] (\betaO - 0.5, \yO) -- (\End + .8, \yO);

\draw[implies] (indstep) edge[bend right=90] (PsiSmall);
\node[anchor=east, black!50!blue] at (\betaA - 0.45, -1) {\small \cref{clm:weak_gap_implies_phi_1_small}};

\draw[implies] (PsiSmall) edge[bend right=90,in=160] (ExPsiLinear);
\node[anchor=east, black!50!blue] at (\betaA + 0.1, -2.3) {\small \cref{lem:newexists_s_st_ex_psi_linear}};

\node[anchor=east, black!50!blue] at (\betaB - 0.4, -2.85) {\small \textbf{First Precondition}};

\draw[implies] (ExPsiLinear) edge[bend right=70] (PsiPoly);
\node[anchor=east, black!50!blue] at (\betaB - 0.4, -3.6) {\small Markov \& U.-Bound};

\draw[implies] (PsiPoly) edge[bend right=40] (PhiPoly);
\node[anchor=east, black!50!blue] at (\betaB - 0.4, -4.3) {\small \cref{lem:psi_potential_poly_implies}~$(i)$};

\draw[implies] (PhiPoly) edge[bend right=90] (ExistsPhiLinear);
\node[anchor=east, black!50!blue] at (\betaB - 0.1, -5) {\small \cref{lem:recovery_phi_2}};

\draw[implies] (ExistsPhiLinear) edge[bend right=90] (PhiLinear);
\node[anchor=east, black!60!green] at (\betaB +0.5, -6.5) {\small Starting Point};

\draw[implies,black!60!green] (\betaC + 4, -2.25) to[bend left=30] (\betaC + 3.5, -6);
\node[text width=4cm, anchor=east,black!60!green] at (\betaC + 5.1, -4.9) {\small Bounded Difference (\cref{lem:psi_potential_poly_implies}~$(ii)$)};

\draw[implies,black!60!green] (\betaC + 4.2, -0.8) to[bend left=30] (\betaC + 4.5, -6);
\node[anchor=east,black!60!green, text width=4cm] at (\betaC + 7.5, -1.8) {\small Drop inequality \\ (\textbf{Second Precondition})};

\node[black!60!green] at (\betaC + 3, -7.5) {\small %
Completion of the proof of \cref{thm:super_exp_potential_concentration} (\cref{sec:super_exp_concentration_proof_completion})};

\end{tikzpicture}}
\caption{Outline for the proof of \cref{thm:super_exp_potential_concentration}. Results in blue are given in \cref{sec:super_exp_concentration_deterministic_claims,sec:super_exp_concentration_probabilistic_drop}, while results in green are used in the completion of the proof in \cref{sec:super_exp_concentration_proof_completion}. }
\label{fig:super_exp_concentration_proof_outline}
\end{figure}

\subsubsection{Deterministic Relations between the Potential Functions}\label{sec:super_exp_concentration_deterministic_claims}

We collect several basic facts about the super-exponential potential functions $\Phi_1 := \Phi_1(\phi_1, z)$ and $\Phi_2 := \Phi_2(\phi_2, z)$ satisfying the preconditions of \cref{thm:super_exp_potential_concentration}.

We start with a simple upper bound on $\Phi_1^s$ using a weak upper bound on the gap at step $s$.

\begin{clm} \label{clm:weak_gap_implies_phi_1_small} For any step $s \geq 0$ where $\Gap(s) \leq \log^2 n$, we have that
$
\Phi_1^{s} \leq e^{\frac{1}{2} \cdot \log^4 n}.
$
\end{clm}
\begin{proof}
Since $\phi_1 \leq (\log n)/6$, we have that
$
\Phi_1^{s} = \sum_{i = 1}^n e^{\phi_1 \cdot (y_i^s - z)^+} \leq n \cdot e^{\phi_1 \cdot \log^2 n} \leq e^{\frac{1}{2} \cdot \log^4 n}. 
$
\end{proof}

The next claim is a simple ``smoothness'' argument showing that the potential cannot decrease quickly within $\lceil n/\log^2 n \rceil$ steps. The derivation is elementary and relies on the fact that the average load changes by at most $2/\log^2 n$.
\begin{clm} \label{clm:phi_j_does_not_drop_quickly}
For any step $s \geq 0$ and any step $r \in [s, s + \lceil n/\log^2 n \rceil]$, we have that $\Phi_2^{r} \geq 0.99 \cdot \Phi_2^{s}$.
\end{clm}
\begin{proof}
The normalized load after $r - s$ steps can decrease by at most $\frac{r - s}{n} \leq \frac{2}{\log^2 n}$. Hence, for any bin $i \in [n]$,
\begin{align*}
\Phi_{2,i}^{r} 
 & = e^{\phi_2 \cdot (y_i^{r} - z )^{+}}
 \geq e^{\phi_2 \cdot (y_i^{s} - \frac{r-s}{n} - z)^{+}}
 \geq e^{\phi_2 \cdot (y_i^{s} - z)^+ - \phi_2 \cdot \frac{2}{\log^2 n}}
 = \Phi_{2,i}^{s} \cdot e^{-\frac{2\phi_2}{\log^2 n} }
 \geq \Phi_{2,i}^{s} \cdot e^{-o(1)} \\
 & \geq 0.99 \cdot \Phi_{2,i}^{s},
\end{align*}
for sufficiently large $n$, using that $\phi_2 \leq (\log n)/6$. By aggregating over all bins, we get the claim.
\end{proof}

The next claim is crucial for applying the concentration inequality, since the second statement bounds the maximum additive change of $\Phi_2^{s}$ in any step $s$ where $\Phi_1^{s}$ is $\poly(n)$:
\begin{lem} \label{lem:psi_potential_poly_implies}
For any step $s \geq 0$ where $\Phi_1^{s} \leq 6n^{12}$, we have that
\begin{align*}
(i) & \quad \Phi_2^{s} \leq n^{7/6}, \\ 
(ii) & \quad |\Phi_2^{s+1} - \Phi_2^{s} | \leq n^{1/3}.
\end{align*}
\end{lem}
\begin{proof} 
Consider an arbitrary step $s \geq 0$ with $\Phi_1^{s} \leq 6n^{12}$. We start by upper bounding the normalized load $y_i^s$ of any bin $i \in [n]$, 
\[
y_i^s 
  \leq z + \frac{\log \Phi_{1,i}^s }{\phi_1} 
  \leq z + \frac{\log \Phi_1^s }{\phi_1}
  \leq z + \frac{\log(6n^{12})}{\phi_1}
  \leq z + \frac{14\log n}{\phi_1}.
\]

\textit{First statement}. Now, we upper bound the contribution of any bin $i \in [n]$ to $\Phi_2^s$,
\begin{align}
\Phi_{2,i}^{s} 
  = \exp \bigl( \phi_2 \cdot \bigl( y_i^{s} - z  \bigr)^{+} \bigr) 
  \leq \exp\left( \frac{14 \cdot \phi_2}{\phi_1} \cdot \log n  \right) \leq n^{1/6}, \label{eq:n6_upper_bound}
\end{align}
using that $\phi_2 \leq \frac{\phi_1}{84}$. Hence, by aggregating over all bins, 
\begin{align} \label{eq:phi_j_less_n_4_3}
\Phi_2^s \leq n \cdot n^{1/6} = n^{7/6}.
\end{align}

\textit{Second statement}. We will derive lower and upper bounds for $\Phi_2^{s+1}$. For the upper bound, let $i = i^{s+1} \in [n]$ be the bin where the $(s+1)$-th ball is allocated, then
\[
\Phi_2^{s+1} \leq \Phi_2^{s} + \Phi_{2,i}^s \cdot e^{\phi_2} \leq \Phi_2^{s} + n^{1/6} \cdot n^{1/6} = \Phi_2^{s} + n^{1/3},
\]
using that $\phi_2 \leq (\log n)/6$ and \cref{eq:n6_upper_bound}. For the lower bound, we pessimistically assume that all bin loads decrease by $1/n$ in step $s+1$, so
\[
\Phi_2^{s+1} \geq \Phi_2^{s} \cdot e^{- \phi_2/n} \stackrel{(a)}{\geq} \Phi_2^{s} \cdot \left(1 - \frac{\phi_2}{n}\right) \stackrel{(b)}{\geq} \Phi_2^{s} - \frac{n \cdot n^{1/6}}{n} \cdot \log n \geq \Phi_2^{s} - n^{1/3},
\]
using in $(a)$ that $e^u \geq 1 + u$ (for any $u$) and in $(b)$ that $\phi_2 \leq \log n$ and $\Phi_2^s \leq n \cdot n^{1/6}$ by \cref{eq:phi_j_less_n_4_3}. Combining the two bounds we get the second statement.
\end{proof}

\subsubsection{Recovery Phase} \label{sec:super_exp_concentration_probabilistic_drop}

In this section, we will show for an auxiliary process $\tilde{\mathcal{P}}$ (to be defined below) that the potential function $\Phi_2$ satisfies $\Phi_2^s \leq 6n$ in \textit{at least one step} $s \in [t - n \log^4 n, t]$ \Whp

First, we show that for the original process $\mathcal{P}$ in the statement of \cref{thm:super_exp_potential_concentration}, the potential 
\begin{align} \label{eq:tilde_phi_1_def}
\tilde{\Phi}_1^{s} := \tilde{\Phi}_1^{s}(t) := \Phi_1^{s} \cdot \mathbf{1}_{\cap_{r\in [t - 2n \log^4 n, s]} \mathcal{K}^{r} \cap \mathcal{H}^{t - 2n\log^4 n}}
\end{align}
is small in expectation for \textit{all} steps $s \geq t-n\log^4 n$. Note that there is a ``recovery time'' until the expectation becomes small, of at most $n \log^4 n$ steps after the ``weak'' bound $\Gap(t-2n\log^4 n) \leq \log^2 n$ which follows from the third precondition \cref{eq:gap_and_condition_k_precondition} in \cref{thm:super_exp_potential_concentration}.

\begin{lem} \label{lem:newexists_s_st_ex_psi_linear}
Consider the potential $\tilde{\Phi}_1 := \tilde{\Phi}_1(t)$ for any step $t \geq 0$. Then, for any step $s \geq t - n \log^4 n$,\[
\Ex{\tilde{\Phi}_1^{s}} \leq 6n.
\]
\end{lem}
\begin{proof}
By the precondition \cref{eq:phi_1_drop_precondition} of \cref{thm:super_exp_potential_concentration},  for any step $s \geq t - n \log^4 n$, \begin{equation*}
\Ex{\left. \Phi_1^{s+1} \,\right|\, \mathfrak{F}^{s}, \mathcal{K}^s  } \leq \Phi_1^{s} \cdot \Big(1 - \frac{1}{n} \Big) + 2.
\end{equation*}
Next note that whenever $\neg \mathcal{K}^s$ holds, it follows deterministically that $\{ \tilde{\Phi}_1^{s} = \tilde{\Phi}_1^{s+1} = 0 \}$, and hence 
\begin{align} \label{eq:tilde_psi_drop}
\Ex{\left. \tilde{\Phi}_1^{s+1} \,\right|\, \mathfrak{F}^{s}  } \leq \tilde{\Phi}_1^{s} \cdot \Big(1 - \frac{1}{n} \Big) + 2.
\end{align}
We will now upper bound $\Ex{\left. \tilde{\Phi}_1^{s} \,\right|\, \mathfrak{F}^{t - 2n \log^4 n}, \mathcal{H}^{t - 2n \log^4 n}}$ for any step $s \geq t - n \log^4 n$. When $\mathcal{H}^{t - 2n \log^4 n}$ holds, by \cref{clm:weak_gap_implies_phi_1_small}, it also follows that $\tilde{\Phi}_1^{t - 2n \log^4 n} \leq \Phi_1^{t - 2n \log^4 n} \leq e^{\frac{1}{2} \cdot \log^4 n}$. Hence applying \cref{lem:geometric_arithmetic}~$(i)$ (for $a = 1 - \frac{1}{n}$ and $b = 2$) using \cref{eq:tilde_psi_drop},
\begin{align*}
\Ex{ \tilde{\Phi}_1^{s} \,\,\left|\,\, \mathfrak{F}^{t - 2n \log^4 n}, \mathcal{H}^{t - 2n \log^4 n} \right.} 
 & \leq \Ex{\tilde{\Phi}_1^s \,\, \left| \,\, \mathfrak{F}^{t - 2n \log^4 n},  \tilde{\Phi}_1^{t - 2n \log^4 n} \leq e^{\frac{1}{2} \cdot \log^4 n} \right.} \\
 & \leq e^{\frac{1}{2} \cdot \log^4 n} \cdot \left(1 - \frac{1}{n} \right)^{s - (t - 2n \log^4 n)} + 2n \\
 & \stackrel{(a)}{\leq} e^{\frac{1}{2} \cdot \log^4 n} \cdot \left(1 - \frac{1}{n} \right)^{n \log^4 n} + 2n \\
 & \stackrel{(b)}{\leq} e^{\frac{1}{2} \cdot \log^4 n} \cdot e^{-\log^4 n} + 2n \leq 1 + 2n \leq 6n,
\end{align*}
using in $(a)$ that $s \geq t - n \log^4 n$ and in $(b)$ that $e^u \geq 1 + u$ (for any $u$). Hence, the claim follows,
\[
\Ex{\tilde{\Phi}_1^{s}} = \Ex{\left. \tilde{\Phi}_1^{s} \,\,\right|\, \mathcal{H}^{t - 2n\log^4 n}} \cdot \Pro{\mathcal{H}^{t - 2n\log^4 n}} + 0 \cdot \Pro{\neg \mathcal{H}^{t - 2n\log^4 n}} \leq 6n. \qedhere
\]
\end{proof}

\paragraph{The auxiliary process $\tilde{\mathcal{P}}$.} Let $\mathcal{P}$ be the process in the statement of \cref{thm:super_exp_potential_concentration}, define $t_1 := t - n \log^4 n$ and the stopping time $\sigma := \inf\{ s \geq t_1 : \{ \Phi_1^s > 6n^{11} \} \cup \neg \mathcal{K}^s \}$. Now we define the auxiliary process $\tilde{\mathcal{P}}_{t_1}$ such that 
\begin{itemize}
 \item in steps $s \in [0, \sigma)$ it follows the allocations of $\mathcal{P}$, and
 \item in steps $s \in [\sigma, \infty)$ it allocates to the (currently) least loaded bin, which is a bin with normalized load $\leq 1 \leq z$.
\end{itemize}
This way the process trivially satisfies $\mathcal{K}^s$ and therefore also the drop inequalities (\cref{eq:phi_1_drop_precondition} and \cref{eq:phi_2_drop_precondition}) for any step $s \geq t_1$. Furthermore, starting with $\Phi_1^{t_1} \leq 6n^{11}$, for any $s \geq t_1$ it deterministically satisfies $\Phi_1^s \leq 6n^{12}$, since $\Phi^{\sigma} \leq 6n^{11} \cdot e^{\phi_1} \leq 6n^{12}$ (as $\phi_1 \leq (\log n)/6$) and for $s \geq \sigma$ the potential does not increase.

For this starting condition, we will also show that the other potential function $\Phi_2$ \Whp~becomes linear in \textit{at least one} step in $[t - n \log^4 n, t]$, by using \cref{lem:newexists_s_st_ex_psi_linear}.

\begin{lem}[Recovery] \label{lem:recovery_phi_2}
For any step $t \geq 0$ and the auxiliary process $\tilde{\mathcal{P}}_{t - n \log^4 n}$, it holds that
\[
\Pro{\left. \bigcup_{s \in [t - n \log^4 n, t]} \left\{ \Phi_2^{s} \leq 6n \right\} \,\right|\, \mathfrak{F}^{t - n \log^4 n}, \Phi_1^{t - n \log^4 n} \leq 6n^{11}} \geq 1 - n^{-6}.
\]
\end{lem}
\begin{proof}
Fix any step $t \geq 0$ and let $t_1 := t - n \log^4 n$ be the starting point of the analysis. Further, assume that $\{ \Phi_1^{t_1} \leq 6n^{11} \}$ holds, which also implies that the auxiliary process $\tilde{\mathcal{P}}_{t_1}$ satisfies $\mathcal{K}^s$ for any step $s \geq t_1$. Further, for any $s \geq t_1$, we define the potential function
\[
\widehat{\Phi}_2^s := \Phi_2^{s} \cdot \mathbf{1}_{\cap_{r \in [t_1, s]} \{\Phi_2^{r} > 6n\}}.
\]
We will show that $\widehat{\Phi}_2$ drops in expectation by a multiplicative factor in every step. We start by showing for $\Phi_2$ that when $\Phi_2^s > 6n$, then it drops in expectation by a multiplicative factor. By the second precondition \cref{eq:phi_2_drop_precondition} of \cref{thm:super_exp_potential_concentration},
\begin{align*}
& \Ex{\left. \Phi_2^{s+1} \,\right|\, \Phi_1^{t_1} \leq 6n^{11}, \mathfrak{F}^{s}, \Phi_2^{s} > 6n }
  = \Ex{\left. \Phi_2^{s+1} \,\right|\, \Phi_1^{t_1} \leq 6n^{11}, \mathfrak{F}^{s}, \mathcal{K}^s, \Phi_2^{s} > 6n } \\
 & \qquad \leq \Phi_2^{s} \cdot \Big(1 - \frac{1}{n} \Big) + 2
 \leq \Phi_2^{s} \cdot \Big(1 - \frac{1}{2n} \Big) -6n \cdot \frac{1}{2n} + 2 
 \leq \Phi_2^{s} \cdot \Big(1 - \frac{1}{2n} \Big).
\end{align*}
Whenever the event $\{ \Phi_2^s \leq 6n \}$ holds, it follows deterministically that $\widehat{\Phi}_2^s = \widehat{\Phi}_2^{s+1} = 0$. So, for the potential $\widehat{\Phi}_2$ we obtain the drop inequality (with one fewer condition),
\begin{align} \label{eq:phi_i_plus_1_large_phi}
\Ex{\left. \widehat{\Phi}_2^{s+1} \,\right|\, \Phi_1^{t_1} \leq 6n^{11}, \mathfrak{F}^{s}}
 \leq \widehat{\Phi}_2^{s} \cdot \Big(1 - \frac{1}{2n} \Big).
\end{align}
By inductively applying~\cref{eq:phi_i_plus_1_large_phi} for $n \log^4 n$ steps starting at $t_1$, we have that
\begin{align*}
\Ex{\left. \widehat{\Phi}_2^{t} \,\,\right|\,\, \mathfrak{F}^{t_1}, \Phi_1^{t_1} \leq 6n^{11} }
 & \leq \Ex{\left. \widehat{\Phi}_2^{t} \,\,\right|\,\, \mathfrak{F}^{t_1}, \widehat{\Phi}_2^{t_1} \leq n^{7/6} } \\
 & \leq \widehat{\Phi}_2^{t_1} \cdot \Big(1 - \frac{1}{2n} \Big)^{n \log^4 n} 
 \leq n^{7/6} \cdot e^{- \frac{1}{2} \log^4 n} \leq n^{-6},
\end{align*}
for sufficiently large $n$, using that $e^u \geq 1 + u$ (for any $u$). Hence, by Markov's inequality,
\[
\Pro{\left. \widehat{\Phi}_2^{t} \leq 1 \,\right|\, \mathfrak{F}^{t_1}, \Phi_1^{t_1} \leq 6n^{11}} \geq 1 - n^{-6}.
\]
Since it deterministically holds that $\{ \Phi_2^{t} \geq n \}$ for any step $t \geq 0$, it follows that if $\{ \widehat{\Phi}_2^{t} \leq 1 \}$ holds, then also $\{ \widehat{\Phi}_2^{t} = 0 \}$. So, we conclude that $\mathbf{1}_{\cap_{r \in [t_1, t]} \{ \Phi_2^{r} > 6n \} } = 0$ holds, i.e.,
\[
\Pro{\left. \bigcup_{r \in [t_1, t]} \left\{ \Phi_2^{r} \leq 6n \right\} ~\right|~ \mathfrak{F}^{t_1}, \Phi_1^{t_1} \leq 6n^{11} } \geq 1 - n^{-6}. \qedhere
\]
\end{proof}

\subsubsection{Completing the Proof of Theorem~\ref{thm:super_exp_potential_concentration}}\label{sec:super_exp_concentration_proof_completion}

For some values of $g$, the number of steps $k$ in the layered induction may be $\omega(1)$ and so we may need to apply \cref{thm:super_exp_potential_concentration} up to $\omega(1)$ times, which means that we cannot afford to lose a polynomial factor in the error probability. To overcome this, we will partition the time-interval into consecutive intervals of length $\lceil n/\log^2 n \rceil$. Then, we will prove that at the end of each such interval the potential is small \Whp, and finally use a simple smoothness argument (\cref{clm:phi_j_does_not_drop_quickly}) to show that the potential is small \Whp~in \textit{all} steps.

\begin{thm} \label{thm:super_exp_potential_concentration}
\ThmSuperExponentialPotentialConcentrationWithLabels{0}
\end{thm} 

We will start by proving the following lemma for the auxiliary process $\tilde{\mathcal{P}}_{t - n \log^4 n}$.
\begin{lem} \label{lem:concentration_for_auxiliary_process}
For any step $t \geq 0$ and the auxiliary process $\tilde{\mathcal{P}}_{t - n \log^4 n}$, it holds that
\[
\Pro{\left. \bigcap_{s \in [t, \tilde{t}]} \left\{ \Phi_2^{s} \leq 8n \right\} \,\,\right|\,\, \Phi_1^{t - n\log^4 n} \leq 6n^{11} } 
  \geq 1 - \frac{1}{4} \cdot (\log^8 n) \cdot P.
 \]
\end{lem}
\begin{proof}
Our goal is to apply Azuma's inequality to $\Phi_2$. However, there are two challenges: $(i)$ we cannot afford to take the union bound over all $\poly(n)$ steps and $(ii)$ $\Phi_2$ is a super-martingale only when it is sufficiently large. To deal with $(i)$ we will apply Azuma's inequality to sub-intervals of length at most $\lceil n/\log^2 n \rceil$ and then use a smoothness argument (\cref{clm:phi_j_does_not_drop_quickly}) to deduce that $\Phi_2$ is small in the steps in between. For $(ii)$ we will define $X^s := X^s(\Phi_2^s)$ in a way to ensure that it is super-martingale at every step and also satisfies the bounded difference inequality.

More specifically, consider an arbitrary step $r \in[t - n \log^4 n, t]$ and partition the interval $(r, \tilde{t}]$ into 
\[
\mathcal{I}_1 := (r, r + \Delta],~ \mathcal{I}_2 = (r + \Delta, r + 2 \Delta],~\ldots~, ~\mathcal{I}_q := (r + (q-1) \Delta, \tilde{t}], 
\]
where $\Delta := \lceil n/\log^2 n \rceil$ and $q := \big\lceil \frac{\tilde{t}-r}{\Delta} \big\rceil \leq \big\lceil \frac{t + n \log^5 n - r}{\Delta} \big\rceil \leq 2\log^7 n$. In order to prove that $\Phi_2$ is at most $8n$ in every step in $(r, \tilde{t}]$, we will use our auxiliary lemmas (\cref{sec:super_exp_concentration_deterministic_claims,sec:super_exp_concentration_probabilistic_drop}) and Azuma's super-martingale concentration inequality~ (\cref{lem:azuma})
to establish that $\Phi_2$ is at most $7n$ at each of the steps $r + \Delta, r + 2\Delta, \ldots , r + (q-1)\Delta, \tilde{t}$. Finally, by using a smoothness argument (\cref{clm:phi_j_does_not_drop_quickly}), we will establish that $\Phi_2$ is at most $8n$ at all steps in $(r, \tilde{t}]$, which is the conclusion of the theorem.

For each interval $i \in [q]$, we define $X_i^{r + (i-1)\Delta} := \max\{\Phi_2^{r + (i-1)\Delta}, 5n + n^{1/3}\}$ and for any $s \in (r + (i-1)\Delta, r + i \Delta]$,
\[
X_i^{s} := \begin{cases}
\Phi_2^{s} & \text{if there exists }u \in [r + (i -1) \Delta, s)\text{ such that }\Phi_2^{u} \geq 5n, \\
5n+n^{1/3} & \text{otherwise}.
\end{cases}
\]
Note that whenever the first condition in the definition of $X_i^{s}$ is satisfied, it remains satisfied until the end of the interval, i.e., until step $r+i \cdot \Delta$. Our next aim is to establish the preconditions of Azuma's inequality (\cref{lem:azuma}) for $X_i^s$. 

For convenience, we define $t_1 := t - n \log^4 n$ and the event 
\[
\mathcal{Z}^{t_1} := \{ \Phi_1^{t_1} \leq 6n^{11} \}.
\]

\begin{lem}\label{lem:x_i_preconditions}
Fix any interval $\mathcal{I}_i$ for $i \in [q]$. Then for the random variables $X_i^{s}$, for any step $s \in (r+(i-1)\Delta, r+i\Delta]$, it holds that
\begin{align*}
    \Ex{ X_{i}^{s} \, \left| \, \mathcal{Z}^{t_1}, \mathfrak{F}^{s-1} \right.} \leq X_{i}^{s-1}, 
\end{align*}
and
\[
 \left( \left. \left| X_i^{s} - X_i^{s-1} \, \right| \,~ \right| \, \mathcal{Z}^{t_1}, \mathfrak{F}^{s-1} \right) \leq 2 n^{1/3}.
\]
\end{lem}

\begin{figure}
    \centering
    \scalebox{0.7}{
    \begin{tikzpicture}[
  IntersectionPoint/.style={circle, draw=black, very thick, fill=black!35!white, inner sep=0.05cm}
]

\definecolor{MyBlue}{HTML}{9DC3E6}
\definecolor{MyYellow}{HTML}{FFE699}
\definecolor{MyGreen}{HTML}{E2F0D9}
\definecolor{CaseTwoGreen}{HTML}{A9D18E}
\definecolor{MyRed}{HTML}{FF9F9F}
\definecolor{MyDarkRed}{HTML}{C00000}
\definecolor{MyOrange}{HTML}{C55A11}
\definecolor{MyGreenLine}{HTML}{A9D18E}

\def\xEnd{16}
\def\xLast{15.40}
\def\yLast{4.5}
\def\ThresholdOne{1}
\def\ThresholdTwo{2.5}
\def\yBottom{-0.8}

\node[anchor=south west, inner sep=0.15cm, fill=MyBlue, rectangle,minimum width=7.7cm] at (0, \yLast) {Case 2(a)};
\node[anchor=south west, inner sep=0.15cm, fill=MyYellow, rectangle,minimum width=\xLast cm - 9.3cm] at (9.3, \yLast) {Case 2(c)};
\node[anchor=south west, inner sep=0.15cm, fill=CaseTwoGreen, rectangle,minimum width=1.8cm] at (7.6, \yLast) {Case 2(b)};

\draw[dashed, thick] (0,\ThresholdOne) -- (\xLast, \ThresholdOne);
\draw[dashed, thick] (0, \ThresholdTwo) -- (\xLast, \ThresholdTwo);

\node[anchor=east] at (0, \ThresholdOne) {$5n$};
\node[anchor=east] at (0, \ThresholdTwo) {$5n + n^{1/3}$};

\node[anchor=west] at (\xEnd, \yBottom) {$s$};

\def\tA{6.42}
\def\tB{8.52}
\def\tC{10.41}
\def\tD{11.71}
\def\tM{13.51}
\def\tE{14.36}

\newcommand{\drawLine}[3]{
\draw[dashed, very thick, #3] (#1, \yBottom) -- (#1, \yLast);
\draw[very thick] (#1, \yBottom) -- (#1, \yBottom -0.2);
\node[anchor=north] at (#1, \yBottom -0.3) {#2};}

\newcommand{\drawPoint}[3]{
\drawLine{#1}{#2}{#3}
\node[IntersectionPoint] at (#1, \cn) {};}

\draw[very thick] (0, \yBottom) -- (0, \yBottom -0.2);
\node[anchor=north] at (0, \yBottom -0.3) {$t_0 + (i - 1) \cdot \Delta$};
\drawLine{8}{$\phantom{1}\sigma\phantom{1}$}{black!30!white};
\drawLine{9}{$\sigma + 1$}{black!30!white};

\draw[very thick] (\xLast, \yBottom) -- (\xLast, \yBottom -0.2);
\node[anchor=north] at (\xLast, \yBottom -0.3) {$t_0 + i \cdot \Delta$};

\draw[MyOrange, very thick] (0, \ThresholdTwo) -- (8, \ThresholdTwo) -- (9, 0.7) -- (10, 1.4) -- (11, 2.7) -- (12, 3.4) -- (13, 2.7) -- (14, 3.2) -- (15, 2.9) -- (\xLast, 3.5);

\draw[MyGreenLine, very thick,yshift=-0.05cm] 
   (0, 0.4) -- (1, 0.2) -- (2, 0.7) -- (3, 0.3) -- (4, 0.5) -- (5, -0.3) -- (6, 0.4) -- (7, 0.8)
   -- (8, 1.8) -- (9, 0.7) -- (10, 1.4) -- (11, 2.7) -- (12, 3.4) -- (13, 2.7) -- (14, 3.2) -- (15, 2.9) -- (\xLast, 3.5);

\draw[->, ultra thick] (0,\yBottom) -- (0, \yLast + 1);
\draw[->, ultra thick] (0,\yBottom) -- (\xEnd, \yBottom);

\begin{scope}[xshift=-1cm, yshift=-0.2cm, local bounding box=bx]
\node[anchor=west] at (\xLast, 0.5) {$X_i^s$};
\node[anchor=west] at (\xLast, 0.0) {$\Phi_2^s$};

\draw[ultra thick, MyOrange] (\xLast - 0.6, 0.55) -- (\xLast - 0.05, 0.55);
\draw[ultra thick, MyGreenLine] (\xLast - 0.6, 0.05) -- (\xLast - 0.05, 0.05);
\end{scope}

\node [fit=(bx),inner sep=0.1cm,draw] {};

\end{tikzpicture}
    }
    \caption{Visualization of three cases in the proof of \cref{lem:x_i_preconditions} for interval $i \in [q]$.}
    \label{fig:x_i_preconditions}
\end{figure}
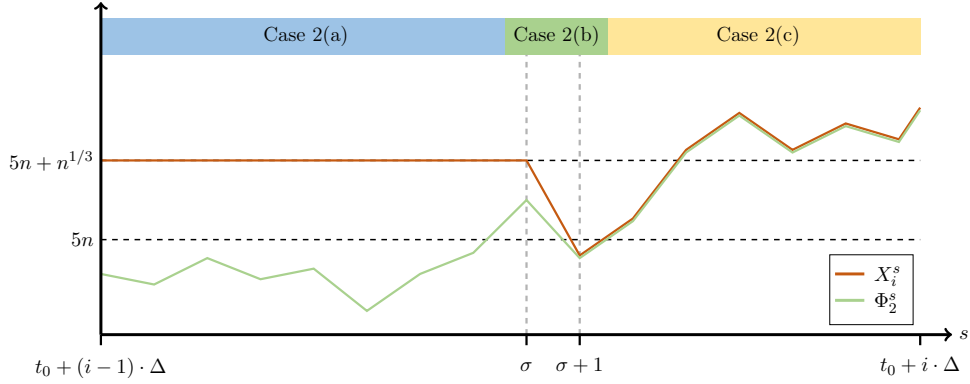

\begin{proof}[Proof of \cref{lem:x_i_preconditions}]
Recall that by conditioning on $\mathcal{Z}^{t_1}$, the auxiliary process $\tilde{\mathcal{P}}_{t_1}$ satisfies $\{ \Phi_1^{s-1} \leq 6n^{12} \}$ and $\mathcal{K}^{s-1}$, for any $s \in (r + (i-1) \Delta, r + i \Delta]$, since $r + (i-1) \Delta \geq t_1$.

By precondition \cref{eq:phi_2_drop_precondition}, when $\{ \Phi_2^{s-1} \geq 4n \}$ also holds, we have that,\begin{align}
\Ex{\Phi_2^{s} \,\,\left|\,\, \mathcal{Z}^{t_1}, \mathfrak{F}^{s-1}, \Phi_2^{s-1} \geq 4n \right.} 
 & = \Ex{\Phi_2^{s} \,\,\left|\,\, \mathcal{Z}^{t_1}, \mathfrak{F}^{s-1}, \mathcal{K}^{s-1}, \Phi_2^{s-1} \geq 4n \right.} \notag \\
 & \leq \Phi_2^{s-1} \cdot \Big(1 - \frac{1}{n}\Big) + 2 \leq \Phi_2^{s-1} - 4n \cdot \frac{1}{n} + 2 \leq \Phi_2^{s-1}. \label{eq:drop}
\end{align}
 Further, by \cref{lem:psi_potential_poly_implies}~$(ii)$, we also have that
\begin{align} \label{eq:concentration_bounded_diff}
\left( \left. \left| \Phi_2^s - \Phi_2^{s-1} \right| ~\right|~ \mathcal{Z}^{t_1}, \mathfrak{F}^{s-1} \right)
  \leq n^{1/3}.
\end{align}

\textbf{Case 1 [$\Phi_2^{r + (i-1)\Delta} \geq 5n + n^{1/3}$]:} In this case $X_i^{s-1} = \Phi_2^{s-1}$ for all $s \in (r + (i-1) \Delta, r + i \Delta]$. By \cref{clm:phi_j_does_not_drop_quickly}, we also have $\Phi_2^{s-1} \geq 0.99 \cdot (5n + n^{1/3}) \geq 4n$ (as $\Delta \leq \lceil n/\log^2 n \rceil$) and the two statements follow by \cref{eq:drop} and \cref{eq:concentration_bounded_diff}.

\textbf{Case 2 [$\Phi_2^{r + (i-1)\Delta} < 5n + n^{1/3}$]:} Let $\sigma := \inf \{ u \geq r + (i-1) \Delta \colon \Phi_2^{u} \geq 5n\}$. We consider the following three cases (see \cref{fig:x_i_preconditions}):
\begin{itemize} \itemsep0pt
    \item \textbf{Case 2(a) [$s-1 < \sigma$]:} Here $X_i^{s} = X_i^{s-1} = 5n + n^{1/3}$, so the two statements hold trivially.
    \item \textbf{Case 2(b) [$s-1 = \sigma$]:} 
    We will first establish that
    \begin{align} \label{eq:phi_2_bound_case_2b}
    5n \leq \Phi_2^{s-1} \leq 5n + n^{1/3}.
    \end{align}
    The lower bound $\Phi_2^{s-1} \geq 5n$ follows by definition of $\sigma$. For the upper bound, we consider the following two cases. If $s - 1 = r + (i-1)\Delta$, then this follows by the assumption for Case 2. Otherwise, we have that $\Phi_2^{s-1} \geq 5n$ and $\Phi_2^{s-2} < 5n$. By \cref{eq:concentration_bounded_diff}, we obtain that  $\Phi_2^{s-1} \leq \Phi_2^{s-2} + n^{1/3} < 5n + n^{1/3}$.
    
    Next, by definition, $X_i^{s-1} = 5n + n^{1/3}$ and $X_i^s = \Phi_2^s$, so by \cref{eq:drop},
    \[\Ex{\left. X_i^{s} \, \right| \, \mathcal{Z}^{t_1}, \mathfrak{F}^{s-1}} = \Ex{\left. \Phi_2^{s}  \, \right| \, \mathcal{Z}^{t_1}, \mathfrak{F}^{s-1} }\leq \Phi_2^{s-1} < X_i^{s-1},
    \] 
    which establishes the first statement. For the second statement, by \cref{eq:concentration_bounded_diff} and \cref{eq:phi_2_bound_case_2b} we have \[
    \left|X_i^{s} - X_i^{s-1}\right| = \left|\Phi_2^{s} - 5n - n^{1/3} \right| \leq \left|\Phi_2^{s-1} - 5n - n^{1/3} \right| + \left|\Phi_2^{s} - \Phi_2^{s-1} \right| \leq 2n^{1/3}.
    \]
    using in the last inequality that \cref{eq:concentration_bounded_diff} and \cref{eq:phi_2_bound_case_2b}.
    \item \textbf{Case 2(c) [$s-1 > \sigma$]:} Here, $X_i^{s-1} = \Phi_2^{s-1}$ and $X_i^s = \Phi_2^s$. Since $\Phi_2^{\sigma} \geq 5n$, by~\cref{clm:phi_j_does_not_drop_quickly} (as $s-\sigma \leq \lceil n/\log^2 n \rceil$), we also have that
    \[
     \Phi_2^{s-1} \geq 0.99 \cdot \Phi_2^{\sigma} \geq 0.99 \cdot 5n \geq 4n,
    \]
    and thus by \cref{eq:drop}, the first statement follows. The second statement follows by \cref{eq:concentration_bounded_diff}.\qedhere
\end{itemize}
\end{proof}

Now we return to the proof of \cref{lem:concentration_for_auxiliary_process}. By~\cref{lem:x_i_preconditions}, we have verified that $X_{i}^{s}$ satisfies the preconditions of Azuma's inequality for any filtration $\mathfrak{F}^{s-1}$ where $\mathcal{Z}^{t_1}$ holds. So, applying \cref{lem:azuma} for $\lambda = \frac{n}{2\log^7 n}$, $N \leq \Delta$ and $D = 2n^{1/3}$, we get for any $i \in [q]$,
\begin{align*}
\Pro{\left. X_{i}^{r + i \Delta} \geq X_{i}^{r + (i-1) \Delta} + \lambda \,\right|\, \mathcal{Z}^{t_1}, \mathfrak{F}^r} &\leq \exp\left(-\frac{n^2/(4\log^{14}n)}{2 \cdot \Delta \cdot (4 n^{2/3})} \right) + 2P \leq 3P,
\end{align*}
since $\Delta \leq \lceil n/\log^2 n \rceil$ and $P \geq n^{-4}$. Taking the union bound over the at most $2\log^7 n$ intervals $i \in [q]$, it follows that
\begin{align*}
\Pro{ \left. \bigcup_{i \in [q]} \left\{ X_{i}^{r + i \Delta} \geq X_{1}^{r} + i \cdot \frac{n}{2\log^7 n} \right\} \,\right|\, \mathcal{Z}^{t_1}, \mathfrak{F}^r} \leq (2\log^7 n) \cdot 3P \leq \frac{1}{8} \cdot (\log^8 n) \cdot P.
\end{align*}
Next, conditional on $(\mathcal{Z}^{t_1}, \mathfrak{F}^r, \Phi_2^r \leq 6n)$, we have the following chain of inclusions:
\begin{align*}
 \bigcap_{i \in [q]} \left\{ X_{i}^{r + i \Delta} \leq X_1^r + i \cdot \frac{n}{2 \log^7 n} \right\}
 &\stackrel{(a)}{\leq}  \bigcap_{i \in [q]} \left\{ X_{i}^{r + i \Delta} \leq 6n + n \right\} \\
 &\stackrel{(b)}{=} 
  \bigcap_{i \in [q]} \left\{ \Phi_2^{r + i \Delta} \leq 7n \right\}
 \\
 &\stackrel{(c)}{\subseteq} \bigcap_{s \in [r,\tilde{t}]} \left\{ \Phi_2^{s} \leq \frac{7}{0.99} \cdot n \right\} \\
 &\stackrel{(d)}{\subseteq} \bigcap_{s \in [t,\tilde{t}]} \left\{ \Phi_2^{s} \leq 8n \right\},
\end{align*}
where $(a)$ holds since $i \leq q \leq 2 \log^7 n$ and $X_1^r \leq \max\{ \Phi_2^r, 5n + n^{1/3} \}$, $(b)$ holds since $\Phi_2^{r + i \Delta} \leq X_{i}^{r + i \Delta}$, $(c)$ holds by applying the smoothness argument of \cref{clm:phi_j_does_not_drop_quickly} to each interval $i \in [q]$, as the length of each interval is at most $\lceil n/\log^2 n \rceil$ and $(d)$ holds since $8 \geq \frac{7}{0.99}$ and $r \leq t$. This implies that
\begin{align} \label{eq:phi_t0_good_whp_conclusion}
\Pro{\left. \bigcap_{s \in [t, \tilde{t}]} \left\{ \Phi_2^{s} \leq 8n \right\} ~\right|~ \mathcal{Z}^{t_1}, \mathfrak{F}^{r}, \Phi_2^{r} \leq 6n } \geq 1 -  \frac{1}{8} \cdot (\log^8 n) \cdot P.
\end{align}
Next define $\tau := \inf\{ r \geq t -  n \log^4 n \colon \Phi_2^{r} \leq 6n \}$. By \cref{lem:recovery_phi_2},
\begin{align} \label{eq:concentration_stopping_time_recovery}
\Pro{\tau \leq t \,\left|\, \mathcal{Z}^{t_1} \right.} \geq \Pro{\tau \leq t \,\left|\, \mathfrak{F}^{t_1}, \mathcal{Z}^{t_1} \right.} \geq 1 - n^{-6}.
\end{align}
We get the conclusion for the $\tilde{\mathcal{P}}_{t_1}$ process, by combining this with \cref{eq:phi_t0_good_whp_conclusion} and \cref{eq:concentration_stopping_time_recovery},
\begin{align*}
& \Pro{\left. \bigcap_{s \in [t, \tilde{t}]} \left\{ \Phi_2^{s} \leq 8n \right\} \,\,\right|\,\, \mathcal{Z}^{t_1} } \\
 &  \geq \sum_{r = t - n \log^4 n}^{t} \Pro{\left. \bigcap_{s \in [t, \tilde{t}]} \left\{ \Phi_2^{s} \leq 8n \right\} ~\right|~ \mathcal{Z}^{t_1}, \mathfrak{F}^{r}, \Phi_2^{r} \leq 6n } \cdot \Pro{\left. \tau = r \,\right|\, \mathcal{Z}^{t_1}} \\
 & \!\!\!\!\stackrel{(\text{\ref{eq:phi_t0_good_whp_conclusion}})}{\geq} \left( 1 - \frac{1}{8} \cdot (\log^8 n) \cdot P \right) \cdot \Pro{\tau \leq t \,\left|\, \mathcal{Z}^{t_1} \right.} \\
 &\!\!\!\! \stackrel{(\text{\ref{eq:concentration_stopping_time_recovery}})}{\geq} \left( 1 - \frac{1}{8} \cdot (\log^8 n) \cdot P \right) \cdot \left( 1 - n^{-6}\right) \geq 1 - \frac{1}{4} \cdot (\log^8 n) \cdot P. \qedhere
\end{align*}
\end{proof}

We now return to the proof of \cref{thm:super_exp_potential_concentration} for the original process $\mathcal{P}$.

\begin{proof}[Proof of \cref{thm:super_exp_potential_concentration}]
Let $t_1 := t - n \log^4 n$. 
We start by showing that the processes $\mathcal{P}$ and $\tilde{\mathcal{P}}_{t_1}$ agree with high probability in the interval $[t_1, \tilde{t}]$ (and so at every step $s \in [t_1, \tilde{t}]$ we have $\Phi_2^s = \Phi_{\tilde{\mathcal{P}}_{t_1}, 2}^s$),
\begin{align} 
\Pro{\bigcap_{s \in [t_1, \tilde{t}]} \left\{ y^s = y_{\tilde{\mathcal{P}}_{t_1}}^s \right\} } 
  & \geq \Pro{\bigcap_{s \in [t - 2n \log^4 n, \tilde{t}]} \left\{ \Phi_1^{s} \leq 6n^{11} \right\} \cap \bigcap_{s \in [t - 2n \log^4 n, \tilde{t}]} \mathcal{K}^s } \notag \\ 
  & \stackrel{(a)}{\geq} \Pro{\bigcap_{s \in [t - 2n \log^4 n, \tilde{t}]} \left\{ \tilde{\Phi}_1^{s} \leq 6n^{11} \right\} \cap \mathcal{H}^{t - n \log^4 n} \cap \bigcap_{s \in [t - 2n \log^4 n, \tilde{t}]} \mathcal{K}^s } \notag \\ 
  & \stackrel{(b)}{\geq} 1 - n^2 \cdot n^{-11} - P \stackrel{(c)}{\geq} 1 - 2P, \label{eq:concentration_p_and_tilde_p_agree}
\end{align}
using in $(a)$ the definition of $\tilde{\Phi}_1$ in \cref{eq:tilde_phi_1_def}, in $(b)$ \cref{lem:newexists_s_st_ex_psi_linear}, Markov's inequality and union bound over $\tilde{t} - (t - 2n \log^4 n) \leq n^2$ steps and precondition \cref{eq:gap_and_condition_k_precondition} and in $(c)$ that $P \geq n^{-4}$. Note that this also implies that
\begin{align} \label{eq:concentration_starting_point_phi_1}
\Pro{\Phi_1^{t - n \log^4 n} \leq 6n^{11}} \geq 1 - 2P.
\end{align}

By \cref{lem:concentration_for_auxiliary_process}, for the auxiliary process $\tilde{\mathcal{P}}_{t_1}$, we have that
\[
\Pro{\left. \bigcap_{s \in [t, \tilde{t}]} \left\{ \Phi_{\tilde{\mathcal{P}}_{t_1}, 2}^{s} \leq 8n \right\} \,\,\right|\,\, \Phi_1^{t - n \log^4 n} \leq 6n^{11} } 
  \geq 1 - \frac{1}{4} \cdot (\log^8 n) \cdot P.
 \]
By combining with \cref{eq:concentration_starting_point_phi_1} (since by definition $\mathcal{P}$ and $\tilde{\mathcal{P}}_{t_1}$ agree in steps $s \leq t_1$) we get, 
\[
\Pro{\bigcap_{s \in [t, \tilde{t}]} \left\{ \Phi_{\tilde{\mathcal{P}}_{t_1}, 2}^{s} \leq 8n \right\}} 
  \geq \left( 1 - \frac{1}{4} \cdot (\log^8 n) \cdot P \right) \cdot \left( 1 - 2P\right) \geq 1 - \frac{1}{2} \cdot (\log^8 n) \cdot P.
\]
Finally, using $\Pro{A \cap B} \geq \Pro{A} - \Pro{\neg B}$ with \cref{eq:concentration_p_and_tilde_p_agree}, we conclude that $\mathcal{P}$ and $\tilde{\mathcal{P}}_{t_1}$ agree in all steps of the interval $[t_1, \tilde{t}]$ and so
\[
\Pro{\bigcap_{s \in [t, \tilde{t}]} \left\{ \Phi_2^{s} \leq 8n \right\}} 
  \geq  1 - \frac{1}{2} \cdot (\log^8 n) \cdot P  - 2P \geq 1 - (\log^8 n) \cdot P. \qedhere
\]
\end{proof}

\section{Upper Bound of \texorpdfstring{$\Oh(\frac{g}{\log g} \cdot \log \log n)$}{O(g/log g * log log n)} for  \texorpdfstring{$g$-\textsc{Adv-Comp}}{g-Adv-Comp} with \texorpdfstring{$g \leq \log n$}{g <= log n}: Layered Induction} \label{sec:g_adv_layered_induction}

In this section, we will complete the proof of the $\Oh\big(\frac{g}{\log g} \cdot \log \log n\big)$ upper bound for any $g \leq \log n$. We do this by first proving the key lemma for the drop of the super-exponential potentials (\cref{lem:gadv_precondition_satisfied}) and then complete the layered induction in \cref{thm:gadv_gap_bound}.

We begin by recalling from \cref{sec:g_adv_k_def} the definition for the number of steps in the layered induction $k$. Recall that $\alpha_1 := \frac{1}{6 \kappa}$, for $\kappa > 0$ the constant in \cref{lem:g_adv_good_gap_after_good_lambda} and $k := k(g) \geq 2$ is the unique integer satisfying,
\[
  (\alpha_1 \log n)^{1/k} \leq g < (\alpha_1 \log n)^{1/(k-1)}.
\]
Note that the integer $k$ always satisfies $k=\Oh(\log \log n)$, since $g > 1$. 

Further, recall that for any integer $1 \leq j \leq k-1$ we are using the super-exponential potentials $\Phi_j := \Phi_j(\phi_j, z_j)$ with smoothing parameters $\phi_j := \alpha_2 \cdot (\log n) \cdot g^{j - k}$ with $\alpha_2 := \frac{1}{84} \cdot \alpha_1$ and offsets $z_j := c_5 \cdot g + \big\lceil\frac{4}{\alpha_2} \big\rceil \cdot j \cdot g$ (for the constant $c_5 > 0$ defined in \cref{eq:g_adv_c5_def}), as well as the potentials $\Psi_j := \Psi_j(\psi_j, z_j)$ with smoothing parameters $\psi_j := \alpha_1 \cdot (\log n) \cdot g^{j-k}$.

We are now ready to prove the key lemma for the layered induction.

\begin{lem} \label{lem:gadv_precondition_satisfied}
Consider the \GAdvComp setting for any $g \in [\log(2C), \alpha_1 \log n]$, for the constant $C > 0$ defined in \cref{thm:g_adv_strong_base_case} and $\alpha_1 > 0$ defined in \cref{eq:g_adv_alpha_1_def}. Further, let $k := k(g) \geq 2$ be the unique integer such that $(\alpha_1 \log n)^{1/k} \leq g < (\alpha_1 \log n)^{1/(k-1)}$. Then, for any integer $1 \leq j \leq k - 1$ and any step $s \geq 0$, $\Phi_{j-1}^s \leq Cn$ implies $\mathcal{K}_{\psi_j, z_j}^s$. 
\end{lem}
\begin{proof}
Consider an arbitrary step $s$ with $\Phi_{j-1}^s \leq Cn$. Recall the definition of $\mathcal{K}_{\psi_j, z_j}^s$,
\[
\mathcal{K}_{\psi_j, z_j}^s(q^s) := \left\{ \forall i \in [n] \colon\  y_i^s \geq z_j - 1 \ \ \Rightarrow \ \ q_i^s \leq \frac{1}{n} \cdot e^{-\psi_j}\right\}.
\]
Thus, we want to bound the probability to allocate to a bin $i \in [n]$ with load $y_i^s \geq z_j - 1$. In order to do this, we will bound the number of bins $\ell \in [n]$ for which the adversary can reverse the comparison of $i$ and $\ell$, by bounding the ones with load $y_{\ell}^s \geq z_j - 1 - g$. Recall that $z_j := c_5 \cdot g + \big\lceil \frac{4}{\alpha_2} \big\rceil \cdot j \cdot g$. In the analysis below we make use of the following simple bound,
\begin{align} \label{eq:z_j_diff_bound}
z_j - 1 - g - z_{j-1} 
 = \left\lceil \frac{4}{\alpha_2} \right\rceil \cdot g - 1 - g
 \geq \frac{3}{\alpha_2} \cdot g,
\end{align}
using that $\alpha_2 \leq 1/2$. We consider the cases $j = 1$ and $j > 1$ separately as $\Phi_{0}$ has a slightly different form than $\Phi_{j-1}$ for $j > 1$.

\textbf{Case 1 [$j = 1$]:} The contribution of any bin $\ell \in [n]$ with load $y_{\ell}^s \geq z_1 - 1 - g$ to $\Phi_0^s$ is,
\begin{align*}
\Phi_{0,\ell}^s 
 = e^{\alpha_2 \cdot (y_{\ell}^s - z_0)^+}
 \geq e^{\alpha_2 \cdot (z_1 - 1 - g - z_0)^+}
 \stackrel{(\text{\ref{eq:z_j_diff_bound}})}{\geq} e^{3g}.
\end{align*}
Hence, when $\{ \Phi_{0}^s \leq Cn \}$ holds, the number of such bins is at most\[
Cn \cdot e^{-3g} = Cn \cdot e^{-g} \cdot e^{-2g} \leq \frac{n}{2} \cdot e^{-2g},
\]
using that $g \geq \log(2C)$. Hence, the probability of allocating a ball to a bin $i \in [n]$ with $y_i^s \geq z_1 - 1$ is at most that of sampling $i$ and a bin $\ell \in [n]$ with $y_{\ell}^s \geq z_1 - 1 - g$, i.e., at most \[
q_i^s \leq 2 \cdot \frac{1}{n} \cdot \frac{1}{2} \cdot e^{-2g} \leq \frac{1}{n} \cdot e^{-\psi_1},
\]
using that $\psi_1 := \alpha_1 \cdot (\log n) \cdot g^{1-k} \leq g$, as $g \geq (\alpha_1 \log n)^{1/k}$.

\textbf{Case 2 [$j > 1$]:} The contribution of any bin $\ell \in [n]$ with load $y_{\ell}^s \geq z_j - 1 - g$ to $\Phi_{j-1}^s$ is,
\begin{align*}
\Phi_{j-1, \ell}^s 
 = e^{\alpha_2 \cdot (\log n) \cdot g^{j-1-k} \cdot (y_{\ell}^s - z_{j-1})}
 \stackrel{(\text{\ref{eq:z_j_diff_bound}})}{\geq} e^{\alpha_2 \cdot (\log n) \cdot g^{j-1-k} \cdot (\frac{3}{\alpha_2} \cdot g)}
 = e^{ 3 \cdot (\log n) \cdot g^{j-k}}.
\end{align*}
Hence, when $\{ \Phi_{j-1}^{s} \leq Cn \}$ holds, the number of such bins is at most
\[
 Cn \cdot  e^{-3 \cdot (\log n) \cdot g^{j-k}} \leq \frac{n}{2} \cdot  e^{-2 \cdot (\log n) \cdot g^{j-k}},
\]
using that $
 e^{- (\log n) \cdot g^{j-k}} 
   \leq e^{-(\log n) \cdot g^{2-k}}
   \leq e^{-(\log n) \cdot g \cdot \frac{1}{\alpha_1 \log n}}
   \leq e^{-\log(2C)} = \frac{1}{2C}
$, since $j > 1$, $g \leq (\alpha_1 \log n)^{1/(k-1)}$, $g \geq \log(2C)$ and $\alpha_1 \leq 1$.

Hence, the probability of allocating a ball to a bin $i \in [n]$ with $y_i^s \geq z_j - 1$ is at most that of sampling $i$ and a bin $\ell \in [n]$ with $y_{\ell}^s \geq z_j - 1 - g$, i.e., at most \[
q_i^s \leq 2 \cdot \frac{1}{n} \cdot \frac{1}{2} \cdot e^{-2 \cdot (\log n) \cdot g^{j-k}} = \frac{1}{n} \cdot e^{-2 \cdot (\log n) \cdot g^{j-k}} \leq \frac{1}{n} \cdot e^{-\psi_j},
\]
recalling that $\psi_j := \alpha_1 \cdot (\log n) \cdot g^{j-k}$ for $\alpha_1 \leq 1$. 

Combining the two cases, we conclude that the event $\mathcal{K}_{\psi_j, z_j}^s$ holds at step $s$.
\end{proof}

\begin{thm} \label{thm:gadv_gap_bound} \label{thm:g_adv_strongest_bound}
Consider the \GAdvComp setting for any $g \in (1, \log n]$. Then, there exists a constant $\tilde{\kappa} > 0$ such that for any step $m \geq 0$, 
\[
\Pro{\Gap(m) \leq \tilde{\kappa} \cdot \frac{g}{\log g} \cdot \log \log n} \geq 1 - n^{-3}.
\]
\end{thm}

\begin{proof}
Let $g_{\min} := \max\big\{ \log(2C), \frac{\alpha_2}{4\sqrt{\alpha_1}} \big\}$. We consider three cases depending on the value of $1 \leq g \leq \log n$:

\textbf{Case 1 [$\min\big\{\frac{\alpha_2}{4}, c_6 \big\} \cdot \log n \leq g \leq  \log n$]:} (for $c_6 > 0$ as defined in \cref{eq:g_adv_c6_def}) In this case, the $\Oh(\log n)$ upper bound follows by the $\Oh(g + \log n)$ upper bound of \cref{thm:g_adv_g_plus_logn_gap}.

\textbf{Case 2 [$g < g_{\min}$]:} For $1 \leq g \leq g_{\min}$, $g$ is constant and the $\Oh(\log \log n)$ upper bound on the gap will follow by considering the $\tilde{g}$-\AdvComp setting with $\tilde{g} = \lceil g_{\min} \rceil$. This setting encompasses \GAdvComp as $\tilde{g} \geq g$ and is analyzed in Case 3.

\textbf{Case 3 [$g_{\min} \leq g < \min\big\{\frac{\alpha_2}{4}, c_6\big\} \cdot \log n$]:} Recall that for any $g < \alpha_1 \log n$, we defined the unique integer $k := k(g) \geq 2$ satisfying,
\begin{align*}
  (\alpha_1 \log n)^{1/k} \leq g < (\alpha_1 \log n)^{1/(k-1)},
\end{align*}
and as explained in \cref{eq:tilde_g_stronger_assumption}, since $g < \frac{\alpha_2}{4} \log n$, we may assume that the following stronger condition holds
\[
(\alpha_1 \log n)^{1/k} \leq g < \left(\frac{\alpha_2}{4} \log n\right)^{1/(k-1)},
\]
where the inequalities are valid using \cref{clm:tilde_g_justification} and that $g \geq \frac{\alpha_2}{4\sqrt{\alpha_1}}$.

Let $t_j := m - 2n(k-j) \log^4 n$ for any integer $0 \leq j \leq k - 1$. We will proceed by induction on the potential functions $\Phi_j$. The base case follows by applying~\cref{thm:g_adv_strong_base_case} (using $g \leq c_6 \log n$ and $t_0 \geq m - n \log^5 n$),
\begin{align} \label{eq:gadv_proof_base_case}
\Pro{\bigcap_{t \in [t_0, m]} \left\{ \Phi_0^{t} \leq Cn \right\} }  
  \geq 1 - n^{-4}.
\end{align}

We will now prove the induction step.
\begin{lem}[\textbf{Induction step}]\label{lem:new_inductive_step}
Consider the \GAdvComp setting for any $g \geq \max\big\{ \log(2C), \frac{\alpha_2}{4\sqrt{\alpha_1}} \big\}$ satisfying $(\alpha_1 \log n)^{1/k} \leq g < (\frac{\alpha_2}{4} \log n)^{1/(k-1)}$ for some integer $k \geq 2$, where $C > 0$ is the constant defined in \cref{thm:g_adv_strong_base_case} and $\alpha_1, \alpha_2 > 0$ are defined in \cref{eq:g_adv_alpha_1_def} and \cref{eq:g_adv_alpha_2_def}. Then, for any integer $1 \leq j \leq k-1$ and any step $m \geq 0$, if it holds that
\[
  \Pro{\bigcap_{t \in [t_{j-1}, m]} \{ \Phi_{j-1}^{t} \leq Cn\} } \geq 1 - \frac{(\log n)^{8(j-1)}}{n^4},
\]
then it also follows that
\[
  \Pro{ \bigcap_{t \in [t_j, m]} \{ \Phi_{j}^{t} \leq Cn \} } \geq 1 - \frac{(\log n)^{8j}}{n^4}.
\]
\end{lem}
\begin{proof}[Proof of \cref{lem:new_inductive_step}]
Consider an arbitrary integer $j$ with $1 \leq j \leq k - 1$ and assume that
\[
  \Pro{\bigcap_{t \in [t_{j-1},m]} \{ \Phi_{j-1}^{t} \leq Cn\} } \geq 1 - \frac{(\log n)^{8(j-1)}}{n^4}.
\]
By \cref{lem:gadv_precondition_satisfied}, we have that $\{ \Phi_{j-1}^t \leq Cn \}$ implies $\mathcal{K}_{\psi_j, z_j}^t$. Furthermore, $\{ \Phi_{j-1}^{t_{j-1}} \leq Cn \}$ also implies $\{ \Gap(t_{j-1}) \leq \log^2 n \}$. Hence, it also holds that
\begin{align} \label{eq:precondition_3_application}
\Pro{\left\{ \Gap(t_{j-1}) \leq \log^2 n \right\} \cap \bigcap_{t \in [t_{j-1}, m]} \mathcal{K}_{\psi_j, z_j}^t} \geq 1 - \frac{(\log n)^{8(j-1)}}{n^4}.
\end{align}
Applying \cref{lem:general_drop_superexponential} for the potentials $\Phi_j$ and $\Psi_j$, since $\psi_j \geq \phi_j = \alpha_2 \cdot (\log n) \cdot g^{j-k} \geq \alpha_2 \cdot (\log n) \cdot g^{1-k} \geq 4$ (as $g < (\frac{\alpha_2}{4} \log n)^{1/(k-1)}$), for any step $t \geq 0$ it holds that
\begin{align} \label{eq:phi_1_drop_precondition_application}
\Ex{\left. \Phi_j^{t+1} \,\right|\, \mathfrak{F}^t, \mathcal{K}_{\psi_j, z_j}^t } \leq \Phi_j^{t} \cdot \Big(1 - \frac{1}{n} \Big) + 2,
\end{align}
and 
\begin{align} \label{eq:phi_2_drop_precondition_application}
\Ex{\left. \Psi_j^{t+1} \,\right|\, \mathfrak{F}^t, \mathcal{K}_{\psi_j, z_j}^t } \leq \Psi_j^{t} \cdot \Big(1 - \frac{1}{n} \Big) + 2.
\end{align}
Hence, by \cref{eq:phi_1_drop_precondition_application}, \cref{eq:phi_2_drop_precondition_application} and \cref{eq:precondition_3_application}, the preconditions of \cref{thm:super_exp_potential_concentration} are satisfied for starting step $t_j := m - 2n(k - j) \cdot \log^4 n$, $P := (\log n)^{8(j-1)}/n^4$ and terminating step at $\tilde{t} := m$, and so we conclude (since $C \geq 8$) that
\[
\Pro{\bigcap_{t \in [t_j, m]} \left\{ \Phi_j^t \leq Cn \right\}} \geq 1 - \frac{(\log n)^{8j}}{n^4}. \qedhere
\]
\end{proof}

Returning to the proof of \cref{thm:g_adv_strongest_bound}, inductively applying \cref{lem:new_inductive_step} for $k-1$ times and using \cref{eq:gadv_proof_base_case} as a base case, we get that
\[
\Pro{\bigcap_{t \in [t_{k-1}, m]} \left\{\Phi_{k-1}^t \leq Cn \right\}} \geq 1 - \frac{(\log n)^{8(k-1)}}{n^4} \geq 1 - n^{-3},
\]
using in the last step that $k = \Oh(\log \log n)$. When $\{ \Phi_{k-1}^m \leq Cn \}$ occurs, the gap at step $m$ cannot be more than $z_k := c_5g + \big\lceil \frac{4}{\alpha_2} \big\rceil \cdot k \cdot g$, since otherwise we would get a contradiction
\begin{align*}
Cn \geq \Phi_{k-1}^{m} & \geq \exp \Big(\alpha_2 \cdot (\log n) \cdot g^{(k-1)-k} \cdot (z_k - z_{k-1} ) \Big) \\
 & = 
 \exp\Big(\alpha_2 \cdot (\log n) \cdot g^{-1} \cdot \Big( \Big\lceil \frac{4}{\alpha_2} \Big\rceil \cdot g \Big)\Big) \geq  \exp(4 \cdot \log n) =n^4.
\end{align*}
Hence, $\Gap(m) \leq z_k = c_5 g + \big\lceil \frac{4}{\alpha_2} \big\rceil \cdot k \cdot g$. 

By the assumption on $g$, we have
\[
g < \left(\frac{\alpha_2}{4} \log n\right)^{1/(k-1)} ~ \Rightarrow ~ \log g < \frac{1}{k-1} \cdot \log \left(\frac{\alpha_2}{4} \log n\right) ~\Rightarrow ~  k < 1 + \frac{\log \left(\frac{\alpha_2}{4} \log n\right)}{\log g},
\]
using that $g > 1$. Since $\alpha_2 > 0$ and $c_5 > 0$ are constants, we conclude that there exists a constant $\tilde{\kappa} > 0$ such that
\[
\Pro{\Gap(m) \leq \tilde{\kappa} \cdot \frac{g}{\log g} \cdot \log \log n } \geq 1 - n^{-3}. \qedhere
\]
\end{proof}

In the above, we actually proved the following slightly stronger corollary, which we will use in \cref{sec:upper_bounds_for_delay_settings}. 
This is based on the insight that, in order to prove the above gap bound at step $m$, the only assumption on the steps $[0,t_0)$, with $t_0:=m - n \log^5 n - \Delta_r$, is that the coarse gap bound of $\Oh(g \log (ng))$ must hold at step $t_0$ (see, e.g., \cref{lem:g_adv_recovery} and \cref{thm:g_adv_strong_base_case}). While during the interval $[t_0,m]$ the process is required to be an instance of \GAdvComp, in the interval $[0,t_0)$ the process can be arbitrary as long as the coarse gap bound holds at step $t_0$.

\begin{cor} \label{cor:upper_bound_tight}
Consider any $g \in (1, \log n]$, $m \geq 0$, $t_0 := m - n \log^5n - \Delta_r$ for $\Delta_r := \Delta_r(g) > 0$ as defined in \cref{lem:g_adv_recovery} and $c_3 > 0$ the constant in \cref{thm:g_adv_warm_up_gap}. Further, consider a process which, in steps $[t_0, m]$, is an instance of \GAdvComp setting. Then, there exists a constant $\tilde{\kappa} > 0$, such that 
\[
\Pro{\left. \Gap(m) \leq \tilde{\kappa} \cdot \frac{g}{\log g} \cdot \log \log n \, \right\vert \, \mathfrak{F}^{t_0}, \, \max_{i \in [n]} \left| y_i^{t_0} \right| \leq c_3 g \log(ng)} \geq 1 - n^{-3}.
\]
\end{cor}

\section{Upper Bounds for Probabilistic Noise and Delay Settings}\label{sec:g_adv_delay_noise_settings}

In this section, we present upper bounds for the probabilistic noise settings and the delay settings. These upper bounds follow from the results (and analysis) for \GAdvComp in Sections~\ref{sec:g_adv_warm_up}-\ref{sec:g_adv_layered_induction}.

\subsection{An Upper Bound for the Probabilistic Noise Setting}

\begin{pro}\label{pro:prob_upper_bound}
Consider the $\rho$\textsc{-Noisy-Comp} setting with $\rho(\delta)$ being any non-decreasing function in $\delta$ with $\lim_{\delta \rightarrow \infty} \rho(\delta) =1$. For any $n \in \mathbb{N}$, define $\delta^*:=\delta^*(n)=\min\{ \delta \geq 1 \colon \rho(\delta) \geq 1 - n^{-4}\}$. Then, there exists a constant $\kappa > 0$, such that for any step $m \geq 0$,
\begin{align*}
    \Pro{\max_{i \in [n]} \left|y_i^m\right| \leq \kappa \cdot \delta^* \log ( n \delta^* ) } \geq 1-n^{-3}. 
\end{align*}
\end{pro}
Note that for the \SigmaNoisyLoad process where $\rho(\delta)$ has
Gaussian tails (see \cref{eq:gaussian}), we have $\delta^*=\Oh(\sigma \cdot \sqrt{\log n})$.
The choice of $\delta^*$ in \cref{pro:prob_upper_bound}~ensures that in most steps, all possible comparisons among bins with load difference greater than $\delta^*$ will be correct, implying that the process satisfies the condition of \GAdvComp with $g=\delta^*$.

\begin{proof}
We will analyze the hyperbolic cosine potential $\Gamma := \Gamma(\gamma)$ as defined in \cref{eq:gamma_def}, with $\gamma := - \log(1-\frac{1}{8 \cdot 48})/\delta^*$.
We first state a trivial upper bound on
$
 \Ex{ \left. \Gamma^{t+1} \, \right| \, \mathfrak{F}^{t}}
$
in terms of $\Gamma^t$, which holds deterministically for all steps $t \geq 0$ (cf.~\cref{lem:g_adv_v_smoothness}~$(i)$),
\begin{align*}
 \Gamma^{t+1} 
= \sum_{i=1}^{n} \Gamma_i^{t+1} \leq 
\sum_{i=1}^n e^{\gamma} \cdot \Gamma_i^{t} = e^{\gamma} \cdot \Gamma^{t} \leq (1 + 2\gamma) \cdot \Gamma^t,
\end{align*}
using that $e^{\gamma} \leq 1 + 2\gamma$ for $0 < \gamma \leq 1$.

We will now provide a better upper bound, exploiting that with high probability all possible comparisons between bins that differ by at least $\delta^*$ will be correct. Again, consider any step $t\geq 0$. Let us assume that in step $t$, we  first determine the outcome of the load comparisons among all $n^2$ possible bin pairs. Only then we sample two bins, and allocate the ball following the pre-determined outcome of the load comparison. For any two bins $i_1, i_2 \in [n]$ with $|x_{i_1}^{t}-x_{i_2}^{t}| \geq \delta^*$, we have
\[
 \rho(|x_{i_1}^{t}-x_{i_2}^{t}|) \geq \rho(\delta^*) \geq 1-n^{-4}.
\]
Hence by the union bound over all $n^2$ pairs, we can conclude that with probability at least $1-n^{-2}$, all comparisons among bin pairs with load difference at least $\delta^*$ are correct. Let us denote this event by $\mathcal{G}^{t}$, so $\Pro{ \mathcal{G}^{t} } \geq 1-n^{-2}$ for every step $t \geq 0$. Conditional on $\mathcal{G}^t$, the process in step $t$ is an instance of \GAdvComp with $g=\delta^*$. Therefore, by \cref{thm:g_adv_warm_up_gap}~$(i)$, there exists a constant $c_1 \geq 1$, so that the hyperbolic cosine potential satisfies
\[
 \Ex{ \left. \Gamma^{t+1} \, \right| \, \mathfrak{F}^t, \mathcal{G}^t } \leq \Gamma^{t} \cdot \left(1 - \frac{\gamma}{96 n} \right) + c_1.
\]

Now combining our two upper bounds on $\Ex{ \left. \Gamma^{t+1} \, \right| \, \mathfrak{F}^t}$, we conclude 
\begin{align*}
\Ex{ \left. \Gamma^{t+1} \, \right| \, \mathfrak{F}^{t}} &\leq 
\Ex{ \left. \Gamma^{t+1} \, \right| \, \mathfrak{F}^t, \mathcal{G}^t } \cdot \Pro { \mathcal{G}^t }
+ \Ex{\left. \Gamma^{t+1} \, \right| \, \mathfrak{F}^t} \cdot (1-\Pro { \mathcal{G}^t }) \\
&\stackrel{(a)}{\leq} 
\Ex{ \left. \Gamma^{t+1} \, \right| \, \mathfrak{F}^t, \mathcal{G}^t } \cdot (1-n^{-2})
+ \Ex{\left. \Gamma^{t+1} \, \right| \, \mathfrak{F}^t} \cdot n^{-2} \\
&\leq \left( \Gamma^{t} \cdot \left(1 - \frac{\gamma}{96n} \right)  + c_1 \right) \cdot (1 - n^{-2}) + (1 + 2\gamma) \cdot \Gamma^{t} \cdot n^{-2} \\
&\leq \Gamma^{t} \cdot \left(1 - \frac{\gamma}{96n}\right) + c_1 
 - \Gamma^{t} \cdot n^{-2} + \Gamma^t \cdot \frac{\gamma}{96n} \cdot n^{-2} + \Gamma^{t} \cdot n^{-2} + \Gamma^t \cdot 2\gamma \cdot n^{-2} \\
&\leq \Gamma^{t} \cdot \left(1 - \frac{\gamma}{100 n} \right)  + c_1,
\end{align*}
using in $(a)$ that $\Pro { \mathcal{G}^t } \geq 1 - n^{-2}$. So, using \cref{lem:geometric_arithmetic}~$(ii)$ (with $a = 1 - \frac{\gamma}{100n}$ and $b = c_1$) since $\Gamma^{0}=2n \leq \frac{100 c_1}{\gamma} \cdot n$, we get for any $t \geq 0$,
\[
\Ex{\Gamma^t} \leq \frac{100c_1}{\gamma} \cdot n.
\]
Using Markov's inequality yields,
$
 \Pro{ \Gamma^t > \frac{100 c_1}{\gamma} \cdot n^4} \leq n^{-3}.
$
Now the claim follows, since the event $\left\{ \Gamma^t \leq \frac{100 c_1}{\gamma} \cdot n^4 \right\}$ for $\gamma = \Theta\big(\frac{1}{\delta^*} \big)$, implies that %
\[
\max_{i \in [n]} \left| y_i^t \right| \leq \frac{1}{\gamma} \log \left(\frac{100 c_1}{\gamma} \cdot n^4 \right) = \Oh( \delta^* \cdot \log(n \delta^*) ). \qedhere
\]
\end{proof}

\subsection{Upper Bounds for Delay Settings} \label{sec:upper_bounds_for_delay_settings}

In this section, we will prove tight upper bounds for the \TwoChoice process in the \TauDelay setting and the \BBatch process for a range of values for the delay parameter $\tau$ and batch size $b$. In particular for $\tau = n$, we show:

\begin{thm} \label{thm:batching}
Consider the \TauDelay setting with $\tau = n$. Then, there exists a constant $\kappa > 0$ such that for any step $m \geq 0$,
\[
 \Pro{\Gap(m) \leq \kappa \cdot \frac{\log n}{\log \log n}} \geq 1 - n^{-2}.
\]
\end{thm}

This implies the same upper bound for the \BBatch process for $b = n$, since it is an instance of the \TauDelay setting with $\tau = n$. This bound improves the $\Oh(\log n)$ bound in~\cite[Theorem 1]{BCEFN12} and can be easily seen to be asymptotically tight due to the $\Omega\big({\scriptstyle \frac{\log n}{\log \log n}}\big)$ lower bound for \OneChoice for with $n$ balls (\cref{obs:two_choice_batching_lower_bound_small_b}).

We will analyze \TauDelay using a more general approach which also works for other choices of the delay parameter $\tau \leq n \log n$. First note that \TauDelay is an instance of $g_1$-\AdvComp with $g_1:=\tau-1$. This follows since for all steps, a bin could sampled (and be allocated to) at most $\tau-1$ times during the last $\tau-1$ steps. However, for a typical step we expect each bin to be incremented much less frequently during the last $\tau-1$ steps, and thus the process is with high probability an instance of $g_2$-\AdvComp for some $g_2 \ll g_1$. 

To make this more specific, consider an arbitrary instance of $g_1$-\AdvComp and define $D^t$ as the set of pairs of bins (of unequal load) whose comparison is reversed with non-zero probability by the adversary in step $t + 1$,
\[
D^t(\mathfrak{F}^t) := \left\{ 
(i, j) \in [n] \times [n] : y_{i}^t > y_{j}^t \wedge \Pro{A^{t+1}(\mathfrak{F}^t, i, j) = i } > 0 
\right\},
\]
and then define the largest load difference that could be reversed by the process at step $t+1$,
\[
g^t(\mathfrak{F}^t) := \max_{ (i, j) \in D^t(\mathfrak{F}^t) } \left| y_i^t - y_j^t \right|.
\]
This can be seen as the ``effective $g$-bound'' of the process in step $t+1$.
\begin{lem} \label{lem:key_batching}
Consider the \textsc{$g_1$-Adv-Comp} setting for any $g_1 \in [1, n \log n]$, and consider any $g_2 \in [1, \log^2 n]$. If for every step $t \geq 0$, we have
\[
\Pro{g^t \leq g_2} \geq 1 - n^{-6},
\]
then there exists a constant $\tilde{\kappa} > 0$ such that for any step $m \geq 0$,
\[
\Pro{\Gap(m) \leq \tilde{\kappa} \cdot \frac{g_2}{\log g_2} \cdot \log \log n} \geq 1 - n^{-2}.
\]
\end{lem}

We will now use \cref{lem:key_batching} to prove \cref{thm:batching}.

\begin{proof}[Proof of \cref{thm:batching}]
Any bin can be allocated at most $n-1$ times during an interval of $\tau-1=n-1$ steps, so an adversary with $g_1 = n-1$ who also remembers the entire history of the process, can simulate \TauDelay at any step $t$, by using the allocation information in steps $[t - n, t]$.

To obtain the bound for $g_2$, note that in \TwoChoice, at each step two bins are sampled for each ball. So in $n-1$ steps of \TwoChoice, there are $2(n-1)$ bins sampled using \OneChoice. By the properties of \OneChoice (\cref{lem:one_choice_lightly}), we have that for any consecutive $n-1$ allocations, with probability at least $1 - n^{-6}$ we sample (and allocate to) no bin more than $11 \log n/\log \log n$ times. In such a sequence of bin samples, we can simulate \TauDelay using $g_2$-\AdvComp with $g_2 = 11 \log n/\log \log n$. Hence, for any step $t \geq 0$,
\[
\Pro{g^t \leq g_2} \geq 1 - n^{-6}.
\]
Since the precondition of \cref{lem:key_batching} holds for $g_1 = n-1$ and $g_2 = 11 \log n/\log \log n$, we get that there exists a constant $\kappa > 0$, such that 
\[
\Pro{\Gap(m) \leq \kappa \cdot \frac{\log n}{\log \log n}} \geq 1 - n^{-2}. \qedhere
\]
\end{proof}

The same argument also applies for any $\tau \in [n \cdot e^{-\log^c n}, n \log n]$. For \OneChoice with $2\tau$ balls the gap is \Whp~$\polylog(n)$ (e.g., see \cref{lem:one_choice_n_polylog_n}) and so we can apply \cref{lem:key_batching} with $g_2 = \polylog(n)$, to obtain the gap bound of
\[
\frac{g_2}{\log g_2} \cdot \log \log n = \Theta(g_2).
\]
\begin{cor} \label{cor:delay_polylog_n_upper_bound}
There exists a constant $\kappa > 0$, such that the \TauDelay setting with any $\tau \in [n \cdot e^{-\log^c n}, n \log n]$, where $c > 0$ is any constant, for any step $m \geq 0$, it holds that
\[
\Pro{\Gap(m) \leq \kappa \cdot \frac{\log n}{\log\left(\frac{4n}{\tau} \cdot \log n \right)}} \geq 1 - n^{-2}.
\]
\end{cor}
\begin{rem} \label{rem:delay_polylog_n_lower_bound}
A matching lower bound holds for the \BBatch process for any batch size $b \in [n \cdot e^{-\log^c n}, n \log n]$. This follows by the lower bound for \OneChoice with $b$ balls (e.g., see \cref{lem:one_choice_lower_bound_gap_whp_lightly}) which matches the gap of \BBatch in the first batch (\cref{obs:two_choice_batching_lower_bound_small_b}). 
\end{rem}

Therefore, \cref{cor:delay_polylog_n_upper_bound} and \cref{rem:delay_polylog_n_lower_bound} establish that \Whp~$\Gap(m) = \Theta\big({\scriptstyle \frac{\log n}{\log((4n/b) \cdot \log n)}}\big)$ for the \BBatch process for any $b \in [n \cdot e^{-\log^c n}, n \log n]$. However, the following remark (which also applies to the \TauDelay setting), establishes that there are regions where the \BBatch process has an asymptotically worse gap than \OneChoice with $b$ balls (see also some empirical results in \cref{fig:bbatch_vs_one_choice} in \cref{sec:experiments}):%

\begin{rem} \label{rem:delay_very_very_lightly_loaded}
For any $\tau$ (or $b$) being $n^{1-\eps}$ for any constant $\eps \in (0,1)$, the \OneChoice process has $\Gap(b) = \Oh(1)$ \Whp~(see \cref{lem:very_very_lightly_loaded}). Hence, by \cref{lem:key_batching} with $g_2 = \Oh(1)$, \TauDelay has for any step $m \geq 0$, $\Gap(m) = \Oh(\log \log n)$ \Whp, which is asymptotically tight by \cref{obs:g_adv_loglogn_lower_bound} for $m=n$.
\end{rem}

We now return to proving \cref{lem:key_batching}.

\begin{proof}[Proof of \cref{lem:key_batching}]
Let $\mathcal{P}$ be a process satisfying the preconditions in the statement and we will define the auxiliary process $\tilde{\mathcal{P}}_{t_0}$ for some step $t_0 \geq 0$ (to be specified below). Consider the stopping time $\sigma := \inf\{ s \geq t_0 : g^s > g_2 \}$. Then, the auxiliary process $\tilde{\mathcal{P}}_{t_0}$ is defined so that
\begin{itemize}
  \item in steps $s \in [0, \sigma)$, it makes the same allocations as $\mathcal{P}$, and
  \item in steps $s \in [\sigma, \infty)$, it makes the same allocations as the $g_2$-\textsc{Bounded} process.
\end{itemize}
This way $\tilde{\mathcal{P}}_{t_0}$ is a $g_2$-\AdvComp process for all steps $s \geq t_0$. Let $\tilde{y}$ be the normalized load vector for $\tilde{\mathcal{P}}_{t_0}$, then it follows by the precondition that \Whp~the two processes agree for any interval $[t_0, m]$ with $m - t_0 \leq n^3$, i.e.,
\begin{align} \label{eq:p_and_tilde_p_agree}
\Pro{\bigcap_{s \in [t_0, m]} \left\{ y^s = \tilde{y}^s \right\}}
  \geq \Pro{\bigcap_{s \in [t_0, m]} \left\{ g^s \leq g_2 \right\}}
  \geq 1 - n^{-6} \cdot n^3 = 1 - n^{-3}.
\end{align}

For $m \leq n^3$, the upper bound follows directly by \cref{thm:g_adv_strongest_bound} for $\tilde{\mathcal{P}}_0$ and taking the union bound with \cref{eq:p_and_tilde_p_agree}, i.e., that $\tilde{\mathcal{P}}_0$ agrees with $\mathcal{P}$.

For $m > n^3$, the analysis is slightly more challenging. We need to show that the process recovers from the weak upper bound obtained from the $g_1$-\AdvComp setting. Let $\Gamma := \Gamma(\gamma)$ be as defined in \cref{eq:gamma_def} with $\gamma := - \log(1 - \frac{1}{8 \cdot 48})/g_2$ for the $\tilde{\mathcal{P}}_{t_0}$ process (i.e., the $\tilde{y}$ load vector). Also, let $t_0 := m - n^3$ and $t_1 := m - n \log^5 n - \Delta_r$ where $\Delta_r := \Delta_r(g_2) = \Theta(n g_2 (\log(ng_2))^2)$ is the recovery time defined in \cref{lem:g_adv_recovery}. 
In this analysis, we consider the following three phases (see \cref{fig:g1_and_g2_simulation}):
\begin{itemize}
  \item $[0, t_0]$: The process $\mathcal{P}$ is an instance of the $g_1$-\AdvComp setting with $g_1 = n \log n$. Hence, by \cref{thm:g_adv_warm_up_gap}, it follows that at step $t_0$ \Whp~$\Gap(t_0) = \Oh(n \log^2 n)$.
  \item $(t_0, t_1]$: The process $\mathcal{P}$ \Whp~agrees with $\tilde{\mathcal{P}}_{t_0}$ which is a $g_2$-\AdvComp process. We will use an analysis similar to that in \cref{sec:g_adv_warm_up} to prove that \Whp~$\Gap(t_1) = \Oh(g_2 \log (ng_2))$.
  \item $(t_1, m]$: The process $\mathcal{P}$ \Whp~continues to agree with $\tilde{\mathcal{P}}_{t_0}$ which is a $g_2$-\AdvComp process and so by \cref{cor:upper_bound_tight} this implies that \Whp~$\Gap(m) = \Oh\big(\frac{g_2}{\log g_2} \cdot \log \log n\big)$.
\end{itemize}

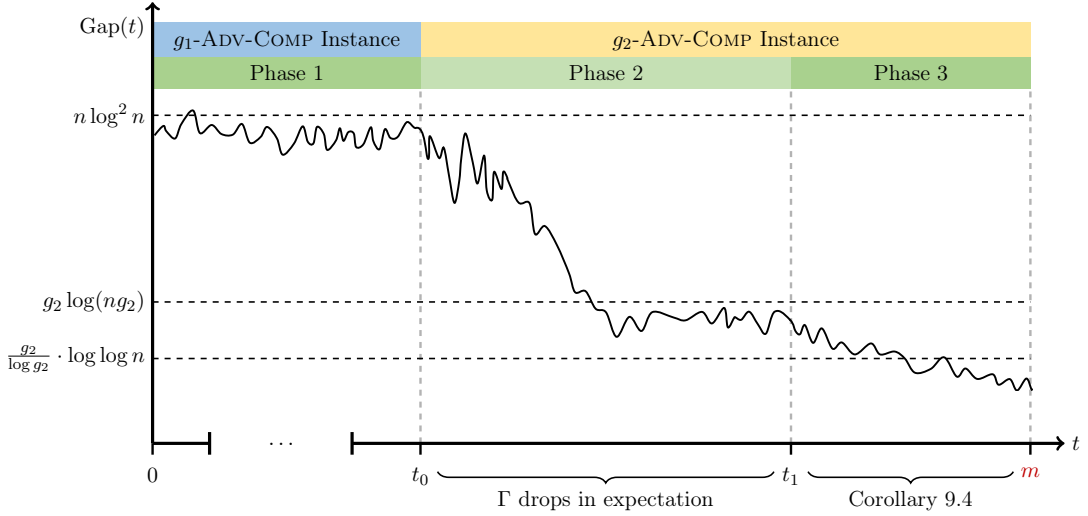
\begin{figure}
    \centering

\scalebox{0.75}{
\begin{tikzpicture}[
  IntersectionPoint/.style={circle, draw=black, very thick, fill=black!35!white, inner sep=0.05cm}
]

\definecolor{MyBlue}{HTML}{9DC3E6}
\definecolor{MyYellow}{HTML}{FFE699}
\definecolor{MyGreen}{HTML}{E2F0D9}
\definecolor{CaseTwoGreen}{HTML}{A9D18E}
\definecolor{MyRed}{HTML}{FF9F9F}
\definecolor{MyDarkRed}{HTML}{C00000}
\definecolor{MyOrange}{HTML}{C55A11}
\definecolor{MyLightGreen}{HTML}{C5E0B4}

\def\xEnd{16}
\def\xLast{15.40}
\def\yLast{6}
\def\ThresholdOne{1}
\def\ThresholdTwo{2.5}
\def\yBottom{-0.8}

\def\tA{4.7}
\def\tB{11.2}

\def\yA{5}
\def\yB{1.7}
\def\yC{0.7}

\node[anchor=south west, inner sep=0.15cm, fill=MyBlue, rectangle, minimum width=\tA cm] at (0, \yLast) {$g_1$-\AdvComp Instance};
\node[anchor=south west, inner sep=0.15cm, fill=MyYellow, rectangle, minimum width=\xLast cm - \tA cm] at (\tA, \yLast) {$g_2$-\AdvComp Instance};

\newcommand{\drawLine}[3]{
\draw[dashed, very thick, #3] (#1, \yBottom) -- (#1, \yLast -0.54);
\draw[very thick] (#1, \yBottom) -- (#1, \yBottom -0.2);
\node[anchor=north] at (#1, \yBottom -0.3) {#2};}

\draw[very thick] (0, \yBottom) -- (0, \yBottom -0.2);
\node[anchor=north] at (0, \yBottom -0.3) {$0$};
\drawLine{\tA}{$t_0$}{black!30!white};
\drawLine{\tB}{$t_1$}{black!30!white};
\drawLine{\xLast}{\textcolor{MyDarkRed}{$m$}}{black!30!white};

\node[anchor=east] at (0, 6.55) {$\Gap(t)$};

\node[anchor=south west, inner sep=0.15cm, fill=CaseTwoGreen, rectangle, minimum width=\tA cm] at (0, \yLast -0.54) {Phase 1};
\node[anchor=south west, inner sep=0.15cm, fill=MyLightGreen, rectangle, minimum width=\tB cm - \tA cm] at (\tA, \yLast -0.54) {Phase 2};
\node[anchor=south west, inner sep=0.15cm, fill=CaseTwoGreen, rectangle, minimum width=\xLast cm - \tB cm] at (\tB, \yLast -0.54) {Phase 3};

\draw[dashed, thick] (0,\yA) -- (\xLast, \yA);
\draw[dashed, thick] (0, \yB) -- (\xLast, \yB);
\draw[dashed, thick] (0, \yC) -- (\xLast, \yC);

\node[anchor=east] at (0, \yA) {$n \log^2 n$};
\node[anchor=east] at (0, \yB) {$g_2 \log(ng_2)$};
\node[anchor=east] at (0, \yC) {$\frac{g_2}{\log g_2} \cdot \log \log n$};

\node[anchor=west] at (\xEnd, \yBottom) {$t$};

\newcommand{\drawBrace}[4]{
\draw [
    thick,
    decoration={
        brace,
        raise=0.5cm,
        amplitude=0.2cm
    },
    decorate
] (#2 - 0.3, \yBottom) -- (#1 + 0.3, \yBottom) 
node [anchor=north,yshift=-0.7cm,#4] {#3}; }

\drawBrace{\tA}{\tB}{$\Gamma$ drops in expectation}{pos=0.5};
\drawBrace{\tB}{\xLast}{\cref{cor:upper_bound_tight}}{pos=0.5};

\draw[->, ultra thick] (0,\yBottom) -- (0, \yLast + 1);
\draw[->, ultra thick] (0,\yBottom) -- (\xEnd, \yBottom);

\node[fill=white, rectangle, anchor=west, minimum width=2.5cm, minimum height=0.2cm] at (1, \yBottom) {$\mathbf{\ldots}$};
\draw[ultra thick] (1, \yBottom + 0.2) -- (1, \yBottom - 0.2);
\draw[ultra thick] (3.5, \yBottom + 0.2) -- (3.5, \yBottom - 0.2);

\node[anchor=south west] (plt) at (-0.1, 0) {\includegraphics[width=15.4cm]{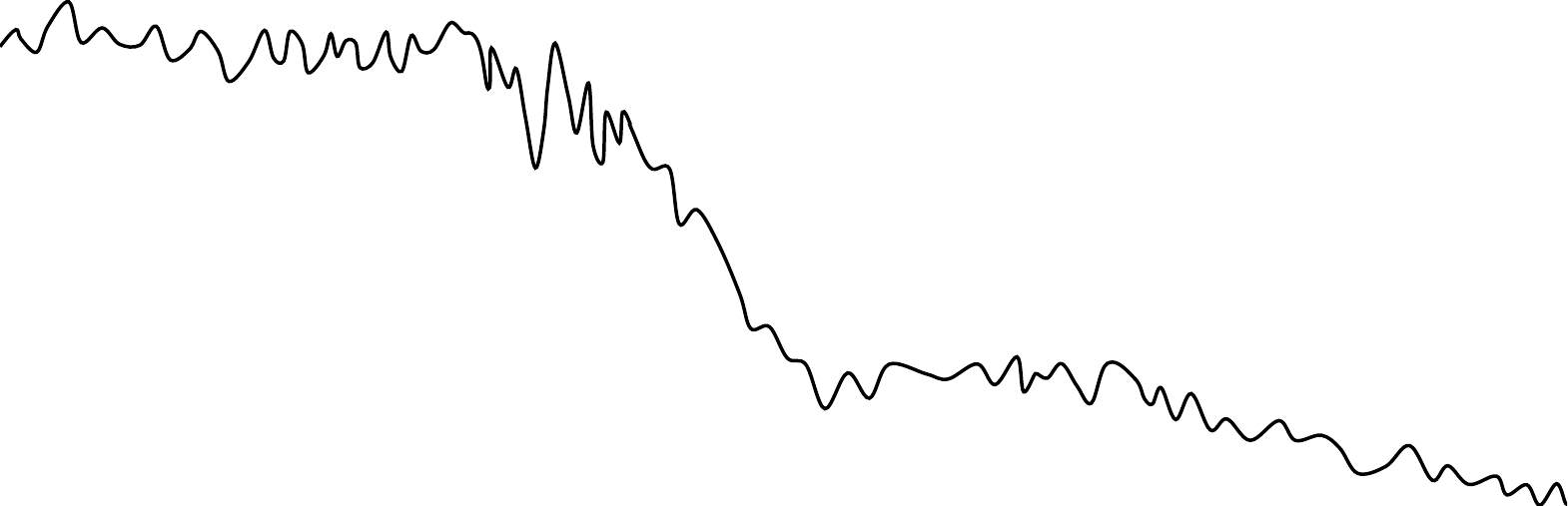}};
\end{tikzpicture}
}

    \caption{The three phases in the proof of \cref{lem:key_batching}.}
    \label{fig:g1_and_g2_simulation}
\end{figure}

\textbf{Phase 1 $[0, t_0]$:} Using \cref{thm:g_adv_warm_up_gap}~$(iii)$ (for $g := n \log n$) and for $c_3 > 0$ being the constant defined in \cref{eq:g_adv_c3_def}, we have that,
\begin{equation} \label{eq:relaxed_bound_at_t0}
\Pro{\max_{i \in [n]} \left| \tilde{y}_i^{t_0} \right| \leq 3 c_3 n \log^2 n} \geq 1 - n^{-14},
\end{equation}
using $c_3 g \log(ng) \leq c_3 n \cdot \log n \cdot \log (n^2 \cdot \log n) \leq 3 c_3 n \log^2 n$.

\textbf{Phase 2 $(t_0, t_1]$:} Let us now turn our attention to the interval $(t_0, t_1]$, where $\tilde{\mathcal{P}}_{t_0}$ is a $g_2$-\AdvComp process. So, by \cref{thm:g_adv_warm_up_gap}~$(i)$ (for $g := g_2$), there exists a constant $c_1 \geq 1$, such that for any step $t \geq t_0$
\[
 \Ex{\left. \Gamma^{t+1} \,\right|\, \mathfrak{F}^t} \leq \Gamma^t \cdot \Big( 1 - \frac{\gamma}{96n} \Big) + c_1.
\]
At step $t_0$, when $\{ \max_{i \in [n]} |\tilde{y}_i^{t_0}| \leq 3c_3 n \log^2 n \}$ holds, we also have that $\Gamma^{t_0} \leq 2n \cdot e^{3\gamma c_3 n \log^2 n}$. Hence, applying \cref{lem:geometric_arithmetic}~$(i)$ (with $a = 1 - \frac{\gamma}{96n}$ and $b = c_1$), for step $t_1$ we have
\begin{align*}
\Ex{ \Gamma^{t_1} \,\left|\, \mathfrak{F}^{t_0}, \max_{i \in [n]} \left| \tilde{y}_i^{t_0}\right| \leq 3 c_3 n \log^2 n \right.}
 & \leq \Ex{ \Gamma^{t_1} \,\left|\, \mathfrak{F}^{t_0}, \Gamma^{t_0} \leq 2n \cdot e^{3\gamma c_3 n \log^2 n} \right.} \\
 & \leq \Gamma^{t_0} \cdot \Big( 1 - \frac{\gamma}{96n} \Big)^{t_1-t_0} + \frac{96c_1}{\gamma} \cdot n \\
 & \stackrel{(a)}{\leq} 2n \cdot e^{3\gamma c_3 n \log^2 n} \cdot e^{-\frac{\gamma}{96n} \cdot \frac{1}{2} n^3 } + \frac{96c_1}{\gamma} \cdot n \\
 & \leq \frac{100c_1}{\gamma} \cdot n.
\end{align*}
using in $(a)$ that $e^u \geq 1 + u$ and $t_1 - t_0 \geq \frac{1}{2} n^3$. By Markov's inequality, we have that,
\[
\Pro{\left. \Gamma^{t_1} \leq \frac{100c_1}{\gamma} \cdot n^{4} \,\, \right\vert \,\, \mathfrak{F}^{t_0}, \max_{i \in [n]} \left| \tilde{y}_i^{t_0}\right| \leq 3c_3 n \log^2 n  } \geq 1 - n^{-3}.
\]
When the event $\left\{ \Gamma^{t_1} \leq \frac{100c_1}{\gamma} \cdot n^{4} \right\}$ holds, it implies that
\begin{align*}
\Gap(t_1) &\leq \frac{1}{\gamma} \cdot \left( \log\left( \frac{100c_1}{\gamma} \right) + 4 \log n \right)  
\stackrel{(a)}{=} \frac{c_3 g_2}{16} 
\cdot \left( \Oh(1) + \log \left( \frac{c_3 g_2}{16} \right) + 4 \log n \right) 
\leq c_3 g_2 \log(ng_2),
\end{align*}
using in (a) that $c_3 := \frac{16}{\gamma g_2}$
(defined in~\cref{eq:g_adv_c3_def}). Therefore, 
\begin{align} \label{eq:relaxed_bound_at_t1}
\Pro{\left. \max_{i \in [n]} \left| \tilde{y}_i^{t_1} \right| \leq c_3 g_2 \log(ng_2) \,\, \right\vert \,\, \mathfrak{F}^{t_0}, \max_{i \in [n]} \left|\tilde{y}_i^{t_0}\right| \leq 3 c_3 n \log^2 n  } \geq 1 - n^{-3}.
\end{align}

\textbf{Phase 3 $(t_1, m]$:} Now, we turn our attention to the steps in $(t_1, m]$, where $\tilde{\mathcal{P}}_{t_0}$ is again a $g_2$-\AdvComp process.
Therefore, applying \cref{cor:upper_bound_tight} (for $t_0 := t_1 = m - n \log^5 n - \Delta_r$ and $g := g_2)$, there exists a constant $\tilde{\kappa} > 0$ such that
\begin{align} \label{eq:relaxed_bound_at_m}
\Pro{\Gap_{\tilde{\mathcal{P}}_{t_0}}(m) \leq \tilde{\kappa} \cdot \frac{g_2}{\log g_2} \cdot \log \log n \,\, \left\vert \,\, \mathfrak{F}^{t_1}, \max_{i \in [n]} \left|\tilde{y}_i^{t_1}\right| \leq c_3g_2 \log(ng_2) \right.} \geq 1 - n^{-3}.
\end{align}
By combining \cref{eq:relaxed_bound_at_t0}, \cref{eq:relaxed_bound_at_t1} and \cref{eq:relaxed_bound_at_m}, we have that%
\[
\Pro{\Gap_{\tilde{\mathcal{P}}_{t_0}}(m) \leq \tilde{\kappa} \cdot \frac{g_2}{\log g_2} \cdot \log \log n } \geq \big(1 - n^{-14}\big) \cdot \big(1 - n^{-3} \big) \cdot \big(1 - n^{-3}\big) \geq 1 - 3n^{-3}.
\]
Finally, by \cref{eq:p_and_tilde_p_agree} we have that \Whp~$\mathcal{P}$ and $\tilde{\mathcal{P}}_{t_0}$ agree in every step in $[t_0, m]$, so by taking the union bound we conclude\[
\Pro{\Gap_{\mathcal{P}}(m) \leq \tilde{\kappa} \cdot \frac{g_2}{\log g_2} \cdot \log \log n } \geq 1 - 3n^{-3} - n^{-3} \geq 1 - n^{-2}. \qedhere
\]
\end{proof}

\section{Lower Bounds}\label{sec:lower_bounds}

In this section, we will state and prove lower bounds on the gap of various processes and settings. Recall that the \GBounded and \GMyopicComp processes are specific instances of the \GAdvComp setting. Our main result is for \GMyopicComp, where we will prove a lower bound of $\Omega(g + \frac{g}{\log g} \cdot \log \log n)$, which matches the upper bounds of Sections~\ref{sec:g_adv_g_plus_logn_bound}-\ref{sec:g_adv_layered_induction} for all $g \geq 0$ (\cref{cor:g_myopic_combined_lower_bound}). \cref{tab:lower_bounds_overview} gives a summary of the results.

 \begin{table}[h]
 \centering
 \begin{tabular}{cccc}
 \textbf{Process} & \textbf{Range} & \textbf{Lower Bound} & \textbf{Reference} \\
 \hline 
 Any \GAdvComp & $0 \leq g $ & $\log_2 \log n + \Omega(1)$ & Obs~\ref{obs:g_adv_loglogn_lower_bound}  \\
 \GMyopicComp & $2 \leq g$ & $\Omega(g)$ & Prop~\ref{pro:g_myopic_g_lower_bound} \\
 \GMyopicComp & $10 \leq g \leq \frac{1}{8} \cdot \frac{\log n}{\log \log n}$ & $\Omega\big(\frac{g}{\log g} \cdot \log \log n \big)$ & Thm~\ref{thm:g_myopic_layered_induction_lower_bound} \\
 \SigmaNoisyLoad & $2 \cdot (\log n)^{-1/3} \leq \sigma$ & $\Omega( \min\{ 1, \sigma \} \cdot (\log n)^{1/3})$ & Prop~\ref{pro:sigma_lower_bound} \\
 \SigmaNoisyLoad & $32 \leq \sigma$ & $\Omega(\min\{\sigma^{4/5}, \sigma^{2/5} \cdot \sqrt{\log n}\})$ & Prop~\ref{pro:sigma_lower_bound} \\
 \BBatch & $b \in [n \cdot e^{\log^{\Omega(1)} n}, n \log n]$ & $\Omega\left(\frac{\log n}{\log(\frac{4n}{b} \cdot \log n)} \right)$ & Obs~\ref{obs:two_choice_batching_lower_bound_small_b}
 \end{tabular}
 \caption{Overview of the lower bounds for different noise settings. All of these hold for a particular value of $m$ with high probability.}
 \label{tab:lower_bounds_overview}
 \end{table}

 We first start with a simple lower bound which follows immediately by majorization with the \TwoChoice process without noise. This lower bound holds regardless of which strategy the adversary uses.
\begin{obs} \label{obs:g_adv_loglogn_lower_bound}
There is a constant $\kappa >0$ such that for any $g \geq 0$ and any instance of the \GAdvComp setting,  \[
\Pro{\Gap(n) \geq \log_{2} \log n - \kappa} \geq 1 - n^{-1}.
\]
\end{obs}
\begin{proof}
For the \TwoChoice process without noise, it was shown in \cite[Theorem 3.3]{ABKU99} that there is a constant $\kappa > 0$ such that $\Pro{ \Gap(n) \geq \log_{2} \log n - \kappa} \geq 1-n^{-1}$. 

At any step $t \geq 0$, probability allocation vector $p$ of \TwoChoice without noise is majorized the probability allocation vector $q^t$ of any instance of \TwoChoice in the \GAdvComp setting, as $q^t$ is formed by $p$ and possibly reallocating some probability mass from light to heavy bins. Hence,
the lower bound follows by majorization (see \cref{lem:majorisation}). 
\end{proof}

We proceed by analyzing the \GMyopicComp process by coupling its allocations with that of a \OneChoice process.

\begin{pro}\label{pro:g_myopic_g_lower_bound}
Consider the \GMyopicComp process. Then, $(i)$~for any $g \in [2, 6 \log n]$ and for $m= \frac{1}{2} \cdot ng$, it holds that
\[
 \Pro{ \Gap(m) \geq \frac{1}{35} \cdot g } \geq 1-n^{-2}.
\]
Further, $(ii)$~for any $g \geq 6 \log n$ and for $m = ng^2/(32 \log n)$,
\[
 \Pro{ \Gap(m)\geq \frac{1}{60} \cdot g } \geq 1-2n^{-2}. 
\]
\end{pro}
Recall that for $g = \Omega(\polylog(n))$, our upper bound on the gap is
\[
 \Oh \left(g + \frac{g}{\log g} \cdot \log \log n \right) = \Oh(g).
\]
This means that the lower bound in \cref{pro:g_myopic_g_lower_bound} is matching for those $g$. For smaller values of $g$, a stronger lower bound will be presented in \cref{thm:g_myopic_layered_induction_lower_bound}.
\begin{proof}
\textit{First statement.} Consider \GMyopicComp with $m = \frac{1}{2} \cdot ng$ balls and define the stopping time $\tau:= \inf \{t \geq 0 \colon \max_{i \in [n]} x_i^t \geq g \}$.
Note that $\tau \leq m$ implies there is a bin $j$ with $x_j^{\tau} \geq g$, and hence
\[
 \Gap(m) \geq x_j^m - \frac{m}{n} \geq x_j^{\tau} - \frac{1}{2} g \geq g - \frac{1}{2} g > \frac{1}{35} g.
\]
Let us now assume $\tau > m$. In that case, all bins have an absolute load in $[0,g]$ and are therefore indistinguishable. Hence during steps $1,2,\ldots,m$, the \GMyopicComp process behaves exactly like \OneChoice. 

By \cref{lem:one_choice_cnlogn} (for $c := \frac{1}{2} \cdot \frac{g}{\log n} \geq \frac{1}{\log n}$ since $g \geq 2$) it holds for \OneChoice that,
\[
 \Pro{ \Gap(m) \geq \frac{1}{10\sqrt{2}} \cdot \sqrt{ g \cdot \log n} } \geq 1-n^{-2},
\]
which, as $g \cdot \log n \geq \frac{1}{6} g^2$, implies that
\[
 \Pro{ \Gap(m) \geq \frac{1}{35} \cdot g } \geq 1-n^{-2}.
\]
\textit{Second statement.} %
We will first prove for the \GMyopicComp process that \whp~for any step $1 \leq t \leq m$ and any bin $i \in [n]$, it holds that $|y_i^t| \leq \frac{g}{2}$. To this end, fix any bin $i \in [n]$ and define for any $1 \leq t \leq m$, $Z^t = Z^t(i) := \sum_{j=1}^t (Y_j - \frac{1}{n})$, where the $Y_j$'s are independent Bernoulli variables with parameter $1/n$ each. Clearly, $Z^t$ forms a martingale, and $\Ex{Z^t}=0$ for all $1 \leq t \leq m$. Further,
let $\tau := \inf \{ t \geq 1 \colon |Z^t| > \frac{g}{2} \}$. Then $Z^{t \wedge \tau}$ is also martingale. Further, 
we have
\[
 \Var{ Z^{t+1} \, \mid \, Z^{t} } = \frac{1}{n} \cdot \left(1 - \frac{1}{n} \right) \leq \frac{1}{n} =: \sigma^2.
\]
Also $|Z^{t+1} - Z^{t}| \leq 1=:M$. Hence by a martingale inequality \cref{lem:cl06_thm_6_1},
\begin{align*}
 \Pro{ \left| Z^{m \wedge \tau} \right|   \geq \lambda } \leq 2 \cdot \exp\left(- \frac{\lambda^2}{2 ( \sum_{i=1}^m \sigma^2 + M \lambda/3) } \right),
\end{align*}
with $\lambda = \frac{g}{2}$,
\begin{align*}
 \Pro{ | Z^{m \wedge \tau}|  \geq \frac{g}{2} } &\leq 2 \cdot \exp\left(- \frac{g^2}{8 ( m \cdot \frac{1}{n} + \frac{g}{6}) } \right) 
 = 2 \cdot \exp\left(- \frac{ g^2}{ 8 \cdot (\frac{g^2}{32 \log n} +\frac{g}{6}) } \right) \leq 2 n^{-3},
\end{align*}
where the last inequality holds since 
$
8 \cdot \left(\frac{g^2}{32 \log n} + \frac{g}{6} \right) \leq 8 \cdot \left(\frac{g^2}{32 \log n} + \frac{g^2}{32 \log n} \right)  = \frac{g^2}{2 \log n},
$
using that $g \geq 6 \log n$.

If the event $|Z^{m \wedge \tau} | \leq \frac{g}{2}$ (or equivalently $\tau > m$) occurs, then this implies that bin $i$ deviates from the average load by at most $\frac{g}{2}$ in all steps $1,2,\ldots,m$. By the union bound, this holds with probability $1-2n^{-2}$ for all $n$ bins. Consequently, the \GMyopicComp process behaves exactly like \OneChoice until time $m$ with probability $1-2n^{-2}$. 

By \cref{lem:one_choice_cnlogn} (for $c := \frac{g^2}{32 \log^2 n}$) it holds for \OneChoice that,
\begin{align*}
\Pro{ \Gap(m) \geq \frac{g}{10 \cdot \sqrt{32}}} \geq 1-n^{-2}. 
\end{align*}
The claim follows by taking the union bound and using that $\frac{g}{10 \cdot \sqrt{32}} \geq \frac{g}{60}$.
\end{proof}

\begin{thm}\label{thm:g_myopic_layered_induction_lower_bound}
Consider the \GMyopicComp process for any $g \in [10, \frac{1}{8} \cdot \frac{\log n}{\log \log n}]$. Then, there exists $\ell := \ell(g,n)$ (defined in \cref{eq:ell_definition}), such that for $m=n \cdot \ell$, it holds that
\[
 \Pro{ \Gap(m) \geq \frac{1}{8} \cdot \frac{g}{\log g} \cdot \log \log n } \geq 1-n^{-\omega(1)}.
\]
\end{thm}
The proof of this lower bound is similar to the method used by Azar, Broder, Karlin and Upfal~\cite{ABKU99} to prove the $\Omega(\log \log n)$ lower bound for \TwoChoice without noise, in the sense that it follows a layered induction approach, but some additional care is needed. For example, the induction step
size in the load is $g$ and not $1$, and the outcome of a load comparison depends on the load difference. %
\begin{proof}
In the proof we will divide the allocation of  $m=n \cdot \ell$ balls into $\ell$ consecutive phases, each of which lasts for $n$ steps, where
\begin{align} \label{eq:ell_definition}
 \ell := \left\lfloor  \frac{\log (\frac{1}{8} \log n/\log g)}{\log g} \right\rfloor.
\end{align}
First, we verify that $\ell \geq 1$,
\begin{align*}
 \ell 
 = \left\lfloor
 \frac{\log (\frac{1}{8} \log n/\log g)}{\log g} \right\rfloor 
 \stackrel{(a)}{\geq} \left\lfloor \frac{\log (\frac{1}{8} \log n/\log g)}{\log (\frac{1}{8} \log n/\log n)} \right\rfloor \stackrel{(b)}{\geq}  \left\lfloor  \frac{\log (\frac{1}{8} \log n/\log n)}{\log (\frac{1}{8} \log n/\log n)} \right\rfloor =1,
\end{align*}
where $(a)$ used that $g \leq \frac{1}{8} \log n/\log \log n$ and $(b)$ that $g \leq \log n$.

Next we define for any $k=1,2,\ldots,\ell$ the following event:
\[
 \mathcal{E}_k := \biggl\{ \Bigl| \bigl\{ i \in [n] \colon x_i^{k \cdot n} \geq k \cdot g \bigr\} \Bigr|
 \geq n \cdot g^{-\sum_{j=1}^{k} g^{j} } \biggr\}.
\]
The main goal of this proof is to show that $\mathcal{E}_{\ell}$ occurs with high probability. Assuming for the moment that $\mathcal{E}_{\ell}$ indeed occurs, let us verify that the lower bound on the gap follows. 
First, note that $\mathcal{E}_{\ell}$ implies the existence of a bin $i \in [n]$ with $x_i^m \geq \ell \cdot g$ and thus $\Gap(m) \geq x_i^m - \frac{m}{n} \geq \ell \cdot g - \ell = \ell \cdot (g-1)$, since by the choice of $\ell=\left\lfloor  \frac{\log (\frac{1}{8} \log n/\log g)}{\log g} \right\rfloor$,
\begin{align} \label{eq:g_to_4g_ell}
   g^{-\sum_{j = 1}^k g^j} \stackrel{(a)}{\geq} g^{-4g^{\ell}} \geq g^{- \frac{1}{2} \log n/\log g} = n^{-1/2},
\end{align}
where in $(a)$ we used that 
for any $g \geq 2$, 
\begin{align}
\sum_{j=1}^k g^j \leq 2 \sum_{j=0}^{\ell} g^j  = 2 \cdot \frac{g^{\ell+1}-1}{g-1} \leq   2  \cdot \frac{g^{\ell+1}}{\frac{1}{2} g} = 4 g^{\ell}. \label{eq:geo_series}
\end{align}

Thus $n \cdot g^{-\sum_{j=1}^{k} g^{j} } \geq 1$. Secondly, we verify that $\ell \cdot (g - 1) = \Omega\big(\frac{g}{\log g} \cdot \log \log n\big)$,
\begin{align}
 \ell \cdot (g-1) &\stackrel{(a)}{\geq} \frac{1}{4} \cdot \frac{\log (\frac{1}{8} \log n/\log g)}{\log g} \cdot g \notag \\
 & = \frac{1}{4} \cdot \frac{\log(1/8) + \log \log n - \log \log g}{\log g} \cdot g \notag \\
 &\stackrel{(b)}{\geq} \frac{1}{8} \cdot \frac{\log \log n}{\log g} \cdot g, \label{eq:ek_gap_lower_bound}
\end{align}
where $(a)$ holds since $\lfloor u \rfloor \geq u/2$ for $u \geq 1$ (as $\ell \geq 1$) and 
$g-1 \geq g/2$, since $g \geq 2$, and $(b)$ holds since for sufficiently large $n$, $(1/4) \cdot \log \log n \geq -\log(1/8)$ and $(1/4) \cdot \log \log n \geq \log \log g$ as $g \leq \log n$. 

In order to establish that $\mathcal{E}_{\ell}$ occurs \whp, we will proceed by induction and prove that for any $k \geq 1$, with $\epsilon_k := n^{-\omega(1)}$ and $\mathcal{E}_{0} := \Omega$,
\[
 \Pro{ \mathcal{E}_k \, \mid \, \mathcal{E}_{k-1} } \geq 1 - \epsilon_k.
\]

\textbf{Induction Base ($k=1$).} Here we consider the allocation of the first $n$ balls into $n$ bins. We are interested in the number of bins which reach load level $g$ during that phase. Note that as long as the loads of both sampled bins are smaller than $g$, the allocation follows that of \OneChoice; if one of the bins has a load which is already larger than $g$, then the load difference may force the process to place a ball in the lighter of the two bins. In the following, we will (pessimistically) assume that the allocation of \emph{all} $n$ balls follows \OneChoice. 

Instead of the original \OneChoice process which produces a load vector $(x_i^n)_{i \in [n]}$, 
we consider the Poisson Approximation (\cref{lem:poisson}) and analyze the load vector $(\tilde{x}_i^n)_{i \in [n]}$, where the $\tilde{x}_i^n$, $i \in [n]$ are independent Poisson random variables with mean $n/n=1$. Clearly, for any $i \in [n]$,
\begin{align*}
 \Pro{ \tilde{x}_i^n \geq g} &\geq \Pro{ \tilde{x}_i^n = g} = e^{-1} \cdot \frac{1^g}{g!} 
 \geq 2 \cdot g^{-g},
\end{align*}
using in the last inequality Stirling's approximation (e.g.~\cite[Lemma 5.8]{MU17}) and $g \geq 5$,
\[
g! \leq e \cdot g \cdot \left( \frac{g}{e} \right)^g = \frac{e^{-1}}{2} \cdot g^g  \cdot \left( \frac{2e^2 g}{e^g} \right) \leq \frac{e^{-1}}{2} \cdot g^g.
\]

Let $Y:=\left| \left\{ i \in [n] \colon \tilde{x}_i^n \geq g \right\} \right|$. Then $\Ex{Y} \geq 2 n \cdot g^{-g}$. By a standard Chernoff Bound,
\begin{align*}
\Pro{ Y \leq n \cdot g^{-g} } \leq \Pro{ Y \leq \frac{1}{2} \cdot \Ex{Y} } \leq \exp\left( -\frac{1}{8} \cdot \Ex{Y} \right) \leq n^{-\omega(1)},
\end{align*}
where we have used  $g^{-g} \geq g^{-4g^{\ell}} \geq n^{-1/2}$, due to \cref{eq:g_to_4g_ell}. Hence by the Poisson Approximation (\cref{lem:poisson}), $\Pro{ \mathcal{E}_1 } \geq 1 - 2 \cdot n^{-\omega(1)} = 1 - n^{-\omega(1)}$.

\textbf{Induction Step ($k-1 \longrightarrow k$, $k \geq 2$).} For the induction step, we analyze phase $k = 2, \ldots , \ell$ and we will lower bound $\Pro{ \mathcal{E}_{k} \, \mid \, \mathcal{E}_{k-1}}$. Assuming $\mathcal{E}_{k-1}$ occurs, there are at least $n \cdot g^{-\sum_{j=1}^{k-1} g^{j}}$ bins whose load is at least $(k-1) \cdot g$ at the beginning of phase $k$. Let us call such a set of bins $\mathcal{B}_{k-1}$, which can be assumed to satisfy with equality:
\begin{align}
| \mathcal{B}_{k-1} | = n \cdot g^{-\sum_{j=1}^{k-1} g^{j}}.  \label{eq:bk_size}
\end{align}
Next note that whenever we sample two bins $i_1,i_2$ from $\mathcal{B}_{k-1}$, we allocate the ball to a random bin among $\{i_1, i_2\}$ if both bins have not reached load level $k \cdot g$. 
Therefore, in order to lower bound the number of bins in $\mathcal{B}_{k-1}$ which reach load level $k \cdot g$ by the end of phase $k$, we may pessimistically assume that if two bins in $\mathcal{B}_{k-1}$ are sampled, the ball will be \emph{always} placed in a randomly sampled bin. Note that the probability that a ball will be allocated into the set $\mathcal{B}_{k-1}$ is lower bounded by
$
 \bigl( \frac{| \mathcal{B}_{k-1} |}{n} \bigr)^2. 
$
Let $Z$ denote the number of balls allocated to $\mathcal{B}_{k-1}$ in phase $k$. Then $\Ex{Z} \geq n \cdot \left( \frac{| \mathcal{B}_{k-1} |}{n} \right)^2$. %
Using a Chernoff Bound for $\tilde{m} := \frac{2}{3} \cdot n \cdot \left( \frac{| \mathcal{B}_{k-1} |}{n} \right)^2$,
\[
\Pro{Z \leq \tilde{m}} \leq \Pro{Z \leq \frac{2}{3} \cdot \Ex{Z}} \leq \exp\left( - \frac{1}{18} \cdot \Ex{Z}\right) \leq n^{-\omega(1)},
\]
where we used that
\[
 \Ex{Z} \geq n \cdot g^{-2 \sum_{j=1}^{k-1} g^j}
 \stackrel{(a)}{\geq} n \cdot g^{-2 g^{k}}
 \stackrel{(b)}{\geq} n \cdot n^{-1/2},
\]
using in $(a)$ that for any integer $g \geq 2$, 
$\sum_{j=1}^{k-1} g^j \leq \frac{g^{k}-1}{g-1} \leq g^{k}$, and in
$(b)$ that $k \leq \ell$ and the property of $\ell$ in \cref{eq:g_to_4g_ell}.

Conditioning on this event, in phase $k$ we have a \OneChoice process with $\tilde{m}$ balls into $\tilde{n}:=\left| \mathcal{B}_{k-1} \right|$ bins, which w.l.o.g.~will be labeled $1,2,\ldots,\tilde{n}$. Again, we apply the Poisson approximation and define $(\tilde{x}^{\tilde{m}})_{i \in [\tilde{n}]}$ as $\tilde{n}$ independent Poisson random variables with mean $\lambda$ given by
\[
 \lambda := \frac{ \tilde{m}}{\tilde{n}} = \frac{ \frac{2}{3} \cdot n \cdot \left( \frac{| \mathcal{B}_{k-1} |}{n} \right)^2}{|\mathcal{B}_{k-1}|} = \frac{2}{3} \cdot \frac{| \mathcal{B}_{k-1} |}{n} = \frac{2}{3} \cdot g^{-\sum_{j=1}^{k-1} g^j},
\]
using \cref{eq:bk_size}.
With that, it follows for any bin $i \in \mathcal{B}_{k-1}$,
\begin{align*}
 \Pro{ \tilde{x}_i \geq g} &
  \geq \Pro{ \tilde{x}_i = g} 
  = e^{-\lambda} \cdot \frac{\lambda^g}{g!} 
  \stackrel{(a)}{\geq} e^{-1} \cdot \frac{ \bigl( \frac{2}{3} \cdot g^{-\sum_{j=1}^{k-1} g^j } \bigr)^g}{g!} 
  \stackrel{(b)}{\geq} 2 \cdot g^{-\sum_{j=1}^{k} g^j },
\end{align*}
using in $(a)$ that $\lambda \leq 1$ and in $(b)$ Stirling's approximation (e.g.~\cite[Lemma 5.8]{MU17}) and that $g \geq 10$,
\[
g! 
 \leq e \cdot g \cdot \left( \frac{g}{e} \right)^g 
 = \frac{e^{-1}}{2} \cdot \left( \frac{2}{3} \right)^g \cdot g^g \cdot \left( 2e^2 g \cdot \left( \frac{3}{2e}\right)^g \right)
 \leq \frac{e^{-1}}{2} \cdot \left( \frac{2}{3} \right)^g \cdot g^{g}.
\]
Let $Y:=\left| \left\{ i \in [\tilde{n}] \colon \tilde{x}_i \geq g \right\} \right|$.
Then $\Ex{Y} \geq 2 \tilde{n} \cdot g^{-\sum_{j=1}^k g^j}$.
Thus by a Chernoff Bound,
\begin{align*}
\Pro{ Y \leq \tilde{n} \cdot g^{-\sum_{j=1}^k g^j } } \leq \Pro{ Y \leq \frac{1}{2} \cdot \Ex{Y} } \leq \exp\left( -\frac{1}{8} \cdot \Ex{Y} \right) \leq n^{-\omega(1)},
\end{align*}
where we used that
\begin{align*}
 \Ex{Y} \geq 2 \tilde{n} \cdot g^{-\sum_{j=1}^k g^j }  &= 2 n \cdot g^{-\sum_{j=1}^{k-1} g^j} \cdot g^{-\sum_{j=1}^{k} g^j} \geq 2 n \cdot g^{-2 \sum_{j=1}^{k} g^j} \stackrel{(\text{\ref{eq:geo_series}})}{\geq} 2 n \cdot g^{-4 g^k} \stackrel{(a)}{\geq} 
 2 n^{1/2},
\end{align*}
where in $(a)$ we used $k \leq \ell$ and \cref{eq:g_to_4g_ell}. Thus by the union bound and Poisson Approximation, $\Pro{ \mathcal{E}_k \, \mid \, \mathcal{E}_{k-1} } \geq 1 - n^{-\omega(1)} - 2 \cdot n^{-\omega(1)} = 1 - n^{-\omega(1)}$, which completes the induction. 

Finally, by a simple union bound,
\begin{align*}
 \Pro{ \mathcal{E}_{\ell } } &\geq 1 - \sum_{k=1}^{\ell} \Pro{ \neg \mathcal{E}_{k} \, \mid \, \mathcal{E}_{k-1} } \geq 1 - \ell \cdot n^{-\omega(1) } \geq 1 -n^{-\omega(1)}.
\end{align*}
As verified in \cref{eq:ek_gap_lower_bound}, $\mathcal{E}_{\ell}$ implies $\Gap(m) \geq \frac{1}{8} \cdot \frac{g}{\log g} \cdot \log \log n$, and therefore the proof is complete.
\end{proof}

Combining \cref{obs:g_adv_loglogn_lower_bound}, \cref{pro:g_myopic_g_lower_bound} and \cref{thm:g_myopic_layered_induction_lower_bound}, we get:
\begin{cor} \label{cor:g_myopic_combined_lower_bound}
Consider the \GMyopicComp process for any $g \geq 1$. Then, there exists an $m := m(g) \geq 0$, such that
\[
\Pro{\Gap(m) = \Omega\left(g + \frac{g}{\log g} \cdot \log \log n \right)} \geq 1 - n^{-1}.
\]
\end{cor}

We proceed with two lower bounds for the \SigmaNoisyLoad process.
\begin{pro}  \label{pro:sigma_lower_bound}
Consider the \SigmaNoisyLoad process with $\rho(\delta) := 1 - \frac{1}{2} \cdot \exp \bigl( - ( \delta/\sigma)^2 \bigr)$ for some (not necessarily constant) $\sigma > 0$.
Then, $(i)$ for any $\sigma \geq 2 \cdot (\log n)^{-1/3}$,
\[
 \Pro{ \Gap(n) \geq \min\left\{ \frac{1}{8} \cdot (\log n)^{1/3}, \frac{1}{2} \sigma \cdot (\log n)^{1/3} \right\} }
 \geq 1-2n^{-1}.
\]
Further, $(ii)$ for any $\sigma \geq 32$, it holds for $m=\frac{1}{2} \sigma^{4/5} \cdot n$, 
\[
 \Pro{ \Gap(m) \geq \min \left\{ \frac{1}{2} \sigma^{4/5}, \frac{1}{30} \sigma^{2/5}  \cdot \sqrt{\log n}  \right\} } \geq 1-2n^{-2}.
\]
\end{pro}

\begin{proof}
\textit{First statement.} Let $\tau:= \inf\{t \geq 1 \colon \max_{i \in [n]} x_i^t \geq \sigma \cdot (\log n)^{1/3} \}$. If $\tau \leq n$, then
\[
 \Gap(n) \geq x_i^{n} - \frac{n}{n} \geq x_i^{\tau} - 1 \geq \sigma \cdot (\log n)^{1/3} - 1 \geq \frac{1}{2} \sigma \cdot (\log n)^{1/3}.
\]
Consider now the case $\tau > n$. For any step $t \leq \tau$, we will perform a ``sub-sampling'' of the correct comparison in (possibly) two stages as follows. Let $i_1 = i_1^t$ and $i_2 = i_2^t$ be the two sampled bins, and $\delta$ be their load difference. Let $Z_1 \sim \mathrm{Ber}(1- \exp(- (\delta/ \sigma)^2))$ and $Z_2 \sim \mathrm{Ber}(1/2)$ be two independent random variables, which can be thought as the outcome of two biased coin flips that will be used to determine whether the load comparison is correct. If $Z_1 = 1$, then the comparison is correct (regardless of what $Z_2$ is). However, if $Z_1=0$, then the comparison is correct if and only if $Z_2=1$. Overall, the probability of a correct comparison (if $\delta > 0$) is equal to
\begin{align*}
 \Pro{Z_1=1} + \Pro{Z_1=0} \cdot \Pro{Z_2 = 1 \, \mid \, Z_1=0} &= 1- \exp\left(- (\delta / \sigma)^2\right) +  \exp\left(- (\delta / \sigma)^2\right) \cdot \frac{1}{2} \\ &= 1- \frac{1}{2} \exp\left(- (\delta / \sigma)^2\right).
\end{align*}

Further, conditional on $Z_1=0$, the ball will be placed in a random bin among $\{i_1, i_2\}$.
Hence as long as $t \leq \tau$, we can couple the allocation of each ball by \GMyopicComp to an allocation by \OneChoice with probability at least $\exp(-(\delta/\sigma)^2) \geq \exp(- (\log n)^{2/3})$. Using a Chernoff bound, with probability $1-n^{-\omega(1)}$, we we can couple the allocation of at least $n/2 \cdot \exp( -(\log n)^{2/3})$ balls out of the first $n$ balls with that of \OneChoice. Consequently, using \cref{lem:one_choice_lower_bound_max_load_whp_lightly} 
the maximum load (and gap) is at least $\frac{1}{8}(\log n)^{1/3}$ with probability at least $1-n^{-1}$. Combining the two cases we get the claim.

\textit{Second Statement.}
Consider any $\sigma \geq 32$ and define the stopping time $\tau:=\inf\{t \geq 1 \colon \max_{i \in [n]} x_i^t \geq \sigma^{4/5} \}$.
Let $m:=\frac{1}{2} \sigma^{4/5} \cdot n$.
If $\tau \leq m$, then there is a bin $i \in [n]$ with $x_i^{\tau} \geq \sigma^{4/5}$, and
\[
 \Gap(m) \geq x_i^{m} - \frac{m}{n} \geq x_i^{\tau} - \frac{1}{2} \sigma^{4/5} \geq \frac{1}{2} \sigma^{4/5}. 
\]
Otherwise, in each step until $m$, the load difference between any two sampled bins is at most $\sigma^{4/5}$, and therefore each ball is coupled with the allocation of a \OneChoice process with probability at least $
\exp\big( - (\sigma^{4/5}/\sigma)^2 \big) \geq 1 - \frac{1}{\sigma^{2/5}}.
$

Let $X$ be the number of \OneChoice allocations in the first $m$ steps. Using the following standard Chernoff bound,
\[
\Pro{X \geq (1 - \delta) \cdot \Ex{X}} \geq 1 - e^{-\frac{1}{2} \delta^2 \cdot \ex{X}},
\]
with $\delta = \frac{1}{\sigma^{2/5}} \leq 1$ (since $\sigma \geq 1$) and $\Ex{X} \geq \left(1 - \frac{1}{\sigma^{2/5}} \right) \cdot m$, we get that
\[
\Pro{X \geq \frac{1}{2} \left(1 - 2 \cdot \frac{1}{\sigma^{2/5}} \right) \cdot \sigma^{4/5} \cdot n} \geq 1 - e^{-\frac{1}{2} \delta^2 \cdot \ex{X}} \geq 1 - n^{-\omega(1)}.
\]
Therefore, using \cref{lem:one_choice_cnlogn} (for $c := \frac{1}{2} \cdot \left(1 - 2 \cdot \frac{1}{\sigma^{2/5}}\right) \cdot \sigma^{4/5} \cdot \frac{1}{\log n} \geq \frac{1}{\log n}$), we get that
\begin{align*}
\Gap(m)
  & \geq \frac{1}{10} \cdot \sqrt{\frac{1}{2} \left(1 - \frac{2}{\sigma^{2/5}}\right) \cdot \sigma^{4/5} \cdot \log n} + \frac{1}{2} \left(1 - \frac{1}{\sigma^{2/5}}\right) \cdot \sigma^{4/5} 
   - \frac{m}{n} \\
  & = \frac{1}{10} \cdot \sqrt{\frac{1}{2} \left(1 - \frac{2}{\sigma^{2/5}}\right) \cdot \sigma^{4/5} \cdot \log n} + \frac{1}{2} \left(1 - \frac{1}{\sigma^{2/5}}\right) \cdot \sigma^{4/5} 
  - \frac{1}{2} \sigma^{4/5} \\
  & = \frac{1}{10} \cdot \sqrt{\frac{1}{2} \left(1 - \frac{2}{\sigma^{2/5}}\right) \cdot \sigma^{4/5} \cdot \log n} - \frac{1}{2} \sigma^{2/5} \\
  & \stackrel{(a)}{\geq} \frac{1}{20} \sigma^{2/5} \cdot \sqrt{\log n} - \frac{1}{2} \cdot \sigma^{2/5}  \\
  & \geq \frac{1}{30} \sigma^{2/5} \cdot \sqrt{\log n},
\end{align*}
using in $(a)$ that $\sigma \geq 32$. So, combining the two cases we get the claim.
\end{proof}

Finally, we will prove a simple (but tight) lower bound for \BBatch, which already holds in the first batch of $b$ balls, by a \OneChoice argument.

\begin{obs}\label{obs:two_choice_batching_lower_bound_small_b}
There exists a constant $\kappa > 0$, such that for the \BBatch process for any $b \in [n \cdot e^{\log^c n}, n \log n]$ with $c > 0$ any constant, it holds that
\[
\Pro{\Gap(b) \geq \kappa \cdot \frac{\log n}{\log\left( \frac{4n}{b} \cdot \log n \right)}} \geq 1 - n^{-1}.
\]
\end{obs}
\begin{proof}
During the allocation of the first batch consisting of $b$ balls, the load information of all bins is never updated, i.e., their load ``estimate'' equals zero. Hence in these first $n$ steps, as ties are broken randomly, \BBatch behaves exactly like \OneChoice. Hence, by \cref{lem:one_choice_lower_bound_gap_whp_lightly} there exists a constant $\kappa > 0$,\[
\Pro{\Gap(b) \geq \kappa \cdot \frac{\log n}{\log\left( \frac{4n}{b} \cdot \log n \right)}} \geq 1 - n^{-1}. \qedhere
\]
\end{proof}

\section{Experimental Results} \label{sec:experiments}

In this section, we empirically analyze the gap for the \GBounded, \GMyopicComp and \SigmaNoisyLoad processes for various values of $g$, $\sigma$, $n$ and $m$ (\cref{fig:noisy_combined} and \cref{tab:g_combined_empirical_gap_distribution}). For sufficiently large $g$, the gap plots for \GBounded and \GMyopicComp are almost linear, as suggested by our upper bound of $\Theta(g)$ for any $g = \log^{\Omega(1)} n$. We also present the empirical gap for \BBatch for various values of batch size $b$ around $n$ and compare with the \OneChoice gap for the case of $b$ balls (\cref{fig:bbatch_vs_one_choice} and \cref{tab:batch_gap_distribution}). As our analysis in \cref{sec:upper_bounds_for_delay_settings} suggests, these gaps are close when $b$ is close to $n$ and for smaller $b$, e.g., when $b = n^{1-\Omega(1)}$, the gaps of the two processes diverge (\cref{rem:delay_very_very_lightly_loaded}).

\begin{figure}[H]
    \centering
    \includegraphics[scale=0.8]{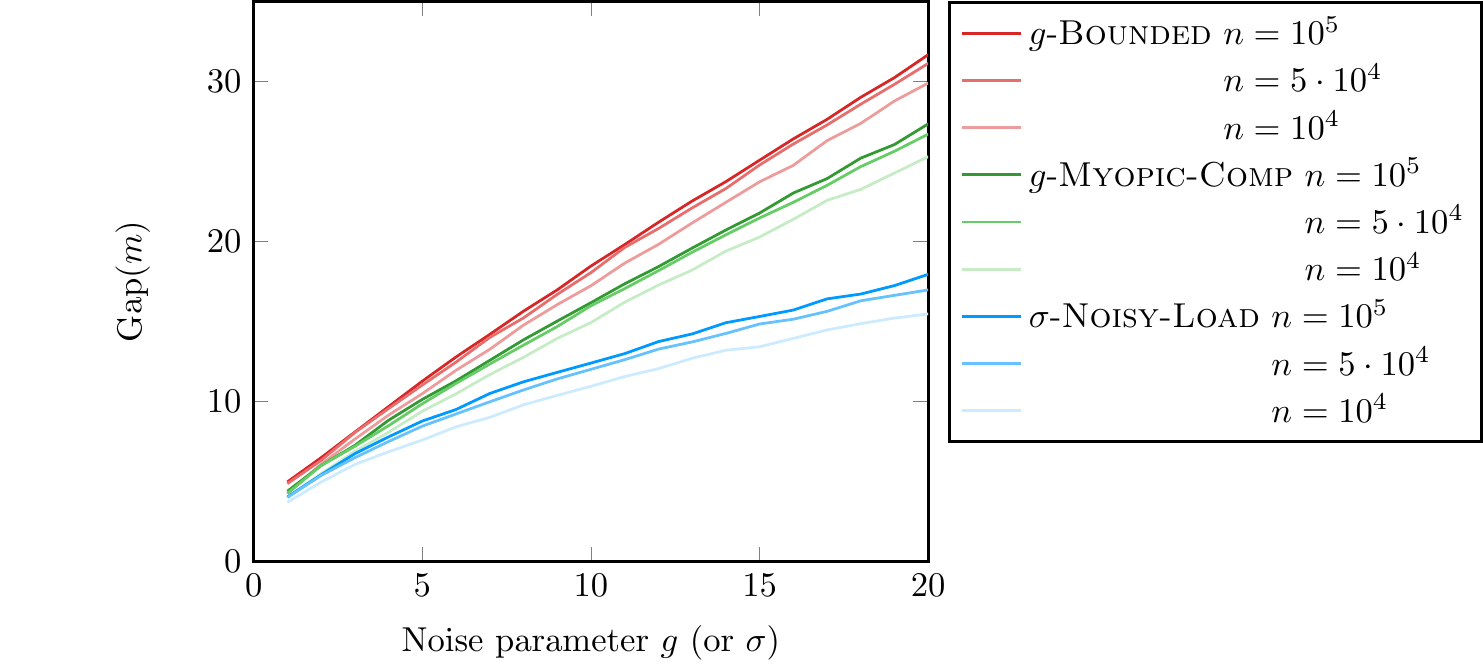}
    \caption{Average gap for the \GBounded and \GMyopicComp with $g \in \{1, \ldots, 20\}$ and \SigmaNoisyLoad with $\sigma \in \{1, \ldots, 20 \}$, for $n \in \{10^4, 5 \cdot 10^4, 10^5 \}$ and $m = 1000 \cdot n$, over $100$ runs.}
    \label{fig:noisy_combined}
\end{figure}

\begin{figure}[H]
    \centering
    \includegraphics[scale=0.8]{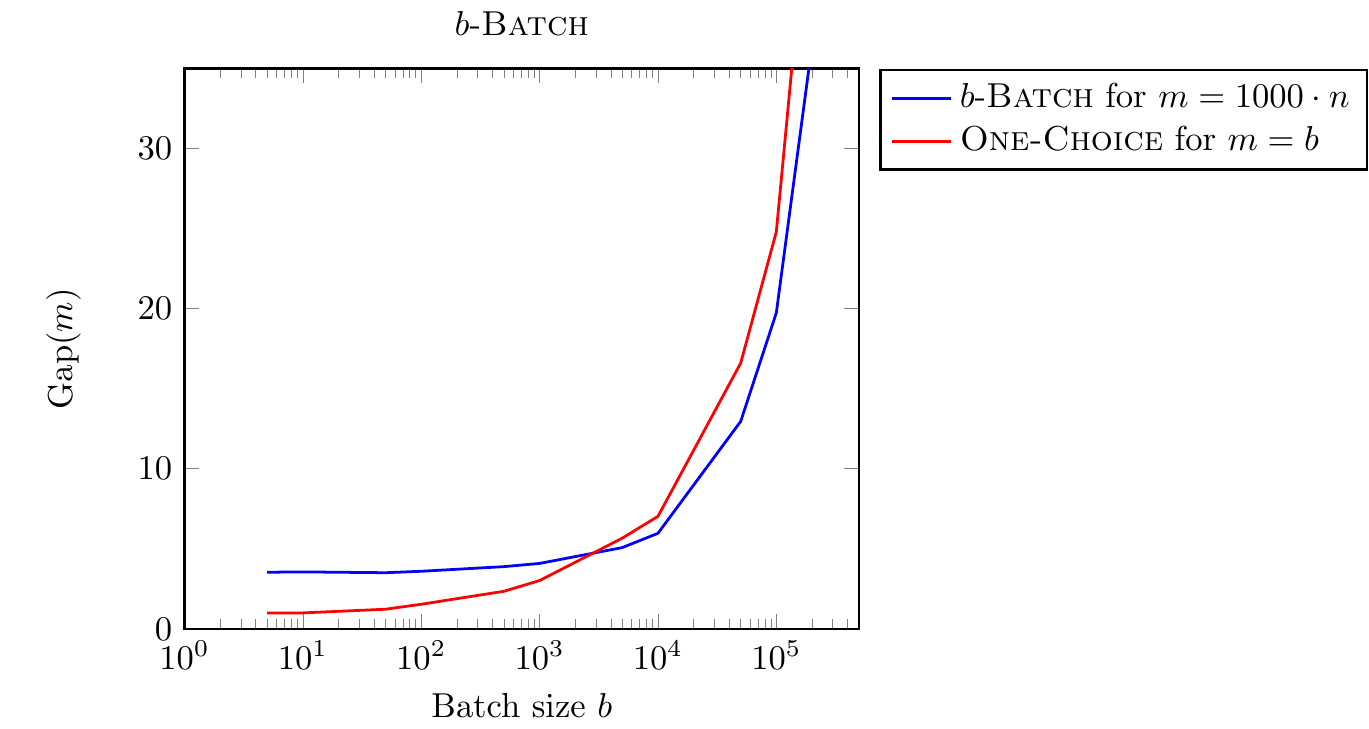}
    \caption{Average gap for the \BBatch process with $b \in \{ 5, 10, 5 \cdot 10, 10^2, \ldots 10^5, 5 \cdot 10^5 \}$, for $n = 10^4$ and $m = 1000 \cdot n$ and for $\OneChoice$ for $b$ balls, over $100$ runs.}
    \label{fig:bbatch_vs_one_choice}
\end{figure}

\colorlet{GA}{black!40!white}
\colorlet{GB}{black!70!white}
\colorlet{GC}{black}
\newcommand{\CA}[1]{\textcolor{GA}{#1}} %
\newcommand{\CB}[1]{\textcolor{GB}{#1}} %
\newcommand{\CC}[1]{\textcolor{GC}{#1}} %
\newcommand{\CI}[2]{\FPeval{\result}{min(30 + 3 * #1, 100)} \colorlet{tmpC}{black!\result!white} {\textcolor{tmpC}{#2}}}
\newcommand{\W}{\phantom{1}}

\begin{table}[H]
    \centering
   \footnotesize{
    \begin{tabular}{|c|c|c|c|c|c|c|c|}
\multicolumn{7}{c}{\bfseries \GBounded} \\
    \hline
$n$ & $g = 0$ & $g = 1$ & $g = 2$ & $g = 4$ & $g = 8$ & $g = 16$ \\ \hline
        $10^4$ &
\makecell{\CI{46}{\textbf{2} : 46\%}\\
\CI{54}{\textbf{3} : 54\%}} &%
\makecell{\CI{74}{\textbf{4} : 74\%}\\
\CI{26}{\textbf{5} : 26\%}}&
\makecell{
\CI{1}{\textbf{5} : \W1\%}\\
\CI{89}{\textbf{6} : 89\%}\\
\CI{10}{\textbf{7} : 10\%}} &
\makecell{
\CI{1}{\textbf{\W8} : \W1\%}\\
\CI{82}{\textbf{\W9} : 82\%}\\
\CI{17}{\textbf{10} : 17\%}} &
\makecell{
\CI{1}{\textbf{13} : \W1\%}\\
\CI{35}{\textbf{14} : 35\%}\\
\CI{51}{\textbf{15} : 51\%}\\
\CI{11}{\textbf{16} : 11\%}\\
\CI{2}{\textbf{17} : \W2\%}
 } &
\makecell{
\CI{4}{\textbf{23} : \W4\%}\\
\CI{37}{\textbf{24} : 37\%}\\
\CI{43}{\textbf{25} : 43\%}\\
\CI{11}{\textbf{26} : 11\%}\\
\CI{5}{\textbf{27} : \W5\%}
 } \\ \hline
        $5 \cdot 10^4$ &
\makecell{\CI{4}{\textbf{2} : \W4\%}\\
\CI{96}{\textbf{3} : 96\%}}&%
\makecell{\CI{13}{\textbf{4} : 13\%}\\
\CI{86}{\textbf{5} : 86\%}\\
\CI{1}{\textbf{6} : \W1\%}}&
\makecell{
\CI{67}{\textbf{6} : 67\%}\\
\CI{33}{\textbf{7} : 33\%}}&
\makecell{
\CI{46}{\textbf{\W9} : 46\%}\\
\CI{51}{\textbf{10} : 51\%}\\
\CI{3}{\textbf{11} : \W3\%}}&
\makecell{
\CI{3}{\textbf{14} : \W3\%}\\
\CI{72}{\textbf{15} : 72\%}\\
\CI{24}{\textbf{16} : 24\%}\\
\CI{1}{\textbf{17} : \W1\%}}&
\makecell{
\CI{25}{\textbf{25} : 25\%}\\
\CI{47}{\textbf{26} : 47\%}\\
\CI{23}{\textbf{27} : 23\%}\\
\CI{4}{\textbf{28} : \W4\%}\\
\CI{1}{\textbf{29} : \W1\%}}\\ \hline
$10^5$ &
\makecell{\CI{100}{\textbf{3} : 100\%}}&%
\makecell{\CI{1}{\textbf{4} : \W1\%}\\
\CI{99}{\textbf{5} : 99\%}}&
\makecell{
\CI{50}{\textbf{6} : 50\%}\\
\CI{50}{\textbf{7} : 50\%}}&
\makecell{
\CI{32}{\textbf{\W9} : 32\%}\\
\CI{67}{\textbf{10} : 67\%}\\
\CI{1}{\textbf{11} : \W1\%}} &
\makecell{
\CI{39}{\textbf{15} : 39\%}\\
\CI{57}{\textbf{16} : 57\%}\\
\CI{4}{\textbf{17} : \W4\%}} &
\makecell{
\CI{9}{\textbf{25} : \W9\%}\\
\CI{50}{\textbf{26} : 50\%}\\
\CI{33}{\textbf{27} : 33\%}\\
\CI{7}{\textbf{28} : \W7\%}\\
\CI{1}{\textbf{29} : \W1\%}}\\ \hline
    \end{tabular}
    }

   \footnotesize{
    \begin{tabular}{|c|c|c|c|c|c|c|c|}
\multicolumn{7}{c}{\bfseries \GMyopicComp} \\
    \hline
$n$ & $g = 0$ & $g = 1$ & $g = 2$ & $g = 4$ & $g = 8$ & $g = 16$ \\ \hline
        $10^4$ &
\makecell{
\CI{46}{\textbf{2} : 46\%} \\
\CI{54}{\textbf{3} : 54\%} } &
\makecell{
\CI{97}{\textbf{4} : 97\%} \\
\CI{3}{\textbf{5} : \W3\%} } &
\makecell{
\CI{49}{\textbf{5} : 49\%} \\
\CI{51}{\textbf{6} : 51\%} } &
\makecell{
\CI{2}{\textbf{7} : \W2\%} \\
\CI{87}{\textbf{8} : 87\%} \\
\CI{11}{\textbf{9} : 11\%} } &
\makecell{
\CI{37}{\textbf{12} : 37\%} \\
\CI{50}{\textbf{13} : 50\%} \\
\CI{12}{\textbf{14} : 12\%} \\
\CI{1}{\textbf{15} : \W1\%} } &
\makecell{
\CI{14}{\textbf{20} : 14\%} \\
\CI{47}{\textbf{21} : 47\%} \\
\CI{29}{\textbf{22} : 29\%} \\
\CI{8}{\textbf{23} : \W8\%} \\
\CI{2}{\textbf{25} : \W2\%} } \\ \hline
        $5 \cdot 10^4$ &
\makecell{
\CI{4}{\textbf{2} : \W4\% } \\
\CI{96}{\textbf{3} : 96\% }} &
\makecell{
\CI{73}{\textbf{4} : 73\%} \\
\CI{27}{\textbf{5} : 27\%} } &
\makecell{
\CI{1}{\textbf{5} : \W1\%} \\
\CI{97}{\textbf{6} : 97\%} \\
\CI{2}{\textbf{7} : \W2\%} } &
\makecell{
\CI{50}{\textbf{8} : 50\%} \\
\CI{50}{\textbf{9} : 50\%} } &
\makecell{
\CI{1}{\textbf{12} : \W1\%} \\
\CI{50}{\textbf{13} : 50\%} \\
\CI{44}{\textbf{14} : 44\%} \\
\CI{5}{\textbf{15} : \W5\%}  } &
\makecell{
\CI{10}{\textbf{21} : 10\%} \\
\CI{44}{\textbf{22} : 44\%} \\
\CI{39}{\textbf{23} : 39\%} \\
\CI{6}{\textbf{24} : \W6\%} \\
\CI{1}{\textbf{26} : \W1\%} } \\ \hline
        $10^5$ &
\makecell{
\CI{100}{\textbf{3} : 100\% }} & 
\makecell{
\CI{59}{\textbf{4} : 59\%} \\
\CI{41}{\textbf{5} : 41\%} } &
\makecell{
\CI{99}{\textbf{6} : 99\%} \\
\CI{1}{\textbf{7} : \W1\%} } &
\makecell{
\CI{19}{\textbf{\W8} : 19\%} \\
\CI{78}{\textbf{\W9} : 78\%} \\
\CI{3}{\textbf{10} : \W3\%} } &
\makecell{
\CI{21}{\textbf{13} : 21\%} \\
\CI{72}{\textbf{14} : 72\%} \\
\CI{7}{\textbf{15} : \W7\%}  } &
\makecell{
\CI{24}{\textbf{22} : 24\%} \\
\CI{51}{\textbf{23} : 51\%} \\
\CI{24}{\textbf{24} : 24\%} \\
\CI{1}{\textbf{26} : \W1\%} } \\ \hline
    \end{tabular}
}

   \footnotesize{
    \begin{tabular}{|c|c|c|c|c|c|c|c|}
    \multicolumn{7}{c}{\bfseries \SigmaNoisyLoad} \\
    \hline
$n$ & $\sigma = 0$ & $\sigma = 1$ & $\sigma = 2$ & $\sigma = 4$ & $\sigma = 8$ & $\sigma = 16$ \\ \hline
        $10^4$ &
\makecell{
\CI{46}{\textbf{2} : 46\%} \\
\CI{54}{\textbf{3} : 54\%}
 } &
\makecell{
\CI{29}{\textbf{3} : 29\%} \\
\CI{71}{\textbf{4} : 71\%} } &
\makecell{
\CI{9}{\textbf{4} : \W9\%} \\
\CI{84}{\textbf{5} : 84\%} \\
\CI{7}{\textbf{6} : \W7\%} } &
\makecell{
\CI{20}{\textbf{6} : 20\%} \\
\CI{73}{\textbf{7} : 73\%} \\
\CI{7}{\textbf{8} : \W7\%} } &
\makecell{
\CI{36}{\textbf{\W9} : 36\%} \\
\CI{50}{\textbf{10} : 50\%} \\
\CI{12}{\textbf{11} : 12\%} \\
\CI{2}{\textbf{12} : \W2\%} } &
\makecell{
\CI{2}{\textbf{12} : \W2\%} \\
\CI{33}{\textbf{13} : 33\%} \\
\CI{42}{\textbf{14} : 42\%} \\
\CI{16}{\textbf{15} : 16\%} \\
\CI{6}{\textbf{16} : \W6\%} \\
\CI{1}{\textbf{18} : \W1\%} } \\ \hline
        $5 \cdot 10^4$ &
\makecell{
\CI{4}{\textbf{2} : \W4\% } \\
\CI{96}{\textbf{3} : 96\% }
} &
\makecell{
\CI{98}{\textbf{4} : 98\%} \\
\CI{2}{\textbf{5} : \W2\%} } &
\makecell{
\CI{61}{\textbf{5} : 61\%} \\
\CI{39}{\textbf{6} : 39\%} } &
\makecell{
\CI{51}{\textbf{\W7} : 51\%} \\
\CI{48}{\textbf{\W8} : 48\%} \\
\CI{1}{\textbf{10} : \W1\%} } &
\makecell{
\CI{1}{\textbf{\W9} : \W1\%} \\
\CI{37}{\textbf{10} : 37\%} \\
\CI{52}{\textbf{11} : 52\%} \\
\CI{8}{\textbf{12} : \W8\%} \\
\CI{2}{\textbf{13} : \W2\%}  } &
\makecell{
\CI{24}{\textbf{14} : 24\%} \\
\CI{45}{\textbf{15} : 45\%} \\
\CI{24}{\textbf{16} : 24\%} \\
\CI{6}{\textbf{17} : \W6\%} \\
\CI{1}{\textbf{18} : \W1\%} } \\ \hline
        $10^5$ &
\makecell{
\CI{100}{\textbf{3} : 100\% }
} & 
\makecell{
\CI{95}{\textbf{4} : 95\%} \\
\CI{5}{\textbf{5} : \W5\%} } &
\makecell{
\CI{58}{\textbf{5} : 58\%} \\
\CI{41}{\textbf{6} : 41\%} \\
\CI{1}{\textbf{7} : \W1\%} } &
\makecell{
\CI{26}{\textbf{\W7} : 26\%} \\
\CI{69}{\textbf{\W8} : 69\%} \\
\CI{4}{\textbf{\W9} : \W4\%} \\
\CI{1}{\textbf{10} : \W1\%} } &
\makecell{
\CI{13}{\textbf{10} : 13\%} \\
\CI{56}{\textbf{11} : 56\%} \\
\CI{26}{\textbf{12} : 26\%} \\
\CI{4}{\textbf{13} : \W4\%} \\
\CI{1}{\textbf{14} : \W1\%}  } &
\makecell{
\CI{1}{\textbf{14} : \W1\%} \\
\CI{49}{\textbf{15} : 49\%} \\
\CI{35}{\textbf{16} : 35\%} \\
\CI{8}{\textbf{17} : \W8\%} \\
\CI{6}{\textbf{18} : \W6\%} \\
\CI{1}{\textbf{19} : \W1\%} } \\ \hline
    \end{tabular}
    }
    \caption{Empirical gap distribution for \GBounded, \GMyopicComp and \SigmaNoisyLoad with  $g, \sigma \in \{ 0, 1, 2, 4, 8, 16 \}$, for $n \in \{ 10^4, 5 \cdot 10^4, 10^5\}$ and $m = 1000 \cdot n$, over $100$ runs.}

    \label{tab:g_combined_empirical_gap_distribution}
\end{table}

\begin{table}[H]
    \centering
   \footnotesize{
\begin{tabular}{|c|c|c|c|c|c|c|}
\multicolumn{6}{c}{\bfseries \BBatch/\OneChoice} \\
\hline
Process & $b = 10$ & $b = 10^2$ & $b = 10^3$ & $b = 10^4$ & $b = 10^5$ \\ \hline
\makecell{$b$-\textsc{Batch}\\$m = 1000 \cdot n$} &
\makecell{
\CI{44}{\textbf{3} : 44\%} \\
\CI{56}{\textbf{4} : 56\%} } &
\makecell{
\CI{40}{\textbf{3} : 40\%} \\
\CI{60}{\textbf{4} : 60\%} } &
\makecell{
\CI{91}{\textbf{4} : 91\%} \\
\CI{9}{\textbf{5} : \W9\%} } &
\makecell{
\CI{29}{\textbf{5} : 29\%} \\
\CI{49}{\textbf{6} : 49\%} \\
\CI{18}{\textbf{7} : 18\%} \\
\CI{4}{\textbf{8} : \W4\%}  } &
\makecell{
\CI{1}{\textbf{16} : \W1\%} \\
\CI{8}{\textbf{17} : \W8\%} \\
\CI{15}{\textbf{18} : 15\%} \\
\CI{28}{\textbf{19} : 28\%} \\
\CI{18}{\textbf{20} : 18\%} \\
\CI{12}{\textbf{21} : 12\%} \\
\CI{14}{\textbf{22} : 14\%} \\
\CI{1}{\textbf{24} : \W1\%} \\
\CI{2}{\textbf{25} : \W2\%} \\
\CI{1}{\textbf{26} : \W1\%} } \\ \hline
\makecell{\textsc{One-Choice}\\$m = b$ balls} &
\makecell{
\CI{100}{\textbf{1} : 100\%}
} &
\makecell{
\CI{47}{\textbf{1} : 47\%} \\
\CI{52}{\textbf{2} : 52\%} \\
\CI{1}{\textbf{3} : \W1\%} } &
\makecell{
\CI{5}{\textbf{2} : \W5\%} \\
\CI{88}{\textbf{3} : 88\%} \\ 
\CI{7}{\textbf{4} : \W7\%} } &
\makecell{
\CI{22}{\textbf{6} : 22\%} \\
\CI{56}{\textbf{7} : 56\%} \\
\CI{19}{\textbf{8} : 19\%} \\
\CI{3}{\textbf{9} : \W3\%} } &
\makecell{
\CI{2}{\textbf{21} : \W2\%} \\
\CI{12}{\textbf{22} : 12\%} \\
\CI{13}{\textbf{23} : 13\%} \\
\CI{21}{\textbf{24} : 21\%} \\
\CI{18}{\textbf{25} : 18\%} \\
\CI{17}{\textbf{26} : 17\%} \\
\CI{4}{\textbf{27} : \W4\%} \\
\CI{8}{\textbf{28} : \W8\%} \\
\CI{4}{\textbf{29} : \W4\%} \\
\CI{1}{\textbf{31} : \W1\%} } \\ \hline
    \end{tabular}
    }
    \caption{Empirical gap distribution for \BBatch with $b \in \{ 5, 10, 5 \cdot 10, 10^2, \ldots 10^5, 5 \cdot 10^5 \}$, for $n = 10^4$ and $m=1000 \cdot n$ and for \OneChoice with $m = b$, over $100$ runs.}

    \label{tab:batch_gap_distribution}
\end{table}

\section{Conclusions}\label{sec:conclusion}

In this work, we proposed a number of different noise settings, and analyzed the gap of the \TwoChoice process. As our main result, we established for all $g \geq 1$ the asymptotic value of $\Gap(\cdot)$ in the adversarial setting \GAdvComp. This shows that $\Gap(\cdot)$ undergoes a delicate phase transition at the point $g=\log^{\Theta(1)} n$, around which the impact of $g$ on the gap switches from superlinear to linear. Further, we presented tight bounds for the delay and batched settings for a range of batch sizes $b$ and delay parameters $\tau$, where load information of the bins may be stale. 
Finally, we also proposed a smoothed probabilistic setting where each load comparison between two bins $i_1$ and $i_2$ is correct with some probability $\rho(\delta)$, where $\delta = |x_{i_1}^t-x_{i_2}^t|$ is their load difference.
Assuming that the load estimates of bins are independent Gaussian perturbations leads to $\rho(\delta)=1-\frac{1}{2} \exp(-(\delta/\sigma)^2)$. For this choice of $\rho$, we proved that the gap is some function that is polynomial in $\sigma$ and logarithmic in $n$. 

There are several directions for future work. One interesting avenue is establishing tight(er) bounds for \SigmaNoisyLoad or other choices of $\rho$, as well as determining tight(er) bounds for the remaining values of $b$ and $\tau$ in the delay settings. Also, it would be interesting to investigate the noisy setting for other balanced allocations processes, such as the \MeanThinning or the $(1+\beta)$ process.

\clearpage

\bibliographystyle{ACM-Reference-Format-CAM}
\bibliography{bibliography}

\appendix

\clearpage

\section{Analysis Tools}

\subsection{Concentration Inequalities}

We state the following well-known concentration inequalities.

\begin{lem}[\textbf{Azuma's Inequality for Super-Martingales {\cite[Problem 6.5]{DubPan}}}] \label{lem:azuma}
Consider a super-martingale $X^0, \ldots, X^N$ satisfying $|X^{i} - X^{i-1}| \leq a_i$ for any $i \in [N]$, then for any $\lambda > 0$,
\[
\Pro{X^N \geq X^0 + \lambda} \leq \exp\left(- \frac{\lambda^2}{2 \cdot \sum_{i=1}^N a_i^2} \right).
\]
\end{lem}

\begin{lem}[\textbf{\cite[Theorems~6.1 \& 6.5]{CL06}}] \label{lem:cl06_thm_6_1}
Consider a martingale $X^0, \ldots, X^N$ with filtration $\mathfrak{F}^0, \ldots, \mathfrak{F}^N$ satisfying $|X^{i} - X^{i-1}| \leq M$ and $\Var{X^i \mid \mathfrak{F}^{i-1}} \leq \sigma_i^2$ for any $i \in [N]$, then for any $\lambda > 0$,
\[
\Pro{\left| X^N - \Ex{X^N}\right| \geq \lambda} \leq 2 \cdot \exp\left(- \frac{\lambda^2}{2 \cdot (\sum_{i=1}^N \sigma_i^2 + M\lambda/3 )} \right).
\]
\end{lem}

\subsection{Facts about \OneChoice}

In this part, we collect several (mostly folklore) results for the maximum load and the gap of \OneChoice. In contrast to e.g., \cite{RS98}, we are aiming to obtain bounds on the gap with probability at least $1 - n ^{-6}$ and not just with probability at least $1 - o(1)$, but which are not necessarily tight up to lower order terms. Furthermore, the bounds in \cite{RS98} hold for any $m = \Omega\big(\frac{n}{\polylog(n)} \big)$, but we wish to also cover $m = n \cdot \exp\big(- \log^{\Oh(1)} n\big)$. We use these bounds to analyze the \BBatch and \TauDelay settings for various values of $b$ and $\tau$ in \cref{sec:g_adv_delay_noise_settings}.

We first restate the so-called Poisson Approximation method.
\begin{lem}[\textbf{\cite[Corollary 5.11]{MU17}}]\label{lem:poisson}
Let $(x^T)_{i \in [n]}$ be the load vector after $T$ steps of \OneChoice. Further, let $(\tilde{x}^T)_{i \in [n]}$ be $n$ independent Poisson random variables with parameter $\lambda=T/n$ each. Further, let $\mathcal{E}$ be any event which is determined by $x^T$, and further assume that $\Pro{ \mathcal{E}}$ is either monotonically increasing in $T$ or monotonically decreasing in $T$. Further, let $\tilde{\mathcal{E}}$ be the corresponding event determined by $\tilde{x}^T$. Then,
\[
 \Pro{ \mathcal{E}} \leq 2 \cdot \Pro{ \tilde{\mathcal{E}}}.
\]
\end{lem}

\begin{lem} \label{clm:one_choice_mgf}
Consider the \OneChoice process, for any $\alpha > 0$, any bin $i \in [n]$, and any step $m \geq 0$,
\[
\ex{e^{\alpha x_i^m}} \leq e^{\frac{m}{n} \cdot (e^{\alpha} - 1)}.
\]
\end{lem}
\begin{proof}
We will proceed inductively to show that $\ex{e^{\alpha x_i^m}} \leq e^{\frac{m}{n} \cdot (e^{\alpha} - 1)}$. The base case follows since $\ex{e^{\alpha x_i^0}} = 1 \leq 1$. For $m \geq 1$, let $Z_{i}^m$ indicate whether the $m$-th ball was allocated to bin $i \in [n]$ and assume that $\ex{e^{\alpha x_i^{m-1}}} \leq e^{\frac{m-1}{n} \cdot (e^{\alpha} - 1)}$ holds, then
\begin{align*}
 \Ex{e^{\alpha x_i^m}} & = \Ex{e^{\alpha (Z_{i}^m +  x_i^{m-1})}} = \Ex{e^{\alpha Z_{i}^m}} \cdot \Ex{e^{\alpha x_i^{m-1}}} \\
 & = \Big(\frac{1}{n} \cdot e^{\alpha} + \Big(1 - \frac{1}{n}\Big) \cdot e^0 \Big) \cdot \Ex{e^{\alpha x_i^{m-1}}}
 = \Big(1 + \frac{1}{n} \cdot (e^{\alpha} - 1)\Big) \cdot \Ex{e^{\alpha x_i^{m-1}}} \\
 & \stackrel{(a)}{\leq} e^{\frac{1}{n} \cdot (e^{\alpha} - 1)} \cdot \Ex{e^{\alpha x_i^{m-1}}} \\
 & \leq e^{\frac{1}{n} \cdot (e^{\alpha} - 1)} \cdot e^{\frac{m-1}{n} \cdot (e^{\alpha} - 1)} = e^{\frac{m}{n} \cdot (e^{\alpha} - 1)}.
\end{align*}
using in $(a)$ that $1 + u \leq e^u$ (for any $u$).
\end{proof}

We now proceed to obtain an upper bound for the maximum load of \OneChoice for any $m \leq 2n \log n$. As we will show in \cref{lem:one_choice_lower_bound_max_load_whp_lightly}, this bound is asymptotically tight.

\begin{lem}[\textbf{cf.~\cite[Lemma 14]{ACMR98}}] \label{lem:one_choice_n_polylog_n}
Consider the \OneChoice process for any $m \leq 2n \log n$. Then,
\[
\Pro{\max_{i \in [n]} x_i^m \leq 11 \cdot \frac{\log n}{\log (\frac{4n}{m} \cdot \log n)}} \geq 1 - n^{-6}.
\]
\end{lem}
\begin{proof}
Using \cref{clm:one_choice_mgf}, for the given $m$ and $\alpha = \log (\frac{4n}{m} \cdot \log n) > 0$, we have that for any bin $i \in [n]$,
\[
\ex{e^{\alpha x_i^m}} \leq e^{\frac{m}{n} \cdot e^{\log (\frac{4n}{m} \cdot \log n)}} = n^4.
\]
Hence, by Markov's inequality,
\[
\Pro{e^{\alpha x_i^m} \leq n^{11}} \geq 1 - n^{-7}. 
\]
When this event holds, we have 
\[
x_i^m 
 \leq \frac{1}{\alpha} \cdot \log(n^{11})
 \leq 11 \cdot \frac{\log n}{\log (\frac{4n}{m} \cdot \log n)}.
\]
By taking the union bound over all bins $i \in [n]$, we get the claim.
\end{proof}

For $m = \Oh(n)$, this recovers the well-known $\Oh(\frac{\log n}{\log \log n})$ bound.

\begin{cor}[\textbf{cf.~\cite[Chapter 5]{MU17}}] \label{lem:one_choice_lightly}
Consider the \OneChoice process for $m = 2n$. Then,
\[
\Pro{\max_{i \in [n]} x_i^m \leq 11 \cdot \frac{\log n}{\log \log n}} \geq 1 - n^{-6}.
\]
\end{cor}

For $m = \Oh(n^{1-\eps})$ for some constant $\eps \in (0, 1)$, this shows that $\Gap(m) = \Oh(1)$.

\begin{cor} \label{lem:very_very_lightly_loaded}
Consider the \OneChoice process with $m = 2n^{1-\eps}$ for any constant $\eps \in (0, 1)$. Then,
\[
\Pro{\max_{i \in [n]} x_i^m \leq \frac{11}{\eps}} \geq 1 - n^{-6}.
\]
\end{cor}

In the following lemma, we prove that the \OneChoice bound obtained in \cref{lem:one_choice_n_polylog_n} is asymptotically tight.

\begin{lem}[\textbf{cf.~Lemma 14 in \cite{ACMR98}}] \label{lem:one_choice_lower_bound_max_load_whp_lightly}
Consider the \OneChoice process with $m \leq n \log n$. Then, there exists a constant $\kappa > 0$, such that 
\[
\Pro{\max_{i \in [n]} x_i^m \geq \frac{1}{4} \cdot \frac{\log n}{\log(\frac{4n}{m} \cdot \log n)}} \geq 1 - n^{-1}.
\]
\end{lem}
\begin{proof}

We will bound the probability of event $\mathcal{E}$, that the maximum load is less than $M = \frac{1}{4} \cdot \frac{\log n}{\log(\frac{4n}{m} \cdot \log n)}$. Clearly, the probability of the event $\mathcal{E}$ is monotonically increasing
in the number of balls (while keeping $M$ fixed). %

Following \cref{lem:poisson}, it suffices to bound the probability of the event $\tilde{\mathcal{E}}$ which is that the maximum value of $n$ independent Poisson random variables with parameter $\lambda=\frac{m}{n}$ is less than $M$. 
We want to show that
\[
\Pro{\tilde{\mathcal{E}}} \leq \left( 1 - \frac{e^{-\frac{m}{n}} \left( \frac{m}{n}\right)^M}{M!} \right)^n \leq \exp{ \left(-n \cdot \frac{e^{-m/n} \left( \frac{m}{n}\right)^M}{M!} \right)} < n^{-1}.
\]
This is equivalent to showing that
\begin{align*}
-n \cdot \frac{e^{-m/n} \left( \frac{m}{n}\right)^M}{M!} < -\log{n} & \Longleftrightarrow 
\log{n} - \frac{m}{n} + M \cdot \log\left( \frac{m}{n} \right) -\log{M!} > \log \log n \\
& \Longleftrightarrow 
\log{n} + M \cdot \log\left( \frac{m}{n} \right) > \frac{m}{n} + \log{M!} + \log \log n.
\end{align*} 
Using $\log M! \leq M \cdot (\log M - 1) + \log M$ (e.g., in \cite[Equation 9.1]{F08}), we deduce that
\begin{align*}
& \frac{m}{n} + M \cdot \left(\log M - 1 - \log \left( \frac{m}{n}\right) \right) + \log M + \log \log n \\
& \qquad < \frac{3}{4} \cdot \log n + \frac{\log n}{4 \cdot \log (\frac{4n \log n}{m})} \cdot \left( \log \log n - \log \left( 4 \cdot \log\left( \frac{4n\log n}{m}\right) \right)  - \log\left( \frac{m}{n} \right) \right) \\
& \qquad = \frac{3}{4} \cdot \log n + \frac{\log n}{4 \cdot \log (\frac{4n \log n}{m})} \cdot \left( \log \left(\frac{4n\log n}{m} \right) - \log \left( 16 \cdot \log\left( \frac{4n\log n}{m}\right) \right)  \right) \\
& \qquad \leq \frac{3}{4} \cdot \log n + \frac{\log n}{4 \cdot \log (\frac{4n \log n}{m})} \cdot \log \left(\frac{4n\log n}{m} \right) \\
& \qquad = \frac{3}{4} \cdot \log n + \frac{1}{4} \cdot \log n = \log n,
\end{align*}
for sufficiently large $n$ and using that $\log\big(\frac{4n \log n}{m}\big) > 0$ as $m \leq n \log n$. Hence, we get the desired lower bound.
\end{proof}

The next standard result was also used in~\cite[Section~4]{PTW15} and is based on~\cite{RS98}. %
\begin{lem}[\textbf{cf.~\cite[Lemma A.2]{LS22RBB}}] \label{lem:one_choice_cnlogn}
Consider the \OneChoice process for $m = c n \log n$ where $c \geq 1/\log n$.
Then,
\[
\Pro{ \Gap(m) \geq \frac{\sqrt{c}}{10} \cdot \log n} \geq 1 - n^{-2}.
\]
\end{lem}

Combining the previous two lemmas, we also get the asymptotically tight bound on the gap

\begin{lem}[\textbf{cf.~\cite[Lemma 14]{ACMR98}}] \label{lem:one_choice_lower_bound_gap_whp_lightly}
Consider the \OneChoice process with $m \leq n \log n$. Then, there exists a constant $\kappa > 0$, such that 
\[
\Pro{\Gap(m) \geq \kappa \cdot \frac{\log n}{\log(\frac{4n}{m} \cdot \log n)}} \geq 1 - n^{-1}.
\]
\end{lem}
\begin{proof}
For sufficiently small constant $C \in (0, 1)$, for any $m \leq Cn \log n$ we have that 
\[
\frac{1}{4} \cdot \frac{\log n}{\log(\frac{4n}{m} \cdot \log n)} \geq 2 \cdot \frac{m}{n},
\]
and hence the conclusion follows by \cref{lem:one_choice_lower_bound_max_load_whp_lightly} for $\kappa = \frac{1}{8}$. For $m > Cn \log n$, the stated bound follows from \cref{lem:one_choice_cnlogn}.
\end{proof}

\subsection{Auxiliary Probabilistic Claims}

For convenience, we state and prove the following well-known result.

\begin{lem} \label{lem:geometric_arithmetic}
Consider any sequence of random variables $(X^i)_{i \in \mathbb{N}}$ for which there exist $0 < a < 1$ and $b > 0$, such that every $i \geq 1$,
\[
\Ex{X^i \mid X^{i-1}} \leq X^{i-1} \cdot a + b.
\]
Then, $(i)$~for every $i \geq 0$, 
\[
\Ex{X^i \mid X^0}
\leq X^0 \cdot a^i + \frac{b}{1 - a}.
\]
Further, $(ii)$~if $X^0 \leq \frac{b}{1-a}$ holds, then for every $i \geq 0$,
\[
\Ex{X^i} \leq \frac{b}{1 - a}.
\]
\end{lem}
\begin{proof}
\textit{First statement.} We will prove by induction that for every $i \in \mathbb{N}$, 
\[
\Ex{X^i \mid X^0} \leq X^0 \cdot a^i + b \cdot \sum_{j = 0}^{i-1} a^j.
\]
For $i = 0$, $\Ex{X^0 \mid X^0} \leq X^0$. Assuming the induction hypothesis holds for some $i \geq 0$, then since $a > 0$,
\begin{align*}
\Ex{X^{i+1} \mid X^0} & = \Ex{\Ex{X^{i+1} \mid X^i}\mid X^0} \leq \Ex{X^{i}\mid X^0} \cdot a + b \\
 & \leq \left(X^0 \cdot a^i + b \cdot \sum_{j = 0}^{i-1} a^j \right) \cdot a + b \\
 & = X^0 \cdot a^{i+1} +b \cdot \sum_{j = 0}^i a^j.
\end{align*}
The claims follows using that $\sum_{j = 0}^i a^j \leq \sum_{j=0}^{\infty} a^j = \frac{1}{1-a}$, for any $a \in (0,1)$.

\textit{Second statement.} We will prove this claim by induction. Then, assuming that $\Ex{X^i} \leq \frac{b}{1-a}$ holds for $i \geq 0$, we have for $i+1$
\begin{align*}
\Ex{X^{i+1}}
  & = \Ex{\Ex{X^{i+1} \mid X^{i}}}  \leq \Ex{X^{i} } \cdot a + b \leq \frac{b}{1-a} \cdot a + b = \frac{b}{1-a}. \qedhere 
\end{align*}
\end{proof}

\subsection{Majorization}

We begin with the following well-known majorization inequality which appears in Theorem 1 \cite[Chapter XII]{MPFbook}. For a proof see e.g., \cite[Lemma A.2]{LSS21}.

\begin{lem}[\textbf{cf.~\cite[Lemma A.2]{LSS21}}]\label{lem:quasilem}Let $(a_k)_{k=1}^n , (b_k)_{k=1}^n $ be non-negative and $(c_k)_{k=1}^n$ be non-negative and non-increasing. If $\sum_{k=1}^i a_k \leq \sum_{k=1}^i b_k$ holds for all $1\leq i\leq n$ then, \begin{equation*} %
\sum_{k=1}^n a_k\cdot c_k \leq \sum_{k=1}^n b_k\cdot c_k.\end{equation*}
\end{lem} 

Next we state an auxiliary result, which is implicit in~\cite{PTW15}. Theorem 3.1 in \cite{PTW15} proves the required majorization for time-independent probability allocation vectors, but as mentioned in the proof of Theorem 3.2 in~\cite{PTW15}, the same result generalizes to time-dependent probability allocation vectors.
\begin{lem}[\textbf{cf.~\cite[Section 3]{PTW15}}]\label{lem:majorisation}
Consider two allocation processes $\mathcal{Q}$ and $\mathcal{P}$. The allocation process $\mathcal{Q}$ uses at each step a fixed probability allocation vector $q$. The allocation process $\mathcal{P}$ uses a time-dependent probability allocation vector $p^{t}$, which may depend on $\mathfrak{F}^{t}$ but majorizes $q$ at each step $t \geq 0$. Let $y^{t}(\mathcal{Q})$ and $y^{t}(\mathcal{P})$ be the two normalized load vectors, sorted non-increasingly. Then there is a coupling such that for all steps $t \geq 0$, $y^{t}(\mathcal{P})$ majorizes $y^{t}(\mathcal{Q})$, so in particular, $y_1^t(\mathcal{P}) \geq y_1^t(\mathcal{Q})$.
\end{lem}

\clearpage

\section{Omitted Proofs and Auxiliary Claims}

\subsection{Omitted Proofs from Section~\ref{sec:g_adv_warm_up}}

\label{sec:g_adv_gap_remains_small_proof}

{\renewcommand{\thecor}{\ref{cor:g_adv_gap_remains_small}}
	\begin{cor}[\textbf{Restated, page~\pageref{cor:g_adv_gap_remains_small}}]
\CorGAdvGapRemainsSmall
	\end{cor} }
	\addtocounter{cor}{-1}

\begin{proof}
We will be using the hyperbolic cosine potential $\Gamma := \Gamma(\gamma)$ with smoothing parameter $\gamma := - \log(1 - \frac{1}{8 \cdot 48})/g$ as we did in the proof of \cref{thm:g_adv_warm_up_gap}.

Consider an arbitrary step $t_0$ with $\max_{i \in [n]} \left| y_i^{t_0} \right| \leq g (\log(ng))^2$. Then, it follows that
\[
\Gamma^{t_0} \leq 2n \cdot e^{\gamma g (\log(ng))^2}.
\]
By \cref{thm:g_adv_warm_up_gap}~$(i)$, there exists a constant $c_1 \geq 1$ such that \[
\Ex{\left. \Gamma^{t+1} \,\right|\, \mathfrak{F}^t} \leq \left(1 - \frac{\gamma}{96n} \right) + c_1,
\]
and using \cref{lem:geometric_arithmetic}~$(i)$ (for $a := 1 - \frac{\gamma}{96n}$ and $b := c_1$) at step $t_1 \geq t_0$, we have that
\[
\Ex{ \Gamma^{t_1} \,\left|\, \mathfrak{F}^{t_0}, \Gamma^{t_0} \leq 2n \cdot e^{\gamma g (\log(ng))^2} \right. } \leq 2n \cdot  e^{\gamma g (\log(ng))^2} \cdot \left( 1 - \frac{\gamma}{96n} \right)^{t_1 - t_0} + \frac{96c_1n}{\gamma} \leq 4n \cdot  e^{\gamma g (\log(ng))^2},
\]
recalling that $\gamma = \Theta\big(\frac{1}{g}\big)$. Hence, by Markov's inequality, we have that
\begin{align*}
& \Pro{\left. \Gamma^{t_1} \leq 4n \cdot e^{\gamma g (\log(ng))^2} \cdot (ng)^{14} \,\right|\, \mathfrak{F}^{t_0}, \max_{i \in [n]} \left| y_i^{t_0} \right| \leq g (\log(ng))^2}  \\
& \quad \geq \Pro{\left. \Gamma^{t_1} \leq 4n \cdot e^{\gamma g (\log(ng))^2} \cdot (ng)^{14} \,\right|\, \mathfrak{F}^{t_0}, \Gamma^{t_0} \leq 2n \cdot e^{\gamma g (\log(ng))^2}} \geq 1 - (ng)^{-14}.
\end{align*}
When the event $\big\{ \Gamma^{t_1} \leq 4n \cdot e^{\gamma g (\log(ng))^2} \cdot (ng)^{14} \big\}$ holds, then it also follows that
\[
\max_{i \in [n]} \left| y_i^{t_1} \right|
  \leq \frac{\log \Gamma^{t_1}}{\gamma} 
  \leq g (\log(ng))^2 + \frac{1}{\gamma} \cdot \left( \log (4n) + 14 \log( ng ) \right) \leq 2 g (\log(ng))^2,
\]
for sufficiently large $n$ and using that $\gamma = \Theta\big(\frac{1}{g}\big)$. Hence, we get the conclusion.
\end{proof}

\subsection{Omitted Proofs from Section~\ref{sec:g_adv_g_plus_logn_bound}}

{\renewcommand{\thelem}{\ref{lem:g_adv_general_quadratic}}
	\begin{lem}[\textbf{Restated, page~\pageref{lem:g_adv_general_quadratic}}]
\GAdvGeneralQuadratic
	\end{lem} }
	\addtocounter{lem}{-1}

\label{lem:g_adv_general_quadratic_proof}
\begin{proof}
\textit{First statement.} For any bin $i \in [n]$, its expected contribution to $\Upsilon^{t+1}$ is given by,
\begin{align*}
\Ex{\left. \Upsilon_i^{t+1} \,\right|\, y^t} 
 & = \Big(y_i^t + 1 - \frac{1}{n} \Big)^2 \cdot r_i^t + \Big(y_i^t - \frac{1}{n}\Big)^2 \cdot (1 - r_i^t) \\
 & = (y_i^t)^2 + 2 \cdot \Big(1 - \frac{1}{n}\Big) \cdot y_i^t \cdot r_i^t - 2 \cdot \frac{1}{n} \cdot y_i^t \cdot (1 - r_i^t) + \Big(1 - \frac{1}{n}\Big)^2 \cdot r_i^t + \frac{1}{n^2} \cdot (1 - r_i^t) \\
 & = (y_i^t)^2 + 2 \cdot \Big(r_i^t - \frac{1}{n}\Big) \cdot y_i^t + \Big(1 - \frac{1}{n}\Big)^2 \cdot r_i^t + \frac{1}{n^2} \cdot (1 - r_i^t).
\end{align*}
Hence, by aggregating over all bins we get,
\begin{align*}
\Ex{\left. \Upsilon^{t+1} \,\right|\, y^t} 
 & = \sum_{i = 1}^n \left[ (y_i^t)^2 + 2 \cdot \Big(r_i^t - \frac{1}{n}\Big) \cdot y_i^t + \Big(1 - \frac{1}{n}\Big)^2 \cdot r_i^t + \frac{1}{n^2} \cdot (1 - r_i^t) \right] \\
 & = \Upsilon^t + \sum_{i = 1}^n 2 \cdot \Big(r_i^t - \frac{1}{n}\Big) \cdot y_i^t + \Big(1 - \frac{1}{n}\Big)^2 + \frac{1}{n} - \frac{1}{n^2} \\
 & = \Upsilon^t + \sum_{i = 1}^n 2 \cdot \Big(r_i^t - \frac{1}{n}\Big) \cdot y_i^t + 1 - \frac{1}{n} \\
 & = \Upsilon^t + \sum_{i = 1}^n 2 \cdot r_i^t \cdot y_i^t - \sum_{i = 1}^n \frac{2}{n} \cdot y_i^t + 1 - \frac{1}{n} \\
 & \stackrel{(a)}{=} \Upsilon^t + \sum_{i = 1}^n 2 \cdot r_i^t \cdot y_i^t + 1 - \frac{1}{n} \\
 & \leq \Upsilon^t + \sum_{i = 1}^n 2 \cdot r_i^t \cdot y_i^t + 1,
\end{align*}
using in $(a)$ that $\sum_{i = 1}^n y_i^t = 0$.
Therefore, by subtracting $\Upsilon^t$, statement $(i)$ follows.

\textit{Second statement.} Let $M :=  \max_{i \in [n]} |y_i^t|$.
We will upper bound the change $\Delta\Upsilon_i^{t+1}$ for an arbitrary bin $i \in [n]$, by considering the following two cases: %

\textbf{Case 1:} Ball at step $t+1$ is allocated to bin $i$. So,
\[
\left|\Delta\Upsilon_i^{t+1}\right| = \Big|\Big(y_i^t + 1 - \frac{1}{n}\Big)^2 - (y_i^t)^2 \Big| = \Big|2 \cdot \Big(1 - \frac{1}{n} \Big) \cdot y_i^t + \Big(1 - \frac{1}{n}\Big)^2 \Big| \leq 2M + 1.
\]

\textbf{Case 2:} Ball at step $t+1$ is not allocated to bin $i$. So,
\[
\left|\Delta\Upsilon_i^{t+1}\right| = \left|\Big(y_i^t - \frac{1}{n}\Big)^2 - (y_i^t)^2\right| = \left|-\frac{2}{n} \cdot y_i^t + \frac{1}{n^2}\right| \leq \frac{2M}{n} + \frac{1}{n}.
\]
Aggregating over all bins $i \in [n]$ yields
\begin{align*}
\left|\Delta\Upsilon^{t+1}\right| \leq \sum_{i = 1}^n \left|\Delta\Upsilon_i^{t+1}\right| 
&\leq 2M + 1 + (n-1) \cdot \left( \frac{2M}{n} + \frac{1}{n} \right) \leq 4M + 2. \qedhere
\end{align*}
\end{proof}

\label{sec:g_adv_lambda_tilde_is_supermartingale}

{\renewcommand{\thelem}{\ref{lem:g_adv_lambda_tilde_is_supermartingale}}
	\begin{lem}[\textbf{Restated, page~\pageref{lem:g_adv_lambda_tilde_is_supermartingale}}]
\LambdaTildeSuperMartingale
	\end{lem} }
	\addtocounter{lem}{-1}

\begin{proof}
Recalling the definition of $\tilde{\Lambda}$ in \cref{eq:tilde_lambda},  
\begin{align*}
\lefteqn{ \ex{ \tilde{\Lambda}_{t_0}^{s+1} \,\big|\, \mathfrak{F}^s} } \\
 &=\ex{ \Lambda^{s+1} \cdot \mathbf{1}_{\mathcal{E}_{t_0}^{s}} \,\big|\, \mathfrak{F}^s} \cdot \exp\Big( - \frac{3\alpha}{n} \cdot B_{t_0}^{s} \Big) \cdot \exp\Big( \frac{\alpha \eps}{n} \cdot G_{t_0}^{s} \Big) \\ 
 & = \ex{ \Lambda^{s+1} \cdot \mathbf{1}_{\mathcal{E}_{t_0}^{s}} \,\big|\, \mathfrak{F}^s} \cdot \exp\left(\frac{\alpha \eps}{n} \cdot \mathbf{1}_{\mathcal{G}^s} - \frac{3\alpha}{n} \cdot \mathbf{1}_{\neg \mathcal{G}^s}\right) \cdot \exp\Big( - \frac{3\alpha}{n} \cdot B_{t_0}^{s - 1} \Big) \cdot \exp\Big( \frac{\alpha \eps}{n} \cdot G_{t_0}^{s - 1} \Big).
\end{align*}
Thus, to prove the statement, it suffices to show that \begin{equation}\label{eq:sufficient2}\ex{ \Lambda^{s+1} \cdot \mathbf{1}_{\mathcal{E}_{t_0}^{s}} \,\big|\, \mathfrak{F}^s} \cdot \exp\left(\frac{\alpha \eps}{n} \cdot \mathbf{1}_{\mathcal{G}^s} - \frac{3\alpha}{n} \cdot \mathbf{1}_{\neg \mathcal{G}^s}\right) 
 \leq \Lambda^{s} \cdot \mathbf{1}_{\mathcal{E}_{t_0}^{s-1}}.\end{equation}To show \cref{eq:sufficient2}, we consider two cases based on whether $\mathcal{G}^s$ holds.

\medskip 

\noindent\textbf{Case 1 [$\mathcal{G}^s$ holds]:} 
Recall that when $\mathcal{G}^s$ holds, then $\Delta^{s} \leq Dng$. Further, when $\Lambda^s \leq cn$ holds (for $c > 0$ the constant in \cref{lem:g_adv_good_step_drop}), then $\mathbf{1}_{\mathcal{E}_{t_0}^{s}} = 0$. Thus, using \cref{lem:g_adv_good_step_drop}~$(ii)$, 
\[
\ex{ \Lambda^{s+1} \cdot \mathbf{1}_{\mathcal{E}_{t_0}^{s}} \,\big|\, \mathfrak{F}^s, \mathcal{G}^s} \leq \Lambda^{s} \cdot \mathbf{1}_{\mathcal{E}_{t_0}^{s-1}}  \cdot \left(1 - \frac{\alpha\eps}{n} \right) \leq \Lambda^{s} \cdot \mathbf{1}_{\mathcal{E}_{t_0}^{s-1}} \cdot \exp\left(- \frac{\alpha \eps}{n} \right).
\]
Hence, since in this case $\mathbf{1}_{\mathcal{G}^s}=1$, the left hand side of \cref{eq:sufficient2} is equal to
\begin{align*}
\ex{\Lambda^{s+1} \cdot \mathbf{1}_{\mathcal{E}_{t_0}^{s}} \,\big|\, \mathfrak{F}^s, \mathcal{G}^s} \cdot \exp\left(\frac{ \alpha \eps}{n} \right) &\leq \left(\Lambda^{s} \cdot \mathbf{1}_{\mathcal{E}_{t_0}^{s-1}} \cdot \exp\left(- \frac{\alpha \eps}{n} \right)\right) \cdot \exp\left( \frac{ \alpha \eps}{n} \right)  = \Lambda^{s} \cdot \mathbf{1}_{\mathcal{E}_{t_0}^{s-1}} .
\end{align*}
\noindent\textbf{Case 2 [$\mathcal{G}^s$ does not hold]:} By \cref{lem:g_adv_bad_step_increase}~$(ii)$,  we get \[
\ex{\Lambda^{s+1} \cdot \mathbf{1}_{\mathcal{E}_{t_0}^{s}} \,\big|\, \mathfrak{F}^s, \neg \mathcal{G}^s} 
\leq \Lambda^{s} \cdot \mathbf{1}_{\mathcal{E}_{t_0}^{s-1}} \cdot \left(1 + \frac{3 \alpha}{n}\right) \leq \Lambda^{s} \cdot \mathbf{1}_{\mathcal{E}_{t_0}^{s-1}}\cdot \exp\left(\frac{3 \alpha }{n} \right).
\]
Hence,  since in this case $\mathbf{1}_{\mathcal{G}^s}=0$, the left hand side of \cref{eq:sufficient2} is equal to
\begin{align*}
\ex{\Lambda^{s+1} \cdot \mathbf{1}_{\mathcal{E}_{t_0}^{s}} \,\big|\, \mathfrak{F}^s, \neg \mathcal{G}^s} \cdot \exp\left(- \frac{3 \alpha }{n} \right) &\leq \left( \Lambda^{s} \cdot \mathbf{1}_{\mathcal{E}_{t_0}^{s-1}}\cdot \exp\left(\frac{3 \alpha }{n} \right)\right)   \cdot \exp\left( -\frac{3 \alpha}{n}   \right) = \Lambda^{s} \cdot \mathbf{1}_{\mathcal{E}_{t_0}^{s-1}}.
\end{align*} Since \cref{eq:sufficient2} holds in either case, we deduce that $(\tilde{\Lambda}_{t_0}^s)_{s \geq t_0}$ forms a super-martingale.\end{proof}

{\renewcommand{\thelem}{\ref{lem:g_adv_bounds_on_quadratic}}
	\begin{lem}[\textbf{Restated, page~\pageref{lem:g_adv_bounds_on_quadratic}}]
\GAdvBoundsOnQuadratic
	\end{lem} }
	\addtocounter{lem}{-1}

\label{sec:g_adv_bounds_on_quadratic_proof}

\begin{proof}
\textit{First statement.} We begin by proving some basic inequalities between exponential, quadratic and linear terms. Let $\hat{u} := (4/\alpha) \cdot \log(4/\alpha)$. Note that $e^u \geq u$ (for any $u \geq 0$) and hence for any $u \geq \hat{u}$,
\begin{align*}
e^{\alpha u/2} = e^{\alpha u/4} \cdot e^{\alpha u/4} \geq \frac{\alpha u}{4} \cdot e^{\alpha \hat{u}/4} =  \frac{\alpha u}{4} \cdot \frac{4}{\alpha} = u,
\end{align*}
and $e^{\alpha u} = e^{\alpha u/2} \cdot e^{\alpha u/2} \geq u \cdot u = u^2$. Therefore, for every $u \geq 0$,
\begin{align} \label{eq:combined_quad_exp}
u^2 \leq \max\left\{\hat{u}^2, e^{\alpha u}\right\}.
\end{align}
Recall that for any bin $i \in [n]$, $\Lambda_i^t := e^{\alpha \cdot (y_i^t - c_4g)^+} + e^{\alpha \cdot (-y_i^t - c_4g)^+}$. Hence, 
\begin{align}
 & \left( (y_i^t - c_4g)^+ \right)^2 + \left( (-y_i^t - c_4g)^+ \right)^2  \stackrel{(a)}{\leq} \max \left\{ 2\hat{u}^2, \Lambda_i^{t} \right\} 
 \stackrel{(b)}{\leq} \max \left\{ 2\hat{u}^2 \cdot \Lambda_i^{t}, \Lambda_i^{t} \right\} 
 \stackrel{(c)}{=} 2\hat{u}^2 \cdot \Lambda_i^t, \label{eq:quadratic_potential_term_bound}
\end{align} 
where in $(a)$ we used \cref{eq:combined_quad_exp} first with $u = (y_i^t - c_4g)^+$ and then with $u = (-y_i^t - c_4g)^+$, in $(b)$ that $\Lambda_i^{t} \geq 1$ for any $i \in [n]$ and in $(c)$ that $\hat{u} \geq 1$, since $\alpha \in (0, 1)$.

We now proceed to upper bound the quadratic potential,
\begin{align*}
\Upsilon^t & \leq \sum_{i = 1}^n \Big[ \left( (y_i^t - c_4g)^+ + c_4g\right)^2 + \left( (-y_i^t - c_4g)^+ + c_4g\right)^2 \Big] \\
 & \stackrel{(a)}{\leq} 2 \cdot \sum_{i = 1}^n \Big[ \left( (y_i^t - c_4g)^+\right)^2 + \left( (-y_i^t - c_4g)^+ \right)^2 + 2 \cdot (c_4g)^2 \Big] \\
 & \!\!\stackrel{(\text{\ref{eq:quadratic_potential_term_bound}})}{\leq} 4\hat{u}^2 \cdot \Lambda^t + 4c_4^2 \cdot ng^2  \\
 & \stackrel{(b)}{\leq} 4\hat{c}\hat{u}^2 \cdot n + 4c_4^2 \cdot ng^2  \\
 & \leq (4 \hat{c}\hat{u}^2 + 4c_4^2) \cdot ng^2,
\end{align*}
using in $(a)$ that $(a+b)^2 \leq 2 \cdot (a^2 + b^2)$ (for any $a, b$) and in $(b)$ that $\Lambda^t \leq \hat{c}\cdot n$. Therefore, for the constant $c_s := c_s(\alpha, c_4, \hat{c}) := 4 \hat{c}\hat{u}^2 + 4c_4^2$, we get the first statement.

\textit{Second statement.} For any bin $i \in [n]$ we have,
\[
|y_i^t| \leq c_4g + \frac{1}{\alpha} \log \Lambda^t.
\]
Hence, using that $(a+b)^2 \leq 2 \cdot (a^2 + b^2)$,
\[
(y_i^t)^2 \leq \Big(c_4g + \frac{1}{\alpha} \log \Lambda^t\Big)^2 
\leq 2 \cdot \Big( c_4^2 g^2 + \frac{1}{\alpha^2} \cdot (\log \Lambda^t)^2 \Big)
\leq c_r \cdot \Big(g^2 + (\log \Lambda^t)^2 \Big),
\]
for some constant $c_r := c_r(\alpha, c_4) = \max\left\{ 2c_4^2, \frac{2}{\alpha^2} \right\} \geq 1$. By aggregating the contributions over all bins, we get the second statement.
\end{proof}

{\renewcommand{\thelem}{\ref{lem:g_adv_stabilization}}
	\begin{lem}[\textbf{Restated, page~\pageref{lem:g_adv_stabilization}}]
\GAdvStabilization
	\end{lem} }
	\addtocounter{lem}{-1}

\label{sec:g_adv_stabilization_proof}
\begin{proof}
By \cref{lem:g_adv_bounds_on_quadratic}~$(i)$, $\Lambda^{t_0} \leq 2cn$ implies that deterministically $\Upsilon^{t_0} \leq c_s n g^2$ for constant $c_s := c_s(\alpha, c_4, 2c) \geq 1$ and \[
\max_{i \in [n]} \left| y_i^{t_0} \right| 
  \leq c_4g + \frac{1}{\alpha} \cdot \log(2cn)
  \leq g (\log(ng))^2,
\]
for sufficiently large $n$ using that $c_4, \alpha, c > 0$ are constants. Let $t_1 :=t_0 + \Delta_s$. By \cref{lem:g_adv_many_good_steps} with $T := c_s ng \cdot \max\{\log n, g\} \geq c_s ng^2$ (and $o(n^2 g^3)$) and $\hat{c} := \frac{\Delta_s \cdot g}{T} = \frac{60}{\alpha\eps r} \geq 1$ as $\alpha, \eps, r \leq 1$, we have that
\begin{align}
 & \Pro{ G_{t_0}^{t_1-1} \geq r \cdot \Delta_s \;\Big|\; \mathfrak{F}^{t_0}, \Lambda^{t_0} \leq 2cn } \notag \\
 & \quad \geq \Pro{ G_{t_0}^{t_1-1} \geq r \cdot \Delta_s \;\Big|\; \mathfrak{F}^{t_0}, \Upsilon^{t_0} \leq T, \max_{i \in [n]} \left| y_i^{t_0} \right| \leq g (\log(ng))^2} \notag \\ 
 & \quad \geq 1 - 2 \cdot (ng)^{-12}. \label{eq:many_good_quantiles_whp_new}
\end{align}
By \cref{lem:g_adv_lambda_tilde_is_supermartingale}, $(\tilde{\Lambda}_{t_{0}}^{t})_{t \geq t_{0}}$ is a super-martingale, so $\ex{\tilde{\Lambda}_{t_0}^{t_1}\mid \mathfrak{F}^{t_0}, \Lambda^{t_0} \leq 2cn } \leq \tilde{\Lambda}_{t_0}^{t_0} =  \Lambda^{t_0}$. Hence, using Markov's inequality we get $\Pro{\tilde{\Lambda}_{t_0}^{t_1} >\Lambda^{t_0} \cdot (ng)^{12}\mid \mathfrak{F}^{t_0}, \Lambda^{t_0} \leq 2cn } \leq  (ng)^{-12}$. Thus, by the definition of $\tilde{\Lambda}_{t_0}^{t_1}$ in \cref{eq:tilde_lambda}, we have 
\begin{equation} \label{eq:supermartingale_markov_new}
\Pro{\Lambda^{t_1} \cdot \mathbf{1}_{\mathcal{E}_{t_0}^{t_1-1}} \leq \Lambda^{t_0} \cdot (ng)^{12} \cdot \exp\left(  \frac{3 \alpha}{n} \cdot B_{t_0}^{t_1-1}  -\frac{\alpha \eps}{n} \cdot G_{t_0}^{t_1-1} \right) \,\, \Bigg| \,\, \mathfrak{F}^{t_0}, \Lambda^{t_0} \leq 2cn } \geq  1 - (ng)^{-12}.    
\end{equation}
Further, if in addition to the two events $\{\tilde{\Lambda}_{t_0}^{t_1} \leq \Lambda^{t_0} \cdot (ng)^{12}\}$ and $\{\Lambda^{t_0} \leq 2cn \}$, also the event $\{G_{t_0}^{t_1-1} \geq r \cdot \Delta_s \}$ holds, then 
\begin{align*}
\Lambda^{t_1} \cdot \mathbf{1}_{\mathcal{E}_{t_0}^{t_1-1}} 
 & \leq \Lambda^{t_0} \cdot (ng)^{12} \cdot \exp\bigg(  \frac{3 \alpha}{n} \cdot B_{t_0}^{t_1-1}  -\frac{\alpha \eps}{n} \cdot G_{t_0}^{t_1-1} \bigg) \\
 & \leq 2cn \cdot (ng)^{12} \cdot \exp\left( \frac{3 \alpha}{n} \cdot (1 - r) \cdot \Delta_s - \frac{\alpha \eps}{n} \cdot r \cdot \Delta_s \right) \\
 & \stackrel{(a)}{=} %
 2cn \cdot (ng)^{12} \cdot \exp\bigg( - \frac{ \alpha \eps}{n} \cdot \frac{r}{2} \cdot \Delta_s \bigg) \\
 & = 2cn \cdot (ng)^{12} \cdot \exp\bigg( - \frac{\alpha \eps}{n} \cdot \frac{r}{2} \cdot  \frac{60 c_s}{\alpha \eps r} \cdot n \cdot \max\{\log n, g\} \bigg) \\ 
 & \stackrel{(b)}{\leq} 2cn \cdot (ng)^{12} \cdot \exp\left( - 30 \cdot \max\{\log n, g\} \right) \\
 & \leq 2cn \cdot (ng)^{12} \cdot \exp\left( - 15 \cdot \log (ng) \right) \\
 & \leq 1,
\end{align*}
where we used in $(a)$ that $r = \frac{6}{6 + \eps}$ implies $\frac{3\alpha}{n} \cdot (1-r) = \frac{3\alpha}{n} \cdot \frac{\eps}{6 + \eps} = \frac{\alpha\eps}{n} \cdot \frac{r}{2}$ and in $(b)$ that $c_s \geq 1$ and $\alpha \leq 1/2$.
Also $\Lambda^{t_1} \geq 2n$ holds deterministically, so we can deduce from the above inequality that $\mathbf{1}_{\mathcal{E}_{t_0}^{t_1-1}}=0$, that is,
\[\Pro{ \neg \mathcal{E}_{t_0}^{t_1 -1} \; \Bigg|\; \mathfrak{F}^{t_0},  \;\; \; \tilde{\Lambda}_{t_0}^{t_1} \leq \Lambda^{t_0} \cdot (ng)^{12}, \;\;\; \Lambda^{t_0} \leq 2cn,  \;\;\; G_{t_0}^{t_1-1} \geq r \cdot \Delta_s } = 1.\] 
Recalling the definition of $\mathcal{E}_{t_0}^{t_1-1} := \bigcap_{t \in [t_0, t_1-1]} \{ \Lambda^t > cn \}$, and taking the union bound over \cref{eq:many_good_quantiles_whp_new} and \cref{eq:supermartingale_markov_new} yields 
\[
\Pro{ \bigcup_{t \in [t_0, t_0 + \Delta_s]} \{ \Lambda^t \leq cn \} \; \bigg| \; \mathfrak{F}^{t_0}, \Lambda^{t_0} \leq 2cn  }\geq 1 - 2 \cdot (ng)^{-12} - (ng)^{-12}\geq 1 -  (ng)^{-11}. \qedhere 
\]
\end{proof} 

\subsection{Omitted Proofs from Section~\ref{sec:g_adv_upper_bound_for_small_g_outline}}

\begin{clm} \label{clm:tilde_g_justification}
Consider $\alpha_1, \alpha_2 > 0$ as defined in \cref{eq:g_adv_alpha_1_def} and \cref{eq:g_adv_alpha_2_def} respectively. Then, for any $g \geq \frac{\alpha_2}{4\sqrt{\alpha_1}}$ and $k := k(g) \geq 2$ being the unique integer such that $(\alpha_1 \log n)^{1/k} < g \leq (\alpha_1 \log n)^{1/(k-1)}$, it holds that
\[
(\alpha_1 \cdot (\log n))^{1/k} \leq \left(\frac{\alpha_2}{4} \cdot (\log n)\right)^{1/(k-1)}.
\]
\end{clm}
\begin{proof}
Let $R := \frac{\alpha_2}{4\sqrt{\alpha_1}} = \frac{\sqrt{\alpha_1}}{4 \cdot 84} \leq 1$ (using that $\alpha_2 = \frac{\alpha_1}{84}$ and $\alpha_1 \leq 1$). By rearranging the target inequality,
\begin{align*}
\left(\frac{\alpha_2}{4} \cdot (\log n)\right)^{1/(k-1)} \cdot (\alpha_1 \cdot (\log n))^{-1/k} 
 & = \exp\left(\frac{1}{k-1} \cdot \log\left(\frac{\alpha_2}{4} \log n\right)  - \frac{1}{k} \cdot \log(\alpha_1 \log n) \right) \\
 & \stackrel{(a)}{\geq} \exp\left(\frac{1}{k-1} \cdot \left( \log R + \frac{1}{k} \cdot \log \log n\right) \right) \\
 & \stackrel{(b)}{\geq} \exp\left(\frac{1}{k-1} \cdot \left( \log R - \log R \right) \right) = 1.
\end{align*}
using in $(a)$ that $-\frac{1}{k} \log \alpha_1 \geq - \frac{1}{2(k-1)} \log \alpha_1$ and in $(b)$ that $k \leq \frac{1}{-\log R} \cdot \log \log n$ since $g \geq R^{-1}$.
\end{proof}

\subsection{Omitted Proofs from Section~\ref{sec:g_adv_base_case}}

Next, we proceed with a simple smoothness argument for the potential $\V$ defined in \cref{eq:g_adv_v_def}.

{\renewcommand{\thelem}{\ref{lem:g_adv_v_smoothness}}
	\begin{lem}[\textbf{Restated, page~\pageref{lem:g_adv_v_smoothness}}]
\GAdvVSmoothness
	\end{lem} }
	\addtocounter{lem}{-1}

\label{sec:g_adv_v_smoothness_proof}
\begin{proof}\textit{First statement.} In each step the normalized load of any bin can change by at most $1$, i.e., $|y_i^{t+1} - y_i^t| \leq 1$ and so $e^{-\alpha_1} \cdot \V_i^t \leq \V_i^{t+1} \leq e^{\alpha_1} \cdot \V_i^t$. By aggregating over all bins, we get the claim.

\textit{Second statement.} For any bin $i \in [n]$, in $T$ steps the normalized load can decrease by at most $T/n$, i.e., $y_i^{s_1} \geq y_i^t - \frac{T}{n}$. So, the overload term is bounded by
\[
e^{\alpha_1 (y_i^t - c_4g)^+} \leq e^{\alpha_1 \frac{T}{n}} \cdot e^{\alpha_1 (y_i^{s_1} - c_4g)^+} \leq e^{\alpha_1 \frac{T}{n}} \cdot V_i^{s_1}.
\]
Similarly, $y_i^{t} \geq y_i^{s_0} - \frac{T}{n}$, and so the underload term is bounded by
\[
e^{\alpha_1 (-y_i^t - c_4g)^+} \leq e^{\alpha_1 \frac{T}{n}} \cdot e^{\alpha_1 (-y_i^{s_0} - c_4g)^+} \leq e^{\alpha_1 \frac{T}{n}} \cdot V_i^{s_0}.
\]
Hence, by aggregating over all bins and using the preconditions $\V^{s_0} \leq \hat{c} n$ and $\V^{s_1} \leq \hat{c} n$,\[
\V^t = \sum_{i = 1}^n \left[ e^{\alpha_1 (y_i^t - c_4g)^+} + e^{\alpha_1 (-y_i^t - c_4g)^+} \right] \leq e^{\alpha_1 \frac{T}{n}} \cdot \sum_{i = 1}^n \left( \V_i^{s_1} + \V_i^{s_0} \right) = e^{\alpha_1 \frac{T}{n}} \cdot \left( \V^{s_1} + \V^{s_0}\right)
\leq e^{\alpha_1 \frac{T}{n}} \cdot 2\hat{c}n. \qedhere
\]
\end{proof}

The following lemma shows that by choosing a large enough offset $c_5 > 0$ in the potential $\Psi_0 := \Psi_0(\alpha_1, c_5g)$ (defined in \cref{eq:g_adv_psi_0_def}), when $\V^t = e^{\Oh(g)} \cdot cn$, then $\Psi_0^t = \Oh(n)$.

{\renewcommand{\thelem}{\ref{lem:g_adv_v_psi0_relation}}
	\begin{lem}[\textbf{Restated, page~\pageref{lem:g_adv_v_psi0_relation}}]
\GAdvVPsiRelation
	\end{lem} }
	\addtocounter{lem}{-1}

\begin{proof}
We start by upper bounding $\Psi_0^t$,
\begin{align*}
\Psi_0^t & =\sum_{i = 1}^n e^{\alpha_1 (y_i^t - c_5g)^+} \\
&= \sum_{i \in [n] \colon y_i^t \geq c_5g} e^{\alpha_1 (y_i^t - c_5g)} + \sum_{i \in [n] \colon y_i^t < c_5g} e^{0} \\
&\leq e^{-\alpha_1 (c_5/2) g}  \sum_{i \in [n] \colon y_i^t \geq c_5g} e^{\alpha_1 (y_i^t - (c_5/2) g)} + n \\
&\stackrel{(a)}{\leq} e^{-\alpha_1 (c_5/2) g}  \sum_{i \in [n] \colon y_i^t \geq c_4g} e^{\alpha_1 (y_i^t - c_4 g)} + n \\
&\leq e^{-\alpha_1 (c_5/2) g}  \sum_{i = 1}^n e^{\alpha_1 (y_i^t - c_4 g)^+} + n \\
& = e^{-\alpha_1 (c_5/2) g}  \cdot \V^t + n,
\end{align*}
where in $(a)$ we used that $c_5/2 \geq c_4$.
Now, using that $\V^t \leq e^{\alpha_1 \cdot \hat{c} \cdot g} \cdot 2e^{2\alpha_1} cn$ and $c_5 \geq 2 \cdot \hat{c}$, we conclude 
\[
\Psi_0^t \leq e^{-\alpha_1 (c_5/2) g} \cdot \V^t + n \leq 2e^{2\alpha_1} \cdot cn + n = Cn. \qedhere
\]
\end{proof}

{\renewcommand{\thelem}{\ref{lem:g_adv_good_v_every_ng_steps}}
	\begin{lem}[\textbf{Restated, page~\pageref{lem:g_adv_good_v_every_ng_steps}}]
 \GAdvGoodVEveryNgRounds
	\end{lem} }
	\addtocounter{lem}{-1}

\label{sec:g_adv_good_v_every_ng_steps_proof}
\begin{proof}
Analogously to the proof of \cref{lem:g_adv_good_gap_after_good_lambda}, we begin by defining the event
\[ \tilde{\mathcal{M}}_{t_0}^{t_1} = \left\{\text{for all }t\in [t_0, t_1]\text{ there exists } s\in [t, t + \tilde{\Delta}_s  ] \text{ such that }\V^s \leq e^{2\alpha_1} cn\right\},\]	
that is, if $\tilde{\mathcal{M}}_{t_0}^{t_1}$ holds then we have that $\V^s \leq e^{2\alpha_1} cn$ at least once every $\tilde{\Delta}_s$ steps and so the claim follows. 

Note that if for some step $j_1$ we have that $\V^{j_1} \leq e^{\alpha_1} cn$ and for some $j_2 \geq j_1$ that $\V^{j_2} > e^{2\alpha_1} cn$, then there must exist $j \in (j_1, j_2)$ such that $\V^j \in (e^{\alpha_1} cn, e^{2\alpha_1} cn]$, since for every $t \geq 0$ it holds that $\V^{t+1} \leq e^{\alpha_1} \cdot \V^t$ (\cref{lem:g_adv_v_smoothness}~$(i)$).
Let $t_0 < \tau_1 <\tau_2<\cdots $ and $t_0 =: s_0<s_1<\cdots $ be two interlaced sequences defined recursively for $i\geq 1$ by \[
\tau_i := \inf\{\tau > s_{i-1}: \V^{\tau} \in (e^{\alpha_1}cn, e^{2\alpha_1}cn] \}\qquad\text{and} \qquad s_i := \inf\{s > \tau_i : \V^s \leq e^{\alpha_1} cn\}. 
\] 
  Thus we have
  \[
  t_0 = s_0 < \tau_1 < s_1 < \tau_2 <s_2 < \cdots, 
  \]
  and since $\tau_i>\tau_{i-1}$ we have $\tau_{t_1 - t_0}\geq t_1 - t_0$. Therefore, if the event $\cap_{i=1}^{t_1 - t_0}\{s_i-\tau_i\leq \tilde{\Delta}_s \} $ holds, then also $ \tilde{\mathcal{M}}_{t_0}^{t_1}$ holds.

Recall that by the strong stabilization (\cref{lem:g_adv_strong_stabilization}) we have for any $i=1,2,\ldots, t_1 - t_0$ and any $\tau = t_0 + 1, \ldots, t_1$, \[
 \Pro{ \left.\bigcup_{t \in [\tau_i,\tau_i + \tilde{\Delta}_s]} \left\{\V^{t} \leq e^{\alpha_1} cn \right\} ~\right|~ \mathcal{Z}^{r_0}, \mathfrak{F}^{\tau} , \; e^{\alpha_1}cn < \V^{\tau} \leq e^{2\alpha_1}cn, \tau_i = \tau } \geq  1 - n^{-11},
  \] and by negating and the definition of $s_i$,
  \[
  \Pro{s_i-\tau_i> \tilde{\Delta}_s \,\, \left| \,\, \mathcal{Z}^{r_0}, \mathfrak{F}^{\tau}, e^{\alpha_1}cn < \V^{\tau} \leq e^{2\alpha_1}cn, \tau_i = \tau \right.} \leq n^{-11}.
  \]
  Since the above bound holds for any $i \geq 1$ and $\mathfrak{F}^{\tau}$, with $\tau_i=\tau$, it follows by the union bound over all $i=1,2,\ldots, t_1 - t_0$, as $t_1 - t_0 \leq 2 n \log^5 n$,
  \[
  \Pro{\left. \neg \tilde{\mathcal{M}}_{t_0}^{t_1} \,\,\right|\,\, \mathcal{Z}^{r_0}, \mathfrak{F}^{t_0}, \V^{t_0} \leq cn }\leq (t_1 - t_0) \cdot n^{-11} \leq n^{-9}. \qedhere \]
\end{proof}

\section{Index of Potential Functions and Constants for Upper Bounds}

In this section we summarize the definitions of the various potential functions (\cref{tab:potential_defs}) and constants (\cref{tab:constant_defs}), used in Sections~\ref{sec:g_adv_warm_up}-\ref{sec:g_adv_layered_induction}.

\begin{table}[H]
    \centering
    \resizebox{\textwidth}{!}{
    \begin{tabular}{|c|c|c|c|}
        \hline
        \textbf{Sym} & \textbf{Potential Function} & \textbf{Reference} & \textbf{Page} \\ \hline
        $\Gamma$ & Hyperbolic cosine potential with $\gamma = \Theta(1/g)$ & \cref{eq:gamma_def} & \pageref{eq:gamma_def} \\ \hline
        $\Lambda$ & Hyperbolic cosine potential with $\alpha = \frac{1}{18}$ and offset $c_4g$ & \cref{eq:lambda_def} & \pageref{eq:lambda_def} \\ \hline
        $\Delta$ & Absolute value potential & \cref{eq:abs_def} & \pageref{eq:abs_def} \\ \hline
        $\Upsilon$ & Quadratic potential & \cref{eq:quad_def} & \pageref{eq:quad_def} \\ \hline
        $\tilde{\Lambda}$ & Adjusted potential for $\Lambda$ & \cref{eq:tilde_lambda} & \pageref{eq:tilde_lambda} \\ \hline
        $\V$ & Same as $\Lambda$ with constant smoothing parameter $\alpha_1 \leq \alpha$ & \cref{eq:g_adv_v_def} & \pageref{eq:g_adv_v_def} \\ \hline
        $\tV$ & Adjusted potential for $\V$ & \cref{eq:g_adv_tilde_v_def} & \pageref{eq:g_adv_tilde_v_def} \\ \hline
        $\Phi_j$ & Super-exponential with $\alpha_2 (\log n) \cdot \max\{g^{j-k}, 1\}$ and offset $z_j$ & \cref{eq:g_adv_phi_0_def,eq:g_adv_phi_j_def} & \pageref{eq:g_adv_phi_0_def} \\ \hline
        $\Psi_j$ & Super-exponential with $\alpha_1 (\log n) \cdot \max\{g^{j-k}, 1\}$ and offset $z_j$ & \cref{eq:g_adv_psi_0_def,eq:g_adv_psi_j_def} & \pageref{eq:g_adv_psi_0_def} \\ \hline
    \end{tabular}
    }
    \caption{Symbols, equation and page references for the various potential functions used (excluding some killed potentials used only in single lemmas).}
\label{tab:potential_defs}
\end{table}

\begin{table}[H]
    \centering
    \resizebox{\textwidth}{!}{
    \begin{tabular}{|c|c|c|c|}
    \hline
    \textbf{Sym} & \textbf{Definition} & \textbf{Usage} & \textbf{Reference} \\ \hline
    $\alpha$ & $\gamma := -\log(1 - \frac{1}{8 \cdot 48})/g$ & Smoothing parameter for $\Gamma$ & \cref{thm:g_adv_warm_up_gap} \\ \hline
    $c_1$ & $c_1 := c' + 4$, where $c'$ is from \cite{PTW15} & $\ex{\Delta\Gamma^{t+1} \mid \mathfrak{F}^t} \leq - \frac{\gamma}{96n} \cdot \Gamma^t + c_1$ & \cref{thm:g_adv_warm_up_gap} \\ \hline
    $c_2$ & $c_2 := \frac{96c_1}{-\log(1 - 1/(8 \cdot 48))}$ & $\ex{\Gamma^t} \leq c_2 ng$ & \cref{thm:g_adv_warm_up_gap} \\ \hline
    $c_3$ & $c_3 := \frac{16}{-\log(1 - 1/(8 \cdot 48))}$ & $\max_{i \in [n]} |y_i^t| \leq c_3 g \log(ng)$ & \cref{eq:g_adv_c3_def} \\ \hline
    $D$ & $D := 365$ & Defines when $\Delta$ is small ($\Delta^t \leq Dng$) & \cref{lem:g_adv_many_good_steps} \\ \hline
    $\alpha$ & $\alpha := \frac{1}{18}$ & Smoothing parameter for $\Lambda$ & \cref{eq:lambda_def} \\ \hline
    $c_4$ & $c_4 := 2D = 730$ & Offset $c_4g$ for $\Lambda$ & \cref{eq:lambda_def} \\ \hline
    $r$ & $r := \frac{6}{6 + \eps}$ & Fraction of good steps & \cref{lem:g_adv_many_good_steps} \\ \hline
    $\eps$ & $\eps := \frac{1}{12}$ & Appears in drop inequalities for $\Lambda$ and $\V$ & \cref{lem:g_adv_good_step_drop} \\ \hline
    $c$ & $c := 12 \cdot 18$ & Defines when $\Lambda$ and $\V$ are large & \cref{lem:g_adv_good_step_drop} \\ \hline
    $c_s$ & $c_s := c_s(\alpha, c_4, 2c) \geq 1$ in \cref{lem:g_adv_bounds_on_quadratic} & $\Upsilon^t \leq c_s \cdot ng^2$ when $\Lambda^t \leq 2cn$ & \cref{lem:g_adv_stabilization} \\ \hline
    $\Delta_s$ & $\Delta_s := \frac{60 c_s}{\alpha \eps r} \cdot n \cdot \max\{\log n, g\}$ & $\Lambda^t \leq cn$ every $\Delta_s$ steps & \cref{lem:g_adv_stabilization} \\ \hline    
    $\kappa$ & $\kappa := \frac{2}{\alpha} + c_4 + \frac{\Delta_s}{n \cdot \max\{\log n, g\}}$ & $\Gap(m) \leq \kappa \cdot (g + \log n)$  & \cref{lem:g_adv_good_gap_after_good_lambda} \\ \hline
    $c_r$ & $c_r := c_r(\alpha, c_4) \geq 1$ in \cref{lem:g_adv_bounds_on_quadratic} & $\Upsilon^t \leq 2c_r \cdot n \cdot (c_3 g \log(ng))^2$ when $\Lambda^t \leq e^{c_3 g \log(ng)}$ & \cref{lem:g_adv_recovery} \\ \hline
    $\Delta_r$ & $\Delta_r := \frac{60 c_3^2 c_r}{\alpha\eps r} \cdot n g \cdot (\log (ng))^2$ & For any $t_0$, $\Lambda^t \leq cn$ for some $t \in [t_0, t_0 + \Delta_r]$ & \cref{lem:g_adv_recovery} \\ \hline
    $\alpha_1$ & $\alpha_1 := \frac{1}{6\kappa}$ & Part of smoothing parameters in $V$ and $\Psi_j$ & \cref{eq:g_adv_alpha_1_def} \\ \hline
    $\tilde{c}_s$ & $\tilde{c}_s := \tilde{c}_s(\alpha_1, c_4, e^{2\alpha_1}c) \geq 1$ in \cref{lem:g_adv_bounds_on_quadratic} & $\Upsilon^t \leq \tilde{c}_s \cdot ng^2$ when $\V^t \leq e^{2\alpha_1}cn$ & \cref{lem:g_adv_strong_stabilization} \\ \hline
    $\tilde{\Delta}_s$ & $\tilde{\Delta}_s := \frac{20 \cdot \tilde{c}_s \cdot \log(2c e^{2\alpha_1})}{\alpha_1 \eps r} \cdot ng$ & $\V^t \leq e^{\alpha_1}n$ every $\tilde{\Delta}_s$ steps & \cref{eq:g_adv_tilde_delta_s_def} \\ \hline
    $c_6$ & $c_6 := \frac{r}{9 \cdot 20 \cdot \tilde{c}_s \cdot \log(2c e^{2\alpha_1})}$ & Strong stabilization works for $g \leq c_6 \log n$ & \cref{eq:g_adv_c6_def} \\ \hline
    $c_5$ & $c_5 := 2 \cdot \max\{c_4, \frac{\tilde{\Delta}_s}{ng} \}$ & Offset $c_5g$ for $\Psi_0$ and $\Phi_0$ & \cref{eq:g_adv_c5_def} \\ \hline
    $C$ & $C := 2e^{2\alpha_1} \cdot c + 1$ & Defines when $\Phi_j$ and $\Psi_j$ are large & \cref{lem:g_adv_v_psi0_relation} \\ \hline
    $\alpha_2$ & $\alpha_2 := \frac{1}{84 \cdot 6\kappa}$ & Part of smoothing parameter in $\Phi_j$ & \cref{eq:g_adv_alpha_2_def} \\ \hline
    $z_j$ & $z_j := c_5 \cdot g + \big\lceil\frac{4}{\alpha_2}\big\rceil \cdot j \cdot g$ & Offsets for $\Phi_j$ and $\Psi_j$ & \cref{eq:g_adv_offsets} \\ \hline
    $k$ & $k \in \mathbb{N}$ s.t.~$(\alpha_1 \log n)^{1/k} \leq g \leq (\frac{\alpha_2}{4} \log n)^{1/(k-1)}$ & Number of layered induction steps & \cref{sec:g_adv_layered_induction} \\ \hline
    \end{tabular}
    }
    \caption{Definition of constants and other variables.}
    \label{tab:constant_defs}
\end{table}

\end{document}